\pdfoutput=1
\documentclass[11pt]{article}

\usepackage[utf8]{inputenc}
\usepackage[T1]{fontenc}
\usepackage{lmodern}
\usepackage{microtype}

\usepackage[letterpaper,top=1.05in,bottom=1.15in,left=1.15in,right=1.15in]{geometry}
\usepackage{amsmath,amssymb,amsthm,mathtools}
\usepackage{bm}
\usepackage{booktabs}
\usepackage{longtable}
\usepackage{array}
\usepackage{tabularx}
\usepackage{multirow}
\usepackage{rotating}
\usepackage{caption}
\usepackage{graphicx}
\usepackage{xcolor}
\usepackage{tikz}
\usepackage{pgfplots}
\pgfplotsset{compat=1.17}
\usetikzlibrary{arrows.meta,positioning,calc,fit,backgrounds,shapes.geometric,
                decorations.pathreplacing,patterns,matrix,chains}

\usepackage[round,authoryear]{natbib}
\usepackage[colorlinks=true,linkcolor=linknavy,citecolor=linknavy,
            urlcolor=linknavy,breaklinks=true]{hyperref}
\usepackage{cleveref}

\definecolor{linknavy}{HTML}{1F3864}
\definecolor{cInk}{HTML}{1A1A1A}
\definecolor{cGold}{HTML}{B8860B}
\definecolor{cGoldDeep}{HTML}{8C6508}
\definecolor{cGoldPale}{HTML}{F5EBD3}
\definecolor{cMute}{HTML}{8A8A8A}
\definecolor{cRule}{HTML}{D9C28A}
\definecolor{cSlate}{HTML}{3C5A74}
\definecolor{cSlatePale}{HTML}{E2E9EF}

\theoremstyle{plain}
\newtheorem{theorem}{Theorem}[section]
\newtheorem{proposition}[theorem]{Proposition}

\theoremstyle{definition}
\newtheorem{definition}[theorem]{Definition}
\newtheorem{example}[theorem]{Example}

\theoremstyle{remark}
\newtheorem{remark}[theorem]{Remark}

\crefname{theorem}{Theorem}{Theorems}
\crefname{proposition}{Proposition}{Propositions}
\crefname{lemma}{Lemma}{Lemmas}
\crefname{corollary}{Corollary}{Corollaries}
\crefname{definition}{Definition}{Definitions}
\crefname{remark}{Remark}{Remarks}
\crefname{example}{Example}{Examples}
\crefname{notation}{Notation}{Notations}
\crefname{assumption}{Assumption}{Assumptions}
\crefname{table}{Table}{Tables}
\crefname{figure}{Figure}{Figures}
\crefname{section}{Section}{Sections}
\crefname{equation}{Equation}{Equations}

\newcommand{\agents}{\mathcal{N}}                
\newcommand{\nag}{n}                             

\newcommand{\St}{\mathcal{S}}                    
\newcommand{\Ac}{\mathcal{A}}                    
\newcommand{\Ob}{\mathcal{O}}                    
\newcommand{\Hi}{\mathcal{H}}                    
\newcommand{\Th}{\Theta}                         
\newcommand{\Simplex}[1]{\Delta(#1)}             

\newcommand{\Tr}{\mathcal{T}}                    
\newcommand{\Om}{\mathcal{Z}}                    
\newcommand{\Rew}{\mathcal{R}}                   
\newcommand{\BellOp}{\mathbb{T}}                 
\newcommand{\Exp}{\mathbb{E}}
\newcommand{\Prb}{\mathbb{P}}
\newcommand{\Real}{\mathbb{R}}
\newcommand{\Nat}{\mathbb{N}}
\newcommand{\Ind}{\mathbf{1}}

\newcommand{\game}{G}
\newcommand{\util}{u}
\newcommand{\brc}{\ensuremath{\mathrm{BR}}}                   
\newcommand{\val}{\ensuremath{\mathrm{val}}}
\newcommand{\NE}{\ensuremath{\mathrm{NE}}}
\newcommand{\CE}{\ensuremath{\mathrm{CE}}}
\newcommand{\CCE}{\ensuremath{\mathrm{CCE}}}
\newcommand{\QRE}{\ensuremath{\mathrm{QRE}}}

\newcommand{\pot}{\Phi}

\DeclareMathOperator*{\argmax}{arg\,max}

\DeclareMathOperator{\supp}{supp}

\DeclareMathOperator{\KL}{KL}

\DeclareMathOperator{\Regret}{Reg}
\newcommand{\norm}[1]{\left\lVert #1 \right\rVert}
\newcommand{\abs}[1]{\left\lvert #1 \right\rvert}

\newcommand{\given}{\,\vert\,}

\newcommand{\PTIME}{\textsc{P}}
\newcommand{\NP}{\textsc{NP}}
\newcommand{\coNP}{\textsc{coNP}}
\newcommand{\PSPACE}{\textsc{PSPACE}}
\newcommand{\NEXP}{\textsc{NEXP}}
\newcommand{\EXP}{\textsc{EXP}}
\newcommand{\PPAD}{\textsc{PPAD}}

\newcolumntype{L}[1]{>{\raggedright\arraybackslash}p{#1}}
\newcolumntype{C}[1]{>{\centering\arraybackslash}p{#1}}

\newenvironment{nomen}[2]{%
  \par\addvspace{0.6\baselineskip}%
  \begingroup\small
  \setlength{\LTpre}{0pt}\setlength{\LTpost}{0pt}%
  \begin{longtable}{@{}L{3.5cm} L{4.0cm} L{5.3cm}@{}}
  \caption{#1}\label{#2}\\[2pt]
  \toprule
  \textsc{Game theory} & \textsc{Reinforcement learning} & \textsc{How to read the pair} \\
  \midrule
  \endfirsthead
  \multicolumn{3}{@{}l}{\small\textsc{(\cref{#2}, continued)}}\\[2pt]
  \toprule
  \textsc{Game theory} & \textsc{Reinforcement learning} & \textsc{How to read the pair} \\
  \midrule
  \endhead
  \midrule
  \multicolumn{3}{@{}r@{}}{\small\textsc{continued on the next page}}\\
  \endfoot
  \bottomrule
  \endlastfoot
}{%
  \end{longtable}%
  \endgroup
  \par\addvspace{0.8\baselineskip}%
}

\newenvironment{nomenx}[4]{%
  \par\addvspace{0.6\baselineskip}%
  \begingroup\small
  \setlength{\LTpre}{0pt}\setlength{\LTpost}{0pt}%
  \begin{longtable}{@{}L{3.5cm} L{4.0cm} L{5.3cm}@{}}
  \caption{#3}\label{#4}\\[2pt]
  \toprule
  \textsc{#1} & \textsc{#2} & \textsc{How to read the pair} \\
  \midrule
  \endfirsthead
  \multicolumn{3}{@{}l}{\small\textsc{(\cref{#4}, continued)}}\\[2pt]
  \toprule
  \textsc{#1} & \textsc{#2} & \textsc{How to read the pair} \\
  \midrule
  \endhead
  \midrule
  \multicolumn{3}{@{}r@{}}{\small\textsc{continued on the next page}}\\
  \endfoot
  \bottomrule
  \endlastfoot
}{%
  \end{longtable}%
  \endgroup
  \par\addvspace{0.8\baselineskip}%
}

\newcommand{\exend}{\hfill\ensuremath{\lozenge}}

\newenvironment{rlreading}[1][The reinforcement learning reading]{%
  \par\addvspace{0.35\baselineskip}%
  \noindent{\color{cGoldDeep}\rule[-0.2ex]{1.6pt}{1.6ex}}\;%
  \textsc{#1.}\enspace\ignorespaces
}{%
  \par\addvspace{0.45\baselineskip}%
}

\newcommand{\bimatrix}[8]{%
  \ensuremath{%
  \begin{array}{r|c|c|}
  \multicolumn{1}{r}{} & \multicolumn{1}{c}{#3} & \multicolumn{1}{c}{#4} \\ \cline{2-3}
  #1 & #5 & #6 \\ \cline{2-3}
  #2 & #7 & #8 \\ \cline{2-3}
  \end{array}}}

\newcommand{\threetwo}[8]{%
  \ensuremath{%
  \begin{array}{r|c|c|}
  \multicolumn{1}{r}{} & \multicolumn{1}{c}{#4} & \multicolumn{1}{c}{#5}
    \\ \cline{2-3}
  #1 & #6 \\ \cline{2-3}
  #2 & #7 \\ \cline{2-3}
  #3 & #8 \\ \cline{2-3}
  \end{array}}}

\newcommand{\trimatrix}[9]{%
  \ensuremath{%
  \begin{array}{r|c|c|c|}
  \multicolumn{1}{r}{} & \multicolumn{1}{c}{#4} & \multicolumn{1}{c}{#5}
    & \multicolumn{1}{c}{#6} \\ \cline{2-4}
  #1 & #7 \\ \cline{2-4}
  #2 & #8 \\ \cline{2-4}
  #3 & #9 \\ \cline{2-4}
  \end{array}}}

\title{\bfseries When Is a Game a Markov Decision Process?\\
Correspondences Between Game Theory\\
and Reinforcement Learning}

\author{Eduardo C. Garrido-Merch\'an\\
DIGNUM Centre, Universidad Pontificia Comillas\\
Madrid, Spain\\
\texttt{ecgarrido@comillas.edu}}

\date{\today}

\begin{document}

\maketitle

\begin{abstract}
Reinforcement learning and game theory describe the same object from two sides.
Game theory specifies what a rational outcome of a strategic interaction is;
reinforcement learning specifies a procedure by which an agent could arrive at
one through experience. The two literatures nevertheless maintain separate
vocabularies for identical mathematical structures, and the textbook treatments
that bridge them tend to present the correspondence as a sequence of analogies
rather than as a set of theorems. This monograph supplies the dictionary. We
take every standard class of game, from the normal-form game through repeated
games, stochastic games, partially observable stochastic games, extensive-form
games with imperfect information, Bayesian games, graphical and congestion
games, games with communication, Stackelberg games, potential and supermodular
games, coalitional games, mean-field games and evolutionary games, and we
identify for each one the decision-theoretic model that reinforcement learning
already studies under a different name. For every pairing we state precisely
whether the correspondence is a definitional identity, an exact reduction
supported by a theorem, or a structural analogy that fails under identifiable
conditions, and we exhibit the degeneracy that recovers the single-agent model
when the number of agents is one. The same treatment is applied to solution
concepts: best response and policy improvement, Nash equilibrium and the fixed
point of the Bellman optimality operator, subgame perfection and the principle
of optimality, maximin value and robust control, correlated equilibrium and
shared-randomness policies, quantal response equilibrium and the soft Bellman
operator of maximum-entropy reinforcement learning, dominance and action
masking, the Shapley value and multi-agent credit assignment, the core and the
stability of a credit scheme, the Vickrey-Clarke-Groves transfer and the
counterfactual baseline, Blackwell approachability and no-regret learning. We
then turn the dictionary into practice, mapping each model class
onto the training procedures that are sound for it, cataloguing the classical
learning dynamics together with what each of them converges to, and recording
the computational complexity of the associated solution problems so that the reader
can see which reductions buy tractability and which merely relabel the
difficulty. A final part is devoted to the correspondences that do not hold: the
non-stationarity that invalidates the independent-learner reduction, the
impossibility results that forbid uncoupled convergence to Nash equilibrium, the
inadequacy of Nash equilibrium as a learning target in general-sum problems, the
gap between the coarse correlated equilibria that no-regret self-play reaches
and the correlated equilibria it is often credited with, the selection of the
risk-dominant rather than the payoff-dominant equilibrium by noisy independent
learning, and the equilibrium-selection gap that separates centralised planning
from decentralised execution in common-payoff games. The intended use is as a
reference. Every section opens with a two-column glossary that places the
game-theoretic and the reinforcement learning vocabulary side by side, every
definition is followed by a worked example with explicit numbers, and every
number printed in an example is recomputed from the model definition by an
accompanying program. Summary tables collect
the correspondences, the solution concepts, the complexity landscape, the
learning dynamics, the evaluation measures and the algorithmic recommendations
in a form that can be consulted without reading the surrounding text.
\end{abstract}


\tableofcontents
\clearpage

\section{Introduction}
\label{sec:intro}

A reinforcement learning researcher who opens a game theory textbook, and a
game theorist who opens a reinforcement learning textbook, each encounter a
large amount of material that they already know, written in a notation they do
not. The partially observable stochastic game of the game theory literature and
the decentralised partially observable Markov decision process of the planning
literature are, up to the restriction of the reward structure, the same object.
The information set of an extensive-form game with perfect recall and the
action-observation history of a partially observable Markov decision process
pick out the same equivalence classes of play. Best response and greedy policy
improvement are the same operation. Subgame perfection and Bellman's principle
of optimality are the same requirement. The Harsanyi transformation of a
Bayesian game into an extensive-form game with an initial chance move, and the
augmentation of a latent-parameter Markov decision process into a partially
observable one, are the same construction applied twice. None of this is
secret, and parts of it are folklore in both communities, but the folklore is
scattered, the statements are usually informal, and the conditions under which
a correspondence is exact rather than suggestive are rarely written down.

This monograph is an attempt to write them down. A game is a Markov decision
process in three distinct senses, and the whole of what follows depends on
keeping them apart. The organising criterion throughout is whether there is a
map between two models that preserves the
quantity one actually cares about, namely the value of a strategy profile and
therefore the set of solutions. Resemblance of form is never sufficient. We will
distinguish three strengths of correspondence and we will label every entry of
every table with one of them. A correspondence is definitional when the two
models are the same tuple with different names, as with a common-payoff
partially observable stochastic game and a decentralised partially observable
Markov decision process. It is an exact reduction when a theorem supplies a
value-preserving and solution-preserving map between models that are not
syntactically identical, as with the equivalence between a decentralised
partially observable Markov decision process and a continuous-state Markov
decision process over occupancy states \citep{dibangoye2016optimally}, or
between a decentralised control problem with partial history sharing and a
single-agent partially observable Markov decision process played by a fictitious
coordinator \citep{nayyar2013decentralized}. It is a structural analogy when the
two models share a form but the map either fails to preserve solutions or
requires assumptions that the general case violates, as with the claim that an
independent learner in a multi-agent environment faces a Markov decision
process, which holds exactly when the other agents are stationary and fails
precisely when they are learning. These three strengths are the three answers to
the question in the title. A game is a Markov decision process by definition in
the first case, by theorem in the second, and by appearance alone in the third.

The distinction matters because the literature is full of statements of the
third kind presented in the voice of the first. Treating the other agents as
part of the environment is the standard justification for applying single-agent
algorithms to multi-agent problems, and it is correct as a statement about a
frozen opponent profile and incorrect as a statement about a training run in
which every agent is simultaneously changing. Likewise, the observation that a
common-payoff game is just a Markov decision process over joint actions is
correct for a centralised planner and silently false for independent learners,
because the reduction preserves the optimal value but destroys the information
about which of several optimal joint actions the agents will each individually
select. The gap between those two readings is the entire subject of coordination
and equilibrium selection, and a dictionary that elides it is worse than no
dictionary at all.

The treatment is intended to be exhaustive over the standard taxonomy of games
and succinct within each entry. We begin in \cref{sec:single-agent} by fixing
the single-agent baseline, the Markov decision process and the partially
observable Markov decision process together with the belief and history
constructions that will be reused throughout. \Cref{sec:normal-form} introduces
the normal-form game as the stateless limit and identifies the one-agent
degeneracy with the multi-armed bandit. \Cref{sec:solution-concepts} develops
the solution concepts as fixed points of operators and shows, for each, the
single-agent object it collapses to, and \cref{sec:dominance} weakens the
rationality assumption in three directions, obtaining dominance and
rationalisability, the bounded-depth models of iterated reasoning, and the
quantal response equilibrium, which is the multi-agent soft Bellman equation.
\Cref{sec:repeated} treats repeated games
and the history-indexed Markov decision process that makes the folk theorem a
statement about equilibrium multiplicity in a large policy class.
\Cref{sec:stochastic} is devoted to stochastic games, the central object of the
correspondence, and to the three reward structures, zero-sum, common-payoff and
general-sum, that determine which single-agent reduction survives.
\Cref{sec:partial-obs} introduces partial observability and the lattice of
models that it generates, and \cref{sec:reductions} collects the exact
reductions that convert these multi-agent models into single-agent ones over
enlarged state spaces, which is where most of the theorem content of this work
is located. \Cref{sec:graphical} treats the structured classes, graphical,
polymatrix and congestion games, whose reinforcement learning counterpart is the
coordination graph. \Cref{sec:extensive} handles extensive-form games, perfect recall
and Kuhn's theorem, and the correspondence between information sets and
action-observation histories. \Cref{sec:bayesian} treats incomplete information,
the Harsanyi transformation and its reinforcement learning counterpart, the
Bayes-adaptive Markov decision process, together with common knowledge and the
hierarchy of interactive beliefs. \Cref{sec:communication} covers the operation
of letting agents choose what others observe, through cheap talk, costly
signalling and information design, and connects it to emergent communication.
\Cref{sec:stackelberg} covers
commitment, bilevel structure and opponent shaping, and \cref{sec:mechanism}
carries the same asymmetry to its limit, in which the designer chooses the game
rather than a move, which is reward design read from the other side and which
brings with it the price of anarchy and two impossibility theorems about
aggregation. \Cref{sec:potential}
connects potential and supermodular games to reward shaping, value factorisation
and the Individual-Global-Max condition, and \cref{sec:coalitions} connects
coalitional game theory to credit
assignment through the Shapley value, the core, the nucleolus and the bargaining
solutions. \Cref{sec:mean-field} and
\cref{sec:evolutionary} treat the two continuum limits, in the population of
agents and in the learning rate respectively, with
\cref{sec:learning-dynamics} between them collecting the classical learning
rules, fictitious play, no-regret dynamics, Blackwell approachability,
stochastic stability and the variational-inequality view that underlies
extragradient and optimistic methods, together with what each converges to.
\Cref{sec:complexity} assembles
the complexity landscape. \Cref{sec:training} is the practical core, mapping
each model class to the training procedures that are sound for it.
\Cref{sec:negative} collects the negative results, and \cref{sec:open} the open
problems.

\begin{figure}[t]
\centering
\begin{tikzpicture}[
  node distance=0pt,
  box/.style={draw=cInk!70, rounded corners=1.5pt, align=center,
              font=\footnotesize, inner sep=4pt, minimum height=7.2mm,
              text width=23mm},
  obs/.style={box, fill=white},
  pobs/.style={box, fill=cSlatePale},
  acc/.style={box, fill=cGoldPale, draw=cGoldDeep},
  ar/.style={-{Latex[length=2mm]}, draw=cMute, thick},
  lbl/.style={font=\scriptsize\itshape, color=cMute}
]

\node[lbl] at (0,1.15) {one agent};
\node[lbl] at (4.6,1.15) {common payoff};
\node[lbl] at (9.5,1.15) {general sum};

\node[obs]  (mdp)   at (0,0)    {MDP};
\node[acc]  (mmdp)  at (4.6,0)  {MMDP\\[-1pt]{\scriptsize(team Markov game)}};
\node[obs]  (sg)    at (9.5,0)  {Stochastic game\\[-1pt]{\scriptsize(Markov game)}};

\node[pobs] (mdp2)  at (0,-1.7) {MDP};
\node[pobs] (decmdp)at (4.6,-1.7){Dec-MDP};
\node[pobs] (posgj) at (9.5,-1.7){Jointly observable\\[-1pt]POSG};

\node[pobs] (pomdp) at (0,-3.4) {POMDP};
\node[acc]  (decpom)at (4.6,-3.4){Dec-POMDP};
\node[obs]  (posg)  at (9.5,-3.4){POSG};

\draw[ar] (mdp)    -- (mmdp);
\draw[ar] (mmdp)   -- (sg);
\draw[ar] (mdp2)   -- (decmdp);
\draw[ar] (decmdp) -- (posgj);
\draw[ar] (pomdp)  -- (decpom);
\draw[ar] (decpom) -- (posg);

\draw[ar] (mdp)    -- (mdp2);
\draw[ar] (mdp2)   -- (pomdp);
\draw[ar] (mmdp)   -- (decmdp);
\draw[ar] (decmdp) -- (decpom);
\draw[ar] (sg)     -- (posgj);
\draw[ar] (posgj)  -- (posg);

\node[lbl, anchor=east, text width=22mm, align=right] at (-1.9,0)    {state fully\\observed};
\node[lbl, anchor=east, text width=22mm, align=right] at (-1.9,-1.7) {state jointly\\determined};
\node[lbl, anchor=east, text width=22mm, align=right] at (-1.9,-3.4) {state partially\\observed};

\end{tikzpicture}
\caption{The lattice of models. Arrows denote inclusion of model classes: the
tail is a special case of the head. Moving right relaxes the reward structure
from a single agent to a shared payoff to arbitrary payoffs; moving down relaxes
the observability of the state. The single-agent column is the subject of
ordinary reinforcement learning; the whole diagram is the subject of this
monograph. The two shaded entries in gold are the classes for which an exact
reduction to a single-agent model over an enlarged state space is available and
is discussed in \cref{sec:reductions}.}
\label{fig:lattice}
\end{figure}
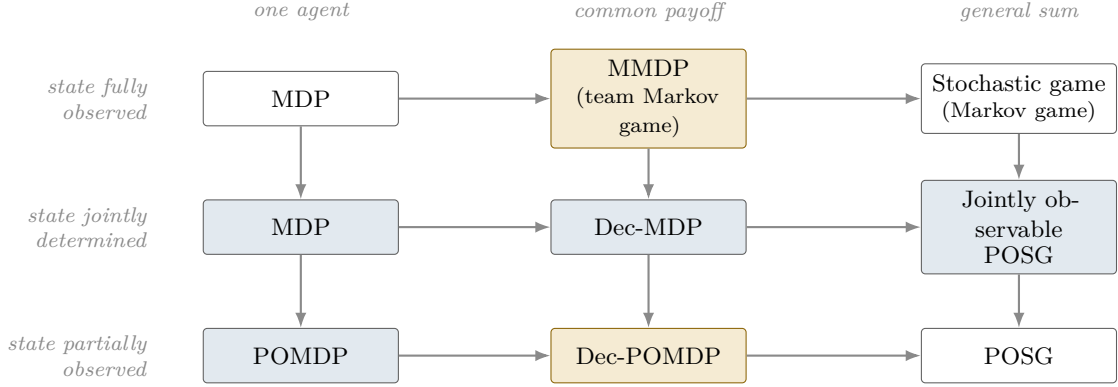

Four conventions are used throughout and are worth stating at the outset. The
first is that we always report the single-agent degeneracy. Every model class we
introduce has a one-agent special case, and in every table we name it, because
the fastest way to check that a proposed correspondence is coherent is to set
$\nag = 1$ and verify that both sides collapse to the same familiar object. The
second is that we distinguish the planning problem, in which the model is known
and a solution is to be computed, from the learning problem, in which the model
is unknown and a solution is to be approximated from sampled interaction. A
reduction that is valid for planning need not be usable for learning, and this
is the most common source of confusion in the area. The occupancy-state
reduction of a decentralised partially observable Markov decision process is
exact, and it requires knowledge of the transition and observation kernels in
order to propagate the occupancy state, so it is a planning tool that becomes a
learning tool only under additional modelling assumptions.

The third convention is the one that gives this work its shape as a reference,
and it is applied without exception, so it is worth stating as a contract.
Every section opens with two paragraphs naming the two readings of its subject,
the one a game theorist brings and the one a reinforcement learning researcher
brings, followed by a three-column table that places the two vocabularies side
by side and glosses each pair in plain words. The tables are numbered and
cross-referenced, so a reader who meets an unfamiliar term can find it without
reading the surrounding argument, and a reader who already knows one column can
read the other as a translation. Within a section the unit of exposition is a
triple. Each definition is followed immediately by a reinforcement learning
reading, set off by a marked run-in heading, which names the corresponding
reinforcement learning object, pairs the symbols one by one, and says what the
definition changes about an implementation; and that is followed by a worked
example with explicit numbers. The triple appears thirty-eight times, once per
definition, and the reader who wants only the dictionary can read the tables and
the readings and skip everything else.

The fourth convention concerns the numbers. Nothing in this subject is difficult
once the objects are small enough to compute by hand, and almost everything in
it is impenetrable when stated only in general form, so the examples carry as
much of the content as the propositions do. Every numerical value printed in an
example, from the equilibrium of a two-by-two matrix to the value of Kuhn poker
and the stationary distribution of a log-linear learning chain, is produced by a
program that recomputes it from the model definition. The program accompanies
this work and checks each printed figure against its recomputation, so the
examples can be audited rather than trusted.

A word is owed about the symbols themselves, because this work translates
concepts between two fields and a reader arriving from either one will meet
notation that is not the notation of its own textbooks. The choice made here is
that the reinforcement learning side is primary. States are $\St$, actions are
$\Ac$, kernels are $\Tr$ and $\Om$, rewards are $\Rew$, policies are $\pi$,
values are $V$ and $Q$, and the discount is $\gamma$, as in
\citet{sutton2018reinforcement}. Game-theoretic objects are grafted onto that
skeleton rather than the reverse, so payoffs are $\util_i$, best responses are
$\brc_i$ and correlated profiles are $\sigma$, and the agent set is the
calligraphic $\agents$. The choice follows from the purpose: the intended reader
is one who already thinks in the first vocabulary and wants the second, and
forcing that reader to relearn $\St$ as a set of histories would defeat the
object.

The cost is borne by the reader coming the other way, and
\cref{tab:osborne-concordance} is the compensation. It places the notation of
\citet{osborne1994course}, the standard graduate reference, beside the notation
used here, object by object. Three differences are structural rather than
cosmetic and are worth naming. First, Osborne and Rubinstein take the preference
relation $\succsim_i$ as primitive and the payoff function as a representation
of it, whereas we write $\util_i$ throughout, which is the special case in which
the preferences are representable and the representation has been fixed. Second,
they build the extensive form on a set of histories $H$, so that a node is a
sequence of actions, whereas we build it on a set of nodes $\mathcal{V}$ and
recover
histories by reading paths; the two are isomorphic and the first is the cleaner
foundation, while the second is nearer to what an implementation stores. Third,
their discount factor is $\delta$ and ours is $\gamma$, which is the single most
frequent source of confusion when the two literatures are read side by side.

\begin{nomenx}{Osborne and Rubinstein}{This work}{Concordance with
\citet{osborne1994course}, the standard graduate reference, whose notation
differs from the one used here. The right-hand column glosses the
pair.}{tab:osborne-concordance}
Player set $N$ & Agent set $\agents$ &
  Finite; $\nag = \abs{\agents}$ in both \\
Action set $A_i$ & $\Ac_i$ &
  Identical object, different font \\
Action profile $a \in A$ & $a \in \Ac$ &
  Identical; both write $a_{-i}$ \\
Preference relation $\succsim_i$ & Payoff $\util_i : \Ac \to \Real$ &
  We assume a representation and fix it \\
Strategic game $\langle N,(A_i),(\succsim_i)\rangle$ &
  $\game = (\agents, \{\Ac_i\}, \{\util_i\})$ &
  The same tuple \\
Mixed strategy $\alpha_i \in \Delta(A_i)$ & $\pi_i \in \Simplex{\Ac_i}$ &
  $\Delta(\cdot)$ is common to both \\
Mixed extension $\langle N,(\Delta(A_i)),(U_i)\rangle$ & $\util_i(\pi)$ &
  Their $U_i$ is our multilinear $\util_i$ \\
Extensive game $\langle N,H,P,f_c,(\mathcal{I}_i),(\succsim_i)\rangle$ &
  Tree on nodes $\mathcal{V}$ with partitions $\mathcal{I}_i$ &
  They index by histories, we by nodes \\
History $h \in H$ & A node of $\mathcal{V}$; trajectory prefix &
  A path from the root \\
Player function $P$ & Which agent acts at a node &
  Same role \\
Chance moves $f_c$ & Kernels $\Tr$, $\Om$ &
  Nature's randomisation \\
Information partition $\mathcal{I}_i$ & $\mathcal{I}_i$ &
  Identical; $I \in \mathcal{I}_i$ is one set \\
Behavioural strategy $\beta_i$ & Policy $\pi_i$ &
  One distribution per information set \\
Discount factor $\delta$ & $\gamma$ &
  The commonest collision between the two \\
Coalitional game $\langle N, v\rangle$ & $(\agents, v)$ &
  They write $v(S)$, we write $v(C)$ \\
Bargaining $\langle X, D, \succsim_1,\succsim_2\rangle$ & $(\mathcal{U}, d)$ &
  We work in utility space directly \\
\end{nomenx}

We assume familiarity with the basic vocabulary of both fields at the level of
\citet{sutton2018reinforcement} for reinforcement learning and
\citet{osborne1994course} or \citet{fudenberg1991game} for game theory. The
reader who wants a broad and gentle entry point to the intersection will find it
in \citet{albrecht2024multi}, whose taxonomy we follow in outline and sharpen in
detail, and in the surveys of \citet{bucsoniu2010multi},
\citet{hernandez2019survey} and \citet{yang2020game}. The present work is
complementary to those: where a survey enumerates algorithms, we enumerate
identities.

\section{The single-agent baseline}
\label{sec:single-agent}

We fix notation for the two single-agent models that every correspondence in
this work degenerates to. Throughout, $\Simplex{X}$ denotes the set of
probability distributions over a finite set $X$, and for a tuple
$x = (x_1,\dots,x_\nag)$ indexed by agents we write $x_{-i}$ for the subtuple
obtained by deleting the $i$-th component.

Before any definitions, a word about the two readings, because this is the
habit the whole monograph asks the reader to acquire. A game theorist looks at
a sequential decision problem and sees a one-player extensive-form game with
chance moves: a tree whose nodes are positions, whose edges are moves, whose
leaves carry payoffs, and whose solution is obtained by backward induction. A
reinforcement learning researcher looks at the same object and sees a Markov
decision process: a state variable that summarises the position, a kernel that
moves it, a reward that accrues along the way, and a solution obtained by
iterating an operator until it stops moving. The two pictures differ in what
they take as primitive. The game theorist takes the tree as primitive and
derives the state as an equivalence class of histories; the reinforcement
learner takes the state as primitive and derives the tree by unrolling. Every
disagreement about notation in this field traces back to that choice, and
knowing which direction one is reading in resolves most of them.

\begin{nomen}{The single-agent vocabulary read from both sides. This table is
the key to \cref{sec:single-agent} and the one on which all later tables
depend.}{tab:nom-single}
Decision node of a one-player tree & State $s \in \St$, or history $h_t$ &
  Where you are and what you may do next \\
Move at a node & Action $a \in \Ac$ & The choice under your control \\
Chance node & Transition kernel $\Tr$, observation kernel $\Om$ &
  The choice not under your control \\
Payoff at a leaf & Discounted return $\sum_t \gamma^t \Rew_t$ &
  What is being maximised \\
Pure strategy & Deterministic policy $\pi : \St \to \Ac$ &
  A complete contingent plan \\
Behavioural strategy & Stochastic policy $\pi : \St \to \Simplex{\Ac}$ &
  A plan that may randomise locally \\
Backward induction & Value iteration, dynamic programming &
  Solve the end first, then work backwards \\
Principle of optimality & Bellman equation &
  Optimal plans are optimal in every subproblem \\
Value of the game at a node & $V^\star(s)$ &
  What the problem is worth from here on \\
Information set & Action-observation history $h_t$, belief $b_t$ &
  What you know when you must choose \\
Nature's move & Latent state, unobserved parameter &
  Uncertainty with a known distribution \\
Security level against nature & Robust value $\max_\pi \min_{p} V^\pi_p$ &
  What you can guarantee with no distribution \\
\end{nomen}

\begin{definition}[Markov decision process]
\label{def:mdp}
A Markov decision process is a tuple
$M = (\St, \Ac, \Tr, \Rew, \gamma, b_0)$ in which $\St$ is a finite set of
states, $\Ac$ a finite set of actions, $\Tr : \St \times \Ac \to \Simplex{\St}$
a transition kernel, $\Rew : \St \times \Ac \times \St \to \Real$ a bounded
reward function, $\gamma \in [0,1)$ a discount factor and
$b_0 \in \Simplex{\St}$ an initial state distribution.
\end{definition}

\begin{rlreading}[The game-theoretic reading]
\Cref{def:mdp} is a one-player extensive-form game with chance moves, in the
sense of \cref{def:efg}, presented in compressed form. Unrolling it gives a tree
whose non-terminal nodes alternate between the single agent and chance: the
agent's nodes are the states $s$, its edges are the actions $a$, the chance
nodes implement $\Tr(\cdot \given s,a)$, and the payoff at a leaf is the
discounted sum of the rewards collected along the path. The state is therefore
not primitive from the game-theoretic side; it is the label of an equivalence
class of histories that a game theorist would write as an information set, and
\cref{prop:infoset-aoh} makes that identification exact. The discount $\gamma$
plays the role that patience plays in \cref{sec:repeated}, and $b_0$ is the
initial chance move. Everything in this monograph is obtained from this tuple by
adding agents, removing observability, or both.
\end{rlreading}

A stationary policy is a map $\pi : \St \to \Simplex{\Ac}$, and its value
function $V^\pi : \St \to \Real$ is the expected discounted return
$V^\pi(s) = \Exp\big[\sum_{t \ge 0} \gamma^t \Rew(s_t,a_t,s_{t+1}) \given s_0 = s\big]$
under the process generated by $\pi$ and $\Tr$. The action-value function is
$Q^\pi(s,a) = \Exp_{s' \sim \Tr(\cdot \given s,a)}[\Rew(s,a,s') + \gamma V^\pi(s')]$.
The entire theory of the single-agent case rests on the observation that the
optimality operator
\begin{align}
(\BellOp V)(s) \;=\; \max_{a \in \Ac} \; \Exp_{s' \sim \Tr(\cdot \given s,a)}
   \Big[\Rew(s,a,s') + \gamma V(s')\Big]
\label{eq:bellman-opt}
\end{align}
is a $\gamma$-contraction in the supremum norm on $\Real^{\St}$, so that by the
Banach fixed point theorem it has a unique fixed point $V^\star$, and that a
policy greedy with respect to $V^\star$ attains it. Two structural facts that we
will invoke repeatedly follow. The first is that the maximisation in
\cref{eq:bellman-opt} decouples across states, which is why dynamic programming
works at all. The second is that the maximum of a linear function over
$\Simplex{\Ac}$ is attained at a vertex, which is why deterministic policies
suffice in the single-agent case. Both facts fail in the multi-agent setting,
the first because the reward of agent $i$ depends on the actions of others and
therefore the per-state problems are coupled through the other agents' policies,
and the second because equilibria of general games require mixing. Nearly every
difficulty catalogued in this monograph is a consequence of one of these two
failures.

The smallest example that exhibits every ingredient is worth carrying through
the rest of the work, so we fix it now and return to it whenever a multi-agent
construction needs a one-agent shadow.

\begin{example}[The slippery corridor]
\label{ex:corridor}
Let $\St = \{L, M, R\}$ be three cells of a corridor, let
$\Ac = \{\mathrm{left}, \mathrm{right}\}$, and let the agent's step succeed with
probability $0.8$ and leave it where it is with probability $0.2$. The cell $R$
is absorbing. A reward of $1$ is collected on the transition that enters $R$ and
nothing is collected anywhere else, and $\gamma = 0.9$. Writing the Bellman
equation at $M$ under the action $\mathrm{right}$,
\begin{align*}
V^\star(M) \;=\; 0.8\,\big(1 + \gamma V^\star(R)\big)
            \;+\; 0.2\,\big(0 + \gamma V^\star(M)\big)
           \;=\; 0.8 + 0.18\, V^\star(M),
\end{align*}
because $V^\star(R) = 0$, which gives $V^\star(M) = 0.8/0.82 = 0.9756$. The same
computation at $L$ gives
$V^\star(L) = 0.8 \cdot 0.9 \cdot 0.9756 + 0.18\, V^\star(L) = 0.8566$, and the
greedy policy moves right everywhere. Three features of this calculation recur
throughout the monograph in disguised form. The equation at $M$ refers to
$V^\star(M)$ on both sides, which is the self-reference that a contraction
resolves; the equation at $L$ refers only to $V^\star(M)$ and $V^\star(L)$, which
is the locality that makes dynamic programming cheap; and the maximisation over
the two actions at $L$ can be performed without knowing what happens at $M$
beyond its value, which is the decoupling across states that
\cref{eq:stage-game} will destroy. \exend
\end{example}

\begin{definition}[Partially observable Markov decision process]
\label{def:pomdp}
A partially observable Markov decision process is a tuple
$P = (\St, \Ac, \Ob, \Tr, \Om, \Rew, \gamma, b_0)$ in which
$(\St,\Ac,\Tr,\Rew,\gamma,b_0)$ is as in \cref{def:mdp}, $\Ob$ is a finite set
of observations and $\Om : \St \times \Ac \to \Simplex{\Ob}$ is an observation
kernel, with $\Om(o \given s', a)$ the probability of observing $o$ after taking
$a$ and arriving at $s'$.
\end{definition}

\begin{rlreading}[The game-theoretic reading]
Adding $\Om$ to \cref{def:mdp} is what a game theorist calls coarsening the
information partition. In the unrolled tree, two nodes that the agent could
previously tell apart now lie in the same information set whenever they share an
action-observation history, so the policy must assign one distribution to the
whole class. The action-observation history $h_t$ is therefore the information
set, the belief $b_t$ is the distribution over nodes within it that a sequential
equilibrium requires in \cref{sec:extensive}, and the requirement that $b_t$ be
the Bayes posterior under the policy being evaluated is the consistency
condition of that solution concept. The one-agent case is the easy case: because
nobody else is conditioning on anything, the belief is a sufficient statistic
and \cref{thm:belief-mdp} applies. \Cref{prop:posg-lattice} records that with
two agents it is not.
\end{rlreading}

Because the state is hidden, a policy must be defined on what the agent has
seen. The action-observation history at time $t$ is
$h_t = (o_0, a_0, o_1, a_1, \dots, a_{t-1}, o_t) \in \Hi_t$, and a general
policy is a map $\pi : \Hi \to \Simplex{\Ac}$ where $\Hi = \bigcup_t \Hi_t$.
Histories grow without bound, and the standard remedy is to replace the history
by a sufficient statistic. The belief state is the posterior over the hidden
state given the history,
\begin{align}
b_t(s) \;=\; \Prb(s_t = s \given h_t),
\qquad
b_{t+1}(s') \;\propto\; \Om(o_{t+1} \given s', a_t)
   \sum_{s \in \St} \Tr(s' \given s, a_t)\, b_t(s),
\label{eq:belief-update}
\end{align}
the second expression being the Bayes filter that propagates it. The belief
state is a sufficient statistic for the history in the sense that
$\Prb(s_t \given h_t) = \Prb(s_t \given b_t)$ and the optimal value is a
function of $b_t$ alone, which yields the first of the reductions we shall be
cataloguing.

\begin{example}[Two ticks of a Bayes filter]
\label{ex:belief}
A sensor reports which of two rooms contains an object. The prior is uniform,
$b_0 = (0.5, 0.5)$, the object does not move, and the sensor is correct with
probability $0.85$. After one reading of ``left'', \cref{eq:belief-update} gives
\begin{align*}
b_1(\mathrm{left}) \;=\;
\frac{0.85 \times 0.5}{0.85 \times 0.5 + 0.15 \times 0.5} \;=\; 0.85 .
\end{align*}
After a second reading of ``left'' the same formula applied to $b_1$ gives
$b_2(\mathrm{left}) = 0.85^2 / (0.85^2 + 0.15^2) = 0.9698$. The arithmetic is
elementary and the structural point is not: the pair $(0.85, 0.15)$ is all the
agent needs to carry forward, so a recurrent network that learns to maintain
two numbers instead of the whole observation history is not a heuristic
compression but an exact one. The multi-agent reductions of
\cref{sec:reductions} are attempts to find the analogous two numbers when there
is more than one agent, and \cref{prop:posg-lattice} will explain why they are
much larger objects. \exend
\end{example}

\begin{theorem}[Belief-MDP reduction; belief sufficiency and the functional
equation, \citealp{astrom1965optimal}; piecewise linearity and convexity,
\citealp{smallwood1973optimal}]
\label{thm:belief-mdp}
Let $P$ be a partially observable Markov decision process with $\St$, $\Ac$ and
$\Ob$ finite and $\Rew$ bounded. Define the belief
Markov decision process $\tilde M = (\Simplex{\St}, \Ac, \tilde\Tr,
\tilde\Rew, \gamma, \delta_{b_0})$ whose states are beliefs, whose transition
kernel is induced by \cref{eq:belief-update} together with the observation
likelihood $\Prb(o \given b,a) = \sum_{s'} \Om(o \given s',a) \sum_s
\Tr(s' \given s,a) b(s)$, and whose reward is
$\tilde\Rew(b,a) = \sum_{s,s'} b(s) \Tr(s' \given s,a) \Rew(s,a,s')$. Then
$\tilde M$ is a Markov decision process with a continuous state space, its
optimal value function $\tilde V^\star$ satisfies
$\tilde V^\star(b_t) = V^\star(h_t)$ for every history $h_t$ with belief $b_t$,
and an optimal stationary policy of $\tilde M$ induces an optimal
history-dependent policy of $P$. Moreover, at every finite horizon
$\tilde V^\star$ is the maximum of a finite set of linear functions of $b$, hence
piecewise linear and convex.
\end{theorem}

\Cref{thm:belief-mdp} is the template for most of this monograph. A model that
appears to lack the Markov property with respect to its natural state variable
is shown to possess it with respect to a statistic of the history, at the cost
of enlarging the state space from a finite set to a simplex. The
multi-agent reductions in \cref{sec:reductions} have exactly this shape, with
the belief replaced by an occupancy state or a public belief state and the
action replaced by a decision rule or a prescription. The price is also the
same: the enlarged problem is a Markov decision process, so dynamic programming
applies in principle, and the state space is a continuum, so dynamic programming
is not immediately executable.

Two further single-agent objects are needed. The first is the robust Markov
decision process, in which the transition kernel is only known to lie in an
uncertainty set $\mathcal{P}(s,a) \subseteq \Simplex{\St}$ and the agent
maximises the worst-case return,
\begin{align}
(\BellOp_{\mathrm{rob}} V)(s) \;=\; \max_{a \in \Ac} \; \min_{p \in \mathcal{P}(s,a)}
   \; \Exp_{s' \sim p}\big[\Rew(s,a,s') + \gamma V(s')\big].
\label{eq:robust-bellman}
\end{align}
Under rectangularity of $\mathcal{P}$, meaning that the uncertainty sets at
distinct state-action pairs are independent, the operator in
\cref{eq:robust-bellman} is again a $\gamma$-contraction and the theory goes
through \citep{iyengar2005robust,nilim2005robust}. We will see in
\cref{sec:stochastic} that \cref{eq:robust-bellman} is literally the Bellman
operator of a two-player zero-sum stochastic game in which the adversary chooses
the transition distribution, which is the cleanest example in this work of two
fields deriving the same operator for unrelated reasons.

\begin{example}[The robust and the nominal choice disagree]
\label{ex:robust-vs-nominal}
An agent chooses between two actions at a state whose successors are worth $10$
and $0$. Action $a$ reaches the good successor with a probability that is only
known to lie in $[0.6, 0.9]$; action $b$ reaches it with probability exactly
$0.7$. Under the midpoint estimate $0.75$ the expected continuation of $a$ is
$7.5$ and that of $b$ is $7.0$, so the nominal optimiser picks $a$; discounted at
$\gamma = 0.9$ these are $6.75$ and $6.3$. Under \cref{eq:robust-bellman} the
adversary sets the probability of $a$ to $0.6$, giving $6.0$ before discounting
and $5.4$ after, against $6.3$ for $b$, so the robust optimiser picks $b$. The
two orderings are opposite, and no amount of data about the average outcome of
$a$ resolves the disagreement, because the two criteria are answering different
questions. The reinforcement learning reading is that a conservative agent and a
risk-neutral agent trained on the same transitions will disagree about the
policy whenever the uncertainty is asymmetric, and \cref{prop:robust-zero-sum}
will identify the conservative agent as the maximiser in a two-player zero-sum
game. \exend
\end{example}

The second is the Bayes-adaptive Markov decision process, in which the model
parameters are themselves unknown and are given a prior. Deferring the details
to \cref{sec:bayesian}, we note here only that the construction is an instance
of \cref{thm:belief-mdp} applied to a partially observable Markov decision
process whose hidden state is the pair consisting of the physical state and the
unknown parameter, and that this is the exact counterpart of the Harsanyi
transformation of a game of incomplete information.

\begin{table}[t]
\centering
\caption{The single-agent vocabulary and the multi-agent terms that generalise
it. The third column names the section in which the generalisation is
developed.}
\label{tab:single-agent-vocab}
\small
\begin{tabular}{@{}L{3.9cm} L{6.5cm} L{3.3cm}@{}}
\toprule
\textsc{Single-agent object} & \textsc{Multi-agent generalisation} & \textsc{Developed in} \\
\midrule
Action $a \in \Ac$ & Joint action $a = (a_1,\dots,a_\nag) \in \prod_i \Ac_i$ & \cref{sec:normal-form} \\
Policy $\pi$ & Policy profile $\pi = (\pi_1,\dots,\pi_\nag)$ & \cref{sec:normal-form} \\
Reward $\Rew$ & Reward profile $(\Rew_1,\dots,\Rew_\nag)$ & \cref{sec:stochastic} \\
Greedy improvement & Best-response correspondence $\brc_i$ & \cref{sec:solution-concepts} \\
Fixed point of $\BellOp$ & Nash equilibrium of the induced game & \cref{sec:solution-concepts} \\
Dominated action & Strictly dominated strategy; rationalisability & \cref{sec:dominance} \\
Softmax policy & Quantal response equilibrium & \cref{sec:dominance} \\
Principle of optimality & Subgame perfection, Markov perfection & \cref{sec:solution-concepts} \\
Deterministic optimal policy & Mixed and behavioural strategies & \cref{sec:extensive} \\
Belief state $b_t$ & Occupancy state, public belief state & \cref{sec:reductions} \\
Robust Bellman operator & Minimax operator of a zero-sum game & \cref{sec:stochastic} \\
Bayes-adaptive MDP & Harsanyi transformation, Bayes-Nash equilibrium & \cref{sec:bayesian} \\
Factored MDP & Graphical game, coordination graph & \cref{sec:graphical} \\
Reward shaping & Potential function, difference rewards & \cref{sec:potential} \\
Reward design & Mechanism design; incentive compatibility & \cref{sec:mechanism} \\
Credit over time (TD) & Credit over agents (Shapley value) & \cref{sec:potential} \\
Return decomposition & Core, nucleolus, bargaining solution & \cref{sec:coalitions} \\
Gradient ascent & Learning dynamics; variational inequality & \cref{sec:learning-dynamics} \\
\bottomrule
\end{tabular}
\end{table}

\section{Normal-form games as the stateless limit}
\label{sec:normal-form}

The normal-form game is what remains of a stochastic game when the state space
is a singleton and the horizon is one. Working out this limit carefully is
worthwhile because every solution concept in the subject is defined first here
and then lifted, and because the lifting is not always the obvious one.

The two readings are easy to state here and are the ones most often confused in
practice. A game theorist sees a payoff tensor: a function from joint actions
to a vector of numbers, one per agent, together with a question about which
joint actions are consistent with rationality. A reinforcement learning
researcher sees a one-step environment: an environment with a single state, a
factored action space, and a vector reward, together with a question about what
a learning rule converges to when every component of the action is chosen by a
separate learner. The payoff tensor and the one-step environment are the same
data structure, and the two questions are different. Keeping them apart is the
discipline this section asks for: the game-theoretic question is about fixed
points of a correspondence, and the reinforcement learning question is about
trajectories of a dynamical system, and they have different answers even when
the underlying tensor is identical.

\begin{nomen}{Normal-form vocabulary read from both sides.}{tab:nom-nfg}
Agent, player $i \in \agents$ & Agent, learner, policy head &
  One decision maker, one reward signal \\
Action set $\Ac_i$ & Action space of agent $i$ &
  The columns of this agent's slice of the tensor \\
Joint action $a \in \Ac$ & Joint action; one environment step &
  One cell of the payoff tensor \\
Payoff $\util_i(a)$ & Reward $\Rew_i$ for a one-step episode &
  The number agent $i$ wants to be large \\
Mixed strategy $\pi_i$ & Stochastic policy over $\Ac_i$ &
  A row of logits after a softmax \\
Strategy profile $\pi$ & Joint policy, product policy &
  What every agent does at once \\
Multilinearity of $\util_i$ & Independence of the action samplers &
  Each agent draws its own action privately \\
Payoff tensor & Reward model of the one-step environment &
  The same array, read by the same indices \\
Common payoff & Shared reward; team reward &
  One scalar for everybody \\
Zero-sum & Adversarial reward $\Rew_2 = -\Rew_1$ &
  One agent's gain is the other's loss \\
General sum & Mixed-motive reward &
  Neither aligned nor opposed \\
Stage game induced at $s$ & One-step game with bootstrapped targets &
  Immediate reward plus discounted value \\
\end{nomen}

\begin{definition}[Normal-form game]
\label{def:nfg}
A normal-form game is a tuple $\game = (\agents, \{\Ac_i\}_{i \in \agents},
\{\util_i\}_{i \in \agents})$ in which $\agents = \{1,\dots,\nag\}$ is a finite
set of agents, $\Ac_i$ is the finite action set of agent $i$, and
$\util_i : \Ac \to \Real$ is the payoff function of agent $i$ defined on the
joint action set $\Ac = \prod_{j \in \agents} \Ac_j$. A mixed strategy of agent
$i$ is $\pi_i \in \Simplex{\Ac_i}$, a profile is $\pi = (\pi_1,\dots,\pi_\nag)$,
and payoffs extend multilinearly,
$\util_i(\pi) = \sum_{a \in \Ac} \big(\prod_j \pi_j(a_j)\big)\, \util_i(a)$.
\end{definition}

\begin{rlreading}
\Cref{def:nfg} is an environment with one state, one step and a factored action
space. The agent set $\agents$ is the set of policy heads, $\Ac_i$ is the action
space of head $i$, the joint action $a$ is what the environment receives in one
call to its step function, and $\util_i$ is the $i$-th component of a vector
reward. The mixed strategy $\pi_i$ is the stochastic policy of head $i$, and the
profile $\pi$ is what a practitioner calls the joint policy. The multilinear
extension $\util_i(\pi)$ is the expected reward of that joint policy, and it is
multilinear for one reason only: each head samples its own action from its own
distribution without looking at the others' samples. That independence is a
property of the sampling code, not of the game, and relaxing it by drawing the
whole joint action from a single distribution over $\Ac$ is what
\cref{def:ce-cce} does. The whole of the difference between a Nash equilibrium
and a correlated equilibrium is the difference between $\nag$ calls to a random
number generator and one.
\end{rlreading}

\begin{example}[Evaluating a joint policy, term by term]
\label{ex:multilinearity}
Take the two-agent game with action sets $\{C, D\}$ and payoffs
\begin{align*}
\bimatrix{C}{D}{C}{D}{3,3}{0,5}{5,0}{1,1},
\end{align*}
which is the prisoner's dilemma and which \cref{ex:four-games} meets again
alongside its three companions. Consider the joint policy in
which the row head cooperates with probability $0.5$ and the column head with
probability $0.25$. The four joint actions have probabilities
$\pi_1(C)\pi_2(C) = 0.125$, $\pi_1(C)\pi_2(D) = 0.375$,
$\pi_1(D)\pi_2(C) = 0.125$ and $\pi_1(D)\pi_2(D) = 0.375$, summing to one.
Weighting the row agent's payoffs by these gives
\begin{align*}
\util_1(\pi) \;=\; 0.125 \cdot 3 + 0.375 \cdot 0 + 0.125 \cdot 5
   + 0.375 \cdot 1 \;=\; 1.375,
\end{align*}
and the same weights against the column agent's payoffs give
$\util_2(\pi) = 0.125 \cdot 3 + 0.375 \cdot 5 + 0.125 \cdot 0 +
0.375 \cdot 1 = 2.625$. Two observations are worth making. The first is that
$\util_1(\pi) = 1.375$ is a number neither agent can read off any single cell of
the matrix; it is an average over an outer product, which is exactly what a
critic trained on sampled joint actions estimates. The second is that the
computation is linear in $\pi_1$ for fixed $\pi_2$, so the best response of
\cref{eq:best-response} maximises a linear function over a simplex and is
attained at a vertex, which is why best responses may always be taken pure even
when equilibria may not. \exend
\end{example}

The multilinearity of $\util_i$ in the profile is the structural fact that makes
the theory work, and it is exactly the independence assumption that each agent
randomises privately. Relaxing it, by allowing the joint distribution over
$\Ac$ to be an arbitrary element of $\Simplex{\Ac}$ rather than a product
measure, produces the correlated equilibrium of \cref{sec:solution-concepts},
and the relaxation is what converts an intractable fixed-point problem into a
linear program.

Three reward structures recur and we name them once.

\begin{definition}[Reward structures]
\label{def:reward-structures}
A normal-form game is common-payoff, also called identical-interest or a team
game, when $\util_i = \util_j$ for all $i,j$. It is constant-sum when
$\sum_{i} \util_i(a) = c$ for all $a \in \Ac$ and some constant $c$, and
zero-sum when $c = 0$; with $\nag = 2$ this means $\util_2 = -\util_1$. It is
general-sum when no restriction is imposed.
\end{definition}

\begin{rlreading}
These are the three shapes the reward signal can take, and the choice among them
is the single most consequential line in an environment specification.
Common payoff means one scalar reward broadcast to every agent, which is the
setting of cooperative benchmarks and the only one in which the word
``optimal'' is unambiguous. Zero-sum with two agents means
$\Rew_2 = -\Rew_1$, so the environment returns one number and negates it for the
opponent, which is the setting of self-play board and card games. General sum
means a reward vector with no relation among its components, which is the
setting of every social dilemma, every market and every system whose agents have
separate owners. Nothing else about a multi-agent problem determines as much:
\cref{tab:sg-structures} is organised entirely by this trichotomy, and the first
of the five questions in \cref{sec:training} asks for it.
\end{rlreading}

The three structures are best met at once, in their smallest instances, because
almost every claim in the rest of this work is first checked against one of
them. We give them here and refer back to them by name throughout.

\begin{example}[The four games that this monograph keeps returning to]
\label{ex:four-games}
Each matrix lists the row agent's payoff first. The prisoner's dilemma, in which
defection strictly dominates cooperation for both agents and the unique
equilibrium $(D,D)$ is worse for both than $(C,C)$,
\begin{align*}
\bimatrix{C}{D}{C}{D}{3,3}{0,5}{5,0}{1,1}
\end{align*}
is the canonical general-sum game and the canonical social dilemma. Matching
pennies, in which the agents want to match and to mismatch respectively,
\begin{align*}
\bimatrix{H}{T}{H}{T}{1,-1}{-1,1}{-1,1}{1,-1}
\end{align*}
is the canonical zero-sum game and the smallest game with no pure equilibrium.
The stag hunt, in which both agents prefer to hunt the stag together but hunting
it alone is the worst outcome,
\begin{align*}
\bimatrix{\mathrm{Stag}}{\mathrm{Hare}}{\mathrm{Stag}}{\mathrm{Hare}}{4,4}{1,3}{3,1}{3,3}
\end{align*}
is the canonical coordination game with a conflict between payoff dominance and
risk dominance, and we return to it in \cref{ex:log-linear}. The game of
chicken, in which each agent prefers to be the one who does not yield,
\begin{align*}
\bimatrix{\mathrm{Swerve}}{\mathrm{Straight}}{\mathrm{Swerve}}{\mathrm{Straight}}{6,6}{2,7}{7,2}{0,0}
\end{align*}
is the canonical anti-coordination game and the one on which correlation pays,
as \cref{ex:chicken-ce} computes. Read as reinforcement learning environments
all four are one-state, one-step, two-agent problems with a two-dimensional
factored action space and a two-component reward; read as games they are the
four qualitatively distinct two-by-two strategic situations. Nothing in the
multi-agent difficulty requires more than this. \exend
\end{example}

The single-agent degeneracy of \cref{def:nfg} is immediate and instructive. With
$\nag = 1$ a normal-form game is a function $\util_1 : \Ac_1 \to \Real$ and the
problem of playing it well is the problem of computing $\argmax_{a}\util_1(a)$.
If $\util_1$ is accessible only through sampling, by playing an action and
observing a noisy payoff, the problem is a multi-armed bandit, and if the sampling is
preceded by an observation drawn from a distribution the problem is a
contextual bandit. The three objects, one-shot optimisation, bandit and
contextual bandit, are therefore the one-agent shadows of the normal-form game,
the repeated normal-form game under bandit feedback, and the repeated
normal-form game with a public state. We record this in
\cref{tab:nfg-degeneracy} together with the corresponding statement for each
reward structure, where the point of interest is that the single-agent
degeneracy is the same object in all three rows: with one agent there is no
distinction between cooperation and competition.

\begin{example}[What the prisoner's dilemma becomes when one agent is deleted]
\label{ex:pd-degeneracy}
Delete the column agent from \cref{ex:four-games} by freezing it at $C$. The row
agent then faces the vector $(\util_1(C,C), \util_1(D,C)) = (3, 5)$, which is a
two-armed bandit whose better arm is $D$. Freeze it at $D$ instead and the
vector is $(0, 1)$, whose better arm is again $D$. The two bandits have the same
argmax, which is precisely what strict dominance means, and this is why the
prisoner's dilemma is solvable without any theory of equilibrium: it is a game
in which both one-agent shadows point the same way for each agent. The stag hunt
is not: freezing the column agent at Stag gives the row agent $(4,3)$ with
argmax Stag, and freezing it at Hare gives $(1,3)$ with argmax Hare. There the
shadow depends on what the other agent does, and that dependence is the whole of
the multi-agent problem. \exend
\end{example}

\begin{table}[t]
\centering
\caption{The normal-form game and its single-agent degeneracies. Feedback
columns indicate what the agent observes after a round.}
\label{tab:nfg-degeneracy}
\small
\begin{tabular}{@{}L{3.4cm} L{3.6cm} L{3.6cm} L{3.4cm}@{}}
\toprule
\textsc{Model} & \textsc{Full payoff feedback} & \textsc{Bandit feedback} & \textsc{With public state} \\
\midrule
$\nag = 1$ & Known objective; one-shot maximisation & Multi-armed bandit & Contextual bandit \\
$\nag \ge 2$, common payoff & Team coordination problem & Cooperative bandit; distributed optimisation & Team Markov game with one state \\
$\nag = 2$, zero-sum & Matrix game; value by linear programming & Adversarial bandit; Exp3 and regret matching & Zero-sum game with public state \\
$\nag \ge 2$, general sum & Normal-form game; equilibrium is \PPAD-complete & No-regret learning converging to \CCE & Stateless general-sum game with context \\
\bottomrule
\end{tabular}
\end{table}

The lifting of a normal-form game into a sequential setting admits two distinct
constructions that are frequently conflated, and distinguishing them is the main
point of this section. The first is the repeated game, treated in
\cref{sec:repeated}, in which the same normal-form game is played at every step
and the agents condition on the history of play. The second is the stage game
induced by a stochastic game at a state, treated in \cref{sec:stochastic}, in
which the payoff entries are the immediate rewards plus the discounted
continuation values. The second construction is the one that
makes dynamic programming work in the multi-agent case, and the continuation
values are what distinguishes it from the first.

\begin{definition}[Induced stage game]
\label{def:stage-game}
Let $\{V_i\}_{i \in \agents}$ be a profile of value functions
$V_i : \St \to \Real$. The stage game induced at state $s$ is the normal-form
game $\game_s[V]$ with agent set $\agents$, action sets $\Ac_i$, and payoffs
\begin{align}
\util^{\,s}_i(a) \;=\; \Exp_{s' \sim \Tr(\cdot \given s,a)}
   \Big[\Rew_i(s,a,s') + \gamma V_i(s')\Big],
\qquad a \in \Ac.
\label{eq:stage-game}
\end{align}
\end{definition}

\begin{rlreading}
\Cref{eq:stage-game} is the bootstrapped target of a centralised critic,
restricted to a single state. A practitioner who trains
$Q_i(s, a_1, \dots, a_\nag)$ and freezes $s$ is holding a payoff tensor in
memory, one entry per joint action, one tensor per agent, and
$\util^{\,s}_i(a)$ is exactly that entry. The entries are not environment
rewards: each is an immediate reward plus a discounted continuation, so the
tensor changes at every sweep of value iteration and the game it defines is a
different game at every iteration. This is the object on which a multi-agent
Bellman backup must act, and the question the rest of this monograph keeps
returning to is what operator should replace the $\max$ of
\cref{eq:bellman-opt} when the thing being backed up is a tensor rather than a
vector.
\end{rlreading}

Comparing \cref{eq:stage-game} with \cref{eq:bellman-opt} isolates the whole
difficulty of the multi-agent extension of dynamic programming. In the
single-agent case the right-hand side of \cref{eq:stage-game} is a vector
indexed by $a$ and the Bellman update takes its maximum. In the multi-agent case
it is a payoff tensor and the Bellman update must take some solution of the game
it defines, a value operator that we write $\mathrm{sol}_i(\game_s[V])$. Which
operator is appropriate, whether it is single-valued, whether it is non-expansive
and whether its fixed point is the equilibrium one wants are the questions that
\cref{sec:stochastic,sec:negative} answer, and the answers are affirmative for
zero-sum games, affirmative for common-payoff games under a coordination
assumption, and negative in general.

\begin{example}[A stage game built from bootstrapped values]
\label{ex:stage-game}
Take two agents at a state $s$ with $\Ac_1 = \Ac_2 = \{0,1\}$, deterministic
transitions that send the joint action $(0,0)$ to a state worth
$(V_1, V_2) = (10, 0)$ and every other joint action to a state worth $(0, 10)$,
immediate rewards all zero, and $\gamma = 0.9$. Then \cref{eq:stage-game} gives
the stage game
\begin{align*}
\bimatrix{0}{1}{0}{1}{9,0}{0,9}{0,9}{0,9},
\end{align*}
in which agent $1$ wants the joint action $(0,0)$ and agent $2$ wants anything
else. The entries are not rewards of the environment; they are rewards plus
bootstrapped continuations, and they change at every sweep of value iteration.
A reinforcement learning practitioner meets this object without naming it every
time a critic conditioned on the joint action is trained, because the
cross-section of that critic at a fixed state is exactly a payoff tensor. The
operator that the Bellman update must apply to this tensor is the choice that
\cref{tab:sg-structures} tabulates. \exend
\end{example}

A last remark concerns the status of randomisation. In a single-agent Markov
decision process the optimal policy may be taken deterministic because the
Bellman optimality operator maximises a linear function over a simplex. In a
normal-form game the analogous statement is false: the matching-pennies game has
no pure Nash equilibrium, and \citet{nash1950equilibrium} guarantees existence
only in mixed strategies. This is the first place at which the single-agent
intuition breaks, and it propagates. Stochastic games require stationary mixed
policies, decentralised partially observable Markov decision processes admit
deterministic optimal joint policies for the finite-horizon planning problem but
not for the independent-learning problem, and extensive-form games require
behavioural strategies whose equivalence to mixed strategies holds only under
perfect recall, which is the content of \cref{thm:kuhn} below.

\begin{example}[Why matching pennies has no deterministic solution]
\label{ex:mp-no-pure}
Enumerate the four joint actions of matching pennies. At $(H,H)$ the column
agent gains by switching to $T$; at $(H,T)$ the row agent gains by switching to
$T$; at $(T,T)$ the column agent gains by switching to $H$; at $(T,H)$ the row
agent gains by switching to $H$. The better-reply graph is a four-cycle with no
sink, so no pure profile is a fixed point of the best-response correspondence.
The unique equilibrium is $\pi_1 = \pi_2 = (\tfrac12, \tfrac12)$ with value $0$,
which is what the linear program of \cref{thm:shapley} returns. The practical
consequence for a reinforcement learning engineer is specific and is usually
discovered the hard way: an algorithm whose policy class is deterministic, such
as a greedy $Q$-learner or a deterministic policy gradient method, cannot
represent the solution of this game, and what it will do instead is cycle
through the four-cycle above for as long as it is trained. A flat learning
curve in such a run is evidence of a limit cycle rather than of convergence,
which is the diagnostic point made again in \cref{sec:evolutionary}. \exend
\end{example}

\section{Solution concepts as fixed points, and what they collapse to}
\label{sec:solution-concepts}

Reinforcement learning has one notion of a solution, the optimal policy, and it
is characterised in two equivalent ways, as a maximiser of the value and as a
fixed point of greedy improvement. Game theory has many notions of a solution,
and almost all of them are obtained by taking one of those two characterisations
and weakening it so that it can accommodate several agents with different
objectives. Organising the solution concepts this way makes the correspondence
with reinforcement learning mechanical rather than mnemonic.

The two readings diverge most sharply here, and the divergence is worth naming
before the definitions. A game theorist regards a solution concept as a
prediction: a set of profiles at which a story about rational agents is
internally consistent, with no claim that anyone will arrive there. A
reinforcement learning researcher regards a solution concept as a target: the
thing an algorithm is supposed to converge to, with the implicit claim that
convergence is the measure of success. The gap between a consistency condition
and a convergence target is the source of most of the confusion in the area, and
it is why \cref{sec:negative} is necessary. A Nash equilibrium is a perfectly
good consistency condition and a bad convergence target; a coarse correlated
equilibrium is a weaker consistency condition and a good convergence target.
Neither statement contradicts the other.

\begin{nomen}{Solution-concept vocabulary read from both sides.}{tab:nom-solconc}
Best response $\brc_i(\pi_{-i})$ & Greedy policy improvement; $\argmax_a Q_i$ &
  Do the best you can, others held fixed \\
Nash equilibrium & Fixed point of simultaneous improvement &
  Nobody gains by moving alone \\
$\varepsilon$-equilibrium & $\varepsilon$-greedy improvement; approximate PI &
  Nobody gains more than $\varepsilon$ alone \\
Mixed strategy & Stochastic policy &
  Randomisation as part of the solution \\
Correlated profile $\sigma \in \Simplex{\Ac}$ &
  Policy conditioned on a shared random seed &
  A joint distribution, not a product \\
Mediator, recommendation & Shared signal, broadcast latent variable &
  A common random bit that both read \\
External regret & Regret of a bandit algorithm &
  Loss against the best fixed action \\
Swap regret & Regret against action-dependent rewrites &
  Loss against the best rewiring of your own play \\
Minmax value & Robust return; worst-case value &
  What you can guarantee unilaterally \\
Pareto frontier & Non-dominated set in multi-objective RL &
  Nobody improves without hurting somebody \\
Utilitarian welfare & Team reward $\sum_i \Rew_i$ &
  Add the rewards and maximise \\
Egalitarian welfare & $\min_i$ return; fairness objective &
  Maximise the worst-off agent \\
Price of anarchy & Gap from independent learners to a planner &
  What decentralisation costs \\
\end{nomen}

The elementary object is the best-response correspondence. For a normal-form
game and an agent $i$,
\begin{align}
\brc_i(\pi_{-i}) \;=\; \argmax_{\pi_i \in \Simplex{\Ac_i}} \; \util_i(\pi_i, \pi_{-i})
\;\subseteq\; \Simplex{\Ac_i},
\label{eq:best-response}
\end{align}
and $\brc(\pi) = \prod_i \brc_i(\pi_{-i})$ is the joint correspondence. With
$\nag = 1$ the right-hand side of \cref{eq:best-response} is the set of
distributions supported on $\argmax_a \util_1(a)$, so best response is greedy
action selection, and in the sequential case it is the policy improvement step
of policy iteration. Every solution concept below is a statement about $\brc$ or
about a relaxation of it.

\begin{definition}[Nash equilibrium]
\label{def:nash}
A profile $\pi^\star$ is a Nash equilibrium when
$\pi^\star_i \in \brc_i(\pi^\star_{-i})$ for every $i \in \agents$, equivalently
when $\pi^\star \in \brc(\pi^\star)$. It is an $\varepsilon$-Nash equilibrium
when $\util_i(\pi^\star) \ge \util_i(\pi_i, \pi^\star_{-i}) - \varepsilon$ for
every $i$ and every $\pi_i$.
\end{definition}

\begin{rlreading}
A Nash equilibrium is a joint policy at which running one more step of policy
improvement on any single agent changes nothing. Spell this out, because it is
the identity on which the rest of the dictionary rests. Freeze the policies
$\pi^\star_{-i}$ of every agent but $i$. By
\cref{prop:fixed-others-mdp} agent $i$ then faces an ordinary Markov decision
process, so it has an ordinary action-value function
$Q_i(a_i) = \util_i(a_i, \pi^\star_{-i})$ and an ordinary greedy policy. The
condition $\pi^\star_i \in \brc_i(\pi^\star_{-i})$ says that
$\pi^\star_i$ is already greedy with respect to that $Q_i$. A Nash equilibrium
is therefore a profile at which $\nag$ simultaneous policy-iteration loops have
all stopped at once, and the symbol dictionary is
$\pi^\star_i \leftrightarrow$ the policy of head $i$,
$\brc_i \leftrightarrow \argmax_{a_i} Q_i$, and
$\pi^\star \in \brc(\pi^\star) \leftrightarrow$ no head's greedy step moves.
The $\varepsilon$ version is the stopping criterion of approximate policy
iteration applied per agent: $\varepsilon$ is the largest improvement any single
head could still make, which is the quantity that \cref{tab:diagnostics} calls
NashConv when summed over agents and exploitability when there are two agents
and the game is zero-sum. What does not carry over from the single-agent case is
uniqueness: a policy-iteration loop on one agent has one stopping point, and
$\nag$ coupled loops may have many, with different returns.
\end{rlreading}

\begin{example}[Finding the equilibria by checking the fixed-point condition]
\label{ex:nash-check}
The definition is a condition to be checked, so check it. In the prisoner's
dilemma, fix the column agent at $C$; the row agent's two values are
$\util_1(C,C) = 3$ and $\util_1(D,C) = 5$, so $\brc_1(C) = \{D\}$. Fix the
column agent at $D$; the values are $0$ and $1$, so $\brc_1(D) = \{D\}$ again.
The row agent's best response is $D$ whatever the column agent does, and by
symmetry the same holds for the column agent, so the joint correspondence
$\brc$ is the constant map sending every profile to $(D,D)$. A constant map has
exactly one fixed point, so $(D,D)$ is the unique Nash equilibrium, with payoffs
$(1,1)$. Notice what the check consisted of: two $\argmax$ computations over a
two-element $Q$-vector, per agent. That is all a Nash equilibrium ever is.

How far is the good outcome from being one? At $(C,C)$ each agent earns $3$ and
the best unilateral deviation earns $5$, so the largest improvement available to
any single agent is $2$. By \cref{def:nash}, $(C,C)$ is an $\varepsilon$-Nash
equilibrium exactly for $\varepsilon \ge 2$, and the number $2$ is the
exploitability that a diagnostic would report for a pair of cooperating
policies.

Now matching pennies. At $(H,H)$ the column agent's values are $-1$ for $H$ and
$+1$ for $T$, so it deviates; the same argument at the other three cells shows
that no pure profile is a fixed point, so the game has no pure equilibrium.
Solving the indifference conditions gives the unique equilibrium
$\pi_1 = \pi_2 = (\tfrac12, \tfrac12)$ with value $0$, which a linear program
confirms. The contrast with the prisoner's dilemma is the one to keep: in the
first game the fixed-point condition is satisfied by a deterministic policy that
greedy $Q$-learning can represent, and in the second it is not, so the same
algorithm run on the two games succeeds in one and cycles forever in the other
for reasons that have nothing to do with its hyperparameters. \exend
\end{example}

A Nash equilibrium is therefore a fixed point of the joint best-response
correspondence, and \citet{nash1950equilibrium} establishes existence in every
finite game by applying Kakutani's theorem to $\brc$, which is non-empty,
convex-valued and upper hemicontinuous. The single-agent degeneracy is exact and
is worth stating as a proposition because it is the hinge of the whole
dictionary.

\begin{proposition}[Nash equilibrium degenerates to Bellman optimality]
\label{prop:nash-degeneracy}
Let $\game$ be a normal-form game with $\nag = 1$. Then $\pi^\star$ is a Nash
equilibrium of $\game$ if and only if $\supp(\pi^\star) \subseteq \argmax_{a}
\util_1(a)$. Let $G$ be a stochastic game with $\nag = 1$, which by
\cref{prop:sg-mdp} is a Markov decision process $M$. Then a stationary profile
$\pi^\star$ is a Markov perfect equilibrium of $G$ if and only if $\pi^\star$ is
an optimal policy of $M$, equivalently if and only if $V^{\pi^\star}$ is the
fixed point of the operator $\BellOp$ in \cref{eq:bellman-opt}.
\end{proposition}

\begin{proof}
For $\nag = 1$ the correspondence $\brc_1$ in \cref{eq:best-response} does not
depend on its argument, since there is no $\pi_{-1}$, so the fixed point
condition $\pi^\star_1 \in \brc_1(\cdot)$ reduces to
$\pi^\star_1 \in \argmax_{\pi_1} \util_1(\pi_1)$, and since $\util_1$ is linear
on the simplex its maximisers are exactly the distributions supported on the
maximisers of $\util_1$ over $\Ac_1$. For the sequential statement, Markov
perfection requires $\pi^\star_1(\cdot \given s)$ to be a best response at every
state $s$ against the continuation value induced by $\pi^\star$ itself, which by
\cref{def:stage-game} is the condition $\supp \pi^\star_1(\cdot \given s)
\subseteq \argmax_a \Exp_{s'}[\Rew(s,a,s') + \gamma V^{\pi^\star}(s')]$ for every
$s$. That is simultaneously the greedy condition and the statement
$V^{\pi^\star} = \BellOp V^{\pi^\star}$, and by uniqueness of the fixed point of
the $\gamma$-contraction $\BellOp$ it holds if and only if
$V^{\pi^\star} = V^\star$.
\end{proof}

\Cref{prop:nash-degeneracy} explains both why Nash equilibrium looks like the
right generalisation of optimality and why it is a poor learning target. It is
the right generalisation in the sense that it is the unique weakening of
``greedy with respect to one's own continuation value'' that is well posed when
the continuation value depends on others. It is a poor learning target because
the properties that make the single-agent fixed point usable, uniqueness and
attainability by a contraction, are both lost: a finite game may have
exponentially many equilibria, they may have different payoffs, and no
best-response-like operator is a contraction in general. We return to this in
\cref{sec:negative}.

\begin{example}[Three equilibria, three payoff vectors, one game]
\label{ex:bos}
The battle of the sexes,
\begin{align*}
\bimatrix{\mathrm{Opera}}{\mathrm{Football}}{\mathrm{Opera}}{\mathrm{Football}}{2,1}{0,0}{0,0}{1,2},
\end{align*}
has two pure equilibria, $(\mathrm{Opera},\mathrm{Opera})$ with payoffs $(2,1)$
and $(\mathrm{Football},\mathrm{Football})$ with payoffs $(1,2)$, and one mixed
equilibrium. To find the mixed one, make each agent indifferent. If the row
agent plays Opera with probability $p$, the column agent's payoffs are $p$ from
Opera and $2(1-p)$ from Football, equal when $p = 2/3$. If the column agent
plays Opera with probability $q$, the row agent's payoffs are $2q$ from Opera
and $1-q$ from Football, equal when $q = 1/3$. At that profile each agent earns
$2/3$, which is strictly worse for both than either pure equilibrium, because
with probability $2/3 \cdot 2/3 + 1/3 \cdot 1/3 = 5/9$ the two miscoordinate and
earn nothing. The reinforcement learning lesson is that ``the algorithm
converged to a Nash equilibrium'' is compatible with a payoff of $0.67$ and with
a payoff of $2.0$ in the same game, so it is not a performance claim. The
lesson is sharpened by \cref{thm:folk}, where the equilibrium set becomes a
continuum. \exend
\end{example}

The relaxations of \cref{def:nash} that reinforcement learning actually attains
are obtained by replacing the product measure $\prod_i \pi_i$ by a joint
distribution. Let $\sigma \in \Simplex{\Ac}$ be a correlated profile, to be read
as the output of a mediator who samples $a \sim \sigma$ and privately recommends
$a_i$ to agent $i$; the construction and the solution concept it supports are
due to \citet{aumann1974subjectivity}.

\begin{definition}[Correlated and coarse correlated equilibrium]
\label{def:ce-cce}
A correlated profile $\sigma \in \Simplex{\Ac}$ is a correlated equilibrium when
for every agent $i$ and every deviation map $\phi_i : \Ac_i \to \Ac_i$,
\begin{align}
\Exp_{a \sim \sigma}\big[\util_i(\phi_i(a_i), a_{-i})\big]
\;\le\; \Exp_{a \sim \sigma}\big[\util_i(a)\big].
\label{eq:ce}
\end{align}
It is a coarse correlated equilibrium when the weaker condition holds in which
$\phi_i$ is restricted to constant maps, that is, when for every $i$ and every
$a_i' \in \Ac_i$,
$\Exp_{a \sim \sigma}[\util_i(a_i', a_{-i})] \le \Exp_{a \sim \sigma}[\util_i(a)]$.
\end{definition}

\begin{rlreading}
Implement $\sigma$ as follows. Before the episode, draw one integer from a
shared random seed that every policy can read; let a lookup table map that
integer to a joint action; give each head only its own component. The head now
conditions on a signal that is correlated with what the others will do, which no
product policy can represent. The correlated-equilibrium condition
\cref{eq:ce} says that each head, after reading its own component and updating
its posterior over the others' components, still prefers to play what it was
told, so obedience is greedy with respect to the conditional $Q$-values. The
coarse condition says only that obeying beats committing in advance to one fixed
action, without reading the signal at all.

The deviation map $\phi_i$ is the object with the least familiar name and the
most familiar content: it is a rewiring of the head's own output, a function
that takes the action the head was about to emit and returns a different one.
The corresponding regret notion is swap regret, and the restriction of $\phi_i$
to constant maps gives external regret, which is the regret that every bandit
and policy-gradient method already measures. So the two solution concepts are
exactly the two regret notions a practitioner can log, and
\cref{ex:cce-strict} shows that logging the wrong one hides a real failure.
\end{rlreading}

The difference between the two conditions is the difference between committing
to a deviation before seeing the recommendation and deciding after seeing it,
and it is exactly the difference between a learner that must choose a fixed
alternative action and one that may condition its alternative on what it was
about to do. We make both the inclusion and its strictness concrete, because in
two-by-two games the two concepts coincide and the textbook examples therefore
hide the distinction.

\begin{example}[Correlation that beats every equilibrium of chicken]
\label{ex:chicken-ce}
In the game of chicken of \cref{ex:four-games} the two pure equilibria give
welfare $7 + 2 = 9$, and the mixed equilibrium, in which each agent swerves with
probability $2/3$, gives each agent $14/3 = 4.667$ and welfare $28/3 = 9.333$.
Now let a mediator draw a joint action from
\begin{align*}
\sigma(\mathrm{Swerve},\mathrm{Swerve}) = \tfrac12, \quad
\sigma(\mathrm{Swerve},\mathrm{Straight}) = \tfrac14, \quad
\sigma(\mathrm{Straight},\mathrm{Swerve}) = \tfrac14,
\end{align*}
and privately tell each agent its own component. Check \cref{eq:ce} for the row
agent. Told to go straight, which happens with probability $1/4$, it knows the
other swerves and earns $7$, while swerving would earn $6$, so it obeys. Told to
swerve, which happens with probability $3/4$, its posterior over the other's
action is $(2/3, 1/3)$ on swerve and straight, so obeying earns
$\tfrac23 \cdot 6 + \tfrac13 \cdot 2 = 14/3$ while going straight earns
$\tfrac23 \cdot 7 + \tfrac13 \cdot 0 = 14/3$, so it is indifferent and obeys.
The same holds by symmetry for the column agent, so $\sigma$ is a correlated
equilibrium, and its welfare is
$\tfrac12 \cdot 12 + \tfrac14 \cdot 9 + \tfrac14 \cdot 9 = 10.5$, strictly
greater than the welfare of every Nash equilibrium of the game. A linear program
over \cref{eq:ce} confirms that $10.5$ is the maximum achievable welfare over
the whole correlated-equilibrium polytope. In reinforcement learning terms the
mediator is a shared random seed broadcast to both policies before they act, and
the example is the precise statement of what that seed buys: a welfare of
$10.5$ instead of $9.333$, obtained with no change to the reward function, no
communication during play, and no departure from unilateral rationality. It is
also the precise statement of what it costs, namely that the policies are no
longer independently executable. \exend
\end{example}

The inclusion $\NE(\game) \subseteq \CE(\game) \subseteq \CCE(\game)$ holds in
every game, and the first inclusion is an equality when restricted to product
measures: a product distribution satisfies \cref{eq:ce} if and only if it is a
Nash equilibrium. Both $\CE$ and $\CCE$ are defined by finitely many linear
inequalities in $\sigma$, hence are non-empty convex polytopes computable by
linear programming, which is the formal reason why the relaxation buys
tractability. The reinforcement learning significance is that these are the
objects that learning dynamics actually converge to. If every agent runs an
algorithm whose external regret
\begin{align}
\Regret_i(T) \;=\; \max_{a_i' \in \Ac_i} \sum_{t=1}^{T}
   \big[\util_i(a_i', a^t_{-i}) - \util_i(a^t)\big]
\label{eq:external-regret}
\end{align}
is sublinear in $T$, then the empirical distribution of joint play
$\bar\sigma_T = \frac{1}{T}\sum_t \delta_{a^t}$ converges to the set of coarse
correlated equilibria, which is immediate from \cref{def:ce-cce} once the
average regret vanishes. If the stronger swap regret is sublinear the limit set
is the set of correlated equilibria; regret matching achieves this directly
\citep{hart2000simple}, and \citet{blum2007external} give a generic reduction
turning any no-external-regret algorithm into a no-swap-regret one.
No such statement holds for Nash equilibrium outside special classes. This is
the single most useful transfer from game theory into multi-agent reinforcement
learning: it tells the practitioner what a self-play run is actually converging
to, and the answer is usually a coarse correlated equilibrium of the empirical
game rather than a Nash equilibrium of anything.

Since the difference between the two relaxations determines what a no-regret run
is entitled to claim, we exhibit it in the smallest game where it can appear.

\begin{example}[A coarse correlated equilibrium that is not correlated]
\label{ex:cce-strict}
In a two-by-two game the two concepts coincide, and the reason is a one-line
computation worth performing. The coarse constraint for the row agent deviating
to its first action is
$\sum_{i,j}\sigma(i,j)[\util_1(1,j) - \util_1(i,j)] \le 0$; the term with $i=1$
vanishes identically, so the constraint reduces to
$\sum_j \sigma(2,j)[\util_1(1,j) - \util_1(2,j)] \le 0$, which is precisely the
conditional constraint \cref{eq:ce} at the row $i = 2$. With only two actions
each conditional constraint is recovered this way, so $\CE = \CCE$ for every
two-by-two game, and a search over four hundred random two-by-two games finds no
separation, as it must not.

Three actions suffice. Take rock-paper-scissors, with the row agent's payoffs
$\util_1(a_1,a_2) \in \{-1,0,1\}$ in the usual pattern, and let $\sigma$ be
uniform on the diagonal, that is, probability $1/3$ on each of
$(R,R)$, $(P,P)$ and $(S,S)$. Both agents earn $0$ under $\sigma$. A constant
deviation to Rock faces the marginal $(\tfrac13,\tfrac13,\tfrac13)$ over the
opponent's actions and earns $\tfrac13(0 - 1 + 1) = 0$, and the same holds for
Paper and Scissors, so no constant deviation gains and $\sigma$ is a coarse
correlated equilibrium. It is not a correlated equilibrium: conditional on the
recommendation Rock the agent knows the opponent plays Rock, so switching to
Paper earns $1$ instead of $0$, a conditional gain of $1$ worth
$\tfrac13 \cdot 1 = \tfrac13$ in expectation. The unique correlated equilibrium
of rock-paper-scissors is the uniform distribution over all nine cells.

The practical reading is exact. Two independent no-external-regret learners
playing rock-paper-scissors may settle into a pattern in which their actions are
perfectly correlated and both earn the equilibrium payoff, and neither will
detect a problem by monitoring its own external regret, because there is none.
A third party who could rewire each agent's action as a function of what it was
about to play would extract a gain of $1/3$ per round. This is why the
evaluation table in \cref{tab:diagnostics} asks for the maximum constant
deviation gain for a coarse correlated target and for the maximum swap gain for
a correlated one: they are different numbers and only the second detects this
failure. \exend
\end{example}

Two further concepts complete the stateless picture. The maximin or security
value of agent $i$ is $\underline{\util}_i = \max_{\pi_i} \min_{\pi_{-i}}
\util_i(\pi_i,\pi_{-i})$, the payoff guaranteed against an arbitrary opponent.
In a two-player zero-sum game the minimax theorem of
\citet{vonneumann1928theorie} gives $\max_{\pi_1}\min_{\pi_2} \util_1 =
\min_{\pi_2}\max_{\pi_1}\util_1$, the common value is the value of the game, and
the set of Nash equilibria coincides with the product of the sets of maximin
strategies, so that equilibrium is interchangeable and payoff-unique. Nothing
in that sentence survives outside the zero-sum case, which is why zero-sum
multi-agent reinforcement learning has clean guarantees and general-sum
multi-agent reinforcement learning does not. The reinforcement learning
counterpart of the maximin value is the robust objective of
\cref{eq:robust-bellman}, and the correspondence is literal rather than
analogical, as \cref{prop:robust-zero-sum} will record.

Pareto optimality and social welfare are the remaining concepts, and they are
the ones that most often express what an engineer actually wants. A profile
$\pi$ is Pareto dominated by $\pi'$ when $\util_i(\pi') \ge \util_i(\pi)$ for
all $i$ with strict inequality for some $i$, and Pareto optimal when it is
dominated by nothing. The utilitarian social welfare is $\sum_i \util_i(\pi)$
and the egalitarian welfare is $\min_i \util_i(\pi)$. With $\nag = 1$ all of
these collapse to maximisation of $\util_1$, which is the degeneracy that makes
single-agent reinforcement learning free of the normative questions that
dominate the multi-agent case. In a common-payoff game Pareto optimality and
welfare maximisation coincide with each other and with the maximisation of the
shared reward, which is why cooperative multi-agent reinforcement learning can
state an unambiguous objective, and the residual difficulty there is not what to
optimise but how to make decentralised agents select the same optimum, which is
equilibrium selection rather than equilibrium computation.

\begin{example}[Four criteria disagree on the prisoner's dilemma]
\label{ex:pd-criteria}
In the prisoner's dilemma of \cref{ex:four-games}, the unique Nash equilibrium
$(D,D)$ has welfare $2$. The welfare maximum is $(C,C)$ with welfare $6$. The
Pareto frontier consists of $(C,C)$, $(C,D)$ and $(D,C)$, with payoff vectors
$(3,3)$, $(0,5)$ and $(5,0)$; the equilibrium is not on it. The maximin value of
each agent is $1$, attained by defecting, because the opponent can always hold
the agent to $1$. The egalitarian optimum is $(C,C)$. So the equilibrium
criterion, the utilitarian criterion, the Pareto criterion and the egalitarian
criterion select different profiles, and only the security criterion agrees with
the equilibrium. A designer who says ``we want the agents to reach equilibrium''
has therefore chosen the welfare $2$ outcome over the welfare $6$ one without
saying so. The same arithmetic recurs in every sequential social dilemma
benchmark, and \cref{sec:repeated} explains the one device, namely history
dependence, that lets the welfare-$6$ outcome become an equilibrium. \exend
\end{example}

\begin{example}[Credit for a shared return, exactly]
\label{ex:shapley-team}
Three agents clear a warehouse. Alone, agents $1$ and $2$ each shift $10$
pallets and agent $3$ shifts $4$. The pairs $\{1,2\}$, $\{1,3\}$ and $\{2,3\}$
shift $26$, $16$ and $16$, and all three together shift $36$. Writing $v$ for
this characteristic function, \cref{eq:shapley} gives
$\phi = (14.67,\, 14.67,\, 6.67)$, which sums to $36$ as efficiency requires.
The computation for agent $3$ averages its marginal contribution over the six
orderings of the three agents: it adds $4$ when it arrives first, $6$ when it
arrives after one of the others, and $10$ when it arrives last, and the
weighted average is $\tfrac13 \cdot 4 + \tfrac13 \cdot 6 + \tfrac13 \cdot 10
= 6.67$. The core of this game is non-empty and contains $\phi$, so the
allocation is also stable in the sense of \cref{def:core}; the nucleolus is
$(14.5, 14.5, 7)$, which differs. In a cooperative multi-agent reinforcement
learning system $v(C)$ is the return obtained when the agents outside $C$ are
replaced by a default policy, $\phi_i$ is the credit assigned to agent $i$, and
the difference between $\phi_3 = 6.67$ and the difference reward
$v(\{1,2,3\}) - v(\{1,2\}) = 10$ is exactly the difference between averaging
over all orderings and evaluating one of them, which \cref{prop:difference-shapley}
makes precise. \exend
\end{example}

\begingroup\small
\begin{longtable}{@{}L{3.0cm} L{4.3cm} L{3.4cm} C{0.7cm}@{}}
\caption{Solution concepts and their reinforcement learning counterparts. The
column $\nag = 1$ gives the single-agent degeneracy, which is in every case the
optimal policy or one of its characterisations. The last column records the
strength of the correspondence: D for definitional identity, E for an exact
reduction supported by a theorem, S for a structural analogy whose failure
conditions are stated in the text.}
\label{tab:solution-concepts}\\[2pt]
\toprule
\textsc{Game-theoretic concept} & \textsc{Reinforcement learning counterpart} & \textsc{Degeneracy at $\nag=1$} & \\
\midrule
\endfirsthead
\multicolumn{4}{@{}l}{\small\textsc{(\cref{tab:solution-concepts}, continued)}}\\[2pt]
\toprule
\textsc{Game-theoretic concept} & \textsc{Reinforcement learning counterpart} & \textsc{Degeneracy at $\nag=1$} & \\
\midrule
\endhead
\midrule
\multicolumn{4}{@{}r@{}}{\small\textsc{continued on the next page}}\\
\endfoot
\bottomrule
\endlastfoot
Best response $\brc_i$ & Greedy policy improvement; argmax over $Q$ & Greedy policy & D \\
Nash equilibrium & Fixed point of simultaneous policy improvement & Optimal policy & D \\
$\varepsilon$-Nash & Approximate policy iteration; $\varepsilon$-greedy improvement & $\varepsilon$-optimal policy & D \\
Markov perfect equilibrium & Stationary optimal policy of each agent given the others & Optimal stationary policy & D \\
Subgame perfect equilibrium & Bellman principle of optimality; backward induction & Optimal policy & D \\
Maximin / security value & Robust objective, worst-case return & $\max_\pi V^\pi$ & E \\
Minimax theorem value & Unique value of the robust Bellman fixed point & $V^\star$ & E \\
Correlated equilibrium & Policy with a shared signal; swap-regret limit & Optimal policy & E \\
Coarse correlated equilibrium & Limit of no-external-regret self-play & Optimal policy & E \\
Strictly dominated strategy & Action provably never greedy; action masking & Suboptimal action & D \\
Rationalisable strategy & Survivor of iterated best-response reasoning & Optimal action & D \\
Quantal response equilibrium & Fixed point of softmax policies; soft Bellman & Softmax policy & E \\
Level-$k$, cognitive hierarchy & Depth-limited opponent modelling & Greedy policy & S \\
Pareto optimality & Non-dominated multi-objective policy & Optimal policy & D \\
Social welfare optimum & Team reward $\sum_i \Rew_i$; centralised objective & Optimal policy & D \\
Nash bargaining solution & Product-of-advantages objective & Optimal policy & E \\
Stackelberg equilibrium & Bilevel optimum; leader with best-responding follower & Optimal policy & E \\
Evolutionarily stable strategy & Asymptotically stable rest point of learning dynamics & Global optimum & S \\
Core of a coalitional game & Stability of a credit-assignment scheme & Trivial & S \\
Nucleolus & Lexicographic min-max credit allocation & Total return & E \\
Shapley value & Per-agent credit; marginal contribution baseline & Total return & E \\
Folk-theorem equilibrium & Non-Markovian policy conditioning on history & Optimal policy & D \\
Trembling-hand perfection & Robustness of the policy to execution noise & Optimal policy & S \\
Sequential equilibrium & On-policy belief filter plus greedy action & Optimal policy & E \\
\end{longtable}
\endgroup

\Cref{tab:solution-concepts} collects these correspondences together with the
ones developed later in the monograph. Three entries deserve comment because
they are the ones most often stated loosely. Subgame perfection and the Bellman
principle are the same requirement, namely that a strategy be optimal in every
subproblem and not merely along the realised path, and the equivalence is
definitional once one observes that a subgame of a single-agent sequential
problem is the problem restarted at a state. Markov perfect equilibrium is the
restriction of subgame perfection to stationary state-dependent strategies and
is therefore the exact multi-agent analogue of restricting attention to
stationary policies in a discounted Markov decision process; the restriction is
without loss in the single-agent case and with loss in the multi-agent case,
because non-Markovian strategies support equilibria that stationary ones do not,
which is the content of the folk theorem in \cref{sec:repeated}. The Shapley
value entry is marked as an exact reduction because the axioms that characterise
it uniquely are exactly the properties one demands of a credit-assignment
scheme, and because the difference reward used in cooperative multi-agent
reinforcement learning is a single term of the Shapley sum, as
\cref{prop:difference-shapley} makes precise.

\section{Dominance, iterated reasoning, and bounded rationality}
\label{sec:dominance}

Nash equilibrium presumes that every agent's belief about the others is exactly
right. Three families of concepts weaken that presumption in different
directions, and all three have reinforcement learning counterparts that are in
daily use and are rarely recognised as such. Dominance and rationalisability
weaken it by asking only what survives common knowledge of rationality, which is
the game-theoretic content of action pruning and action masking. Level-$k$ and
cognitive hierarchy models weaken it by bounding the depth of recursion, which
is the content of depth-limited opponent modelling. Quantal response weakens it
by replacing exact maximisation with a softmax, which is literally the content
of entropy-regularised reinforcement learning. The refinements at the end of the
section go the other way and strengthen Nash equilibrium by requiring robustness
to small mistakes, which is the content of training a policy that must survive
execution noise.

The two readings here are almost the same sentence spoken with different
emphasis. A game theorist asks what a perfectly rational agent could believe,
and treats departures from maximisation as a modelling device for experimental
data. A reinforcement learning researcher asks what a finite computation can
produce, and treats departures from maximisation as either an exploration
mechanism or a regulariser that makes an optimisation problem well behaved. The
mathematics is identical and the motivation is different, which is why the
softmax appears in one literature as a model of human error and in the other as
a device for making a Bellman operator smooth.

\begin{nomen}{Bounded-rationality vocabulary read from both sides.}{tab:nom-dominance}
Strictly dominated strategy & Action that is never greedy for any opponent
  policy & Safe to delete from the action space \\
Weakly dominated strategy & Action never strictly better, sometimes worse &
  Unsafe to delete; order of deletion matters \\
Iterated elimination (IESDS) & Iterative action masking &
  Prune, then prune again given the pruning \\
Rationalisable strategy & Survivor of unbounded best-response recursion &
  What a rational agent could justify playing \\
Common knowledge of rationality & Unbounded recursion depth in modelling &
  I know that you know that I know \\
Level-$k$ agent & Opponent model unrolled $k$ times &
  Best response to a level-$(k{-}1)$ model \\
Cognitive hierarchy & Mixture over modelling depths &
  Best response to a distribution of depths \\
Quantal response, logit choice & Boltzmann or softmax policy &
  Better actions more likely, not certain \\
Precision $\lambda$ & Inverse temperature $1/\tau$ &
  How sharply the agent maximises \\
Quantal response equilibrium & Fixed point of softmax best responses &
  Soft Bellman fixed point with several agents \\
$\varepsilon$-equilibrium & $\varepsilon$-greedy improvement stopping rule &
  Nobody gains more than $\varepsilon$ \\
Trembling-hand perfection & Robustness to action-execution noise &
  Survives a small probability of any action \\
Proper equilibrium & Noise weighted by how costly the mistake is &
  Expensive mistakes are rarer \\
\end{nomen}

\begin{definition}[Dominance]
\label{def:dominance}
In a normal-form game, a strategy $\pi_i \in \Simplex{\Ac_i}$ strictly dominates
$\pi_i'$ when $\util_i(\pi_i, a_{-i}) > \util_i(\pi_i', a_{-i})$ for every
$a_{-i} \in \Ac_{-i}$, and weakly dominates it when the inequality holds weakly
everywhere and strictly somewhere. A pure action is strictly dominated when some
mixed strategy strictly dominates it.
\end{definition}

\begin{rlreading}
A strictly dominated action is one that is not greedy with respect to any
$Q$-function the game can produce, whatever the other agents do, so deleting it
from the action space cannot change the optimal policy of any agent. This is the
formal licence behind action masking, and it is the only licence there is. It is
specific to strict dominance: a weakly dominated action is greedy for some
opponent behaviour, so masking it changes the set of equilibria, and the order
in which weakly dominated actions are deleted changes the answer. One detail
matters in practice and is easy to miss. The dominating object in
\cref{def:dominance} is a mixed strategy, not an action, so an action can be
safely maskable even though no single other action beats it; a mask computed by
comparing $Q$-values pairwise will therefore miss deletions that a linear
program would find, and \cref{ex:mixed-dominance} exhibits the smallest case.
\end{rlreading}

\begin{example}[An action worth deleting that no other action beats]
\label{ex:mixed-dominance}
Give the row agent three actions against the column agent's two, with row
payoffs
\begin{align*}
\threetwo{a_1}{a_2}{a_3}{\text{left}}{\text{right}}%
{4 & 0}{0 & 4}{1 & 1}.
\end{align*}
Compare $a_3$ with $a_1$: it loses in the left column, $1 < 4$, and wins in the
right, $1 > 0$, so $a_1$ does not dominate it. Compare $a_3$ with $a_2$: the
same in reverse. No pure action dominates $a_3$. Now compare it with the mixed
strategy that plays $a_1$ and $a_2$ with probability one half each, whose payoff
vector is $(\tfrac12 \cdot 4 + \tfrac12 \cdot 0,\; \tfrac12 \cdot 0 +
\tfrac12 \cdot 4) = (2, 2)$, which beats $(1,1)$ in both columns. So $a_3$ is
strictly dominated by a mixture and by nothing else, and may be masked.

The practical reading is that a safe action mask is the output of a linear
program, one per action, asking whether some distribution over the remaining
actions beats it in every column; it is not the output of a pairwise comparison
of $Q$-values. The prisoner's dilemma is the easy case where the pairwise test
suffices, since $D$ beats $C$ by $2$ against $C$ and by $1$ against $D$, so
defection dominates in the pure sense. Most games are not the easy case. \exend
\end{example}

\begin{definition}[Iterated elimination and rationalisability;
\citealp{bernheim1984rationalizable,pearce1984rationalizable}]
\label{def:rationalisable}
Let $\Ac_i^0 = \Ac_i$ and let $\Ac_i^{k+1}$ be the set of actions in $\Ac_i^k$
that are not strictly dominated when the opponents are restricted to
$\Ac_{-i}^k$. The set $\Ac_i^\infty = \bigcap_k \Ac_i^k$ is the set of actions
surviving iterated elimination of strictly dominated strategies. An action is
rationalisable when it is a best response to some belief over the opponents'
rationalisable strategies, the belief being allowed to be correlated across
opponents. For two agents the two sets coincide; for more than two they
coincide when correlated beliefs are permitted.
\end{definition}

\begin{rlreading}
Iterated elimination is action masking applied repeatedly, in which each round
masks the actions that are never greedy given that the previous round's masks
are common knowledge. The first round is computable from an agent's own
experience, since an action whose $Q$-value has been lower under every opponent
policy ever observed is a candidate for deletion. No later round is, because
round $k$ needs the fact that the other agents performed round $k-1$, which is
information about their reasoning rather than about their play. The
rationalisable set is the limit of the recursion, so it is what an agent could
justify playing under unbounded mutual reasoning, and the gap between it and the
Nash equilibrium set measures exactly how much equilibrium assumes beyond
rationality. \Cref{ex:iesds} runs the recursion to a unique answer and
\cref{ex:rationalisable} exhibits a game in which it deletes nothing at all.
\end{rlreading}

\begin{example}[Iterated elimination solves a game that dominance alone does not]
\label{ex:iesds}
Consider the three-by-three game in which the row agent's payoff is
$\util_1(i,j) = 4 - i - j$ and the column agent's is $\util_2(i,j) = 2 + j - i$,
with $i, j \in \{0,1,2\}$ indexing the rows and columns,
\begin{align*}
\trimatrix{0}{1}{2}{0}{1}{2}%
{4,2 & 3,3 & 2,4}%
{3,1 & 2,2 & 1,3}%
{2,0 & 1,1 & 0,2}.
\end{align*}
Row $0$ strictly dominates rows $1$ and $2$ for the row agent, since
$\util_1$ decreases in $i$ for every $j$, and column $2$ strictly dominates
columns $0$ and $1$ for the column agent, since $\util_2$ increases in $j$ for
every $i$. Four eliminations leave the single profile $(0,2)$ with payoffs
$(2,4)$, which is therefore the unique Nash equilibrium and is reached without
ever writing down a fixed-point condition. The reinforcement learning reading is
that an agent which observes only its own reward can discover the first round of
elimination from data, because a dominated action has a lower $Q$-value under
every opponent policy it has ever seen, but cannot discover the second round
without reasoning about the other agent's incentives, because the second round
depends on the opponent having performed the first. Independent learners
therefore recover the first level of this hierarchy and no more, which is the
precise sense in which they are level-$1$ reasoners. \exend
\end{example}

\begin{example}[Rationalisability is strictly weaker than equilibrium]
\label{ex:rationalisable}
In matching pennies every action of every agent is rationalisable: Heads is a
best response to the belief that the opponent plays Heads, Tails to the belief
that it plays Tails, and both beliefs are themselves supported on rationalisable
actions. So iterated elimination removes nothing and the rationalisable set is
the whole game, while the Nash equilibrium set is the single profile
$(\tfrac12,\tfrac12) \times (\tfrac12,\tfrac12)$. The gap is informative. An
algorithm that only ever deletes provably bad actions will never solve matching
pennies, because there are no provably bad actions; solving it requires the
additional assumption that the opponent's belief about you is correct, which is
exactly what equilibrium adds to rationality. This is the cleanest statement of
what a self-play assumption buys and of what it presupposes. \exend
\end{example}

Bounding the depth of the recursion in \cref{def:rationalisable} gives the
second family. A level-$0$ agent plays a fixed non-strategic distribution,
usually uniform; a level-$k$ agent best responds to the belief that all others
are level $k-1$ \citep{stahl1995players}; and a cognitive-hierarchy agent best
responds to a distribution over lower levels, usually Poisson
\citep{camerer2004cognitive}. The reinforcement learning counterpart is exact
and is implemented whenever an agent conditions its policy on a learned model of
the opponent that is itself trained against a lower-depth model, which is the
structure of recursive reasoning methods \citep{wen2019probabilistic} and of the
finitely nested interactive partially observable Markov decision process of
\cref{sec:bayesian}.

\begin{example}[Depth of reasoning in the guessing game]
\label{ex:beauty}
Every agent picks a number in $[0,100]$ and whoever is closest to two thirds of
the average wins. A level-$0$ agent picks uniformly, so its expected choice is
$50$. A level-$1$ agent best responds by picking $\tfrac23 \cdot 50 = 33.33$. A
level-$2$ agent picks $\tfrac23 \cdot 33.33 = 22.22$, a level-$3$ agent
$14.81$, a level-$4$ agent $9.88$. The unique Nash equilibrium, obtained by
iterated elimination of dominated strategies carried to the limit, is that
everybody picks $0$. The sequence makes the point that equilibrium is the
$k \to \infty$ limit of a recursion whose intermediate terms are all different
and all individually rational given the corresponding belief. For multi-agent
reinforcement learning the practical reading is that the depth of an opponent
model is a hyperparameter with a visible effect on behaviour, that training
against a frozen opponent produces a level-$1$ policy by construction, that
training against a policy that was itself trained against a frozen opponent
produces a level-$2$ policy, and that self-play with a replay buffer of past
opponents produces something close to a cognitive hierarchy whose depth
distribution is set by the sampling scheme over the buffer. Reporting the depth
is therefore reporting a property of the algorithm, not of the game. \exend
\end{example}

The third family replaces exact maximisation by a smooth approximation, and
here the correspondence with reinforcement learning is an identity rather than
an analogy.

\begin{definition}[Quantal response equilibrium;
\citealp{mckelvey1995quantal}]
\label{def:qre}
Fix a precision $\lambda \ge 0$. A profile $\pi^\lambda$ is a logit quantal
response equilibrium when for every agent $i$ and every $a_i \in \Ac_i$,
\begin{align}
\pi^\lambda_i(a_i) \;=\;
\frac{\exp\!\big(\lambda\, \util_i(a_i, \pi^\lambda_{-i})\big)}
     {\sum_{a_i' \in \Ac_i} \exp\!\big(\lambda\, \util_i(a_i', \pi^\lambda_{-i})\big)} .
\label{eq:qre}
\end{align}
At $\lambda = 0$ the unique solution is uniform play, and as
$\lambda \to \infty$ every limit point of $\pi^\lambda$ is a Nash equilibrium.
\end{definition}

\begin{rlreading}
Replace $\lambda$ by $1/\tau$ and \cref{eq:qre} is the line of code that turns a
$Q$-vector into a Boltzmann policy, $\pi_i(a_i) \propto \exp(Q_i(a_i)/\tau)$,
with $Q_i(a_i) = \util_i(a_i, \pi_{-i})$ the value of action $a_i$ against the
current opponents. The only thing \cref{def:qre} adds to that line is the word
``equilibrium'': it requires the $Q$-vector that produced the softmax to be the
one the softmax itself induces. So a quantal response equilibrium is the
fixed point of a self-play loop in which every agent uses a softmax policy at a
fixed temperature and nobody ever anneals, which is the loop that a great many
multi-agent implementations actually run. The precision $\lambda$ is the inverse
temperature, $\lambda = 0$ is uniform exploration, and $\lambda \to \infty$ is
the greedy limit; \cref{prop:qre-softq} shows that at one agent the whole
construction collapses to the soft Bellman operator of maximum-entropy
reinforcement learning. The practical corollary is blunt: a run at a fixed
temperature is not converging to a Nash equilibrium of the game, it is
converging to a Nash equilibrium of a different game whose payoffs carry an
entropy bonus, and \cref{ex:soft-value} prices that bonus.
\end{rlreading}

\begin{example}[One logit response, computed]
\label{ex:logit-response}
Before any fixed point, the response alone. In matching pennies the row agent's
payoffs are $+1$ for matching and $-1$ otherwise. Suppose the column agent plays
Heads with probability $0.7$. Then
$\util_1(H, \pi_2) = 0.7(1) + 0.3(-1) = 0.4$ and
$\util_1(T, \pi_2) = 0.7(-1) + 0.3(1) = -0.4$, so the $Q$-vector is
$(0.4, -0.4)$ and the logit response puts weight
\begin{align*}
\pi_1(H) \;=\; \frac{e^{0.4\lambda}}{e^{0.4\lambda} + e^{-0.4\lambda}}
            \;=\; \frac{1}{1 + e^{-0.8\lambda}}
\end{align*}
on Heads. At $\lambda = 0$ this is $0.5$, at $\lambda = 1$ it is $0.690$, at
$\lambda = 5$ it is $0.982$, and in the limit it is $1$, the greedy action.
Against a uniform column agent both payoffs are $0$ and the response is uniform
at every precision, which is why the equilibrium of this particular game is
unaffected by the temperature while the path to it, computed in
\cref{ex:qre-path} for an asymmetric variant, is not. The arithmetic is the
arithmetic of a softmax layer over two logits, and that is the whole of the
correspondence. \exend
\end{example}

\begin{proposition}[Quantal response is the multi-agent soft Bellman equation]
\label{prop:qre-softq}
Let $\nag = 1$ and write $\tau = 1/\lambda$. Then \cref{eq:qre} reduces to the
softmax policy $\pi(a) \propto \exp(\util_1(a)/\tau)$, which is the maximiser of
the entropy-regularised objective
$\max_{\pi \in \Simplex{\Ac}} \; \Exp_{a \sim \pi}[\util_1(a)]
 + \tau\, H(\pi)$, with optimal value the log-sum-exp
$\tau \log \sum_a \exp(\util_1(a)/\tau)$. Consequently the sequential version
of \cref{eq:qre}, obtained by applying it to the stage game of
\cref{def:stage-game}, degenerates at $\nag = 1$ to the soft Bellman operator
$(\BellOp_\tau V)(s) = \tau \log \sum_a \exp\big(
\Exp_{s'}[\Rew(s,a,s') + \gamma V(s')]/\tau \big)$,
which is the operator of maximum-entropy reinforcement learning
\citep{haarnoja2017reinforcement} and the regularised Markov decision process of
\citet{geist2019theory}.
\end{proposition}

\begin{proof}
With one agent the right-hand side of \cref{eq:qre} does not depend on
$\pi_{-i}$, so the fixed-point condition is an explicit formula. The variational
identity
$\max_{\pi \in \Simplex{\Ac}} \{ \sum_a \pi(a) u(a) - \tau \sum_a \pi(a)\log\pi(a) \}
 = \tau \log \sum_a e^{u(a)/\tau}$
with maximiser $\pi(a) \propto e^{u(a)/\tau}$ is the Donsker-Varadhan
variational formula for the log-partition function, obtained by setting the
gradient of the Lagrangian to zero. Substituting $u(a) = \util^{\,s}_1(a)$ from
\cref{eq:stage-game} gives the stated operator.
\end{proof}

\Cref{prop:qre-softq} explains why entropy regularisation stabilises multi-agent
training for a reason that has nothing to do with exploration. The quantal
response operator is a contraction for $\lambda$ small enough and has a unique
fixed point, whereas the Nash operator is neither, so adding entropy is a
convexification of an ill-posed fixed-point problem. This is the mechanism
behind regularised Nash dynamics \citep{perolat2021poincare} and behind the
unified treatment of \citet{sokota2023unified}, and \cref{sec:evolutionary}
exhibits the same term as the mutation component of the replicator-mutator
dynamics.

\begin{example}[Tracking the quantal response path to the equilibrium]
\label{ex:qre-path}
Take the zero-sum game with the row agent's payoff matrix
$\big(\begin{smallmatrix} 2 & -1 \\ -1 & 1 \end{smallmatrix}\big)$. Its unique
Nash equilibrium has both agents playing their first action with probability
$0.4$ and has value $0.2$: the row agent is indifferent when
$3q - 1 = 1 - 2q$, that is $q = 0.4$, and symmetrically for the column agent.
Solving \cref{eq:qre} numerically along a path in $\lambda$ gives the first
action probabilities $(0.500, 0.500)$ at $\lambda = 0$, $(0.489, 0.391)$ at
$\lambda = 1$, $(0.418, 0.387)$ at $\lambda = 5$, $(0.404, 0.396)$ at
$\lambda = 20$ and $(0.401, 0.399)$ at $\lambda = 100$, converging to
$(0.4, 0.4)$. The path is the exact analogue of the regularisation path of an
entropy-regularised policy optimiser as the temperature is annealed, and the
fact that it is a path rather than a point is what makes annealing work: at each
$\lambda$ the solution is unique and locally stable, so a continuation method
tracks it, whereas solving at $\lambda = \infty$ directly means solving the
original fixed-point problem with none of its difficulties removed. \exend
\end{example}

\begin{example}[How much the softmax costs, in value]
\label{ex:soft-value}
A state offers three actions with $Q$-values $(1,\, 0.5,\, 0)$. The hard maximum
is $1$. The soft maximum $\tau \log \sum_a e^{Q(a)/\tau}$ is $1.680$ at
$\tau = 1$, with the best action taken with probability $0.506$, and $1.0007$ at
$\tau = 0.1$, with the best action taken with probability $0.993$. The soft
value always exceeds the hard value, by at most $\tau \log \abs{\Ac}$, which is
$\tau \log 3 = 1.0986\tau$ here. A practitioner who anneals $\tau$ from $1$ to
$0.1$ therefore watches a reported value drop by roughly $0.68$ without any
deterioration in the policy, because the drop is the removal of the entropy
bonus rather than a loss of return. In the multi-agent setting the same bonus
is what separates the quantal response equilibrium from the Nash equilibrium,
so the same accounting applies to exploitability measurements, which should be
computed against the unregularised game. \exend
\end{example}

The refinements go in the opposite direction. A Nash equilibrium may rely on an
agent choosing an action that is optimal only because an event of probability
zero is assumed never to occur, and the refinements rule this out by requiring
optimality to survive small perturbations.

\begin{definition}[Trembling-hand perfection and properness;
\citealp{selten1975reexamination,myerson1978refinements}]
\label{def:trembling}
A profile $\pi^\star$ is trembling-hand perfect when there is a sequence of
completely mixed profiles $\pi^m \to \pi^\star$ such that $\pi_i^\star$ is a
best response to $\pi^m_{-i}$ for every $i$ and every $m$. It is proper when in
addition the perturbations can be taken so that a costlier mistake receives a
probability at most $\varepsilon$ times that of a cheaper one.
\end{definition}

\begin{rlreading}
The sequence $\pi^m$ is an exploration schedule being annealed to zero, and
trembling-hand perfection is the requirement that the profile survive the whole
schedule rather than only its limit. An $\varepsilon$-greedy learner, or a
policy deployed on hardware that occasionally executes the wrong action, is
living at a strictly positive $m$ and will therefore be pushed off any
equilibrium that this refinement rejects. Properness is the same requirement
with the trembles weighted by how costly the mistake is, which is what a
Boltzmann schedule does and what uniform $\varepsilon$-greedy exploration does
not, so the two standard exploration mechanisms target two different refinements
of the same solution concept. \Cref{ex:trembling} exhibits an equilibrium that a
tremble of $0.01$ destroys.
\end{rlreading}

\begin{example}[An equilibrium that a single tremble destroys]
\label{ex:trembling}
Consider the common-payoff game
\begin{align*}
\bimatrix{a}{b}{c}{d}{1,1}{0,0}{1,1}{1,1}.
\end{align*}
The game has three pure Nash equilibria, $(a,c)$, $(b,c)$ and $(b,d)$, all with
payoff $1$, and only one of them survives perturbation. Take $(a,c)$ first. If
the column agent plays $d$ with any probability $\varepsilon > 0$, the row
agent's payoff from $a$ is $1 - \varepsilon$ while its payoff from $b$ is $1$,
so $a$ is no longer a best response and $(a,c)$ is not trembling-hand perfect;
at $\varepsilon = 0.01$ the two values are $0.99$ and $1.00$. Now take $(b,d)$.
If the row agent plays $a$ with probability $\varepsilon$, the column agent's
payoff from $d$ is $1 - \varepsilon$ while its payoff from $c$ is $1$, so $d$ is
no longer a best response either, by the mirror-image argument and with the
mirror-image numbers. The profile $(b,c)$ survives every tremble, because $b$
pays $1$ whatever the column agent does and $c$ pays $1$ whatever the row agent
does. The reinforcement learning reading is direct: $a$ is weakly dominated for the
row agent and $d$ is weakly dominated for the column agent, and an agent trained
with any exploration noise at all, or deployed on hardware with a non-zero
probability of executing the wrong action, will be pushed off both and onto the
undominated pair. Trembling-hand perfection is thus the
equilibrium concept that an $\varepsilon$-greedy learner implicitly targets, and
the refinement is not a mathematical nicety but a description of what noisy
execution does to a policy. Properness refines this further by making the
probability of a tremble decrease in the cost of the mistake, which is what a
Boltzmann exploration schedule does and what uniform $\varepsilon$-greedy
exploration does not. \exend
\end{example}

The last item in this section is the one that a practitioner uses without
naming. An $\varepsilon$-Nash equilibrium, in which no agent can gain more than
$\varepsilon$ by deviating, is the stopping criterion of every approximate
method in the area, and its single-agent degeneracy is the $\varepsilon$-optimal
policy returned by approximate policy iteration. The relation between
$\varepsilon$-equilibria and exact ones is however badly behaved in a way that
has no single-agent analogue: an $\varepsilon$-equilibrium may be far in
strategy space from every exact equilibrium, and its payoff may differ from
every equilibrium payoff by much more than $\varepsilon$. This is why
\cref{tab:diagnostics} asks for exploitability rather than for a distance in
policy space, and why the complexity statements of \cref{sec:complexity}
distinguish an approximation in the deviation gain from an approximation in the
strategy.

\section{Repeated games and the history-indexed Markov decision process}
\label{sec:repeated}

A repeated game is a normal-form game $\game$ played at every step
$t = 0,1,2,\dots$, with each agent observing some function of the realised joint
action before the next round. The monitoring structure is what determines the
difficulty. Under perfect monitoring each agent observes the full joint action
$a^t$; under imperfect public monitoring the agents observe a common signal
$y^t$ drawn from a distribution depending on $a^t$; under private monitoring
each agent observes its own signal. A strategy of agent $i$ is a map
$\pi_i : \Hi_i \to \Simplex{\Ac_i}$ from the agent's observation history to
distributions over its actions, and the payoff is the discounted sum
$\util_i(\pi) = (1-\gamma)\Exp[\sum_t \gamma^t \util_i(a^t)]$, the normalisation
being conventional so that payoffs live in the convex hull of the stage payoffs.

The two readings are instructive because they attach the same object to
different hyperparameters. A game theorist reads the discount factor as patience
and the strategy space as the set of all history-contingent plans, and asks
which payoff vectors can be sustained. A reinforcement learning researcher reads
the discount factor as a parameter of the return and the strategy space as
whatever the architecture can represent, and asks what the algorithm converges
to. The folk theorem is the statement that these two questions have the same
answer and that the answer is almost everything, so that in a repeated game the
architecture and the discount factor determine the outcome at least as much as
the learning rule does. We give the theorem and then insist on this reading,
because it is the one with practical consequences for benchmark design.

\begin{nomen}{Repeated-game vocabulary read from both sides.}{tab:nom-repeated}
Stage game $\game$ & One-step environment; inner reward &
  The game played every period \\
History $h^t$ & Trajectory prefix; recurrent hidden state &
  Everything observed so far \\
Strategy $\pi_i : \Hi_i \to \Simplex{\Ac_i}$ &
  History-dependent policy; recurrent policy &
  A plan contingent on the past \\
Discount factor $\gamma$ & Discount factor $\gamma$ &
  Patience; also the effective horizon \\
Perfect monitoring & Full observability of the joint action &
  You see what everybody did \\
Imperfect public monitoring & Shared noisy observation of the joint action &
  You see a common signal, not the action \\
Private monitoring & Per-agent noisy observation &
  Each agent sees a different signal \\
Minmax value $\underline{\util}_i$ & Worst-case return under adversarial partners &
  The punishment that others can impose \\
Individually rational payoff & Return above the worst-case guarantee &
  Better than being punished forever \\
Grim trigger, tit-for-tat & Finite-state policy on the opponent's last action &
  A small recurrent controller \\
Folk theorem & Expressiveness of the policy class sets the outcome set &
  Memory creates equilibria \\
Sequential social dilemma & Temporally extended mixed-motive environment &
  The dilemma spread over time \\
\end{nomen}

The translation into reinforcement learning vocabulary is immediate and exact.

\begin{proposition}[Repeated game as a stochastic game on histories]
\label{prop:repeated-sg}
A repeated game with perfect monitoring is a stochastic game whose state at time
$t$ is the history $h^t = (a^0,\dots,a^{t-1})$, whose transition is the
deterministic append $h^{t+1} = h^t a^t$, and whose reward is
$\Rew_i(h^t, a^t) = \util_i(a^t)$. Under imperfect public monitoring the state
is the public history $(y^0,\dots,y^{t-1})$ and the game is a partially
observable stochastic game in which the hidden component of the state is empty
but the agents' private action records differ. With $\nag = 1$ the former is a
Markov decision process on the history tree whose optimal policy is the constant
policy $a^t \equiv \argmax_a \util_1(a)$.
\end{proposition}

The statement is a tautology, and its value lies in what it exposes. The state
space of the induced stochastic game grows exponentially in $t$, so the
correspondence does not by itself yield an algorithm; what it yields is the
correct diagnosis of why repeated games are not solved by stationary policies.
In the induced stochastic game, a stationary policy is one that ignores the
history, and the set of Markov perfect equilibria in stationary policies is
exactly the set of profiles that play a Nash equilibrium of the stage game at
every step. Conditioning on history enlarges the equilibrium set enormously,
which is the folk theorem.

\begin{theorem}[Folk theorem, discounted case; \citealp{fudenberg1986folk}]
\label{thm:folk}
Let $\game$ be a finite normal-form game with perfect monitoring and a public
randomisation device, let $\underline{\util}_i$ be the minmax value of agent $i$
in independent mixed strategies of the opponents, and
let $F^\star = \{v \in \Real^{\nag} : v \in \mathrm{conv}\,\util(\Ac),\;
v_i > \underline{\util}_i \ \forall i\}$ be the set of strictly individually
rational feasible payoffs. If $F^\star$ has dimension $\nag$, then for every
$v \in F^\star$ there exists $\bar\gamma < 1$ such that for all
$\gamma \in (\bar\gamma, 1)$ the infinitely repeated game has a subgame perfect
equilibrium with payoff profile $v$. Without public randomisation the conclusion
weakens to a payoff within an arbitrarily small $\varepsilon$ of $v$.
\end{theorem}

\begin{rlreading}
\Cref{thm:folk} is a statement about the policy class, and it is the reason two
runs of the same algorithm on the same social dilemma can report opposite
conclusions. The three ingredients map onto three lines of a configuration
file. The policy's memory decides whether history-contingent strategies are
representable at all, so a feedforward policy can only sustain the stage
equilibrium while one bit of recurrent state can sustain cooperation. The
discount factor decides whether the punishment is worth fearing, and
\cref{ex:grim} computes the threshold exactly. The public randomisation device
is a shared random seed, the same object that \cref{def:ce-cce} requires. Change
any of the three and the set of self-consistent outcomes changes, with the
learning rule untouched. The methodological consequence is that a reported
cooperation rate in an iterated social dilemma measures the triple
(architecture, discount, seed availability) at least as much as it measures the
algorithm, and a comparison that varies the algorithm while holding the triple
implicit is not a controlled comparison.
\end{rlreading}

\begin{example}[The feasible set of the iterated prisoner's dilemma]
\label{ex:ipd-folk}
With the stage payoffs of \cref{ex:four-games} the minmax value of each agent is
$1$, attained when the opponent defects, which a linear program over the stage
game confirms. The feasible set is the convex hull of $(3,3)$, $(0,5)$, $(5,0)$
and $(1,1)$, and $F^\star$ is the part of it strictly above $1$ in both
coordinates. \Cref{thm:folk} therefore asserts that for a patient enough pair of
agents every payoff vector in that two-dimensional region, including $(3,3)$,
including $(4.3, 1.01)$ and including $(1.01, 4.3)$, is the payoff of some
subgame perfect equilibrium. The asymmetric pairs are near the boundary: a
linear program over the hull gives $4.327$ as the largest payoff the first agent
can receive while the second receives $1.01$, so $(4.9, 1.01)$ would lie outside
the feasible set and is not available to any profile, equilibrium or not. The statement ``the agents reached an equilibrium
of the iterated prisoner's dilemma'' therefore conveys almost nothing: it
excludes only payoffs below the minmax and outside the hull. Any empirical claim
about cooperation in this environment is a claim about which equilibrium the
particular learning rule selected, and is therefore a claim about the algorithm
and the architecture rather than about the game. \exend
\end{example}

\begin{example}[The discount threshold for grim trigger, computed]
\label{ex:grim}
Grim trigger cooperates until the opponent defects and defects forever after.
Suppose both agents play it. On the equilibrium path each earns
$R/(1-\gamma) = 3/(1-\gamma)$. A one-shot deviation earns $T = 5$ now and then
$P = 1$ forever, for a total of $5 + \gamma/(1-\gamma)$. Cooperation is
sustainable when
\begin{align*}
\frac{3}{1-\gamma} \;\ge\; 5 + \frac{\gamma}{1-\gamma}
\quad\Longleftrightarrow\quad
3 \;\ge\; 5(1-\gamma) + \gamma
\quad\Longleftrightarrow\quad
\gamma \;\ge\; \frac{T - R}{T - P} \;=\; \frac{2}{4} \;=\; \frac12 .
\end{align*}
At $\gamma = 0.4$ cooperating is worth $5.0$ and deviating is worth $5.667$, so
defection wins by $0.667$; at $\gamma = 0.5$ the two are both worth $6.0$; at
$\gamma = 0.6$ cooperating is worth $7.5$ and deviating $6.5$, so cooperation
wins by $1.0$. The practical consequence is blunt and is routinely overlooked:
whether cooperation is an equilibrium of a sequential social dilemma is decided
by the discount factor in the configuration file, and an experiment that reports
defection at $\gamma = 0.4$ and cooperation at $\gamma = 0.6$ has measured the
threshold $1/2$ rather than a property of the learning algorithm. \exend
\end{example}

\begin{example}[Tit-for-tat needs more patience than grim trigger]
\label{ex:tft}
Tit-for-tat plays what the opponent played last. Against a tit-for-tat opponent
three deviations must be checked, and only the third binds. Permanent defection
earns $5 + \gamma/(1-\gamma)$, which loses to $3/(1-\gamma)$ exactly when
$\gamma \ge 1/2$, as for grim trigger. A single defection followed by a return
earns less than that. Alternating defection and cooperation, however, earns
$T$ in odd periods and $S = 0$ in even ones, for a total of
$(T + \gamma S)/(1-\gamma^2) = 5/(1-\gamma^2)$, and this beats $3/(1-\gamma)$
exactly when $5/(1+\gamma) > 3$, that is when $\gamma < 2/3$. At $\gamma = 0.6$
the three values are $7.5$ for always cooperating, $6.5$ for always defecting
and $7.8125$ for alternating, so tit-for-tat is not a best response to itself;
at $\gamma = 0.8$ they are $15.0$, $9.0$ and $13.889$, and it is. The threshold
is $(T-R)/(R-S) = 2/3$, strictly above the grim-trigger threshold of $1/2$,
because tit-for-tat punishes for exactly one period and therefore punishes less.
The lesson for a practitioner is that the strength of the punishment encoded in
a policy determines the range of discount factors over which cooperation is
self-enforcing, so the recurrent architecture and the discount factor must be
chosen together. \exend
\end{example}

For a reinforcement learning audience the folk theorem is best read as a
statement about the policy class. The set of attainable equilibrium payoffs is
determined by how much of the history the policy is allowed to condition on, and
making the policy more expressive, by giving the agent a recurrent state or a
longer context window, does not merely improve the attainable return but changes
the set of self-consistent outcomes. A stationary policy in the iterated
prisoner's dilemma can only sustain mutual defection; a policy with one bit of
memory can sustain cooperation through a grim trigger or tit-for-tat. The
practical corollary is that reported cooperation rates in iterated social
dilemmas are properties of the joint specification of algorithm, memory
architecture and discount factor, and that comparing an
algorithm with a recurrent policy against one with a feedforward policy on such
a benchmark compares policy classes rather than learning rules.

The second reason the folk theorem matters here is that it is the cleanest
instance of a phenomenon that recurs throughout the monograph: in multi-agent
problems, equilibrium multiplicity is the rule, and therefore a solution concept
alone never determines an outcome. The single-agent case has a unique optimal
value and so a learning algorithm can be judged by how close it comes to it. In
a general-sum repeated game, \cref{thm:folk} says that essentially every
individually rational payoff is an equilibrium payoff, so the statement that an
algorithm converged to an equilibrium carries almost no information. What
carries information is which equilibrium, and the criteria for preferring one,
welfare, fairness, robustness to exploitation, are imported from outside the
equilibrium concept. This is the formal content behind the objection of
\citet{shoham2007if} that the multi-agent learning literature frequently leaves
its question unstated.

A third connection is worth drawing because it is productive rather than
cautionary. A repeated game in which each agent's strategy is implemented by a
finite automaton with $k$ states is a finite stochastic game whose state is the
tuple of automaton states, and the best response to a fixed finite-automaton
opponent is computable by ordinary dynamic programming on a state space of size
$k^{\nag - 1}$. This is the exact sense in which opponent modelling reduces
multi-agent learning to single-agent learning: if agent $i$ maintains a model of
the others as finite-state machines, then agent $i$ faces a Markov decision
process over the product of the physical state and the inferred machine states,
and if the model is uncertain, a partially observable one, which is the subject
of \cref{sec:bayesian}. The reduction is exact when the model class contains the
truth and the others are not learning, and it degrades gracefully in the first
respect and catastrophically in the second.

\begin{example}[Best responding to tit-for-tat is a two-state Markov decision problem]
\label{ex:automaton-mdp}
Fix the opponent to tit-for-tat and let the state be what the opponent will play
next, which under tit-for-tat is what we played last. The state space is
$\{C, D\}$, the action space is $\{C, D\}$, the transition is deterministic and
sends the state to the action just taken, and the rewards are the stage payoffs
of \cref{ex:four-games}. With $\gamma = 0.6$, value iteration on this two-state
Markov decision process returns $V(C) = 7.8125$ and $V(D) = 4.6875$ with the
greedy action being $D$ at the state $C$ and $C$ at the state $D$, which is
exactly the alternating deviation of \cref{ex:tft}. The computation takes four
numbers and reproduces a conclusion that the repeated-game argument reaches with
an infinite-horizon comparison, which is the point: once the opponent is a fixed
finite automaton, nothing multi-agent remains and the entire single-agent
toolbox applies verbatim. The hypothesis that fails in a real training run is
the fixity, and \cref{sec:negative} is devoted to its consequences. \exend
\end{example}

Imperfect monitoring is where repeated games become genuinely harder and where
the reinforcement learning reading changes. Under perfect monitoring a deviation
is detected with certainty, so punishment can be triggered without error. Under
imperfect public monitoring the agents observe only a noisy signal of the joint
action, so a punishment rule must sometimes punish an innocent partner, and
sustaining cooperation requires balancing the deterrent effect of punishment
against its cost on the equilibrium path. The recursive characterisation of the
equilibrium payoff set by \citet{abreu1990toward} is the dynamic-programming
formulation of exactly that balance: the set of equilibrium payoffs is the
largest self-generating set, which is a fixed point of a set-valued operator and
therefore the set analogue of \cref{eq:bellman-opt}. The reinforcement learning
counterpart is the observation that a value function over sets of payoff
vectors, rather than over scalars, is what a general-sum sequential problem
requires, which is also the shape taken by the occupancy-state reduction of
\cref{sec:reductions} and by the distributional view of multi-objective control.

Social dilemmas repeated in time, rather than in a matrix, are the form in which
this material reaches most practitioners. The empirical study of cooperation in the iterated prisoner's dilemma begins
with \citet{axelrod1981evolution}, whose tournaments established that simple
reciprocating strategies outperform both unconditional cooperation and
unconditional defection against a varied population. A sequential social dilemma
\citep{leibo2017multi} is a Markov game in which the temporally extended
policies play the roles of cooperation and defection, so that the dilemma is a
property of the induced empirical game over policies rather than of any single
state. The methodological consequence is that the analysis belongs to
\cref{sec:evolutionary}: one estimates the empirical payoff matrix over a finite
set of trained policies and examines whether it has the structure of
\cref{ex:four-games}, which is what the evaluation suites in this area do
\citep{leibo2021scalable}. The dilemma is thus not assumed but measured, and the
measurement is a normal-form game extracted from a Markov game.

One further construction belongs here because it shows that the repeated setting
admits strategies with no single-agent analogue at all. In the iterated
prisoner's dilemma a memory-one strategy can enforce a linear relation between
its own long-run payoff and the opponent's, regardless of what the opponent
does, and in particular can enforce a fixed extortionate ratio of surpluses
\citep{press2012iterated}. The construction works by choosing the four
conditional cooperation probabilities so that a certain determinant vanishes,
which makes the stationary distribution of the induced Markov chain satisfy the
desired linear constraint identically. The reinforcement learning reading is
that an opponent which adapts to maximise its own return will, against such a
strategy, be steered to a point chosen by the extorter, so opponent shaping in
the sense of \cref{sec:stackelberg} is available even to a memoryless policy
with four parameters. This is the sharpest available counterexample to the idea
that a small policy class implies a weak agent, and it belongs in the same
discussion as the folk theorem because both are statements that memory, however
little, changes the strategic situation qualitatively.

\begin{table}[t]
\centering
\caption{Repeated games and their reinforcement learning reading. The policy
class column indicates what the agent's policy must condition on; the
equilibrium column indicates what is sustainable.}
\label{tab:repeated}
\small
\begin{tabular}{@{}L{3.2cm} L{3.6cm} L{3.5cm} L{3.1cm}@{}}
\toprule
\textsc{Monitoring} & \textsc{Induced model} & \textsc{Policy class needed} & \textsc{Sustainable outcomes} \\
\midrule
None (one shot) & Normal-form game; bandit if $\nag=1$ & Stateless & Stage Nash equilibria \\
Perfect, finite horizon & Stochastic game on the history tree & Finite memory & Stage Nash by backward induction if unique \\
Perfect, discounted infinite & Stochastic game on the history tree & Unbounded or recurrent memory & All of $F^\star$ (\cref{thm:folk}) \\
Imperfect public & POSG with public state & Memory over public signals & Perfect public equilibria; self-generating sets \citep{abreu1990toward} \\
Imperfect private & POSG with private histories & Belief over others' histories & Sequential equilibria; hard \\
Fixed finite-automaton opponents & MDP over product of automaton states & Markov in the product state & Best response, computable \\
Memory-one, zero determinant & Four-parameter stationary policy & One bit of memory & Enforced linear payoff relation \citep{press2012iterated} \\
\bottomrule
\end{tabular}
\end{table}

\section{Stochastic games: the central correspondence}
\label{sec:stochastic}

The stochastic game of \citet{shapley1953stochastic} is the object from which
almost everything else in this monograph is obtained by specialisation or by
relaxation. It is also the model for which the correspondence with
reinforcement learning is tightest, because the Markov decision process is
literally its one-agent instance.

The two readings coincide so nearly here that the only thing worth flagging is
the vocabulary. What a game theorist calls a stochastic game a reinforcement
learning researcher calls a Markov game, and the two names denote the same
tuple; the first emphasises that the transitions are random, the second that
they are memoryless, and both properties are in the definition. What a game
theorist calls a stationary strategy profile a reinforcement learning researcher
calls a joint policy, and what the first calls the value of the game at a state
the second calls the value function. The one genuine difference is the object
that the Bellman update applies at each state: a maximum in the single-agent
case, and a solution operator on a payoff tensor in the multi-agent case. Every
algorithmic question in this section is a question about that operator.

\begin{nomen}{Stochastic-game vocabulary read from both sides.}{tab:nom-sg}
Stochastic game, Markov game & Multi-agent Markov decision process &
  The same tuple under two names \\
State $s \in \St$ & State; observation when fully observed &
  What the transition kernel depends on \\
Stationary strategy $\pi_i(\cdot \given s)$ & Policy of agent $i$ &
  A row of probabilities per state \\
Value of the game at $s$ & $V_i^\pi(s)$ &
  Discounted return of agent $i$ from $s$ \\
Stage game $\game_s[V]$ & Joint-action critic cross-section at $s$ &
  Reward plus discounted continuation \\
Solution operator $\mathrm{sol}_i$ & What replaces $\max$ in the backup &
  $\max$, $\val$, or a Nash selection \\
Shapley operator $\BellOp_{\mathrm{mm}}$ & Minimax Bellman backup &
  Solve a matrix game at every state \\
Markov perfect equilibrium & Stationary joint policy, each one greedy &
  Nobody gains by changing its policy \\
Team Markov game, MMDP & Common-payoff fully observed MARL &
  One reward, many action components \\
Adversary choosing the kernel & Robust MDP uncertainty set &
  Nature as a second player \\
Interchangeability of equilibria & Safety of mixing trained opponents &
  Holds in zero-sum, fails otherwise \\
Turn-based game & Alternating-move environment &
  One agent acts per state \\
$\omega$-regular objective & Safety or liveness specification; shielding &
  A property of the whole trajectory \\
\end{nomen}

\begin{definition}[Stochastic game, Markov game]
\label{def:sg}
A stochastic game is a tuple
\begin{align*}
G \;=\; \big(\agents,\; \St,\; \{\Ac_i\}_{i \in \agents},\; \Tr,\;
          \{\Rew_i\}_{i \in \agents},\; \gamma,\; b_0\big),
\end{align*}
in which $\agents$, $\St$ and $b_0$ are as before, $\Ac_i$ is the
action set of agent $i$ with $\Ac = \prod_i \Ac_i$ the joint action set,
$\Tr : \St \times \Ac \to \Simplex{\St}$ is the transition kernel, and
$\Rew_i : \St \times \Ac \times \St \to \Real$ is the reward function of agent
$i$. Agents act simultaneously and all observe the state. A stationary policy
of agent $i$ is $\pi_i : \St \to \Simplex{\Ac_i}$ and the value of a profile is
$V_i^\pi(s) = \Exp[\sum_t \gamma^t \Rew_i(s_t,a_t,s_{t+1}) \given s_0 = s]$.
\end{definition}

\begin{rlreading}
Compare \cref{def:sg} with \cref{def:mdp} symbol by symbol. The state space, the
transition kernel, the discount and the initial distribution are unchanged. Two
things are replaced. The single action set $\Ac$ becomes a product
$\prod_i \Ac_i$, so the environment's step function now takes a tuple and the
policy is split across $\nag$ heads, each of which controls one coordinate and
sees none of the others at the moment of choosing. The single reward $\Rew$
becomes a vector $(\Rew_1, \dots, \Rew_\nag)$, so the environment returns one
number per head. That is the entire difference, and everything hard about
multi-agent reinforcement learning follows from those two substitutions. The
value $V_i^\pi$ is the ordinary discounted return of head $i$ under the joint
policy $\pi$, and the superscript matters: there is one value function per agent
and all of them depend on all of the policies, so no single scalar is being
optimised and the machinery of \cref{sec:learning-dynamics} replaces the
machinery of optimisation. A practitioner who writes a joint-action critic
$Q_i(s, a_1, \dots, a_\nag)$ has written down exactly this object.
\end{rlreading}

\begin{example}[Evaluating a fixed joint policy, and why it is not the value of
the game]
\label{ex:sg-evaluation}
Let $\St = \{0,1\}$, let each agent have two actions, let every joint action
move the state to the other one, and let $\gamma = 0.9$. The row agent's stage
payoff is matching pennies at state $0$ and
$\big(\begin{smallmatrix} 0 & 3 \\ 2 & -1 \end{smallmatrix}\big)$ at state $1$,
and the column agent's is the negative of it. Fix the stationary profile in
which the row agent plays its first action with probability $0.8$ at both states
and the column agent plays its first action with probability $0.25$ at both.

Policy evaluation is two steps. First collapse each stage game to a scalar by
averaging over the profile, which is \cref{ex:multilinearity} applied twice: at
state $0$ the expected stage payoff is
$0.8(0.25 - 0.75) + 0.2(-0.25 + 0.75) = -0.3$, and at state $1$ it is
$0.8(0.75 \cdot 3) + 0.2(0.25 \cdot 2 - 0.75) = 1.75$. Second solve the two
linear equations $V(0) = -0.3 + 0.9\,V(1)$ and $V(1) = 1.75 + 0.9\,V(0)$, which
give $V(0) = 6.711$ and $V(1) = 7.789$.

Now compare with \cref{ex:zero-sum-sg}, which computes the value of the same
game to be $4.737$ at state $0$. The fixed profile earns the row agent
$1.974$ more than the value of the game. There is no contradiction: the value is
what the row agent can guarantee against a best-responding opponent, and this
column agent is not best responding. The lesson is the one that
\cref{tab:diagnostics} turns into a measurement rule: $V^\pi$ reports how a pair
of policies did against each other and says nothing about either of them
separately, so a return curve from self-play is a property of the pair. The
quantity that is a property of the row agent alone is the value it obtains
against a best response, and computing it requires solving a Markov decision
process, not reading a log. \exend
\end{example}

\begin{proposition}[Degeneracies of the stochastic game]
\label{prop:sg-mdp}
Let $G$ be a stochastic game. Then the following hold. First, if $\nag = 1$ then
$G$ is a Markov decision process in the sense of \cref{def:mdp}, with the same
state space, action space, kernel, reward and discount. Second, if
$\abs{\St} = 1$ then $G$ is a repeated normal-form game in the sense of
\cref{sec:repeated}, restricted to stationary strategies. Third, if
$\Rew_i = \Rew_j$ for all $i,j$ then the centralised control problem over the
joint action space, namely the Markov decision process
$(\St, \Ac, \Tr, \Rew, \gamma, b_0)$, has optimal value equal to the maximal
value attainable by any joint policy of $G$, and a deterministic optimal
stationary joint policy exists.
\end{proposition}

The third clause is the multi-agent Markov decision process of
\citet{boutilier1996planning}, and it is the statement that cooperative
multi-agent planning under full observability is a single-agent problem. It is
also the statement that is most frequently over-read. The reduction preserves
the optimal value and the set of optimal joint policies, and it says nothing
about whether $\nag$ independently acting agents will jointly realise one of
them. If the optimal joint action at a state is not unique, and each agent
independently selects a component of some optimal joint action, the resulting
joint action need not be optimal, and in a game such as the penalty game of
\citet{claus1998dynamics} it can be arbitrarily bad. We state this precisely
because the gap is the entire subject of equilibrium selection in cooperative
multi-agent reinforcement learning.

\begin{remark}[The reduction preserves value but not decentralisability]
\label{rem:mmdp-gap}
Let $G$ be a common-payoff stochastic game and $M$ the induced joint-action
Markov decision process. The map $\pi \mapsto \pi$ from joint policies of $G$ to
policies of $M$ is value-preserving and surjective onto the product policies of
$M$ only if the policy of $M$ factorises. An optimal policy of $M$ chosen
arbitrarily need not factorise as $\pi(a \given s) = \prod_i \pi_i(a_i \given s)$,
and even when an optimal factorised policy exists, there is in general no
decentralised rule by which each agent can select its component without a shared
tie-breaking convention or a shared random seed. The reduction is therefore
exact for the planning problem under centralised execution, and it is a
structural analogy only for decentralised execution.
\end{remark}

\begin{example}[A penalty game in which greedy learners select the worst equilibrium]
\label{ex:penalty}
Two agents share the payoff
\begin{align*}
\trimatrix{a_1 = 0}{a_1 = 1}{a_1 = 2}{a_2 = 0}{a_2 = 1}{a_2 = 2}%
{10 & -20 & 0}{-20 & 6 & 0}{0 & 0 & 4},
\end{align*}
a single number per cell because the game is common-payoff. The joint-action
Markov decision process of \cref{prop:sg-mdp} has optimal value $10$, attained
at $(0,0)$, and a centralised controller that issues joint actions reaches it
immediately. Independent learners do not. The game has three pure equilibria,
$(0,0)$, $(1,1)$ and $(2,2)$, worth $10$, $6$ and $4$. Consider an agent whose
partner is still exploring uniformly. Its expected payoffs are
$\tfrac13(10 - 20 + 0) = -3.33$ for action $0$,
$\tfrac13(-20 + 6 + 0) = -4.67$ for action $1$, and
$\tfrac13(0 + 0 + 4) = 1.33$ for action $2$. The greedy action under an
exploring partner is therefore the one that leads to the worst equilibrium, and
a pair of independent $Q$-learners that bootstrap from such estimates converges
to $(2,2)$ and collects $4$ instead of $10$, a shortfall of $6$ on a problem
whose centralised solution is trivial.

The mechanism is visible in the spreads. Action $0$ has outcomes ranging over
$30$ units and action $2$ over $4$, so averaging over partner uncertainty
penalises the action that attains the optimum far more than the cautious one.
This is relative overgeneralisation, and it is the reason the algorithmic
responses of \citet{lauer2000algorithm}, \citet{matignon2007hysteretic} and
\citet{panait2006lenience} all amount to replacing the average over partner
behaviour by something closer to a maximum: optimistic initialisation,
hysteretic learning rates that raise estimates faster than they lower them, and
leniency that discards a proportion of the worst observed outcomes. Each is an
equilibrium-selection device aimed at this computation, and each is justified by
it rather than by any convergence theorem. \exend
\end{example}

When $\nag = 2$ and $\Rew_2 = -\Rew_1$, the solution operator of
\cref{def:stage-game} is the minimax value of a matrix game, and everything
works.

\begin{theorem}[\citealp{shapley1953stochastic}]
\label{thm:shapley}
Let $G$ be a two-player zero-sum discounted stochastic game and define
$\BellOp_{\mathrm{mm}} : \Real^{\St} \to \Real^{\St}$ by
\begin{align}
(\BellOp_{\mathrm{mm}} V)(s) \;=\;
\max_{\pi_1 \in \Simplex{\Ac_1}} \; \min_{\pi_2 \in \Simplex{\Ac_2}}
\; \Exp_{a \sim \pi_1 \times \pi_2} \Exp_{s' \sim \Tr(\cdot\given s,a)}
\Big[ \Rew_1(s,a,s') + \gamma V(s') \Big].
\label{eq:minimax-bellman}
\end{align}
Then $\BellOp_{\mathrm{mm}}$ is a $\gamma$-contraction in the supremum norm, its
unique fixed point $V^\star$ is the value of $G$, both agents possess optimal
stationary mixed policies obtained by solving the matrix game in
\cref{eq:minimax-bellman} at $V^\star$, and value iteration
$V^{k+1} = \BellOp_{\mathrm{mm}} V^k$ converges geometrically.
\end{theorem}

\begin{example}[Minimax value iteration on a two-state game, by hand]
\label{ex:zero-sum-sg}
Let $\St = \{0, 1\}$, let each agent have two actions, let every joint action
move the state to the other one, and let $\gamma = 0.9$. The stage payoff to the
row agent is matching pennies at state $0$ and
$\big(\begin{smallmatrix} 0 & 3 \\ 2 & -1 \end{smallmatrix}\big)$ at state $1$.
Because the transition does not depend on the action, \cref{eq:minimax-bellman}
separates into $V(s) = \val(A_s) + \gamma V(1-s)$, where $\val$ is the value of
the stage matrix. The matching-pennies matrix has value $0$, and the second
matrix has value $1$, obtained by equating $2 - 2p$ and $4p - 1$ at $p = 1/2$.
Substituting, $V(0) = 0 + 0.9 V(1)$ and $V(1) = 1 + 0.9 V(0)$, so
$V(0) = 0.9(1 + 0.9 V(0))$, giving $V(0) = 0.9/0.19 = 4.7368$ and
$V(1) = 5.2632$. Running $\BellOp_{\mathrm{mm}}$ numerically from $V = 0$
reproduces these to machine precision. The equilibrium policies are
$(\tfrac12,\tfrac12)$ for both agents at state $0$, and at state $1$ the row
agent plays $(\tfrac12,\tfrac12)$ while the column agent plays
$(\tfrac23,\tfrac13)$. Two things are worth noticing. First, the policies are
mixed at every state, so a tabular $Q$-learner with greedy action selection
cannot represent the solution, which is \cref{ex:mp-no-pure} lifted to the
sequential case. Second, the backup at each state is a linear program rather
than a maximum, and this is the only change to the single-agent algorithm; the
contraction, the convergence rate and the stopping criterion are unchanged. That
is the whole content of minimax-Q. \exend
\end{example}

The proof rests on the fact that the value functional of a matrix game is
non-expansive, $\abs{\val(A) - \val(B)} \le \norm{A - B}_\infty$, which lifts
the contraction property of the single-agent operator verbatim. The
reinforcement learning algorithm obtained by replacing the expectation in
\cref{eq:minimax-bellman} by a sample is minimax-Q
\citep{littman1994markov}; its convergence under the usual Robbins-Monro
conditions follows from the contraction property through the general
stochastic-approximation analysis of \citet{szepesvari1999unified}, and is not
established in the paper that introduces the algorithm. Comparison with
\cref{eq:robust-bellman} gives the following identity, which we state because it
is the clearest case in this work of two literatures writing down the same
operator.

\begin{proposition}[Robust control is a zero-sum stochastic game;
\citealp{iyengar2005robust,nilim2005robust,wiesemann2013robust}]
\label{prop:robust-zero-sum}
Let $M_{\mathrm{rob}}$ be a robust Markov decision process with $(s,a)$-rectangular
uncertainty set $\mathcal{P}$, each $\mathcal{P}(s,a)$ convex and compact. Then $M_{\mathrm{rob}}$ is the two-player zero-sum
stochastic game in which agent $1$ chooses $a \in \Ac$, agent $2$ observes $a$
and chooses $p \in \mathcal{P}(s,a)$, the next state is drawn from $p$, and
$\Rew_2 = -\Rew_1$; the operators \cref{eq:robust-bellman} and
\cref{eq:minimax-bellman} coincide and the optimal robust policy is the
equilibrium policy of agent $1$. If instead $\mathcal{P}$ is only
$s$-rectangular, so that the adversary must commit to a kernel before observing
$a$, the induced game is simultaneous-move, the value still exists by the
minimax theorem, and the optimal robust policy is in general stochastic
\citep{wiesemann2013robust}.
\end{proposition}

The second clause has a concrete consequence: robust reinforcement learning
under $s$-rectangular uncertainty requires randomised policies while the
$(s,a)$-rectangular case does not, and the reason is the non-existence of pure
equilibria in matching pennies, imported into control theory. Adversarial reinforcement learning methods that train a
protagonist against a learned disturbance agent
\citep{pinto2017robust,morimoto2005robust} are instantiating this game, and
whether their fixed point is meaningful depends on whether the adversary's
action set corresponds to a rectangular uncertainty set.

\begin{example}[Who moves first in the uncertainty set]
\label{ex:rectangularity}
Suppose a state offers two actions and the adversary may set the probability of
a good outcome within $[0.6, 0.9]$ for action $a$ and within $[0.7, 0.7]$ for
action $b$, the good outcome being worth $10$. Under $(s,a)$-rectangularity the
adversary sees the action and picks $0.6$ for $a$, so the protagonist compares
$6.0$ with $7.0$ and plays $b$ deterministically; this is
\cref{ex:robust-vs-nominal} again, read as a game. Under $s$-rectangularity the
adversary must commit to one kernel before seeing the action, and the
protagonist may randomise to make the commitment less effective. The modelling
question that decides which applies is whether the disturbance can react to the
control, which in a learned adversarial setup is decided by whether the
adversary network receives the protagonist's action as an input. That single
wiring choice determines whether the trained protagonist should be expected to
be deterministic or stochastic, and it is the kind of design detail that the
game-theoretic reading makes visible and that an unstructured reading does
not. \exend
\end{example}

For general-sum discounted stochastic games, \citet{fink1964equilibrium} and
\citet{takahashi1964equilibrium} establish that a stationary Markov perfect
equilibrium exists, by a fixed-point argument on the product of stationary
policy spaces. Existence is all that transfers. The natural algorithmic analogue
of \cref{thm:shapley} is to define an operator that replaces $\val$ by a Nash
value of the stage game $\game_s[V]$, which is the Nash-Q algorithm of
\citet{hu2003nash}, and it fails for three distinct reasons that are worth
separating. First, the stage game may have several equilibria with different
payoffs, so the operator is set-valued and requires an exogenous selection rule
on which all agents must agree. Second, even with a selection rule the resulting
operator is not non-expansive, so the contraction argument is unavailable; the
convergence theorem for Nash-Q assumes that either every stage game encountered
during learning has a global optimum, or every stage game encountered has a
saddle point, with all agents consistently selecting equilibria of that one
type; the disjunction is over the whole run, it cannot be verified in advance,
and it effectively restricts the result to the cooperative and zero-sum cases. Third, and most fundamentally,
\citet{zinkevich2005cyclic} prove that no algorithm which extracts a stationary
policy from $Q$-values alone can compute stationary Nash equilibria of
general-sum stochastic games, by exhibiting pairs of games with identical
$Q$-functions but different unique stationary equilibria, so that no
policy-extraction rule can be correct in both. The friend-or-foe algorithm of
\citet{littman2001friend} responds by requiring the designer to declare each
opponent cooperative or adversarial, thereby forcing the stage game into one of
the two solvable classes, and the correlated-Q family of
\citet{greenwald2003correlated} responds by replacing the Nash operator with a
correlated-equilibrium operator, which restores computability of the stage
solution by linear programming while leaving the selection and non-expansiveness
problems in place.

\begin{example}[The stage operator is the only thing that changes]
\label{ex:stage-operator}
Take the stage game of \cref{ex:stage-game}, whose payoff tensor at the state in
question is $\big(\begin{smallmatrix} (9,0) & (0,9) \\ (0,9) & (0,9)
\end{smallmatrix}\big)$. A single-agent backup would report $\max = 9$ for agent
$1$. A minimax backup is undefined because the game is not zero-sum. A Nash
backup must choose among the equilibria, and enumeration gives exactly two,
$(0,1)$ and $(1,1)$, each worth $0$ to agent $1$; the profile $(0,0)$, which
gives agent $1$ its maximum of $9$, is not an equilibrium because agent $2$
gains by deviating, and $(1,0)$ is not one because agent $1$ does. A
correlated-equilibrium backup is unconstrained by welfare, since every cell of
this tensor has welfare $9$, so the selection must come from somewhere else. So
three different backups return different numbers for the same tensor, and the
choice is not determined by the data. This is the content of the sentence that there is no
canonical multi-agent Bellman operator, and it is why \cref{tab:sg-structures}
has a column for the stage operator rather than taking it for granted. \exend
\end{example}

The most widely used reduction in practice is to ignore the game structure and
run a single-agent algorithm per agent. The following proposition states exactly
when that is sound.

\begin{proposition}[Fixed others induce a Markov decision process]
\label{prop:fixed-others-mdp}
Let $G$ be a stochastic game and fix stationary policies $\pi_j$ for all
$j \ne i$. Then agent $i$ faces the Markov decision process
$M_i[\pi_{-i}] = (\St, \Ac_i, \Tr_i, \Rew_i^{\,\pi_{-i}}, \gamma, b_0)$ with
\begin{align}
\Tr_i(s' \given s, a_i) = \sum_{a_{-i}} \textstyle\prod_{j \ne i}
\pi_j(a_j \given s)\, \Tr(s' \given s, (a_i,a_{-i})),
\quad
\Rew_i^{\,\pi_{-i}}(s,a_i) = \Exp\big[\Rew_i(s,(a_i,a_{-i}),s')\big],
\label{eq:induced-mdp}
\end{align}
and the optimal policy of $M_i[\pi_{-i}]$ is exactly a best response
$\pi_i \in \brc_i(\pi_{-i})$. If the $\pi_j$ are history-dependent with finite
internal memory states $m_j$, the same holds with state space
$\St \times \prod_{j\ne i} \mathcal{M}_j$, and if the $m_j$ are not observed by
agent $i$ the induced model is a partially observable Markov decision process
with hidden component $(m_j)_{j \ne i}$.
\end{proposition}

\Cref{prop:fixed-others-mdp} is exact, and its hypothesis is violated in every
simultaneous learning run. When the $\pi_j$ change between episodes, the kernel
and reward of \cref{eq:induced-mdp} are time-varying, agent $i$'s problem is a
non-stationary Markov decision process, the Bellman operator it is implicitly
iterating has a moving fixed point, and the convergence guarantee of Q-learning
does not apply. This is the non-stationarity problem, and it is a failure of the modelling
assumption rather than a difficulty of implementation. Three
families of responses exist and are catalogued in \cref{sec:training}:
centralise the critic so that it conditions on the joint action and is therefore
stationary with respect to a fixed joint policy, which is the
centralised-training-decentralised-execution idea; model the others explicitly
and treat their parameters as part of the state, which converts the problem into
the partially observable one of \cref{sec:bayesian}; or abandon the fixed-point
target and use a dynamics whose guarantee is a regret bound rather than a
convergence statement, which is the no-regret route to coarse correlated
equilibria.

A final specialisation deserves mention because it connects this material to
formal methods and to safe reinforcement learning. In a turn-based stochastic
game only one agent acts at each state, which removes the need for mixing and
makes backward induction applicable; the objective may then be an $\omega$-regular
property of the trajectory, such as always avoiding a set of states or visiting
a set infinitely often, rather than a discounted sum. Games of this form are the
standard model of reactive synthesis, and the winning strategy of the
protagonist in a safety game is exactly the shield that filters a learned
policy's actions in safe reinforcement learning \citep{alshiekh2018safe}. The
correspondence is that the safety specification is the adversary's payoff, the
environment is the adversary, and the shield is the protagonist's winning
region computed by backward induction on the game graph. The complexity of these
problems is collected with the others in \cref{sec:complexity}, where the
turn-based and the simultaneous-move cases are seen to separate sharply.

\begin{table}[t]
\centering
\caption{Stochastic games by reward structure: the sound single-agent
reduction, the solution operator entering the Bellman update, and whether that
operator is a contraction. The last column names the canonical tabular
algorithm.}
\label{tab:sg-structures}
\small
\begin{tabular}{@{}L{2.5cm} L{3.1cm} L{2.6cm} C{2.1cm} L{2.6cm}@{}}
\toprule
\textsc{Structure} & \textsc{Sound reduction} & \textsc{Stage operator} & \textsc{Contraction} & \textsc{Algorithm} \\
\midrule
$\nag = 1$ & Identity; it is an MDP & $\max$ & Yes & Q-learning \\
Common payoff, centralised & Joint-action MDP & $\max$ over $\Ac$ & Yes & Joint-action Q-learning \\
Common payoff, decentralised & None exact; see \cref{rem:mmdp-gap} & $\max$ plus tie-breaking & Given a selection rule & Distributed Q, FMQ \\
Two-player zero-sum & Minimax-Bellman operator; a robust MDP when the adversary is rectangular (\cref{prop:robust-zero-sum}) & $\val$ (minimax) & Yes & Minimax-Q \\
Two-player constant-sum & Affine reparametrisation of zero-sum & $\val$ & Yes & Minimax-Q \\
General sum & None; existence only & Nash selection & No & Nash-Q, CE-Q \\
General sum, potential & Potential MDP; see \cref{sec:potential} & $\max$ of $\pot$ & Yes & Independent PG \\
Fixed opponents & Induced MDP (\cref{prop:fixed-others-mdp}) & $\max$ & Yes & Any single-agent method \\
Turn-based, $\omega$-regular & Game graph; backward induction & Attractor computation & Exact, finite & Synthesis; shielding \\
\bottomrule
\end{tabular}
\end{table}

\section{Partial observability and the lattice of decentralised models}
\label{sec:partial-obs}

Removing the assumption that every agent observes the state generates the
richest and the most terminologically confused region of the map. The confusion
is avoidable, because the models differ along exactly two axes, the reward
structure and the observability structure, and \cref{fig:lattice} is the
resulting two-dimensional grid. We give the definitions along that grid and
then state the inclusions as a proposition.

The two readings part company here more sharply than anywhere else, and the
split is terminological rather than mathematical. The planning literature, which
grew out of operations research, names its models by what the agents cannot see
and by the fact that they share a reward, producing the decentralised partially
observable Markov decision process and its relatives; its question is to compute
an optimal joint policy, and its solution concept is optimality. The game theory
literature names the same tuple by the fact that the agents may want different
things, producing the partially observable stochastic game; its question is to
characterise equilibria, and its solution concept is equilibrium. A
common-payoff partially observable stochastic game and a decentralised partially
observable Markov decision process are the same tuple. The two literatures ask
different questions about it, and \cref{prop:opt-implies-ne} relates the
answers: every optimum is an equilibrium and most equilibria are not optima.

\begin{nomen}{Decentralised-model vocabulary read from both sides.}{tab:nom-partialobs}
Partially observable stochastic game & General-sum partially observable MARL &
  Many agents, private observations, own rewards \\
Common-payoff POSG & Dec-POMDP &
  The same tuple; one shared reward \\
Jointly observable & Dec-MDP &
  Together the agents see the state \\
Private information $p^i_t$ & Agent $i$'s action-observation history &
  What only agent $i$ knows \\
Common information $c_t$ & What is public; broadcast channel contents &
  What everybody knows everybody knows \\
Interactive belief hierarchy & Nested opponent models (I-POMDP) &
  Beliefs about beliefs about beliefs \\
Optimal joint policy & Team-optimal decentralised controller &
  The planner's answer \\
Nash equilibrium of the POSG & Self-play fixed point &
  The learner's answer \\
Equilibrium selection & Convention, tie-breaking, shared seed &
  Which of several good answers \\
Risk dominance & Outcome favoured under partner uncertainty &
  Safe when the partner may err \\
Payoff dominance & Pareto-best equilibrium &
  Best if coordination succeeds \\
Decentralised execution & No communication at deployment &
  Each policy sees only its own input \\
\end{nomen}

\begin{definition}[Partially observable stochastic game]
\label{def:posg}
A partially observable stochastic game is a tuple
$G = (\agents, \St, \{\Ac_i\}, \{\Ob_i\}, \Tr, \Om, \{\Rew_i\}, \gamma, b_0)$ in
which $(\agents, \St, \{\Ac_i\}, \Tr, \{\Rew_i\}, \gamma, b_0)$ is a stochastic
game, $\Ob_i$ is the observation set of agent $i$ with
$\Ob = \prod_i \Ob_i$, and $\Om : \St \times \Ac \to \Simplex{\Ob}$ is a joint
observation kernel, with $\Om(o \given s', a)$ the probability that the agents
jointly receive $o = (o_1,\dots,o_\nag)$ after joint action $a$ leads to state
$s'$. Agent $i$ observes only $o_i$ and its own actions, so its policy is a map
$\pi_i : \Hi_i \to \Simplex{\Ac_i}$ on its own action-observation histories
$h_i^t = (o_i^0, a_i^0, \dots, a_i^{t-1}, o_i^t)$.
\end{definition}

\begin{rlreading}
Take \cref{def:sg} and give each head its own observation stream. The kernel
$\Om$ is the environment's observation function, $o_i$ is the tensor that
arrives at head $i$, and $h_i^t$ is whatever that head has accumulated, which in
an implementation is a recurrent hidden state or a stacked context window. The
policy signature changes from $\pi_i : \St \to \Simplex{\Ac_i}$ to
$\pi_i : \Hi_i \to \Simplex{\Ac_i}$, and that single change is the difference
between a feedforward and a recurrent architecture. One consequence is worth
stating before any theorem. In \cref{def:pomdp} a single agent could replace its
history by a belief over the state, because its own history carried all the
information in the system about that state. Here it cannot, because what agent
$i$ needs in order to predict the future includes what agent $j$ will do, which
depends on $h_j^t$, which agent $i$ has never seen. The hidden variable is no
longer the world state; it is the world state together with every other agent's
history, and \cref{prop:posg-lattice} records what that costs.
\end{rlreading}

The joint observation kernel is deliberately not assumed to factorise, because
correlated observation noise is common and because the non-factorised form is
what makes the model closed under the operations we perform on it.

\begin{definition}[Decentralised POMDP and decentralised MDP]
\label{def:decpomdp}
A decentralised partially observable Markov decision process is a partially
observable stochastic game with $\Rew_i = \Rew$ for all $i$. It is a
decentralised Markov decision process when it is in addition jointly fully
observable, meaning that there exists a map $\varphi : \Ob \to \St$ with
$\Prb(s_{t} = \varphi(o^t)) = 1$ for every reachable $(s_t, o^t)$, so that the
joint observation determines the state even though no individual observation
need do so.
\end{definition}

\begin{rlreading}
Setting $\Rew_i = \Rew$ is one line of the environment: return the same scalar
to every head. That line changes the question being asked, not the tuple. With
distinct rewards there is no best joint policy, only equilibria; with one shared
reward there is a best joint policy, and the whole apparatus of optimality
returns. Joint full observability is a second, independent line: it holds when
concatenating everyone's observation would reconstruct the state, which is the
situation of a sensor network whose nodes together see everything and
individually see a slice. The two conditions are what separate the four
cooperative models of \cref{fig:lattice}, and a practitioner can read off which
one applies by answering two yes-or-no questions about the simulator.
\Cref{ex:posg-tuple} answers them for a concrete problem.
\end{rlreading}

\begin{example}[One problem, written out as the tuple]
\label{ex:posg-tuple}
Instantiate \cref{def:posg} on the two-rescuer problem that
\cref{ex:dec-rescue} will then solve. The agent set is $\agents = \{1,2\}$. The
state set is $\St = \{\mathrm{left}, \mathrm{right}\}$, recording which door
hides the victim, with $b_0$ uniform. Each agent's action set is
$\Ac_i = \{\mathrm{enter\ left}, \mathrm{enter\ right}\}$ and each observation
set is $\Ob_i = \{\mathrm{signal\ left}, \mathrm{signal\ right}\}$. The joint
observation kernel factorises here, $\Om(o \given s) = \prod_i \Om_i(o_i \given
s)$ with each signal correct with probability $0.8$, so it is an array of
$2 \times 4 = 8$ numbers. The rewards are equal, $\Rew_1 = \Rew_2$, with value
$10$ when both agents enter the right door, $-10$ when both enter the wrong one
and a penalty when they split.

Now answer the two questions. Are the rewards equal? Yes, so by
\cref{def:decpomdp} this is a decentralised partially observable Markov decision
process and the optimal joint policy is a well-posed object. Does the joint
observation determine the state? No: both signals are wrong simultaneously with
probability $0.2 \times 0.2 = 0.04$, and after a pair of left signals the
posterior on the left door is $0.8^2/(0.8^2 + 0.2^2) = 0.941$, not $1$. So it is
not a decentralised Markov decision process, and the model sits in the
Dec-POMDP row of \cref{tab:posg-lattice} rather than the Dec-MDP row. Raise the
sensor accuracy to $1$ and the second answer flips, the posterior becomes
degenerate, and the problem drops a complexity row. That is the entire content
of the distinction, and it is decided by one number in the observation
kernel. \exend
\end{example}

\begin{proposition}[Inclusions and degeneracies]
\label{prop:posg-lattice}
The following hold and are exhibited in \cref{fig:lattice}. A partially
observable stochastic game with $\nag = 1$ is a partially observable Markov
decision process. A decentralised partially observable Markov decision process
with $\nag = 1$ is a partially observable Markov decision process. A partially
observable stochastic game in which $\Ob_i = \St$ and
$\Om(o \given s',a) = \Ind[o_i = s' \ \forall i]$ is a stochastic game. A
decentralised Markov decision process with $\nag = 1$ is a Markov decision
process. The class inclusions
$\mathrm{MDP} \subset \mathrm{MMDP} \subset \mathrm{SG} \subset
\mathrm{POSG}$ and
$\mathrm{MDP} \subset \mathrm{Dec\text{-}MDP} \subset \mathrm{Dec\text{-}POMDP}
\subset \mathrm{POSG}$ hold, and $\mathrm{Dec\text{-}POMDP}$ and
$\mathrm{SG}$ are incomparable, their intersection being $\mathrm{MMDP}$.
\end{proposition}

Two consequences of \cref{def:posg} govern everything that follows. The first is
that there is no belief state. In a single-agent partially observable Markov
decision process the agent's belief $\Prb(s_t \given h_t)$ is a sufficient
statistic because the agent's own history is all the information in the system
that bears on the hidden state. In a partially observable stochastic game, what
agent $i$ needs in order to predict the future is a belief over the state and
over the other agents' histories jointly, because the others' future actions
depend on their histories; and the others' beliefs about agent $i$'s history
matter in turn, generating the infinite hierarchy of interactive beliefs
formalised by \citet{gmytrasiewicz2005framework} in the interactive partially
observable Markov decision process. There is no finite sufficient statistic in
general, and this is the structural reason why decentralised models are harder
than their centralised counterparts by an entire complexity class rather than by
a constant factor. The lower bounds in the last column of
\cref{tab:posg-lattice} for the two general-sum rows are inherited from the
common-payoff restriction and therefore concern the team-optimal joint policy;
the equilibrium problems of a partially observable stochastic game are not
optimisation problems and are classified separately in \cref{sec:complexity}.
The finitely nested interactive partially observable Markov decision process is
\PSPACE-hard at a finite horizon and undecidable at an infinite one, exactly as
in the single-agent case \citep{gmytrasiewicz2005framework}.

The smallest problem on which all of this is visible has one time step, two
agents and two states, and we use it repeatedly.

\begin{example}[The two-rescuer problem]
\label{ex:dec-rescue}
A victim is behind the left or the right door with equal probability. Two
rescuers each receive a private signal that is correct with probability $0.8$.
Each must enter a door. If both enter the door with the victim they score $10$,
if both enter the empty door they score $-10$, and if they split up they score a
penalty $p$. The reward is shared, so this is a one-step decentralised partially
observable Markov decision process in the sense of \cref{def:decpomdp}, and each
agent's type is its own signal, which is \cref{prop:bayesian-stage} in its
smallest instance.

First compute what a centralised planner that sees both signals would do. The
posterior that the victim is on the left given two left signals is
$0.8^2 / (0.8^2 + 0.2^2) = 0.9412$; given one signal of each kind it is exactly
$0.5$. Following the posterior majority, and breaking the tie arbitrarily, gives
an expected score of $6.0$.

Now take $p = -20$. There are four decision rules per agent, from the constant
rules to following or inverting the signal. Enumerating the sixteen pairs, the
best achievable value is $0$, attained by both agents ignoring their signals and
always entering the same fixed door; following the signal gives $-0.4$, which is
worse than ignoring it. The entire informational advantage of the private
signals is destroyed by the risk of splitting up, and the gap to the
centralised planner is the full $6.0$.

Now take $p = -2$. The best pair of decision rules is for both agents to follow
their signals, which yields $5.36$, against $0$ for ignoring them, and the gap
to the planner falls to $0.64$. The qualitative conclusion is that whether
private information is worth using at all is decided by the cost of
miscoordination, not by the quality of the information: the signal accuracy was
$0.8$ in both cases. A practitioner reading a decentralised benchmark should
therefore ask what the miscoordination penalty is before asking whether the
agents learned to use their observations, because in the first regime the
optimal decentralised policy ignores them and a learner that discovers this is
behaving correctly. \exend
\end{example}

The second consequence is that the optimality criterion must be re-examined. In
a decentralised partially observable Markov decision process the objective is
unambiguous, since all agents share a reward, and the planning problem is to
find a joint policy $\pi = (\pi_1,\dots,\pi_\nag)$ maximising the common value.
This is a well-posed optimisation problem rather than an equilibrium problem,
which is why the decentralised partially observable Markov decision process literature
speaks of optimal joint policies while the partially observable stochastic game
literature speaks of equilibria. The relation between the two is that every
optimal joint policy of a decentralised partially observable Markov decision
process is a Nash equilibrium of the corresponding partially observable
stochastic game, and the converse fails badly: the set of Nash equilibria of a
common-payoff game includes every profile from which no agent can unilaterally
improve, and in a coordination problem that set contains arbitrarily poor
profiles. This single inclusion, optimum implies equilibrium but not conversely,
is the precise statement of why independent learners in cooperative tasks
converge to suboptimal conventions.

\begin{proposition}[Optimality and equilibrium in common-payoff games]
\label{prop:opt-implies-ne}
Let $G$ be a decentralised partially observable Markov decision process and
$\pi^\star$ an optimal joint policy. Then $\pi^\star$ is a Nash equilibrium of
$G$ viewed as a partially observable stochastic game. The converse is false:
there exist common-payoff games with Nash equilibria whose value is strictly
smaller than the optimum, and the gap can be made arbitrarily large.
\end{proposition}

\begin{proof}
If some agent $i$ had a profitable unilateral deviation $\pi_i'$ from
$\pi^\star$, then because all rewards coincide the joint policy
$(\pi_i', \pi^\star_{-i})$ would have strictly larger common value, contradicting
optimality. For the converse, take the two-agent single-state game with
$\Ac_1 = \Ac_2 = \{0,1\}$, $\util(0,0) = 0$, $\util(1,1) = K$ and
$\util(a) = -P$ otherwise. The profile $(0,0)$ is a Nash equilibrium for
$P > 0$, since a unilateral move to $1$ yields $-P < 0$, and its value is $0$
while the optimum is $K$.
\end{proof}

The game in the proof is a two-action instance of the penalty structure studied
by \citet{claus1998dynamics}, and it isolates the phenomenon in its minimal
form: the suboptimal equilibrium is
sustained by the risk an agent incurs when it unilaterally departs from a
convention, with no conflict of interest anywhere in the game. Risk dominance and payoff dominance
are the game-theoretic names for the two criteria that conflict here
\citep{harsanyi1988general}, and independent learners with exploration converge
to the risk-dominant equilibrium under a wide range of conditions, which is
exactly the behaviour reported empirically under the name of relative
overgeneralisation \citep{wei2016lenient}. The algorithmic responses, optimistic
initialisation, lenient learning, hysteretic learning and distributed Q-learning
\citep{lauer2000algorithm,matignon2007hysteretic,panait2006lenience}, are all
attempts to bias the dynamics toward the payoff-dominant equilibrium, and they
are best understood as equilibrium selection devices rather than as learning
improvements.

\begin{example}[Why an arbitrary convention counts as an equilibrium]
\label{ex:convention}
Suppose two agents must each send one of two symbols, and they score $1$ if the
first agent's symbol equals the second agent's and $0$ otherwise. Both
``always send $0$'' and ``always send $1$'' are optimal joint policies, and so
is any pair that agrees. Now suppose two independent training runs produce the
first convention and the second respectively. Each run reports a return of $1$
and each has found an optimal joint policy. Pairing one agent from the first run
with one from the second returns $0$. No agent is at fault, no run failed, and
the mean return across runs is uninformative about the result. This is the
reason \cref{tab:diagnostics} asks for cross-play, the return obtained when
agents from independently trained runs are paired, as the diagnostic for
cooperative decentralised problems, and it is the failure that
\citet{hu2020other} is designed to prevent by breaking the symmetry that makes
the two conventions interchangeable. In the vocabulary of this section, the
pathology is that the game has a symmetry group acting on its optimal policies,
self-play selects one orbit element arbitrarily, and nothing in the objective
records which. \exend
\end{example}

\begin{table}[t]
\centering
\caption{The partially observable lattice. The columns record whether the
individual observation determines the state, whether the joint observation does,
and what the model degenerates to when $\nag = 1$. The final column gives the
finite-horizon planning complexity with the model given explicitly; the
infinite-horizon and stationary-equilibrium classifications, which differ, are
collected in \cref{tab:complexity}.}
\label{tab:posg-lattice}
\small
\begin{tabular}{@{}L{2.3cm} L{2.0cm} C{1.3cm} C{1.3cm} L{1.6cm} C{3.6cm}@{}}
\toprule
\textsc{Model} & \textsc{Rewards} & \textsc{$o_i \to s$} & \textsc{$o \to s$} & \textsc{$\nag=1$} & \textsc{Finite-horizon complexity} \\
\midrule
MDP & single agent & yes & yes & MDP & \PTIME-complete \\
MMDP & common & yes & yes & MDP & \PTIME\ in $\abs{\St}\prod_i\abs{\Ac_i}$ \\
SG (Markov game) & general & yes & yes & MDP & \PTIME\ if zero-sum; \PPAD-complete if general-sum \\
POMDP & single agent & no & no & POMDP & \PSPACE-complete \\
Dec-MDP & common & no & yes & MDP & \NEXP-complete (horizon $<\abs{\St}$) \\
Dec-POMDP & common & no & no & POMDP & \NEXP-complete (horizon $<\abs{\St}$) \\
Jointly obs.\ POSG & general & no & yes & MDP & \NEXP-hard for the team-optimal problem \\
POSG & general & no & no & POMDP & \NEXP-hard for the team-optimal problem \\
I-POMDP & general & no & no & POMDP & \PSPACE-hard, finitely nested \\
\bottomrule
\end{tabular}
\end{table}

\section{Exact reductions of decentralised models to single-agent ones}
\label{sec:reductions}

The results in this section are the technical heart of the correspondence. Each
one takes a model in which several agents act on private information and
produces a single-agent sequential decision problem whose solution is in
bijection with the solution of the original. They are the multi-agent analogues
of \cref{thm:belief-mdp}, they share its template, and they share its price: the
reduced problem is a Markov decision process over a space of distributions, so
it is amenable to dynamic programming in principle and not in practice without
further approximation. Collecting them in one place makes visible that the
modern deep algorithms for cooperative and for zero-sum imperfect-information
play are implementations of the same two theorems.

The two readings are, for once, the same reading arrived at twice. The control
theorist constructs a fictitious coordinator because decentralised control
problems are not amenable to dynamic programming in their own variables; the
game theorist constructs a public belief state because imperfect-information
games have no proper subgames in their own variables. Both are performing the
same operation, namely conditioning on common knowledge and lifting the action
space from actions to functions of private information. Recognising this is the
single most useful item in this monograph, and the table below is its
dictionary.

\begin{nomen}{Reduction vocabulary read from both sides.}{tab:nom-reductions}
Common knowledge $c_t$ & Public features; broadcast observations &
  What conditioning is legitimate on \\
Prescription $\gamma_t^i$ & Decision rule mapping private history to action &
  A function, not an action \\
Coordinator & Centralised controller issuing prescriptions &
  A fictitious agent who sees only the public part \\
Common-information belief $\Pi_t$ & Posterior over state and private histories &
  The coordinator's belief state \\
Occupancy state $\xi_t$ & Distribution over state and joint history &
  The sufficient statistic of a Dec-POMDP \\
Public belief state $\beta_t$ & Public state plus per-agent range beliefs &
  The zero-sum counterpart of $\xi_t$ \\
Behavioural prescription & Policy head conditioned on the public state &
  What a network outputs in BAD or ReBeL \\
Piecewise linear convex value & $\alpha$-vector representation &
  Why point-based solvers transfer \\
Deterministic transition on $\xi$ & Noise absorbed into the state &
  Search over sequences of decision rules \\
Type in a stage Bayesian game & Action-observation history at time $t$ &
  The private information you condition on \\
Subgame rooted at $\beta_t$ & Depth-limited search subtree &
  Where search can legitimately restart \\
\end{nomen}

The unifying idea is the common-information decomposition. Partition what the
system knows at time $t$ into the common information $c_t$, defined as the
intersection of the agents' information, and the private information $p_t^i$
retained by each agent. A fictitious coordinator who observes only $c_t$ cannot
choose actions, because actions must depend on private information, but it can
choose prescriptions, that is, functions $\gamma_t^i : \mathcal{P}_t^i \to \Ac_i$
mapping each agent's private information to an action. The coordinator's choice
of $\gamma_t = (\gamma_t^1,\dots,\gamma_t^\nag)$ together with the agents'
mechanical application of their prescriptions reproduces any decentralised
strategy profile.

\begin{theorem}[Common-information reduction; \citealp{nayyar2013decentralized}]
\label{thm:coordinator}
Consider a decentralised stochastic control problem with partial history
sharing, in which the common information $c_t$ is non-decreasing and each agent
$i$ knows $c_t$ and its private information $p_t^i$. Define the coordinator's
problem as the single-agent decision problem with observation $c_t$, action
$\gamma_t$ ranging over the finite set of prescription profiles, and reward the
expected team reward. Then the coordinator's problem is a partially observable
Markov decision process whose hidden state is $(s_t, p_t^1,\dots,p_t^\nag)$, its
belief-MDP state is the common-information belief
$\Pi_t(s,p) = \Prb(s_t = s,\, p_t = p \given c_t)$, and there are
cost-preserving maps in both directions between coordination strategies and
decentralised profiles, so the two problems have the same achievable expected
total reward. Consequently an optimal coordinator strategy yields an optimal
decentralised profile, and dynamic programming over $\Pi_t$ is valid.
\end{theorem}

\Cref{thm:coordinator} converts a decentralised problem into a centralised one
by enlarging the action space from actions to functions of private information.
The exponential blow-up, $\abs{\Ac_i}^{\abs{\mathcal{P}_t^i}}$ prescriptions per
agent, is where the complexity goes, and it is not an artefact: by
\cref{sec:complexity} the original problem is \NEXP-hard, so no polynomial-time
reduction producing a polynomially sized partially observable Markov decision
process can exist unless \PSPACE, \EXP\ and \NEXP\ all coincide.

\begin{example}[Counting prescriptions]
\label{ex:prescriptions}
In the two-rescuer problem of \cref{ex:dec-rescue} each agent has two possible
private observations and two actions, so there are $2^2 = 4$ prescriptions per
agent and $4^2 = 16$ prescription profiles, which is why the enumeration in that
example was feasible by hand. Scale the problem a little. With two agents, four
private observations and three actions there are $3^4 = 81$ prescriptions each
and $6561$ profiles. With three agents, eight private observations and four
actions there are $4^8 = 65536$ prescriptions each and
$65536^3 \approx 2.8 \times 10^{14}$ profiles, which is a single step of the
coordinator's problem. For comparison, the number of deterministic joint
policies of a two-agent decentralised problem with two observations per step
over a horizon of three is $(2^{7})^2 = 16384$, and it grows doubly
exponentially in the horizon. These numbers are the content of the
\NEXP-hardness in \cref{tab:complexity}; the reduction is exact and its cost is
exactly this counting. The practical reading is that occupancy-state and
coordinator methods are usable when the private information is small, which in
practice means when the problem has been engineered so that most of what agents
observe is public. \exend
\end{example}

The decentralised partially observable Markov decision process with no
communication is the special case in which the common information is empty, and
the common-information belief then becomes a distribution over the state and
over the joint action-observation history. This object has its own name and its
own theorem.

\begin{definition}[Occupancy state]
\label{def:occupancy}
Let $G$ be a decentralised partially observable Markov decision process and let
$\Hi_t = \prod_i \Hi_t^i$ be the set of joint action-observation histories of
length $t$. Given the initial belief $b_0$ and the sequence of joint decision
rules $\delta_{0:t-1}$ with $\delta_k = (\delta_k^i : \Hi_k^i \to \Ac_i)_i$, the
occupancy state is
$\xi_t \in \Simplex{\St \times \Hi_t}$ defined by
$\xi_t(s,h) = \Prb(s_t = s,\, h_t = h \given b_0, \delta_{0:t-1})$.
\end{definition}

\begin{rlreading}
The occupancy state is the replacement for the belief state of
\cref{thm:belief-mdp} when there is more than one agent, and it is a bigger
object for a precise reason: it is a distribution not over the world state but
over the world state paired with everybody's recurrent hidden state. A
practitioner who wanted to maintain it would have to track, jointly, where the
world is and what every teammate has seen, because a teammate's future action
depends on its own history and on nothing the agent can observe. The decision
rule $\delta_t$ is not an action but a function from a history to an action,
that is, an entire policy network evaluated at one time step, so the
coordinator's action space is a space of networks. \Cref{ex:occupancy} writes
one out in full for an eight-entry instance, and \cref{ex:prescriptions} counts
how fast it grows.
\end{rlreading}

\begin{example}[An occupancy state written out]
\label{ex:occupancy}
In the two-rescuer problem, after the signals have been drawn and before any
action is taken, the occupancy state is a distribution over the eight pairs
consisting of a door and a pair of signals. With signal accuracy $0.8$ its
entries are
$\xi(\mathrm{left}; L, L) = 0.5 \times 0.8 \times 0.8 = 0.32$,
$\xi(\mathrm{left}; L, R) = \xi(\mathrm{left}; R, L) = 0.5 \times 0.8 \times 0.2
= 0.08$,
$\xi(\mathrm{left}; R, R) = 0.5 \times 0.2 \times 0.2 = 0.02$,
and the four mirror entries for the right door, summing to one. Three features
of this object are worth naming because they recur at every scale. It is a
single vector, not a pair of beliefs, because no factorisation into per-agent
beliefs is valid; it lives in a simplex of dimension seven here and of dimension
$\abs{\St}\prod_i\abs{\Hi^i_t} - 1$ in general, which is where the
blow-up lives; and the map from it and a joint decision rule to the next
occupancy state is deterministic, because all the randomness of the problem has
already been absorbed into the distribution. That determinism is what
\cref{thm:occupancy} exploits. \exend
\end{example}

\begin{theorem}[Occupancy-MDP reduction; \citealp{dibangoye2016optimally}]
\label{thm:occupancy}
Let $G$ be a finite-horizon decentralised partially observable Markov decision
process. The occupancy state $\xi_t$ is a sufficient statistic for the history
of joint decision rules, the map $(\xi_t, \delta_t) \mapsto \xi_{t+1}$ is
deterministic, and the occupancy Markov decision process
$(\Simplex{\St\times\Hi}, \{\delta\}, \text{deterministic transition},
\tilde\Rew)$ with
$\tilde\Rew(\xi_t,\delta_t) = \sum_{s,h} \xi_t(s,h)\, \Rew(s,\delta_t(h))$
is a continuous-state, deterministic-transition Markov decision process whose
optimal value equals the optimal value of $G$ and whose optimal policy induces
an optimal joint policy of $G$. Moreover the optimal value function
$V_t^\star : \Simplex{\St\times\Hi_t} \to \Real$ is piecewise linear and convex.
\end{theorem}

The piecewise linearity and convexity is the same property that
\cref{thm:belief-mdp} gives for the single-agent case, and it is the property
that makes point-based approximation methods applicable, which is how the
heuristic search value iteration family was transplanted from partially
observable Markov decision processes to decentralised ones. The determinism of
the transition is worth emphasising, because it is the formal statement that all
the stochasticity of a decentralised partially observable Markov decision
process has been absorbed into the state; the residual problem is a search over
sequences of decision rules.

The zero-sum counterpart replaces the occupancy state by the public belief
state, and the fact that the construction works there too, with an entirely
different solution concept, shows that the reduction is a property of the
information structure rather than of the reward structure.

\begin{definition}[Public belief state]
\label{def:pbs}
In a two-player zero-sum game with imperfect information and public actions,
let the public state $s^{\mathrm{pub}}_t$ be the common knowledge produced by the
observed public history, and let $\mathcal{S}_i$ be the set of agent $i$'s
information states consistent with it. The public belief state is
$\beta_t = (s^{\mathrm{pub}}_t, \mu_1, \mu_2)$ with
$\mu_i \in \Simplex{\mathcal{S}_i}$ the distribution over agent $i$'s
information states induced by the agents' strategies and Bayes' rule.
\end{definition}

\begin{rlreading}
A value network for an imperfect-information game has to take something as
input, and \cref{def:pbs} says what. Not the world state, because the opponent
cannot condition on it and a position's value depends on what the opponent
believes. Not the public history alone, because two positions with the same
betting sequence are worth different amounts depending on how each agent's
strategy filtered the hands into it. The right input is the public history
together with one probability vector per agent over that agent's possible
private states, which is what $\beta_t$ is. Those vectors are computed by
Bayes' rule from the strategies, so they exist only once the strategies are
fixed, which is why the construction is a self-consistency condition and not a
preprocessing step. Once the input is $\beta_t$ the game has proper subgames
again, depth-limited search becomes legitimate, and the AlphaZero recipe
transfers. \Cref{ex:pbs-poker} computes a $\beta_t$ with three numbers.
\end{rlreading}

\begin{remark}[Public belief state reduction]
\label{thm:pbs}
Assume the rules of the game and the agents' policies are common knowledge. Then
a two-player zero-sum game with imperfect information and public observations
induces a continuous-state perfect-information game whose states are public
belief states and whose actions are behavioural prescriptions
$\mathcal{S}_i \to \Simplex{\Ac_i}$. The two representations are strategically
identical, so the value of the induced game at $\beta_0$ equals the value of the
original and an equilibrium of the induced game induces an equilibrium of the
original. The construction is the specialisation of \cref{thm:coordinator} to
the zero-sum case and is stated informally in this form by
\citet{brown2020rebel}, who credit \citet{nayyar2013decentralized} and
\citet{dibangoye2016optimally}.
\end{remark}

\begin{example}[The public belief state of a poker hand]
\label{ex:pbs-poker}
In Kuhn poker, treated in full in \cref{ex:kuhn}, each agent holds one of three
cards and the betting is public. After the first agent checks, the public state
is the single word ``check'' and the public belief state consists of that word
together with two distributions over the three cards, one per agent, computed by
Bayes' rule from the agents' known strategies. If the first agent's strategy is
to check with probability one on the jack and the queen and with probability
$2/3$ on the king, then after a check the second agent's belief about the first
agent's card is proportional to $(1, 1, 2/3)$, that is
$(0.375, 0.375, 0.25)$. That triple, not the actual card, is the state on which
a depth-limited search may legitimately restart, which is the content of
\cref{thm:pbs}. The reinforcement learning reading is exact and practical: a
value network for an imperfect-information game must take the belief pair as
input, because a network that takes only the public history is evaluating a
position whose value is not well defined, and a network that takes the private
card is evaluating a position the opponent cannot condition on. Every design
error in this area is one of those two. \exend
\end{example}

\Cref{thm:pbs} is \cref{thm:coordinator} with the team objective replaced by a
zero-sum objective and the single coordinator replaced by two coordinators
playing a perfect-information game, and the reason the public belief state is
well defined is the same reason the common-information belief is: both are the
posterior conditional on the common knowledge. The practical consequence is that
the techniques of perfect-information search, specifically the combination of a
learned value function with depth-limited search that underlies AlphaZero
\citep{silver2018general}, become applicable to imperfect-information games once
the state is taken to be a public belief state rather than a world state, which
is the design of ReBeL \citep{brown2020rebel} and of Student of Games
\citep{schmid2023student}. On the cooperative side the identical construction,
with prescriptions computed by a neural network conditioned on the public
belief, is the Bayesian action decoder \citep{foerster2019bayesian}, and the
observation that these two lines of work implement the same theorem under
different reward structures is, in our view, the most useful single item in this
monograph. The same construction, applied to the subsets of agents for whom a
fact is common knowledge, is the multi-agent common-knowledge method of
\citet{dewitt2019multi}, which chooses dynamically how large a group to
coordinate as a function of how much those agents commonly know.

A third reduction of a different flavour converts a decentralised planning
problem into a sequence of one-shot games, which is the bridge between
\cref{sec:partial-obs} and the Bayesian games of \cref{sec:bayesian}.

\begin{proposition}[Stage-wise Bayesian game reduction; \citealp{emery2004approximate,oliehoek2008optimal}]
\label{prop:bayesian-stage}
Let $G$ be a finite-horizon decentralised partially observable Markov decision
process and fix the joint decision rules $\delta_{0:t-1}$, which determine
$\xi_t$. Then the choice of $\delta_t$ that maximises the remaining expected
return is a Bayes-Nash equilibrium, indeed an optimal profile, of the
common-payoff Bayesian game $\mathcal{BG}_t$ whose agent set is $\agents$, whose
type space for agent $i$ is $\Hi_t^i$, whose type distribution is the marginal
of $\xi_t$ on $\Hi_t$, and whose payoff is the immediate reward plus the optimal
continuation value. Solving $G$ exactly amounts to solving the sequence
$\mathcal{BG}_0, \mathcal{BG}_1, \dots$ with exact continuation values, and
replacing the continuation value by a heuristic yields the approximate methods
of \citet{emery2004approximate}.
\end{proposition}

The statement requires the past joint decision rules to be the optimal ones, and
$Q^\star$ to be the action-value under the optimal joint policy rather than a
quantity defined by a maximisation, which is how
\citet{oliehoek2008optimal} state it.

\Cref{prop:bayesian-stage} is the precise sense in which an agent's
action-observation history plays the role of a Harsanyi type. The type of an
agent is the private information that determines its beliefs about the others,
and in a decentralised partially observable Markov decision process at time $t$
that is exactly $h_t^i$. This identification is used again in
\cref{sec:extensive}, where the information set of an extensive-form game is
shown to be the same object once more, and in \cref{sec:bayesian}, where it
drives the correspondence between the Harsanyi transformation and the
Bayes-adaptive Markov decision process.

\begin{table}[t]
\centering
\caption{Exact reductions of multi-agent models to single-agent ones. Each
reduction preserves optimal value and maps solutions bijectively; each enlarges
the state space to a space of distributions and the action space to a space of
functions. The last column names the deep reinforcement learning methods that
implement the reduction approximately.}
\label{tab:reductions}
\small
\begin{tabular}{@{}L{2.5cm} L{2.6cm} L{2.5cm} L{2.0cm} L{3.0cm}@{}}
\toprule
\textsc{Source model} & \textsc{Reduced state} & \textsc{Reduced action} & \textsc{Target} & \textsc{Implemented by} \\
\midrule
POMDP & Belief $b_t \in \Simplex{\St}$ & Action $a_t$ & MDP & Recurrent and belief-conditioned policies \\
Partial history sharing & Common-information belief $\Pi_t$ & Prescription $\gamma_t$ & POMDP & Coordinator architectures \\
Dec-POMDP & Occupancy state $\xi_t$ & Joint decision rule $\delta_t$ & Deterministic MDP & Point-based oMDP solvers; BAD \\
2p zero-sum imperfect info & Public belief state $\beta_t$ & Policy prescription & Perfect-info game & ReBeL; Student of Games \\
Dec-POMDP, group common knowledge & Common knowledge of a subset & Prescription for that subset & Hierarchical POMDP & MACKRL \citep{dewitt2019multi} \\
Dec-POMDP at stage $t$ & Type distribution from $\xi_t$ & Joint type-action map & Bayesian game & Bayesian-game approximation \\
Latent MDP & Belief over $(s,\theta)$ & Action $a_t$ & BAMDP & Meta-RL with context encoders \\
Fixed-opponent SG & State $s_t$ & Action $a_i$ & MDP & Any single-agent method \\
Finite-automaton opponents & $(s_t, m_{-i})$ & Action $a_i$ & MDP & Opponent modelling \\
\bottomrule
\end{tabular}
\end{table}

\section{Graphical, factored and congestion games: structure in the payoff tensor}
\label{sec:graphical}

The models of the previous sections are specified by a payoff tensor with one
entry per joint action, so their description length is exponential in the number
of agents. Every practical multi-agent system avoids this, because real
interactions are local: a vehicle is affected by its neighbours and not by
traffic on another continent, a unit in a battle is affected by the units in its
range, a machine in a plant is affected by its immediate upstream and downstream
machines. Game theory formalises locality by placing the agents on a graph and
requiring each agent's payoff to depend only on its neighbours, which produces
graphical and polymatrix games; reinforcement learning formalises the same
locality as a factored Markov decision process and as a coordination graph. The
two constructions are the same, and the correspondence is the one that converts
an exponential maximisation into a message-passing computation.

The two readings differ in what the graph is for. A game theorist draws the
graph to make a representation succinct and then asks whether the succinctness
survives into the solution problem, which is the question whose answer is in
\cref{sec:complexity}. A reinforcement learning researcher draws the graph to
make an argmax tractable and then asks whether the restriction loses the optimum,
which is the question that \cref{sec:potential} answers through the
Individual-Global-Max condition. Both are asking what structure buys, and the
answers differ: succinct representation does not imply easy equilibrium, and
tractable argmax does imply restricted expressiveness.

\begin{nomen}{Structured-game vocabulary read from both sides.}{tab:nom-graphical}
Graphical game & Coordination graph; factored MARL &
  Payoffs depend on neighbours only \\
Neighbourhood $N(i)$ & Scope of a local payoff factor &
  Who can affect agent $i$ directly \\
Local payoff $\util_i(a_i, a_{N(i)})$ & Local utility term $f_e(a_e)$ &
  One summand of the joint value \\
Polymatrix game & Pairwise-decomposable joint reward &
  Sum of two-agent games on edges \\
Treewidth of the interaction graph & Treewidth of the coordination graph &
  What makes elimination cheap \\
Variable elimination & Max-plus, message passing, DCOP &
  Compute the argmax without enumerating \\
Congestion game & Resource-sharing MARL environment &
  Cost depends on how many choose it \\
Rosenthal potential & Potential of a potential game &
  A scalar whose local differences match payoffs \\
Wardrop equilibrium & Mean-field equilibrium of a routing game &
  Nonatomic agents equalise costs \\
Price of anarchy & Gap from selfish learners to a planner &
  What locality plus selfishness costs \\
Anonymous game & Permutation-invariant policy; mean-field &
  Only the count of each action matters \\
\end{nomen}

\begin{definition}[Graphical games, \citealp{kearns2001graphical}; polymatrix
games]
\label{def:graphical}
A graphical game is a normal-form game together with an undirected graph on
$\agents$ such that $\util_i(a)$ depends only on $a_i$ and on $a_{N(i)}$, where
$N(i)$ is the set of neighbours of $i$. It is a polymatrix game when in addition
the payoff decomposes as a sum over edges,
$\util_i(a) = \sum_{j \in N(i)} \util_{ij}(a_i, a_j)$, so that the game is
determined by one two-agent matrix per edge. The graphical class is the one
introduced by \citet{kearns2001graphical}; the polymatrix class is older and is
not theirs.
\end{definition}

\begin{rlreading}
A graphical game is a factored reward model and its reinforcement learning name
is a coordination graph. The neighbourhood $N(i)$ is the scope of the factors
that touch agent $i$, so an implementation stores one small tensor per
hyperedge instead of one enormous tensor over $\Ac$. Two savings follow and
they are independent. The first is memory: the model shrinks from
$\abs{\Ac}^{\nag}$ entries per agent to $\abs{\Ac}^{d+1}$, where $d$ is the
maximum degree. The second is compute: the joint $\argmax$ that a greedy policy
needs no longer requires enumerating $\prod_i \Ac_i$ and becomes message
passing, which is \cref{prop:coordination-graph}. The polymatrix case, in which
every factor has exactly two agents in scope, is the one in which a pairwise
value decomposition is exactly right rather than an approximation, and
\cref{ex:coordination-graph} and \cref{ex:polymatrix} quantify both savings.
\end{rlreading}

The description length of a graphical game is $\nag \cdot \abs{\Ac}^{d+1}$ where
$d$ is the maximum degree, against $\nag \cdot \abs{\Ac}^{\nag}$ for the general
case, and for a polymatrix game it is $\abs{E} \cdot \abs{\Ac}^2$. This is the
first of the two savings. The second is computational and is the one a
reinforcement learning practitioner uses every time a factored argmax is
required.

\begin{proposition}[Coordination graphs and variable elimination]
\label{prop:coordination-graph}
Let $Q_{\mathrm{tot}}(s, a) = \sum_{e \in E} f_e(s, a_e)$ be a joint
action-value that decomposes over the hyperedges of a graph with scopes $a_e$.
Then $\max_a Q_{\mathrm{tot}}(s,a)$ and a maximising joint action are computable
by variable elimination in time exponential in the induced width of the
elimination ordering and linear in $\nag$, rather than exponential in $\nag$.
When the graph is a tree the induced width is one and the computation is two
sweeps of message passing. This is the construction of
\citet{guestrin2001multiagent}. The max-plus algorithm of
\citet{kok2006collaborative} is the loopy counterpart: it is exact on trees,
where it coincides with the above, and on a graph with cycles it returns an
approximate maximiser rather than the exact one, which is the trade it makes for
being anytime and fully distributed. Its deep instance is the coordination graph
network of \citet{bohmer2020deep}.
\end{proposition}

\begin{example}[What the graph saves, in arithmetic]
\label{ex:coordination-graph}
Take $\nag$ agents with five actions each, arranged in a chain so that agent $i$
interacts only with agents $i-1$ and $i+1$. The joint action space has $5^\nag$
elements: $625$ for four agents, $390{,}625$ for eight, and about
$9.5 \times 10^{13}$ for twenty. Variable elimination along the chain passes one
message per edge, each of which is a table of $5 \times 5 = 25$ entries, for a
total of $(\nag - 1) \cdot 25$ operations: $75$, $175$ and $475$ respectively.
The saving is from exponential to linear in $\nag$ and it is exact, not
approximate, provided the value really does decompose over the chain. The
qualifier is the whole of the difficulty: a value function trained end to end
does not decompose unless the architecture forces it to, and forcing it
restricts the hypothesis class, which is the trade-off that
\cref{prop:monotone-igm} quantifies. A chain is the easy case; for a general
graph the cost of elimination is exponential in the induced width, so a densely
connected set of agents recovers the original difficulty, correctly, because a
densely connected set of agents really is a hard coordination problem. \exend
\end{example}

\begin{example}[Zero-sum on every edge, not zero-sum overall]
\label{ex:polymatrix}
Place three agents on a triangle and put rock-paper-scissors on each edge, so
that $\util_i(a) = \sum_{j \ne i} \util_{\mathrm{RPS}}(a_i, a_j)$. Each pairwise
game has value $0$ and the unique pairwise equilibrium is uniform. Every edge of this game is
itself zero-sum, so it is a pairwise constant-sum polymatrix game, which is the
easy subclass; the payoffs also sum to zero at every joint action, but here that
is a consequence of the edges rather than an independent property. The substance
of \citet{cai2011minmax} is that the two coincide: a game whose payoffs sum to a
constant globally, without any pair being in a constant-sum relation in
isolation, is payoff-equivalent to a pairwise constant-sum polymatrix game, so
the easy subclass is the whole class. Its representation needs three matrices of
nine entries, that is $27$ numbers, against the $27 \times 3 = 81$ numbers of
the explicit three-agent tensor. The structural point, which is the content of
\citet{cai2011minmax}, is that zero-sum polymatrix games retain the good
properties of two-player zero-sum games: the minmax theorem generalises to them,
their equilibria form a convex set, those equilibria are computable in
polynomial time, and no-regret learning converges to them. The same source notes
that unrestricted three-player zero-sum games are already \PPAD-complete, so the
pairwise-separable structure is essentially the furthest that tractability can
be pushed in the many-agent constant-sum setting. Recognising that an
application falls in this class is worth more than any algorithmic
refinement. \exend
\end{example}

Congestion games are the most important structured class in applications,
because resource sharing is the commonest form of multi-agent interaction in
engineering.

\begin{definition}[Congestion game; \citealp{rosenthal1973class}]
\label{def:congestion}
Let $\mathcal{E}$ be a finite set of resources, let each agent $i$ choose a
subset $a_i \subseteq \mathcal{E}$ from an allowed family, and let
$c_e : \Nat \to \Real$ be a cost function for resource $e$. The cost to agent
$i$ is $\sum_{e \in a_i} c_e(n_e(a))$, where $n_e(a)$ is the number of agents
using $e$. Every congestion game is an exact potential game in the sense of
\cref{def:potential-game}, with Rosenthal's potential
$\pot(a) = \sum_{e \in \mathcal{E}} \sum_{k=1}^{n_e(a)} c_e(k)$.
\end{definition}

\begin{rlreading}
A congestion game is the formal model of every multi-agent environment in which
agents compete for a shared resource whose quality degrades with use:
bandwidth, road capacity, a charging station, a processing queue, a piece of
map. The defining restriction is that an agent's cost depends on how many others
chose each resource and not on which others did, so the reward function is
anonymous in the sense used at the end of this section. The payoff of the
restriction is large and is stated by \cref{def:congestion} itself: every such
game has an exact potential, so independent learners that only ever take
improving steps are hill-climbing a single scalar and must stop. That is the
strongest convergence guarantee available anywhere in this monograph to agents
that do not coordinate. What it does not guarantee is that the resting point is
good, and \cref{ex:pigou} measures the shortfall.
\end{rlreading}

The potential guarantees that better-reply dynamics terminate at a pure
equilibrium, which is the strongest convergence statement available to
independent learners and is developed in \cref{sec:potential}. What the
potential does not guarantee is that the equilibrium is good, and the gap
between the equilibrium and the social optimum is the price of anarchy, treated
in full in \cref{sec:mechanism}. The two smallest instances are worth having
explicitly because they are the examples every practitioner should be able to
reconstruct.

\begin{example}[Pigou's network and a price of anarchy of four thirds]
\label{ex:pigou}
One unit of traffic travels from a source to a sink along two parallel links.
The upper link has constant cost $1$ regardless of load; the lower link has cost
equal to the fraction $x$ of traffic using it. Selfish traffic all moves to the
lower link, because for any $x < 1$ the lower link is cheaper, so the
equilibrium flow is $x = 1$ and every unit pays $1$, for a total cost of $1$.
A planner minimising the total cost $x \cdot x + (1-x) \cdot 1$ differentiates
to get $2x - 1 = 0$, so $x = 1/2$ and the total cost is
$\tfrac14 + \tfrac12 = 0.75$. The price of anarchy is $1/0.75 = 4/3$. The
reinforcement learning reading is that a population of independently trained
agents sharing a resource will overload it by exactly this factor, that no
amount of better optimisation fixes it because the equilibrium is exactly what
each agent wants, and that the two remedies are to change the reward, which is
the tolling construction of \cref{sec:mechanism}, or to centralise the decision,
which is the planner's solution. \exend
\end{example}

\begin{example}[Braess's paradox \citep{braess2005paradox}: adding capacity
makes everybody worse off]
\label{ex:braess}
A unit of traffic travels from $s$ to $t$. There are two routes, $s \to A \to t$
with costs $x$ then $1$, and $s \to B \to t$ with costs $1$ then $x$, where $x$
is the flow on the edge in question. By symmetry the equilibrium splits the
traffic in half and every unit pays $\tfrac12 + 1 = 1.5$. Now add a free link
from $A$ to $B$. The route $s \to A \to B \to t$ costs $x_{sA} + 0 + x_{Bt}$,
and taking it is a dominant choice for every unit, so all traffic uses it and
every unit pays $1 + 0 + 1 = 2$. Adding capacity raised the equilibrium cost by
a factor $4/3$ while the social optimum is unchanged. The lesson transfers
verbatim to multi-agent reinforcement learning: enlarging the action space of
selfish learners can lower the equilibrium welfare, so an ablation showing that
a richer action space hurts is not necessarily an optimisation failure and may
be the correct equilibrium of a different game. Diagnosing which requires
comparing against a centralised planner on the same model, which is the tenth
row of \cref{tab:diagnostics}. \exend
\end{example}

The last structured class worth naming here is the anonymous game, in which each
agent's payoff depends on its own action and on the number of other agents
playing each action, but not on their identities. Anonymity is the finite-$\nag$
version of the exchangeability assumption that produces the mean-field limit of
\cref{sec:mean-field}, and in reinforcement learning it is the structural
property that licenses a permutation-invariant architecture, in which the
other agents are summarised by a pooled embedding rather than concatenated. The
correspondence is exact: a permutation-invariant critic is well specified
precisely when the underlying game is anonymous, and using one when the game is
not anonymous introduces a representational error of the same kind as using a
monotone mixing network on a game that violates the Individual-Global-Max
condition.

\begin{table}[t]
\centering
\caption{Structured game classes, the reinforcement learning architecture that
exploits the structure, and what the structure buys. The last column records
whether the saving is in representation, in computation of the argmax, or in the
solution concept itself.}
\label{tab:graphical}
\small
\begin{tabular}{@{}L{2.7cm} L{3.3cm} L{3.4cm} L{3.0cm}@{}}
\toprule
\textsc{Class} & \textsc{Structure assumed} & \textsc{RL counterpart} & \textsc{What it buys} \\
\midrule
General normal form & None & Joint-action critic & Nothing; $\abs{\Ac}^\nag$ entries \\
Graphical game & Payoffs depend on $N(i)$ & Scoped critic factors & Representation: $\abs{\Ac}^{d+1}$ per agent \\
Polymatrix game & Payoffs sum over edges & Pairwise value decomposition & Representation; if zero-sum, a minmax theorem and polynomial-time equilibria \citep{cai2011minmax} \\
Coordination graph & Value sums over hyperedges & Max-plus, DCG & Argmax by elimination (\cref{prop:coordination-graph}) \\
Congestion game & Cost depends on load only & Shared-resource environment & Pure equilibrium exists; potential \citep{rosenthal1973class} \\
Anonymous game & Payoff depends on action counts & Permutation-invariant critic & Representation linear in $\nag$ \\
Mean-field limit & Exchangeable, $\nag \to \infty$ & Representative-agent training & Cost independent of $\nag$ (\cref{sec:mean-field}) \\
\bottomrule
\end{tabular}
\end{table}

\section{Extensive-form games, perfect recall, and information sets as histories}
\label{sec:extensive}

The extensive form represents a game as a tree in which nodes are decision
points, edges are actions, and the points that a player cannot distinguish are
grouped into information sets. It is the representation in which the entire
literature on imperfect-information game solving is written, including
counterfactual regret minimisation and the superhuman poker systems, and it is
the representation that reinforcement learning practitioners most often find
alien. The purpose of this section is to show that it is the familiar partially
observable sequential decision problem in different coordinates.

The two readings are the two ways of writing down the same trajectory
distribution. A game theorist writes the tree and groups its nodes into
information sets, so that a strategy is a labelling of information sets by
distributions over actions, and the quantity to be computed is a realisation
probability for each terminal node. A reinforcement learning researcher writes
a state, an observation kernel and a policy, so that a policy is a map from
histories to distributions over actions, and the quantity to be computed is a
visitation measure. Information set and action-observation history are the same
equivalence class of play, behavioural strategy and policy are the same object,
and realisation plan and occupancy measure are the same vector. The one real
difference is that the tree makes the recursive structure of regret explicit
while the state makes bootstrapping explicit, and the two algorithm families
follow the representation rather than the problem.

\begin{nomen}{Extensive-form vocabulary read from both sides.}{tab:nom-efg}
Node $h$ of the tree & Full history; state-trajectory prefix &
  A position including hidden information \\
Information set $I$ & Action-observation history $h^i_t$ &
  What the acting agent can distinguish \\
Chance node & Transition or observation kernel &
  Nature's move with a known distribution \\
Behavioural strategy $\pi_i(I)$ & Policy $\pi_i(\cdot \given h^i_t)$ &
  Local randomisation at each decision \\
Mixed strategy & Distribution over deterministic policies &
  Randomise once, over whole plans \\
Perfect recall & The policy sees the full own history &
  You never forget what you did \\
Realisation plan & Occupancy measure; visitation distribution &
  Probability of reaching each sequence \\
Sequence form & Bellman flow polytope &
  Linear constraints instead of a simplex \\
Counterfactual reach $\pi^{-i}(h)$ & Discounted visitation of others &
  How often the rest of the world sends you here \\
Counterfactual value & Advantage weighted by reach &
  Local contribution to global regret \\
Subgame & Sub-MDP rooted at a state &
  A part solvable on its own \\
Sequential equilibrium & On-policy belief filter plus greedy play &
  Beliefs consistent with the policy \\
Extensive-form correlated equilibrium & Shared seed revealed as play unfolds &
  A mediator who recommends step by step \\
\end{nomen}

\begin{definition}[Extensive-form game]
\label{def:efg}
An extensive-form game consists of a finite rooted tree with node set
$\mathcal{V}$, terminal nodes $Z \subseteq \mathcal{V}$, a player function
assigning each non-terminal
node to an agent in $\agents \cup \{c\}$ where $c$ is chance, a probability
distribution over the edges at each chance node, payoff functions
$\util_i : Z \to \Real$, and for each agent $i$ a partition $\mathcal{I}_i$ of
the nodes assigned to $i$ into information sets, such that all nodes in an
information set $I$ have the same set of available actions $\Ac(I)$. A
behavioural strategy of agent $i$ is a map
$\pi_i : \mathcal{I}_i \to \Simplex{\Ac(\cdot)}$ assigning a distribution over
actions to each of its information sets. A mixed strategy is a distribution over
pure strategies, a pure strategy being a deterministic choice at every
information set of $i$.
\end{definition}

\begin{rlreading}
Unroll a partially observable environment into the tree of all its possible
trajectories and \cref{def:efg} is what you have written down. A node is a full
history including the hidden state; a chance node is one call to the transition
or observation kernel; an information set is the set of nodes the acting agent
cannot tell apart, which is exactly the set of nodes sharing an
action-observation history, so an information set is a recurrent hidden state.
A behavioural strategy is a policy: one distribution over actions per
information set, which is one forward pass per observed history. A mixed
strategy is something a policy gradient method never manipulates and a
population-based method always does, namely a distribution over whole
deterministic policies. \Cref{prop:infoset-aoh} states the correspondence
formally and \cref{thm:kuhn} says when the two strategy notions are
interchangeable.
\end{rlreading}

\begin{example}[The smallest tree, with and without an observation]
\label{ex:efg-tuple}
Chance flips a fair coin; the agent then picks $L$ or $R$; it scores $1$ if its
pick matches the coin and $0$ otherwise. The tree has one chance node at the
root, two decision nodes below it and four terminal nodes.

Write it as \cref{def:efg} twice. In the blind version the agent does not see
the coin, so both decision nodes lie in one information set, the agent has a
single behavioural strategy parameter $p = \pi(L)$, and its expected payoff is
$0.5 p + 0.5 (1-p) = 0.5$ for every $p$. The whole strategy space is worth the
same, so there is nothing to learn. In the sighted version the agent sees the
coin, the two nodes lie in two singleton information sets, the agent has two
parameters, and matching the coin at each gives an expected payoff of $1$.

The difference, $1 - 0.5 = 0.5$, is the value of information, and it is the same
$0.5$ that \cref{ex:cheap-talk} obtains by letting one agent tell another what
the state is. That is not a coincidence: a communication channel is a device for
splitting an information set, and splitting information sets is the only thing
that raises the value of a one-agent tree. Read as reinforcement learning, the
blind version is a feedforward policy on an unobservable state and the sighted
version is the same policy with the state appended to its input, and the
$0.5$ is the return gap that no optimiser can close on the first. \exend
\end{example}

\begin{definition}[Perfect recall]
\label{def:perfect-recall}
Agent $i$ has perfect recall when, for every information set $I \in
\mathcal{I}_i$, all nodes in $I$ induce the same sequence of $i$'s own past
information sets and actions. A game has perfect recall when every agent does.
\end{definition}

\begin{rlreading}
Perfect recall is a property of the architecture, not of the task. A policy that
conditions on its full action-observation history has it by construction, so any
recurrent network with enough capacity, or any transformer whose context window
covers the episode, is a perfect-recall agent. A feedforward policy on the
current observation does not have it, because two nodes that differ only in what
the agent itself did earlier land in the same information set. A truncated
context window does not have it either, for the same reason, and the truncation
point is exactly where recall is lost. \Cref{thm:kuhn} says what is at stake:
under perfect recall, optimising a policy loses nothing relative to optimising a
distribution over policies, and without it the policy class is strictly weaker,
so local randomisation becomes necessary. \Cref{ex:perfect-recall} exhibits the
failure on a four-cell corridor.
\end{rlreading}

\begin{theorem}[\citealp{kuhn1953extensive}]
\label{thm:kuhn}
In an extensive-form game in which agent $i$ has perfect recall, for every mixed
strategy of agent $i$ there is a behavioural strategy that induces the same
distribution over terminal nodes against every profile of the other agents, and
conversely. Without perfect recall the equivalence fails in one direction only: a
behavioural strategy always admits an outcome-equivalent mixed strategy, while a
mixed strategy need not admit an outcome-equivalent behavioural one, and perfect
recall is a sufficient condition for it to do so. The one-directional
statement presumes that no information set is met twice along a play, that is,
that the agent is not absent-minded.
\end{theorem}

\Cref{thm:kuhn} is the reason reinforcement learning can be applied to
extensive-form games at all. A behavioural strategy is precisely a policy, a
map from what the agent currently knows to a distribution over actions, and a
mixed strategy is a distribution over deterministic policies, which is the
object a population-based or evolutionary method manipulates. Perfect recall is
the condition under which optimising a policy is without loss relative to
optimising a distribution over policies, and it is assumed, usually silently,
whenever a recurrent policy is trained on a partially observable task. The
failure mode is concrete: a feedforward policy on a partially observable task
is an imperfect-recall agent, the equivalence of \cref{thm:kuhn} does not apply
to it, and the optimal memoryless policy may require randomisation and may be
strictly worse than the optimal history-dependent one
\citep{singh1994learning}.

\begin{example}[Perfect recall is a property of the architecture, not the task]
\label{ex:perfect-recall}
Consider a corridor of four cells in which the agent observes only whether a
wall is immediately to its left and to its right, so that the two middle cells
produce the same observation. A policy with memory distinguishes them by
counting steps and has perfect recall; a feedforward policy on the current
observation does not, because the two middle cells lie in one information set
even though the agent's own past actions differ between them. For the
memoryless agent the optimal policy may require randomisation, since moving
right deterministically in the ambiguous observation can trap it and moving left
deterministically can too, and \citet{singh1994learning} show exactly this. The
game-theoretic statement and the reinforcement learning statement are the same
statement: dropping perfect recall breaks \cref{thm:kuhn}, so the behavioural
policy class is strictly weaker than the mixed one, and the gap is filled by
local randomisation. The practical corollary is that the choice between a
feedforward and a recurrent policy is a choice of game, not of function
approximator. \exend
\end{example}

The correspondence with partial observability is now immediate, and we state it
in the direction that is useful.

\begin{proposition}[Information sets are action-observation histories]
\label{prop:infoset-aoh}
Let $P$ be a finite-horizon partially observable Markov decision process.
Unrolling $P$ into its tree of state-action-observation trajectories, with
chance nodes implementing $b_0$, $\Tr$ and $\Om$, yields a one-player
extensive-form game with perfect recall in which two nodes lie in the same
information set if and only if they share the same action-observation history
$h_t$. Conversely, let $E$ be a one-player extensive-form game with perfect
recall; then $E$ is a finite-horizon partially observable Markov decision
process whose states are the nodes of the tree, whose observations are the
information sets, and whose transition kernel is given by the chance
distributions. Under this correspondence, behavioural strategies are policies
$\pi : \Hi \to \Simplex{\Ac}$, and the expected payoff of a behavioural strategy
equals the expected return of the corresponding policy.
\end{proposition}

For $\nag \ge 2$ the same unrolling relates extensive-form games to partially
observable stochastic games, with one caveat that is a genuine difference and
not a notational one. A partially observable stochastic game as defined in
\cref{def:posg} has simultaneous moves at every step, whereas an extensive-form
game has a single acting agent at each node. Each representation can simulate
the other. Simultaneous moves are encoded in the extensive form by letting the
agents move in an arbitrary fixed order while placing all nodes that differ only
in the earlier agents' unobserved choices into a single information set; turn
taking is encoded in a partially observable stochastic game by including the
identity of the agent to move in the state and giving the others a no-op action.
The encodings are value-preserving, so the model classes are equivalent in
expressive power, and they differ in which quantities are natural to compute.
The extensive form makes the tree structure and therefore regret decomposition
explicit; the partially observable stochastic game makes the state and therefore
bootstrapping explicit. Algorithms follow the representation: counterfactual
regret minimisation traverses the tree, temporal-difference methods bootstrap on
the state.

The smallest poker game is the right place to make all of this concrete, and we
work it out completely because every quantity in it can be checked by hand or by
a short program.

\begin{example}[Kuhn poker, counted and solved]
\label{ex:kuhn}
Three cards, a jack, a queen and a king, are shuffled and one is dealt to each
of two agents; the third is unseen. Each agent antes one chip. The first agent
checks or bets one chip. If it checks, the second agent checks, ending the hand
with a showdown for the antes, or bets, after which the first agent folds,
conceding the antes, or calls, producing a showdown for two chips each. If the
first agent bets, the second folds or calls.

Count the objects. The first agent has three information sets at the root, one
per card, and three more after a check followed by a bet, so six in all; the
second agent has three after a check and three after a bet, so six; twelve in
total. A pure strategy of either agent picks one of two actions at each of its
six information sets, so each has $2^6 = 64$ pure strategies and the normal form
is a $64 \times 64$ matrix. The sequence form instead lists the realisation
weight of each sequence of own actions: the empty sequence plus four per card
for each agent, so thirteen variables per agent. The reduction from $64$ to
$13$ is the whole point of the sequence form, and the same reduction in a real
poker variant is the difference between an impossible and a routine linear
program.

Solve it. The sequence-form linear program returns a value of $-1/18 \approx
-0.0556$ chips per hand to the first agent, so the game favours the second
agent slightly. Equilibrium play is not unique: there is a one-parameter family
in which the first agent bluffs with the jack at a rate that is tied, by the
indifference conditions, to the rate at which it bets the king.

Read it as reinforcement learning. The state is the pair of cards plus the
betting history, the observation is the agent's own card plus the public betting
history, the information set is exactly the action-observation history, and a
behavioural strategy is a policy conditioned on that history. The public belief
state of \cref{def:pbs} is the betting history together with the two posterior
distributions over cards, which is \cref{ex:pbs-poker}. The value $-1/18$ is the
number against which exploitability is measured: a trained agent is evaluated
not by its return against its training partner but by the gap between $-1/18$
and its return against a best response, which is computable exactly here by
solving a one-player Markov decision process against the frozen policy. \exend
\end{example}

Three further correspondences complete the picture. The first concerns
representation size. Every extensive-form game has an equivalent normal form
whose action sets are the sets of pure strategies, and these are exponentially
large in the number of information sets, so the normal form is not a usable
computational representation. \citet{koller1992complexity} show that the equilibrium of a two-player zero-sum
extensive-form game with perfect recall is computable in time polynomial in the
tree, by a linear program with exponentially many constraints solved through the
ellipsoid method. The sequence form, arrived at independently by several authors and developed
for computation by \citet{vonstengel1996efficient}, replaces
the simplex of mixed strategies by the polytope of realisation plans, whose
dimension is linear in the number of sequences, and thereby reduces the same
problem to a linear program whose size is linear in the tree;
\citet{koller1996efficient} extend the device to the general-sum case, where the
equilibrium conditions become a linear complementarity problem solved by Lemke's
algorithm rather than a linear program. The reinforcement learning reading is
that in the one-player case the realisation plan is literally the occupancy
measure, the state-action visitation distribution whose use converts policy
optimisation into a linear program over the polytope of admissible occupancies,
and the Bellman flow constraints are the sequence-form constraints. With two
players the construction becomes a saddle point over a pair of such polytopes,
so the correspondence is structural rather than an identity.

The second concerns the solution concepts. Nash equilibrium applied to the
normal form of an extensive-form game permits incredible threats, and the
refinement that forbids them is subgame perfection, which under perfect
information is computed by backward induction and is therefore literally dynamic
programming. Under imperfect information proper subgames are rare, since a
subgame must begin at a node whose information set is a singleton and must
contain entire information sets, and the refinements that replace subgame
perfection are sequential equilibrium \citep{kreps1982sequential} and the
perfect Bayesian equilibrium,
both of which require a belief system assigning a distribution over nodes within
each information set together with a consistency requirement linking beliefs to
strategies. The reinforcement learning counterpart of a belief system is the
belief state of \cref{thm:belief-mdp}, and the consistency requirement is the
statement that the belief is the Bayes posterior under the policy being
evaluated, which is exactly what an on-policy belief filter computes and what an
off-policy one does not. This is the formal reason why belief-conditioned
policies trained off-policy can be inconsistent in a way that has no single-agent
analogue, and it is the problem that off-belief learning \citep{hu2021off} is
designed to remove by fixing the belief to one induced by a reference policy.

Correlation in the sequential setting needs its own concept, because a mediator
who must reveal recommendations as play unfolds is weaker than one who reveals
everything at the start. The extensive-form correlated equilibrium is defined by
requiring that no agent gain by deviating at any information set, given the
recommendation received there and the posterior it induces, and the set it
defines sits between the Nash and the normal-form correlated equilibria in a
way that is specific to the tree. For reinforcement learning the relevant
reading is that a shared random seed sampled once before an episode implements
a normal-form correlated equilibrium, while a shared signal broadcast at each
step implements the extensive-form notion, and these are different objects with
different attainability properties; a system whose agents share a seed at reset
and nothing thereafter is in the first category.

The third concerns regret. Counterfactual regret minimisation
\citep{zinkevich2007regret} decomposes the overall external regret of a
behavioural strategy into a sum of local regrets at information sets, weighted
by counterfactual reach probabilities, and minimising each local regret with a
no-regret rule such as regret matching drives the overall regret to zero. In a
two-player zero-sum game the average strategy therefore converges to a Nash
equilibrium at rate $O(1/\sqrt{T})$. The decomposition theorem is the
game-theoretic analogue of the performance difference lemma of single-agent
reinforcement learning \citep{kakade2002approximately}, which likewise expresses
a global value difference as a discounted sum of local advantages weighted by
visitation probabilities, and the counterfactual reach probability plays the
role of the discounted state visitation measure. Recognising this makes the
deep variants legible: Deep CFR \citep{brown2019deep} is approximate
counterfactual regret minimisation with function approximation of the regret
table, neural fictitious self-play \citep{heinrich2016deep} is approximate
fictitious play with a supervised average-policy network, and regularised Nash
dynamics \citep{perolat2021poincare,perolat2022mastering} replaces the regret
machinery with a continuous-time dynamics whose convergence comes from an
entropy regulariser, which is the same device that makes mirror descent converge
in single-agent policy optimisation.

\begin{example}[Regret matching, three rounds by hand]
\label{ex:regret-matching}
Regret matching plays each action with probability proportional to its
accumulated positive regret, and uniformly when all regrets are non-positive.
Take rock-paper-scissors against an opponent who always plays Rock, with
payoffs $+1$ for a win, $-1$ for a loss and $0$ for a draw.

Round one. All regrets are zero, so the agent plays uniformly. The payoffs of
the three actions against Rock are $(0, 1, -1)$ and the realised expected payoff
is $\tfrac13(0 + 1 - 1) = 0$, so the regrets become $(0, 1, -1)$.

Round two. The positive part of the regret vector is $(0, 1, 0)$, so the agent
plays Paper with probability one. Its payoff is $1$, and the regrets update by
$(0 - 1,\, 1 - 1,\, -1 - 1) = (-1, 0, -2)$, giving cumulative regrets
$(-1, 1, -3)$.

Round three. The positive part is again $(0, 1, 0)$, so the agent plays Paper.
After a thousand rounds the average strategy places weight $0.9993$ on Paper,
which is the best response, and the external regret per round has gone to zero.

Now the same rule in self play, both agents running regret matching from a
slightly perturbed start. The iterates cycle, as \cref{sec:evolutionary}
predicts, and the time averages converge: after twenty thousand rounds the
average strategy is $(0.334, 0.329, 0.337)$, close to uniform, and its
exploitability is $0.014$ and falling at the $O(1/\sqrt{T})$ rate. The two
halves of this example are the two halves of the practical advice: against a
fixed opponent the iterate is what you want, and in self play the average is,
and reporting the wrong one is the methodological error recorded in
\cref{sec:negative}. \exend
\end{example}

\begin{table}[t]
\centering
\caption{Extensive-form vocabulary and its reinforcement learning translation.}
\label{tab:efg-dictionary}
\small
\begin{tabular}{@{}L{3.6cm} L{5.4cm} L{3.9cm}@{}}
\toprule
\textsc{Extensive form} & \textsc{Reinforcement learning} & \textsc{Status} \\
\midrule
Node of the tree & State-history pair; trajectory prefix & Definitional \\
Information set $I \in \mathcal{I}_i$ & Action-observation history $h_t^i$ & Exact under perfect recall \\
Chance node & Stochastic transition or observation & Definitional \\
Behavioural strategy & Policy $\pi_i : \Hi_i \to \Simplex{\Ac_i}$ & Definitional \\
Mixed strategy & Distribution over deterministic policies & Equivalent by \cref{thm:kuhn} \\
Perfect recall & Policy conditions on full own history & Definitional \\
Imperfect recall & Memoryless or truncated-context policy & Definitional \\
Realisation plan; sequence form & Occupancy measure; Bellman flow polytope & Exact for one player \\
Subgame & Sub-MDP rooted at a state & Definitional \\
Subgame perfection & Bellman principle of optimality & Definitional \\
Backward induction & Value iteration on a tree & Definitional \\
Belief system of a sequential equilibrium & On-policy belief filter & Exact \\
Extensive-form correlated equilibrium & Signal broadcast at each step & Definitional \\
Normal-form correlated equilibrium & Shared seed sampled once at reset & Definitional \\
Counterfactual reach probability & Discounted state visitation measure & Exact \\
CFR regret decomposition & Performance difference lemma & Structural \\
Average strategy of CFR & Policy averaging; target network of NFSP & Exact in zero-sum \\
\bottomrule
\end{tabular}
\end{table}

\section{Incomplete information: Harsanyi types and Bayes-adaptive learning}
\label{sec:bayesian}

Partial observability concerns ignorance about the state. Incomplete information
concerns ignorance about the game: who the other agents are, what they want, and
what they know. Game theory handles it by the device of
\citet{harsanyi1967games}, which converts ignorance about the game into
ignorance about a state variable, and reinforcement learning handles it by the
device of Bayesian reinforcement learning, which converts ignorance about the
model into ignorance about a state variable. They are the same device.

The two readings are almost interchangeable once the device is applied, and the
only thing that differs is what the unknown is taken to be. A game theorist
calls the unknown a type and attaches it to an agent, so that uncertainty is
about preferences and information. A reinforcement learning researcher calls the
unknown a task parameter and attaches it to the environment, so that uncertainty
is about dynamics and rewards. Both then put a prior on it, fold it into the
state, and solve the enlarged problem. The one asymmetry is that with several
agents a type is something an agent may wish to conceal, and concealment has no
single-agent counterpart at all, which is where signalling begins.

\begin{nomen}{Incomplete-information vocabulary read from both sides.}{tab:nom-bayesian}
Type $\theta_i$ & Latent task parameter; opponent identity &
  Private information determining preferences \\
Type space $\Th_i$ & Task family; model class &
  The range of what the unknown can be \\
Common prior $p$ & Task distribution; meta-training distribution &
  The shared belief before anything is seen \\
Harsanyi transformation & State augmentation with the latent parameter &
  Make the unknown part of the state \\
Interim payoff & Value conditional on the posterior &
  After you learn your own type \\
Ex-ante payoff & Expected return over the task distribution &
  Before anything is drawn \\
Bayes-Nash equilibrium & Fixed point over type-conditioned policies &
  Best response at every type \\
Bayes-optimal policy & Solution of the BAMDP &
  The $\nag = 1$ case of the above \\
Value of information & Exploration bonus; information gain &
  Why a costly action can be worth taking \\
Signalling & No single-agent counterpart &
  Acting to change what others believe \\
Belief hierarchy & Nested opponent models &
  Beliefs about beliefs, without end \\
Common knowledge & The meet of the agents' partitions &
  What conditioning everyone can agree on \\
Zero-shot coordination & Generalisation to an unseen partner &
  Ex-ante payoff over a partner prior \\
\end{nomen}

\begin{definition}[Bayesian game]
\label{def:bayesian-game}
A Bayesian game is a tuple
$(\agents, \{\Ac_i\}, \{\Th_i\}, p, \{\util_i\})$ in which $\Th_i$ is the finite
type set of agent $i$, $p \in \Simplex{\prod_i \Th_i}$ is a common prior over
type profiles, and $\util_i : \Ac \times \Th \to \Real$ is the payoff of agent
$i$. Each agent observes its own type $\theta_i$ and nothing else, so a strategy
is a map $\pi_i : \Th_i \to \Simplex{\Ac_i}$.
\end{definition}

\begin{rlreading}
A Bayesian game is a task distribution. The type profile $\theta$ is the task
parameter drawn at reset, the common prior $p$ is the sampler that draws it, and
$\util_i(a, \theta)$ is a reward function that depends on that parameter. The
one feature with no single-agent counterpart is that the draw is revealed
asymmetrically: agent $i$ is told $\theta_i$ and nothing else, so the strategy
signature $\pi_i : \Th_i \to \Simplex{\Ac_i}$ is a policy conditioned on a
private context vector. In an implementation $\theta_i$ is whatever is in head
$i$'s observation at reset and not in anybody else's: its own capability, its
own goal, its own system prompt. Reading a multi-agent problem as a Bayesian
game is therefore the act of asking which parts of the reset state are private,
and the answer determines whether an agent can be trained against a fixed
partner or must be trained against a distribution of them.
\end{rlreading}

\begin{definition}[Bayes-Nash equilibrium]
\label{def:bne}
A profile $\pi^\star$ is a Bayes-Nash equilibrium when for every agent $i$ and
every type $\theta_i$ with $p(\theta_i) > 0$,
\begin{align}
\pi^\star_i(\cdot \given \theta_i) \in
\argmax_{\pi_i \in \Simplex{\Ac_i}} \;
\Exp_{\theta_{-i} \sim p(\cdot \given \theta_i)}\;
\Exp_{a_{-i} \sim \pi^\star_{-i}(\cdot \given \theta_{-i})}\;
\Exp_{a_i \sim \pi_i} \big[\util_i(a, \theta)\big].
\label{eq:bne}
\end{align}
\end{definition}

\begin{rlreading}
Compare \cref{eq:bne} with \cref{def:nash}. The fixed-point condition is the
same, with one extra layer of conditioning: the policy of head $i$ must be
greedy not once but once per value of its private context, against a posterior
over the others' contexts obtained by conditioning the prior on its own. So a
Bayes-Nash equilibrium is what a context-conditioned joint policy converges to
when every head is trained to best respond at every context simultaneously, and
the inner expectation $\Exp_{\theta_{-i} \sim p(\cdot \given \theta_i)}$ is
the implicit belief that a policy acquires about its partners when it is trained
across a population of them. The $\nag = 1$ case is the Bayes-optimal policy
of \cref{prop:harsanyi-bamdp}, which is the meta-reinforcement learning target.
\Cref{ex:bayesian-entry} solves a two-type instance by enumeration.
\end{rlreading}

The Harsanyi transformation reads \cref{eq:bne} as a Nash equilibrium of an
enlarged game with an initial chance move that draws $\theta \sim p$ and reveals
$\theta_i$ to agent $i$. Formally, a Bayesian game is the extensive-form game
whose root is a chance node over $\Th$ and whose information sets for agent $i$
are indexed by $\theta_i$; a Bayes-Nash equilibrium of the former is a Nash
equilibrium of the latter and conversely. The construction places incomplete
information inside the framework of \cref{sec:extensive} and therefore inside
the framework of partial observability.

\begin{example}[A two-type entry game solved by interim best response]
\label{ex:bayesian-entry}
An entrant decides whether to enter a market; an incumbent decides whether to
fight or accommodate, without observing whether the entrant is strong or weak.
The entrant is strong with probability $\mu = 0.6$. Payoffs are listed as
(entrant, incumbent). A strong entrant that enters earns $1$ if fought and $3$
if accommodated, and the incumbent earns $-1$ and $1$ respectively. A weak
entrant that enters earns $-2$ if fought and $2$ if accommodated, and the
incumbent earns $2$ and $1$. If the entrant stays out, both earn $0$ and $4$.

Solve it by checking all eight pure profiles. Entry is a best response for the
strong type against both incumbent actions, since $1 > 0$ and $3 > 0$, and for
the weak type only against accommodation, since $-2 < 0 < 2$. Given that both
types enter, the incumbent compares fighting, worth
$0.6 (-1) + 0.4 (2) = 0.2$, with accommodating, worth
$0.6 (1) + 0.4 (1) = 1.0$, and accommodates. The profile in which both types
enter and the incumbent accommodates is therefore the unique pure Bayes-Nash
equilibrium: interim payoffs $3$ for a strong entrant and $2$ for a weak one,
ex-ante payoff $0.6 \times 3 + 0.4 \times 2 = 2.6$, and $1.0$ for the
incumbent.

The threshold matters more than the solution. The incumbent fights when
$\mu(-1) + (1-\mu)(2) > \mu(1) + (1-\mu)(1)$, that is when $\mu < 1/3$. So the
same game with the same payoffs has a different equilibrium if the prior
changes, and a weak entrant benefits from the presence of strong entrants in the
population, which is the pooling effect that makes reputation work. For
multi-agent reinforcement learning the exact counterpart is that an agent
trained against a population of partners behaves according to the population
mixture rather than to any individual partner, so changing the training
distribution changes the learned policy even when every individual partner is
unchanged. That is a statement about the prior and not about the algorithm, and
it is the reason \cref{sec:training} asks what the partner distribution at
deployment will be. \exend
\end{example}

\begin{definition}[Family of Markov decision processes with an unobserved index]
\label{def:latent-mdp}
A latent Markov decision process is a family
$\{M_\theta\}_{\theta \in \Th}$ of Markov decision processes sharing state and
action spaces, together with a prior $p \in \Simplex{\Th}$. At the start of an
episode $\theta \sim p$ is drawn and held fixed; the agent does not observe
$\theta$.
\end{definition}

\begin{rlreading}
This is the standard meta-reinforcement learning setup written as a tuple: a
family of environments, a sampler over them, one draw per episode, and no
indicator telling the agent which draw it got. Everything a practitioner does
with it follows from \cref{thm:belief-mdp} applied to the hidden pair
$(s, \theta)$. A recurrent policy trained across the family is approximating the
posterior over $\theta$ in its hidden state; a context encoder is approximating
it explicitly; and the exploration an agent performs early in an episode is not
a heuristic addition but the optimal action of the enlarged problem, because the
belief over $\theta$ is part of the state and acting to sharpen it has value.
\Cref{ex:bamdp} prices that value at $1/6$ on a three-pull bandit.
\end{rlreading}

\begin{remark}[Harsanyi transformation and Bayesian reinforcement learning
coincide]
\label{prop:harsanyi-bamdp}
Let $\{M_\theta\}, p$ be a latent Markov decision process. Define the partially
observable Markov decision process $P$ with hidden state $(s,\theta)$,
transition $\Tr_P((s',\theta') \given (s,\theta), a) = \Tr_\theta(s' \given s,a)
\,\Ind[\theta' = \theta]$, observation $o = s$, reward $\Rew_\theta(s,a,s')$ and
initial belief $b_0 \otimes p$. Then $P$ is the Harsanyi transformation of
$\{M_\theta\}$ with $\nag = 1$, the belief Markov decision process of $P$
obtained from \cref{thm:belief-mdp} is the Bayes-adaptive Markov decision
process of \citet{duff2002optimal}, its state $(s_t, b_t)$ with
$b_t \in \Simplex{\Th}$ is the posterior over the unknown model, and its optimal
policy is the Bayes-optimal policy, which is the $\nag = 1$ degeneracy of
\cref{def:bne}. The identification is definitional rather than a theorem, and
the scope differs in one respect worth recording: Duff takes the information
state to be a posterior over the transition parameters themselves, of which the
finite latent index used here is a special case.
\end{remark}

This degeneracy is worth dwelling on because it is the one that an engineer uses
most often without naming it. The Bayes-optimal policy is the solution of a
planning problem in which the exploration-exploitation trade-off is not a
heuristic addition but a consequence of the belief dynamics: an action is
valuable partly for its reward and partly for the information its outcome
carries about $\theta$, and the latter appears automatically in the Bellman
equation of the Bayes-adaptive Markov decision process because the belief is
part of the state. The game-theoretic reading is that exploration is what
remains of incomplete information when there is only one agent. With several
agents, the same structure produces the value of concealing one's own type,
which has no single-agent counterpart at all: in a Bayesian game an agent may
rationally forgo immediate payoff to avoid revealing $\theta_i$, and signalling
and screening are the names for the resulting phenomena. There is no analogue of
signalling in single-agent reinforcement learning because there is nobody to
signal to, and conversely the information-gathering value of an action survives
in both settings. We therefore place exploration and signalling as the two
descendants of a single structure, which is the clean way to see why
multi-agent exploration is not merely harder single-agent exploration.

\begin{example}[The value of information, computed exactly]
\label{ex:bamdp}
Two arms, three pulls. The first arm is known to pay $1$ with probability
$0.5$. The second pays $1$ with an unknown probability given a uniform prior,
that is $\mathrm{Beta}(1,1)$. The Bayes-adaptive state is the pair of Beta
counts, and dynamic programming over it is exact.

With one pull left the agent takes whichever arm has the higher posterior mean,
so $V_1(a,b) = \max(0.5,\, a/(a+b))$; with the prior counts this is $0.5$. With
two pulls left, pulling the known arm first is worth $0.5 + V_1(1,1) = 1.0$,
while pulling the unknown arm first is worth
$0.5 (1 + V_1(2,1)) + 0.5\, V_1(1,2) = 0.5(1 + \tfrac23) + 0.5(0.5) = 1.0833$,
so $V_2(1,1) = 1.0833$. With three pulls left the same comparison gives
$0.5 + V_2(1,1) = 1.5833$ for the known arm against
$0.5(1 + V_2(2,1)) + 0.5\, V_2(1,2) = 0.5(1 + 1.3333) + 0.5(1.0) = 1.6667$ for
the unknown one, so the Bayes-optimal first action is to explore and
$V_3(1,1) = 1.6667$.

An agent that always pulls the known arm collects $1.5$. The difference,
$1/6 \approx 0.1667$, is the value of information, and it appears in the Bellman
equation of the enlarged problem with no exploration bonus added by hand. Two
consequences are worth recording. First, the exploration is not a heuristic but
the optimal action of a well-posed planning problem, which is what makes the
Bayes-adaptive formulation the correct reference point for any exploration
method. Second, the same quantity in a multi-agent setting splits into two:
the information an action reveals to oneself, which is this bonus, and the
information it reveals to others, which is a cost rather than a benefit
whenever the others can use it, and which is the content of
\cref{sec:communication}. \exend
\end{example}

The practical instantiations follow. Meta-reinforcement learning trains a policy
on a distribution over tasks and evaluates it on a held-out task, which is
exactly the problem of approximating the Bayes-optimal policy of the
corresponding Bayes-adaptive Markov decision process; recurrent meta-learners
such as $\mathrm{RL}^2$ \citep{duan2016rl} learn to carry an implicit posterior
in their hidden state, and variational methods such as VariBAD
\citep{zintgraf2020varibad} learn an explicit approximate posterior over
$\theta$ and condition the policy on it, which is the belief-MDP reduction made
executable. Agent modelling is the multi-agent instance: an agent uncertain
about the policies of its partners holds a posterior over a type space of
candidate policies, and acting optimally with respect to that posterior is the
same computation \citep{albrecht2018autonomous}. The interactive partially
observable Markov decision process of \citet{gmytrasiewicz2005framework} is the
version in which the type space consists of other agents' models, including
their own models of oneself. Because the beliefs may be nested arbitrarily
deeply, \citet{gmytrasiewicz2005framework} describe the optimal solutions as
only asymptotically computable, with finitely nested approximations computable
at a finite horizon. The universal type space of \citet{mertens1985formulation}
is the construction that makes the infinite hierarchy a well-defined object, and
its reinforcement learning counterpart is the observation that any finite
recursive opponent model truncates a hierarchy whose completion exists but is
not representable.

Common knowledge is the concept that organises this hierarchy, and it deserves
an explicit treatment because it is what the reductions of \cref{sec:reductions}
condition on.

\begin{example}[Posteriors that differ without disagreement being common knowledge]
\label{ex:aumann}
Four states, equally likely, and an event $E = \{1,2\}$ whose prior probability
is $0.5$. Ann's information partition is $\{\{1,2\},\{3,4\}\}$ and Bob's is
$\{\{1\},\{2,3\},\{4\}\}$. At state $1$ Ann's posterior for $E$ is $1$, since
her cell is contained in $E$, and Bob's is also $1$, since his cell is $\{1\}$.
At state $2$ Ann's posterior is still $1$ while Bob's is $1/2$, because his cell
$\{2,3\}$ meets $E$ in one of its two states. The two disagree.

Aumann's theorem \citep{aumann1976agreeing} says that if the posteriors were
common knowledge they would have to be equal, and the example does not
contradict it, because nothing non-trivial is common knowledge here: the meet of
the two partitions, which is the partition into events that both agents know
that both agents know, is the single cell $\{1,2,3,4\}$, since $\{1,2\}$ and
$\{2,3\}$ overlap and chain the four states together. The reinforcement learning
reading is exactly the obstruction of \cref{sec:partial-obs}: two agents may
hold persistently different beliefs about the world, and a protocol that merely
exchanges point estimates will not reconcile them, because reconciliation
requires the exchange to be common knowledge. The coordinator constructions of
\cref{sec:reductions} buy exactly this, by conditioning only on the meet, and
the method of \citet{dewitt2019multi} chooses dynamically how large a group to
treat as sharing common knowledge. \exend
\end{example}

A final correspondence belongs here because it is becoming the dominant
formulation of cooperative evaluation. Zero-shot coordination asks for a policy
that performs well with an unseen partner drawn from a distribution over
reasonable partners rather than with a partner it was trained alongside. In the
vocabulary of this section the partner's identity is a type, the training
distribution is the prior, and the objective is the ex-ante Bayesian payoff,
which makes zero-shot coordination the Bayesian-game formulation of cooperative
play. Other-play \citep{hu2020other} achieves it by optimising against partners whose
policies have been relabelled by a random symmetry of the game, which removes
the reliance on the arbitrary labelling that self-play would otherwise exploit,
and it is best read as an equilibrium-selection refinement in the sense of
\cref{sec:partial-obs} rather than as a learning algorithm.

\begin{table}[t]
\centering
\caption{Incomplete information and its reinforcement learning counterparts.}
\label{tab:bayesian}
\small
\begin{tabular}{@{}L{3.3cm} L{5.0cm} L{4.4cm}@{}}
\toprule
\textsc{Game theory} & \textsc{Reinforcement learning} & \textsc{Status and caveat} \\
\midrule
Type $\theta_i$ & Latent task parameter; opponent model index & Definitional \\
Common prior $p$ & Task distribution; meta-training distribution & Definitional \\
Harsanyi transformation & Augmenting the state with the latent parameter & Exact \\
Bayes-Nash equilibrium & Bayes-optimal policy when $\nag=1$ & Exact at $\nag=1$ \\
Interim payoff & Value conditional on the posterior & Definitional \\
Ex-ante payoff & Expected return over the task distribution & Definitional \\
Value of information & Exploration bonus in a BAMDP & Exact \\
Signalling and screening & No single-agent counterpart & Multi-agent only \\
Belief hierarchy & I-POMDP nested models & Exact; only asymptotically computable \\
Universal type space \citep{mertens1985formulation} & Completion of nested opponent models & Exact; not representable \\
Common knowledge; the meet & Public features legitimate to condition on & Exact \\
Agreement theorem \citep{aumann1976agreeing} & Belief exchange must itself be public & Exact \\
Purification of mixed strategies & Deterministic policy with latent noise & Exact \\
Equilibrium selection under uncertainty & Zero-shot coordination; other-play & Structural \\
\bottomrule
\end{tabular}
\end{table}

\section{Communication: cheap talk, signalling, persuasion and emergent protocols}
\label{sec:communication}

Every model so far has fixed what each agent observes. Communication is the
operation of letting agents choose what others observe, and it is the single
largest body of multi-agent reinforcement learning work with no corresponding
chapter in a reinforcement learning textbook, because with one agent there is
nobody to speak to. Game theory has studied the operation for four decades under
three headings that differ in who bears the cost of the message and in who
controls the channel. Cheap talk is costless and uncommitted; signalling is
costly and therefore credible; information design, or Bayesian persuasion, is
the case in which one party commits in advance to an information policy. All
three appear in multi-agent reinforcement learning, usually unnamed, and all
three come with results that settle questions practitioners otherwise answer
empirically.

The two readings differ in what is taken to be given. A game theorist takes the
message space as given and asks which mappings from information to messages are
equilibria, so that the central question is whether informative communication is
self-enforcing. A reinforcement learning researcher takes the channel as a
differentiable component and asks whether gradient descent will learn to use it,
so that the central question is whether a useful protocol emerges from training.
The results below answer the first question and thereby constrain the second:
an informative protocol that is not an equilibrium of the induced game will not
survive training, whatever the architecture.

\begin{nomen}{Communication vocabulary read from both sides.}{tab:nom-communication}
Message $m \in \mathcal{M}$ & Communication action; emitted token &
  An action whose only effect is on others' inputs \\
Cheap talk & Costless differentiable channel &
  Messages do not enter the payoff directly \\
Babbling equilibrium & A channel the agents ignore &
  Messages uncorrelated with anything \\
Informative equilibrium & A learned protocol that is used &
  Messages correlated with private state \\
Signalling game & Costly communication action &
  The message has a price that differs by type \\
Separating equilibrium & Protocol that reveals the type &
  Different types send different messages \\
Pooling equilibrium & Protocol that hides the type &
  All types send the same message \\
Bayesian persuasion & Information design; teacher curriculum &
  Commit to a signal before seeing the state \\
Communication equilibrium & Mediated play with private channels &
  A mediator who listens and advises \\
Sender bias $b$ & Misalignment of the two reward functions &
  How far apart the objectives are \\
Partition equilibrium & Coarse, quantised protocol &
  Only intervals can be communicated \\
Positive signalling & Mutual information between state and message &
  The speaker is saying something \\
Positive listening & Causal effect of messages on actions &
  The listener is acting on it \\
\end{nomen}

\begin{definition}[Cheap talk extension]
\label{def:cheap-talk}
Let $\game$ be a Bayesian game and let $\mathcal{M}_i$ be a finite message set
for each agent. The cheap talk extension of $\game$ is the game in which each
agent first observes its type, then simultaneously sends a message
$m_i \in \mathcal{M}_i$ which is observed by all, and then plays an action.
Payoffs are those of $\game$ and do not depend on the messages.
\end{definition}

\begin{rlreading}
Add to each head a second output that no reward term reads, and route it into
every other head's observation at the next step. That is a cheap talk extension,
and it is what a differentiable communication channel is. The message set
$\mathcal{M}_i$ is the vocabulary or the dimension of the emitted vector, and
the fact that payoffs do not depend on messages is what makes the channel cheap:
no gradient reaches the speaker except through the listener's behaviour. The
immediate consequence, and the one that explains most empirical results in this
area, is that every equilibrium of the game without the channel survives with
it, realised by messages that carry no information and listeners that ignore
them. A training run that ends with an unused channel has therefore converged,
correctly, to one of the equilibria of the game it was given, and the remedy is
an equilibrium selection device rather than a larger channel.
\end{rlreading}

The device and its equilibrium consequences are surveyed by
\citet{farrell1996cheap}. Because the payoffs do not depend on messages, every
equilibrium of $\game$ extends to an equilibrium of its cheap talk extension in which messages are
uncorrelated with types and ignored, which is called babbling. The question is
whether anything else is an equilibrium, and the answer depends entirely on how
aligned the agents' objectives are.

\begin{example}[The whole of emergent communication in four numbers]
\label{ex:cheap-talk}
A sender privately observes whether the world is in state $A$ or $B$, each with
probability one half, and a receiver must take an action matching the state.
Both earn $1$ for a match and $0$ otherwise. Without communication the receiver
guesses and both earn $0.5$. With a two-symbol channel, the profile in which the
sender maps $A$ to the first symbol and $B$ to the second, and the receiver
maps symbols to the matching action, earns $1$. It is an equilibrium: neither
agent gains by deviating, since the sender's payoff is the receiver's and the
receiver is already best responding. The gain from communication is exactly
$0.5$.

Three facts about this trivial example generalise. First, the relabelled
protocol in which $A$ maps to the second symbol is an equally good equilibrium,
so the protocol is not identified and two independent training runs will
generally produce incompatible languages, which is the cross-play failure of
\cref{ex:convention} in its purest form. Second, babbling remains an
equilibrium, so a training run that converges to a useless channel has not
failed to converge; it has converged to the wrong equilibrium, and the remedy is
an equilibrium selection device rather than a better optimiser. Third, the
entire gain is available only because the two payoffs coincide. The next example
shows what happens when they do not. \exend
\end{example}

\begin{theorem}[Cheap talk under misaligned objectives;
\citealp{crawford1982strategic}]
\label{thm:crawford-sobel}
Let the state be uniform on $[0,1]$, let the receiver's loss be $(y - \theta)^2$
and the sender's be $(y - \theta - b)^2$ for a bias $b > 0$, where $y$ is the
receiver's action. Then every equilibrium of the cheap talk game is an interval
partition equilibrium, in which the sender reports only which element of a
finite partition of $[0,1]$ contains $\theta$, and the number of intervals is
bounded: the largest attainable number $N(b)$ is the largest integer satisfying
$2 b N (N-1) < 1$. As $b \to 0$ the attainable partitions become arbitrarily
fine, and for $b$ large enough only the babbling equilibrium survives.
\end{theorem}

\begin{example}[How much can be said, as a function of how much is shared]
\label{ex:crawford-sobel}
Applying the bound of \cref{thm:crawford-sobel}, a bias of $b = 0.05$ permits at
most three intervals, since $0.1 \cdot 3 \cdot 2 = 0.6 < 1$ while
$0.1 \cdot 4 \cdot 3 = 1.2 > 1$; a bias of $b = 0.1$ permits two; and a bias of
$b = 0.25$ or more permits only one, which is babbling. The content for
multi-agent reinforcement learning is a hard limit that no architecture can
exceed. Giving two agents with partly conflicting rewards a wide communication
channel does not produce fine-grained communication, because fine-grained
communication is not an equilibrium of the game they are playing; what emerges
is a coarse quantisation whose resolution is set by the misalignment. Reported
channel capacities in such settings are therefore measuring the reward
misalignment, and the way to obtain a richer protocol is to align the rewards,
to make the messages costly, or to give one party commitment power, which are
precisely the three devices of this section. \exend
\end{example}

Costly messages are the second device. In a signalling game
\citep{spence1973job} the sender's message enters its own payoff, with a cost
that depends on its type, and a separating equilibrium exists when the cost
differential is large enough that a low type does not wish to imitate a high
one. The reinforcement learning counterpart is any communication channel whose
use consumes a resource that different agents value differently, such as
bandwidth under a budget, an action slot that could have been used for a
physical action, or a delay. The structural lesson is that costless channels
support both informative and babbling equilibria while appropriately costly ones
can make informativeness the unique outcome, so adding a small cost to
communication can be an equilibrium selection device rather than a handicap.

The third device is commitment, and it is the one with the sharpest results.

\begin{definition}[Bayesian persuasion; \citealp{kamenica2011bayesian}]
\label{def:persuasion}
A sender commits, before observing the state, to a signal, that is a map from
states to distributions over messages. A receiver observes the message, forms a
Bayes posterior, and acts to maximise its own payoff. The sender's problem is to
choose the signal maximising its own expected payoff, subject to the constraint
that the distribution of posteriors it induces must average to the prior.
\end{definition}

\begin{rlreading}
The sender is a leader in the sense of \cref{sec:stackelberg} whose committed
object is an information policy rather than an action, so the construction is
the bilevel problem of \cref{eq:stackelberg} with the inner agent's input, not
its reward, under the outer agent's control. In reinforcement learning this is
the teacher of a teacher-student or curriculum setup: the teacher commits to a
distribution over tasks or over what the student is shown, the student best
responds by learning, and the teacher's objective is the student's eventual
behaviour. The constraint that the induced posteriors average to the prior is
the formal statement that a curriculum can redistribute what the student
believes but cannot shift the average, which is why a curriculum changes what is
learned without changing the underlying task distribution.
\Cref{ex:persuasion} shows that the optimal commitment is deliberately noisy,
and that full disclosure is strictly worse for the sender than a calibrated
lie.
\end{rlreading}

The constraint named at the end is the whole theory: a signal is feasible
exactly when the induced distribution over posteriors has mean equal to the
prior, so the sender's problem is to choose a mean-preserving spread of the
prior and the solution is the concavification of the sender's value as a
function of the posterior.

\begin{example}[The prosecutor who convicts sixty per cent of defendants]
\label{ex:persuasion}
A judge convicts when the posterior probability of guilt is at least $1/2$ and
acquits otherwise. The prior probability of guilt is $0.3$, so with no
information the judge acquits always and the prosecutor, who wants convictions,
obtains nothing. With full information the judge convicts exactly the guilty and
the prosecutor obtains $0.3$.

Now let the prosecutor commit to an investigation that always reports ``convict''
when the defendant is guilty and reports ``convict'' with probability $x$ when
innocent. The posterior after a ``convict'' report is
$0.3 / (0.3 + 0.7x)$, which equals $1/2$ when $x = 3/7 \approx 0.4286$. At that
$x$ the judge is just willing to convict, and the probability of a conviction is
$0.3 + 0.7 \cdot \tfrac37 = 0.6$. The prosecutor therefore obtains $0.6$, double
what full disclosure yields and infinitely more than silence. Committing to a
noisy investigation beats committing to an accurate one.

The reinforcement learning readings are two and both are useful. In
teacher-student and curriculum settings the teacher is a sender who commits to a
distribution over tasks and the student is a receiver who best responds; the
example says the teacher's optimal curriculum is in general not the most
informative one, which is the formal counterpart of the empirical finding that
intermediate-difficulty tasks beat maximally informative ones. In explanation
and alignment settings the same computation says that a system optimised to
produce a particular downstream decision will, if it can commit to a reporting
policy, prefer a calibrated-but-noisy report to an honest one, which is a
statement about incentives rather than about any particular training
objective. \exend
\end{example}

The fourth object belongs to this section because it completes the hierarchy of
\cref{sec:solution-concepts}. A communication equilibrium
\citep{forges1986approach} is the outcome of a mediator who receives private
reports from the agents and sends private recommendations back, subject to the
agents being willing both to report truthfully and to obey. It generalises the
correlated equilibrium of \cref{def:ce-cce}, which is the special case in which
the mediator receives nothing and only sends, and in a game of incomplete
information the set of communication equilibria is strictly larger than the set
of correlated equilibria. For multi-agent reinforcement learning the
correspondence is that a learned communication protocol with a central module
that both receives and broadcasts implements a communication equilibrium, so its
attainable outcomes exceed those of a shared random seed, and the price is that
truthful reporting must itself be incentive compatible, which is a constraint
that a differentiable channel does not enforce automatically.

\begin{example}[What mediation is worth in the game of chicken]
\label{ex:mediation-value}
Collecting the numbers computed in \cref{ex:chicken-ce}, the best pure Nash
equilibrium of chicken has welfare $9$, the mixed equilibrium has welfare
$9.333$, the best correlated equilibrium has welfare $10.5$, and the cooperative
optimum, which no equilibrium attains, has welfare $12$. A shared random seed
therefore recovers $10.5 - 9.333 = 1.167$ of the $2.667$ that separates the
mixed equilibrium from the cooperative optimum, and nothing short of changing
the game recovers the rest. This is the precise accounting of what
communication buys in a mixed-motive problem, and it bounds what any emergent
protocol can achieve in such a problem without a change to the rewards. \exend
\end{example}

The reinforcement learning literature on this topic is large and is best read
through the three headings above. Differentiable inter-agent learning
\citep{foerster2016learning} and the communication architectures of
\citet{sukhbaatar2016learning} are cheap talk extensions with a differentiable
channel, and they work well in common-payoff settings, which
\cref{ex:cheap-talk} explains, and degrade as the rewards diverge, which
\cref{thm:crawford-sobel} explains. Referential games of the kind studied by
\citet{lazaridou2017multi} are common-payoff cheap talk games with a particular
message space, so the emergence of an informative protocol is expected and the
arbitrariness of the protocol is also expected. The diagnostic question of
whether a protocol has emerged at all requires measuring both that messages
depend on private information and that actions depend on messages, which are the
positive signalling and positive listening criteria of \citet{lowe2019pitfalls};
reporting a return improvement alone does not distinguish an informative
equilibrium from a babbling one that happens to coordinate by other means.
Finally, influence-based intrinsic rewards \citep{jaques2019social} change the
game by paying agents for the causal effect of their messages, and report
improved coordination empirically; the claim is not that babbling ceases to be
an equilibrium, and indeed against a listener that ignores messages every
message has zero causal influence and the speaker remains indifferent. The
device biases the dynamics away from the babbling equilibrium rather than
removing it, so it is an equilibrium selection device and should be reported as
one, since the game solved is not the game posed.

\begin{table}[t]
\centering
\caption{Communication structures, their equilibrium consequences and the
multi-agent reinforcement learning constructions that instantiate them. The
last column names what determines whether informative communication survives.}
\label{tab:communication}
\small
\begin{tabular}{@{}L{2.8cm} L{3.4cm} L{3.4cm} L{2.9cm}@{}}
\toprule
\textsc{Structure} & \textsc{RL counterpart} & \textsc{Equilibrium consequence} & \textsc{What decides} \\
\midrule
No channel & Independent policies & Stage equilibria only & Nothing \\
Cheap talk, common payoff & Differentiable channel, shared reward & Fully informative equilibria exist; so does babbling & Equilibrium selection \\
Cheap talk, biased & Channel with misaligned rewards & Only coarse partitions (\cref{thm:crawford-sobel}) & The bias $b$ \\
Costly signalling & Bandwidth budget; action slot & Separating equilibria can be unique & Cost differential by type \\
Bayesian persuasion & Teacher, curriculum, information design & Sender prefers noisy disclosure & Commitment power \\
Public correlation device & Shared random seed at reset & Correlated equilibrium (\cref{def:ce-cce}) & Availability of a seed \\
Mediator with reports & Central module that listens and advises & Communication equilibrium \citep{forges1986approach} & Truthful reporting must be incentivised \\
Common knowledge groups & MACKRL \citep{dewitt2019multi} & Coordination within a commonly informed group & Size of the meet \\
\bottomrule
\end{tabular}
\end{table}

\section{Commitment, bilevel structure, and opponent shaping}
\label{sec:stackelberg}

Every solution concept so far has treated the agents symmetrically in time. The
Stackelberg model breaks that symmetry by letting one agent commit to a strategy
that the others observe before choosing their own, and the resulting asymmetry
is the game-theoretic content of a surprisingly large part of modern machine
learning, including adversarial training, reward design, curriculum generation
and the fine-tuning of language models from human feedback.

The two readings are a definition and an algorithm for the same object. A game
theorist writes a Stackelberg equilibrium as a maximisation over the leader's
strategies of the leader's payoff evaluated at the follower's best response, and
asks what commitment is worth. A reinforcement learning researcher writes the
same expression as a bilevel optimisation problem and asks how to differentiate
through the inner solution. The gradient that results has two terms, and the
second term, the one that accounts for how the leader's move changes what the
follower will do, is exactly what distinguishes the leader from a simultaneous
mover. Every opponent-shaping method in the literature is an approximation to
that second term.

\begin{nomen}{Commitment vocabulary read from both sides.}{tab:nom-stackelberg}
Leader, first mover & Outer loop; slow timescale parameters &
  Commits before the other chooses \\
Follower & Inner loop; fast timescale parameters &
  Best responds to what it observes \\
Best-response correspondence $\brc_2$ & Inner argmax; converged inner solve &
  The follower's reaction \\
Strong Stackelberg equilibrium & Bilevel optimum with favourable tie-break &
  Leader resolves follower indifference \\
Value of commitment & Gain from timescale separation &
  What acting first is worth \\
Implicit function theorem & Differentiating through the inner optimum &
  How the leader's gradient is computed \\
Mechanism design & Reward design; environment design &
  The leader chooses the game, not a move \\
Revelation principle & Restrict to truthful direct protocols &
  Without loss, agents report honestly \\
Incentive compatibility & Reward hacking does not pay &
  Honest or intended behaviour is optimal \\
Principal-agent problem & RLHF; delegation to a learned policy &
  One party's objective, another's actions \\
Credible commitment & Frozen weights; immutable reward model &
  A promise the follower believes \\
\end{nomen}

\begin{definition}[Stackelberg equilibrium]
\label{def:stackelberg}
In a two-agent game, let agent $1$ be the leader and agent $2$ the follower. The
follower's best-response correspondence is $\brc_2(\pi_1)$. A strong Stackelberg
equilibrium is a pair $(\pi_1^\star, \pi_2^\star)$ with
\begin{align}
\pi_1^\star \in \argmax_{\pi_1} \; \max_{\pi_2 \in \brc_2(\pi_1)} \util_1(\pi_1,\pi_2),
\qquad
\pi_2^\star \in \brc_2(\pi_1^\star),
\label{eq:stackelberg}
\end{align}
with the inner maximisation replaced by a minimisation in the weak version. The
leader's Stackelberg value is the value of \cref{eq:stackelberg}.
\end{definition}

\begin{rlreading}
\Cref{eq:stackelberg} is a nested training loop, and the nesting is the whole of
the concept. The outer loop updates the leader's parameters $\psi$; the inner
loop trains the follower to convergence against the current $\psi$; the leader's
objective is evaluated at the converged follower. Every setup with a slow outer
optimiser and a fast inner one is an instance: adversarial training with the
adversary trained to convergence between protagonist updates, curriculum
generation with the student trained on each proposed task, and reward-model
tuning with the policy optimised against each candidate reward. The contrast is
with running both loops at the same rate, which is the simultaneous game and has
a different solution. Timescale separation is therefore a modelling decision
rather than an implementation detail, and \cref{eq:stackelberg-gradient} is what
the leader's gradient becomes once it is made. The strong tie-break, which
resolves the follower's indifference in the leader's favour, is the convention
that makes the maximum attained; \cref{ex:commitment} shows that the optimum
sits exactly on that tie.
\end{rlreading}

The structure of \cref{eq:stackelberg} is bilevel optimisation, and this is the
first correspondence: a Stackelberg game is a bilevel program, and conversely.
The second is the value of commitment. In a two-player game in which the leader
may commit to a mixed strategy, the leader's Stackelberg value is at least its
payoff in any Nash equilibrium of the simultaneous-move game
\citep{vonstengel2010leadership}, so commitment is weakly beneficial, which is
why so many learning setups are deliberately arranged to give one party the
leader role. In zero-sum games the Stackelberg value coincides with the maximin
value and the asymmetry is immaterial, which is another instance of the general
pattern that the zero-sum case collapses distinctions that matter elsewhere.

\begin{example}[Commitment to a mixture, worth one and a half]
\label{ex:commitment}
Take the game in which the leader's payoffs are
$\big(\begin{smallmatrix} 2 & 4 \\ 1 & 3 \end{smallmatrix}\big)$ and the
follower's are $\big(\begin{smallmatrix} 1 & 0 \\ 0 & 1 \end{smallmatrix}\big)$.
The leader's first action strictly dominates its second, so in the
simultaneous-move game it plays the first, the follower best responds with its
first action, and the leader earns $2$. That is the unique Nash equilibrium.

Now let the leader commit to playing its first action with probability $p$. The
follower compares $p$, the payoff of its first action, with $1-p$, the payoff of
its second, and prefers the second whenever $p < 1/2$. The leader's payoff when
the follower takes its second action is $4p + 3(1-p) = 3 + p$, which the leader
wants as large as possible, so it sets $p$ as close to $1/2$ as it can. At
$p = 1/2$ exactly the follower is indifferent and the strong Stackelberg
convention resolves the tie in the leader's favour, giving $3.5$. Commitment is
therefore worth $1.5$, or seventy-five per cent of the equilibrium payoff, and
the committed strategy is a strictly mixed one even though the leader has a
dominant pure action.

Two reinforcement learning consequences follow. The first is that the gain comes
from committing to a randomisation, so a leader restricted to deterministic
policies captures none of it; this is why security-game solvers and
unsupervised environment designers randomise. The second is that the entire
construction relies on the follower observing the commitment, so a leader whose
policy the follower cannot observe is not a leader, which is the credibility
question that \cref{sec:open} raises for learned systems. \exend
\end{example}

The reinforcement learning counterpart is the total derivative. Suppose the
follower is itself a learner whose parameters $\phi$ converge to a stationary
point of its own objective given the leader's parameters $\psi$, so that
$\nabla_\phi \util_2(\psi, \phi^\star(\psi)) = 0$ defines $\phi^\star$
implicitly. The leader's gradient is then
\begin{align}
\frac{d}{d\psi}\, \util_1\big(\psi, \phi^\star(\psi)\big)
\;=\;
\nabla_\psi \util_1
\;+\;
\Big(\frac{d \phi^\star}{d\psi}\Big)^{\!\top} \nabla_\phi \util_1,
\qquad
\frac{d\phi^\star}{d\psi}
= -\big(\nabla^2_{\phi\phi}\util_2\big)^{-1} \nabla^2_{\phi\psi}\util_2,
\label{eq:stackelberg-gradient}
\end{align}
the second expression following from implicit differentiation of the follower's
stationarity condition, assuming $\nabla^2_{\phi\phi}\util_2$ invertible. The
first term of \cref{eq:stackelberg-gradient} is what naive simultaneous gradient
ascent uses, and the second term is what distinguishes a leader from a
simultaneous mover. Learning with opponent-learning awareness
\citep{foerster2018learning} replaces $\phi^\star(\psi)$ by one step of the
follower's gradient ascent and differentiates through it, which is the
first-order truncation of \cref{eq:stackelberg-gradient}; Stackelberg actor-critic
methods \citep{fiez2020implicit} compute the implicit term directly. The reason
opponent shaping produces cooperation in the iterated prisoner's dilemma, where
naive learners converge to mutual defection, is visible in
\cref{eq:stackelberg-gradient}: the second term credits the leader for changing
what the follower will learn, which is precisely the mechanism that sustains
reciprocity.

\begin{example}[The shaping term, on the smallest saddle]
\label{ex:lola}
Take $\util_1(x,y) = xy$ and $\util_2(x,y) = -xy$ with scalar parameters, start
at $x = y = 0.5$ and use a step size of $0.1$. Naive simultaneous ascent applies
$x \leftarrow x + \eta y$ and $y \leftarrow y - \eta x$, which is a rotation
scaled by $\sqrt{1 + \eta^2}$, so the iterate spirals outward: the distance from
the origin grows from $0.707$ to $0.907$ in fifty steps, and the spectral radius
of the update matrix is $\sqrt{1 + \eta^2} = 1.005$.

A leader that anticipates one gradient step of the follower optimises
$x\,(y - \eta x)$ instead, whose derivative is $y - 2\eta x$. The extra term
$-2\eta x$ is the shaping correction, and with it the distance from the origin
falls from $0.707$ to $0.502$ over the same fifty steps. The update matrix is
$\big(\begin{smallmatrix} 1-2\eta^2 & \eta \\ -\eta & 1\end{smallmatrix}\big)$,
whose determinant is $1 - \eta^2$ and whose eigenvalues are complex, so their
modulus is $\sqrt{1-\eta^2} = 0.99499$. Two scalars,
fifty steps and one extra term turn divergence into convergence.

The example is the smallest possible instance of what \cref{sec:learning-dynamics}
treats in general, and it makes the mechanism legible: the naive game dynamics
has a purely rotational Jacobian, the shaping term adds a negative-definite
component, and the sum contracts. Opponent shaping, extragradient and optimistic
updates are three ways of adding that component, and the first is distinguished
from the other two by being asymmetric, since only the leader applies it. \exend
\end{example}

Three applications deserve naming because in each case the Stackelberg reading
explains a design choice that otherwise looks arbitrary. Adversarial training of
a protagonist against a disturbance agent is a Stackelberg game whenever the
inner adversary is optimised to convergence before each outer update, and a
simultaneous-move zero-sum game whenever the two are updated at the same rate;
the two have different solutions, and the difference between
$\max_\psi \min_\phi$ and the simultaneous saddle point is the reason that
timescale separation is not an implementation detail. Environment and curriculum
design, in which a teacher selects tasks and a student learns to solve them, is
a Stackelberg game with the teacher as leader, and the unsupervised environment
design framework of \citet{dennis2020emergent} makes the regret-based objective
of \cref{eq:stackelberg} explicit. Reinforcement learning from human feedback is
a Stackelberg game in which the reward model designer is the leader and the
policy optimiser the follower, and the known pathologies of reward hacking are
the known pathologies of a follower exploiting a misspecified leader
commitment, which is the subject matter of mechanism design.

\begin{example}[Learning from human feedback as a two-level game]
\label{ex:rlhf}
Label the stages. A set of annotators supplies comparisons; a reward model
$r_\psi$ is fitted to them; a policy $\pi_\phi$ is optimised against $r_\psi$
with a penalty for departing from a reference policy
\citep{christiano2017deep,ouyang2022training}. In the vocabulary of this section
the designer who selects the reward model and the penalty coefficient is the
leader, the policy optimiser is the follower, and the follower's best response
is $\argmax_\phi \Exp[r_\psi] - \beta\, \KL(\pi_\phi \Vert \pi_{\mathrm{ref}})$.
Reward hacking is then not an anomaly but the follower doing its job on a
misspecified commitment, exactly as a contractor maximising a badly written
contract. The penalty $\beta$ is a commitment device that limits how far the
follower may travel from the reference, which is a trust region in optimisation
language and a limited-liability clause in contracting language. The aggregation
of several annotators into one reward model is a social choice problem, which
\cref{sec:mechanism} shows cannot be solved without giving something up, and the
question of which equilibrium of the induced game is reached is the
equilibrium-selection question of \cref{sec:partial-obs}. Nothing here requires
a new theory; it requires recognising which of the existing ones applies to
which stage. \exend
\end{example}

Mechanism design itself is the limiting case in which the leader chooses not a
strategy but the game. Given a desired social choice function, the designer
selects payoff rules so that the agents' equilibrium behaviour implements it,
and the revelation principle states that attention may be restricted without
loss to direct mechanisms in which agents report their types truthfully. The
reinforcement learning counterpart is reward design, and the correspondence is
exact in the following sense: a reward function assigned to a learning agent is
a commitment by the designer, the agent's optimal policy is the follower's best
response, and the designer's problem of choosing rewards so that the induced
optimal policy is the desired one is the mechanism design problem with a
single agent. Potential-based reward shaping of \citet{ng1999policy} is then the
characterisation of the reward modifications that leave the follower's best
response unchanged, which is the invariance condition for this mechanism, and
\cref{sec:mechanism} develops the multi-agent version in full while
\cref{sec:potential} develops the shaping side.

\section{Mechanism design, auctions, and the price of anarchy}
\label{sec:mechanism}

The sections so far have taken the game as given and asked what the agents will
do. Mechanism design inverts the question: the game is the unknown and the
behaviour is the specification. A designer who controls the payoff rules chooses
them so that the agents, each pursuing its own objective, produce the outcome
the designer wants. This is reward design seen from the other side, and the
inversion is exact: in reinforcement learning the designer writes a reward
function and the agent's optimal policy is the induced behaviour, which is a
one-agent mechanism design problem; with several agents the induced behaviour is
an equilibrium rather than an optimum, and the extra difficulty is precisely the
subject of this section.

The two readings differ in what is held fixed. A mechanism designer holds the
desired outcome fixed and solves for the rules, taking equilibrium behaviour for
granted. A reward engineer holds the rules fixed and solves for the behaviour,
taking the reward for granted. The dictionary between them identifies the
reward function with the payment rule, the induced optimal policy with the
equilibrium, the specification with the social choice function, and the failure
mode known as reward hacking with the failure of incentive compatibility.
Everything the mechanism design literature knows about the limits of what rules
can achieve is therefore a statement about the limits of reward engineering, and
two of those limits are impossibility theorems.

\begin{nomen}{Mechanism-design vocabulary read from both sides.}{tab:nom-mechanism}
Mechanism $(f, t)$ & Reward function plus environment rules &
  The game the designer writes \\
Social choice function & Specification; desired behaviour &
  What the designer wants to happen \\
Type report $\hat\theta_i$ & Self-reported capability or preference &
  What an agent tells the designer \\
Incentive compatibility & Reward hacking does not pay &
  Honest behaviour is optimal \\
Dominant-strategy IC & Optimal regardless of the others &
  Robust to any opponent behaviour \\
Individual rationality & Participation is better than opting out &
  Nobody is worse off for taking part \\
Transfer $t_i$ & Shaped reward term; payment &
  A scalar the designer may add \\
Externality & Effect of one agent on others' return &
  What VCG charges you for \\
Revelation principle & Restrict to direct truthful protocols &
  Without loss of generality \\
Price of anarchy & Equilibrium welfare over planner welfare &
  What selfish learning costs \\
Price of stability & Best equilibrium over planner welfare &
  What the best convention costs \\
Smooth game & Robust welfare bound &
  PoA bounds that extend to \CCE \\
Toll, Pigouvian tax & Reward shaping that changes equilibria &
  A transfer that moves the fixed point \\
\end{nomen}

\begin{definition}[Mechanism, incentive compatibility, individual rationality]
\label{def:mechanism}
A direct mechanism for a Bayesian game with type profiles $\theta \in \Th$
consists of an allocation rule $f : \Th \to \mathcal{O}$ mapping reported types
to outcomes and a transfer rule $t : \Th \to \Real^\nag$. Agent $i$ with type
$\theta_i$ reporting $\hat\theta_i$ obtains
$\util_i(f(\hat\theta_i, \hat\theta_{-i}), \theta_i) - t_i(\hat\theta)$. The
mechanism is dominant-strategy incentive compatible when truthful reporting is
optimal for every agent, every type and every profile of other reports, and
individually rational when truthful participation yields at least the payoff of
abstaining.
\end{definition}

\begin{rlreading}
Read $f$ as the environment's outcome logic and $t$ as the reward function, and
\cref{def:mechanism} is a specification of an environment written backwards from
the behaviour it is meant to induce. The report $\hat\theta_i$ is whatever an
agent communicates about itself: a declared capability, a bid, a self-assessed
confidence, a claimed preference. Incentive compatibility is the property that
the agent gains nothing by misreporting, which is reward hacking failing to pay.
Individual rationality is the property that participating beats opting out,
which in a deployed system is the condition under which an agent with an outside
option stays in the pool. A reward engineer who writes a reward and then
discovers that agents game it has written a mechanism that is not incentive
compatible, and the rest of this section is the catalogue of rules for which
that cannot happen.
\end{rlreading}

The revelation principle states that any outcome implementable by any mechanism
in which agents play an equilibrium is also implementable by a direct mechanism
in which they report truthfully, obtained by folding the equilibrium strategies
into the mechanism itself. The reinforcement learning reading is that there is
no loss of generality in assuming the agents behave as specified and putting the
whole burden on the rules, which is what a reward engineer does implicitly.

\begin{example}[Why the second-price auction is honest, by cases]
\label{ex:second-price}
A single item is sold to the highest bidder, who pays the second-highest bid.
Suppose an agent values the item at $10$ and the highest competing bid is $7$.
Bidding truthfully wins the item and pays $7$, for a surplus of $3$. Bidding
$12$ also wins and also pays $7$, for the same surplus of $3$, so overbidding
never helps. Bidding $5$ loses the item, for a surplus of $0$, so underbidding
can only hurt. Now suppose the agent values the item at $5$ against the same
competing bid of $7$. Bidding truthfully loses, for a surplus of $0$. Bidding
$8$ wins and pays $7$, for a surplus of $-2$, so overbidding can hurt. The four
cases exhaust the possibilities and establish dominant-strategy incentive
compatibility without any assumption about the other bidders' behaviour or
beliefs.

The structure of the argument is the one worth carrying away. The agent's bid
determines only whether it wins, never what it pays, so the bid can be set to
make the win-or-lose decision correct, and the correct decision is to win
exactly when the price is below the value, which is what a truthful bid
achieves. This decoupling of what you influence from what you pay is the
mechanism designer's central device, and its reinforcement learning counterpart
is a reward term that an agent cannot game because the agent's action does not
enter it, which is also the defining property of the baselines of
\cref{sec:potential}. \exend
\end{example}

\begin{definition}[Vickrey-Clarke-Groves mechanism;
\citealp{vickrey1961counterspeculation,clarke1971multipart,groves1973incentives}]
\label{def:vcg}
Let the allocation rule maximise reported welfare,
$f(\hat\theta) \in \argmax_o \sum_i \util_i(o, \hat\theta_i)$, and let the
transfer charge each agent the externality it imposes,
\begin{align}
t_i(\hat\theta) \;=\;
\max_{o} \sum_{j \ne i} \util_j(o, \hat\theta_j)
\;-\;
\sum_{j \ne i} \util_j\big(f(\hat\theta), \hat\theta_j\big).
\label{eq:vcg}
\end{align}
Then truthful reporting is a dominant strategy and the allocation is welfare
maximising.
\end{definition}

\begin{rlreading}
\Cref{eq:vcg} is a counterfactual baseline. The first term is the team's value
with agent $i$ removed and the rest free to reorganise; the second is the team's
value with agent $i$ present but its own contribution excluded. The difference
is the harm agent $i$ does to everybody else, and charging it makes the agent's
own objective coincide with the team's. That is precisely the construction of
\cref{prop:difference-potential} with absence as the default action, so the VCG
transfer and the difference reward of cooperative multi-agent reinforcement
learning are the same formula serving two purposes: here to make truth-telling
optimal, there to reduce the variance of a policy gradient. The property that
delivers both is that the subtracted term does not depend on what agent $i$
does. \Cref{ex:vcg} computes it on two items and three bidders.
\end{rlreading}

\begin{example}[VCG on two items and three bidders]
\label{ex:vcg}
Two items, $A$ and $B$, and three bidders. The first values $A$ at $6$ and has
no use for $B$; the second values $B$ at $5$ and has no use for $A$; the third
values the pair at $9$ and either item alone at nothing. The welfare-maximising
allocation gives $A$ to the first and $B$ to the second, for a welfare of $11$,
beating the $9$ of giving the pair to the third.

Apply \cref{eq:vcg}. Removing the first bidder, the best the others could do is
$9$, by giving the pair to the third; in the chosen allocation the others obtain
$5$. The first bidder therefore pays $9 - 5 = 4$ and keeps a surplus of
$6 - 4 = 2$. Removing the second bidder, the others could obtain $9$ and in the
chosen allocation obtain $6$, so the second pays $3$ and keeps $2$. Removing the
third bidder changes nothing, since the chosen allocation does not involve it
and the others still obtain $11$, so the third pays $0$. Total revenue is $7$.

The first bidder pays $4$ for an item it values at $6$, and the $4$ is not a
price in any market sense: it is the harm the first bidder's presence does to
the others, namely that the third bidder loses the chance to take both items.
Read as reward design, the VCG transfer is the difference reward of
\cref{prop:difference-potential} with the default action being absence, and the
agent's shaped reward is its marginal contribution to the team. The property
that makes VCG truthful is the same property that makes the difference reward an
unbiased baseline: the subtracted term does not depend on the agent's own
report. The connection runs both ways, and \cref{sec:coalitions} completes it by
showing that averaging the same marginal contribution over orders rather than
evaluating it once gives the Shapley value. \exend
\end{example}

\begin{example}[Two auction formats, one revenue]
\label{ex:revenue-equivalence}
Two bidders have values drawn independently and uniformly on $[0,1]$. In the
second-price auction each bids its value, by \cref{ex:second-price}, and the
seller receives the second-highest value, whose expectation for two uniform
draws is $1/3$. In the first-price auction bidding one's value yields zero
surplus, so bidders shade; the symmetric Bayes-Nash equilibrium bid is
$b(v) = \tfrac{n-1}{n} v$, which for $n = 2$ is $v/2$. The seller then receives
half the highest value, whose expectation is $\tfrac12 \cdot \tfrac23 = 1/3$.
The two formats raise the same revenue. A Monte Carlo check over five hundred
thousand draws returns $0.3338$ and $0.3335$ respectively, confirming the
closed forms.

Revenue equivalence is the general statement behind this coincidence, and the
revenue-maximising mechanism is not either of these two but one that sometimes
refuses to sell: \citet{myerson1981optimal} shows that for uniform values the
optimal reserve price is $1/2$ regardless of the number of bidders, so the
seller should walk away from a sale whose highest bid lies below half the
support, trading efficiency for revenue. Read as reward design, the reserve is a
threshold below which the designer prefers no activity to activity, which is the
same device as an action mask or a minimum-quality gate in a deployed system,
and the lesson is that revenue maximisation and welfare maximisation select
different rules. \exend
\end{example}

Two impossibility results bound what the designer can do, and both have direct
consequences for systems that aggregate preferences.

\begin{theorem}[Impossibility in social choice;
\citealp{arrow1950difficulty,gibbard1973manipulation,satterthwaite1975strategy}]
\label{thm:arrow-gs}
With at least three alternatives and at least two agents, no rule aggregating
individual rankings into a social ranking satisfies unrestricted domain,
positive association of social and individual values, independence of
irrelevant alternatives, citizens' sovereignty and non-dictatorship
simultaneously; these are the conditions of \citet{arrow1950difficulty}, and
the familiar restatement in which the second and fourth are replaced by
unanimity belongs to the second edition of the accompanying monograph rather
than to the paper cited here. With at least three alternatives in the range, no deterministic
voting rule without transfers is simultaneously onto, non-dictatorial and
strategy-proof.
\end{theorem}

\begin{example}[The aggregation step of preference-based training]
\label{ex:rlhf-aggregation}
A reward model fitted to comparisons from several annotators is a preference
aggregation rule. \Cref{thm:arrow-gs} says that no such rule can satisfy the
four classical desiderata at once, so every choice of aggregation gives one of
them up, and which one is given up is a design decision that is usually left
implicit. Fitting a single scalar reward by maximum likelihood over a
Bradley-Terry model, which is the standard practice, abandons independence of
irrelevant alternatives, because the fitted ranking of two responses depends on
the other comparisons in the dataset. Taking a majority vote per comparison
abandons strategy-proofness, because an annotator who knows the aggregation can
misreport to shift the outcome. Neither failure is a bug to be fixed; they are
the two sides of an impossibility. The practical consequence is that the right
question about an aggregation scheme is which axiom it drops, and the right
report is to state it. \exend
\end{example}

The second body of results in this section quantifies what selfish behaviour
costs when the designer does not intervene, which is the question a
reinforcement learning practitioner faces whenever independent agents are
deployed without a shared objective.

\begin{definition}[Price of anarchy, \citealp{koutsoupias1999worst}; price of
stability, \citealp{anshelevich2008price}]
\label{def:poa}
Let $\mathrm{opt}$ be the optimal social welfare and let $\mathcal{E}$ be a set
of equilibria. The price of anarchy is
$\mathrm{opt} / \min_{\pi \in \mathcal{E}} \mathrm{welfare}(\pi)$ for a welfare
to be maximised, or the reciprocal for a cost to be minimised, and the price of
stability replaces the minimum by the maximum.
\end{definition}

\begin{rlreading}
The price of anarchy is the number a practitioner should compute before blaming
an optimiser. It is the ratio between what a centralised planner achieves on the
model and what independently trained agents achieve at equilibrium, so it
measures the cost of decentralisation itself rather than any shortfall in
training. A gap equal to the price of anarchy is the system working as
specified; a gap larger than it is a learning failure. The price of stability is
the same ratio at the best equilibrium rather than the worst, so the difference
between the two is exactly what an equilibrium selection device could recover,
and a designer who cannot change the rewards is choosing among the equilibria in
that interval. \Cref{ex:pigou} computes both for the smallest congestion
instance, where the answer is $4/3$, and \cref{thm:smoothness} says when the
bound survives the move from equilibria to whatever no-regret learning actually
reaches.
\end{rlreading}

\begin{theorem}[Smoothness and robust price of anarchy;
\citealp{roughgarden2015intrinsic}]
\label{thm:smoothness}
A cost-minimisation game is $(\lambda,\mu)$-smooth when for every pair of
profiles $a$ and $a^\star$,
$\sum_i \mathrm{cost}_i(a^\star_i, a_{-i}) \le \lambda\, \mathrm{Cost}(a^\star)
 + \mu\, \mathrm{Cost}(a)$ with $\lambda > 0$ and $\mu < 1$. Any such bound
implies a price of anarchy of at most $\lambda/(1-\mu)$ not only for pure Nash
equilibria but, with no quantitative degradation, for mixed Nash equilibria, for
correlated equilibria, and for the average objective value of every outcome
sequence generated by no-regret learners, the last clause being the coarse
correlated case.
\end{theorem}

\Cref{thm:smoothness} is the result in this monograph with the most immediate
practical value, and the reason is the last clause. \Cref{sec:solution-concepts}
established that no-regret self-play converges to the coarse correlated
equilibria of the empirical game and not to anything stronger, which looks like
a weak guarantee. Smoothness converts it into a welfare guarantee: whenever the
underlying game is smooth, the welfare bound that holds at pure equilibria holds
automatically at the limit of independent no-regret learning, with no assumption
that the learners coordinate or converge. For a practitioner the recipe is to
check smoothness of the deployment game once, offline, and then to deploy
independent learners with the welfare bound in hand.

\begin{example}[The selfish-routing bound, and what a toll buys]
\label{ex:tolls}
In Pigou's network of \cref{ex:pigou} the equilibrium cost is $1$ and the
optimum is $0.75$, so the price of anarchy is $4/3$, and the same bound holds
for all of nonatomic selfish routing with affine latencies
\citep{roughgarden2002bad}. Because the bound is proved by a smoothness
argument, \cref{thm:smoothness} extends it to the empirical distribution of play
of independent no-regret learners, so a fleet of selfish routing agents trained
by any no-regret method overloads the network by at most this factor.

Now add a toll. Charge the lower link a price equal to the marginal congestion
it causes, that is $x$, so that a unit using it perceives $x + x = 2x$. Selfish
traffic now equalises $2x$ against $1$ and sends $x = 1/2$ along the lower link,
which is the planner's split, and the resulting physical cost is $0.75$. The
toll is a Pigouvian tax; in the vocabulary of \cref{sec:potential} it is a
reward modification that changes the equilibrium rather than preserving it, and
in the vocabulary of \cref{sec:stackelberg} it is the leader's commitment in a
Stackelberg game whose followers are the routing agents. The three vocabularies
describe one intervention and the arithmetic is the same in all three. \exend
\end{example}

The remainder of the correspondence is a list of problems that practitioners
meet under other names. Automated mechanism design computes the rules by
optimisation rather than by closed form: by search and mixed-integer
programming in \citet{conitzer2004complexity}, and by supervised deep learning
over a parameterised class of mechanisms in \citet{dutting2024optimal}. Neither
uses reinforcement learning, and in a multi-agent system whose designer may set
rewards this is nevertheless the formal version of reward search. The principal-agent problem, in which one party chooses an
incentive scheme and another chooses an effort level that the first observes
only through a noisy outcome, is the model of delegation to a learned policy,
and the standard result that the optimal contract trades off insurance against
incentives is the formal version of the trade-off between a dense shaped reward
and a sparse outcome reward. Two-sided matching \citep{gale1962college} is the
model of task assignment in multi-robot systems, and the deferred-acceptance
algorithm produces a stable assignment in which no pair prefers each other to
their assigned partners, a stability condition that is the matching analogue of
the core of \cref{sec:coalitions}. The economy-level version, in which a learned
planner sets taxes and learned agents respond, is a Stackelberg problem solved
by nested reinforcement learning \citep{zheng2022ai}, and its difficulties are
the difficulties of \cref{eq:stackelberg-gradient} at scale.

\begin{table}[t]
\centering
\caption{Mechanism design read as reward design. The third column names the
reinforcement learning construction and the fourth records what the mechanism
guarantees and at what cost.}
\label{tab:mechanism}
\small
\begin{tabular}{@{}L{2.8cm} L{3.2cm} L{3.2cm} L{3.2cm}@{}}
\toprule
\textsc{Mechanism} & \textsc{Guarantee} & \textsc{RL counterpart} & \textsc{Price paid} \\
\midrule
Second price & Dominant-strategy truthful & Reward independent of own report & Revenue below first price in some settings \\
VCG \citep{vickrey1961counterspeculation} & Truthful and welfare maximising & Difference reward; counterfactual baseline & Not budget balanced \\
Myerson optimal & Revenue maximising in expectation & Reward shaped by a reserve threshold & Efficiency is sacrificed \\
Majority voting & Simple and anonymous & Ensemble reward by vote & Not strategy-proof (\cref{thm:arrow-gs}) \\
Scalarised aggregation & Single scalar reward & Bradley-Terry reward model & Independence of irrelevant alternatives \\
Pigouvian toll & Shifts equilibrium to the optimum & Reward shaping that is not potential-based & Changes the equilibrium set \\
Deferred acceptance \citep{gale1962college} & Stable matching & Task assignment & Not strategy-proof for both sides \\
Learned mechanism \citep{dutting2024optimal,zheng2022ai} & Approximate, data-driven & Bilevel RL & No exact guarantee \\
\bottomrule
\end{tabular}
\end{table}

\section{Potential and supermodular structure, value factorisation, and credit}
\label{sec:potential}

Two distinct notions of potential appear in this field and they are regularly
confused, so we separate them before using either. The first is the potential of
a potential game, a function on joint actions whose differences reproduce every
agent's unilateral payoff differences. The second is the potential of
potential-based reward shaping, a function on states whose discounted
differences are added to the reward without changing the optimal policy. They
are different objects serving different purposes, and both have a role here.

The two readings of this section's material are a structural condition and an
architectural constraint. A game theorist asks which games admit a potential,
because in those games unilateral improvement is a hill-climb on a single scalar
and therefore terminates, which is the strongest convergence property available
to uncoupled learners. A reinforcement learning researcher asks which
architectures force a joint value function to have the property that
decentralised greedy action selection is jointly greedy, which is the
Individual-Global-Max condition. The two questions are the same question asked
about the game and about the hypothesis class respectively, and the
correspondence is that a value factorisation is a commitment to solving only
games with a particular structure.

\begin{nomen}{Potential, factorisation and credit vocabulary read from both sides.}{tab:nom-potential}
Exact potential $\pot$ & A single scalar all agents hill-climb &
  Unilateral payoff differences match $\pot$ \\
Finite improvement property & Better-reply dynamics terminate &
  No cycles; a pure equilibrium exists \\
Markov potential game & Class where independent PG converges &
  The sequential version of the above \\
Supermodular game & Monotone best responses &
  More of one agent's action invites more of the other's \\
Strategic complements & Positive feedback between agents &
  Topkis' lattice structure \\
Difference reward $D_i$ & Counterfactual baseline; COMA advantage &
  Team return minus your absence \\
Shapley value $\phi_i$ & Order-averaged credit &
  The same difference, over all orders \\
Efficiency axiom & Credits sum to the realised return &
  A decomposition that decomposes \\
Individual-Global-Max & Decentralised greedy is jointly greedy &
  Argmax factorises across agents \\
Monotone mixing & QMIX-style architecture &
  Sufficient for IGM, not necessary \\
Potential-based shaping & $\gamma \Psi(s') - \Psi(s)$ added to the reward &
  Invisible to the solution concept \\
Equilibrium-preserving shaping & Per-agent potential terms &
  Preserves the set, not the selection \\
\end{nomen}

\begin{definition}[Potential game; \citealp{monderer1996potential}]
\label{def:potential-game}
A normal-form game is an exact potential game when there exists
$\pot : \Ac \to \Real$ such that for every agent $i$, every
$a_{-i} \in \Ac_{-i}$ and every $a_i, a_i' \in \Ac_i$,
\begin{align}
\util_i(a_i, a_{-i}) - \util_i(a_i', a_{-i})
\;=\;
\pot(a_i, a_{-i}) - \pot(a_i', a_{-i}).
\label{eq:exact-potential}
\end{align}
It is an ordinal potential game when the two sides of \cref{eq:exact-potential}
merely have the same sign, and a weighted potential game when they are
proportional with an agent-dependent positive weight.
\end{definition}

\begin{rlreading}
A potential is a single scalar that every agent is implicitly hill-climbing even
though no agent is optimising it and no agent can see it. That is the property
single-agent reinforcement learning has by construction, since the return is one
number, and that multi-agent reinforcement learning loses. Recovering it is
worth a great deal: $\pot$ is a Lyapunov function for better-reply dynamics, so
independent learners that only take improving steps cannot cycle and must stop
at a local maximum of $\pot$, which is a pure Nash equilibrium. This is the
strongest convergence statement available in this monograph to agents that do
not coordinate, and it is why the structural check in \cref{sec:training} asks
for potential structure before it asks anything else. The sequential version,
the Markov potential game at the end of this section, is the class in which
independent policy gradient provably converges.
\end{rlreading}

The structural consequence just invoked deserves its name. Any sequence of
unilateral strictly improving deviations increases $\pot$ and therefore
terminates, which is the finite improvement property, so every finite ordinal
potential game has a pure Nash equilibrium and better-reply dynamics converge to
one. Two special cases cover most of what arises in practice. A common-payoff game is an exact potential game with
$\pot = \util$, so cooperative tasks are potential games by construction. A game
in which each agent's payoff is its own contribution to a shared objective minus
a term independent of its own action is also an exact potential game, which is
the design principle behind difference rewards.

\begin{example}[A potential, exhibited and used]
\label{ex:potential}
Take the two-agent game with row payoffs
$\big(\begin{smallmatrix} 2 & 4 \\ 5 & 1 \end{smallmatrix}\big)$ and column
payoffs $\big(\begin{smallmatrix} 2 & 4 \\ 4 & 0 \end{smallmatrix}\big)$, which
are not equal, so the game is not common-payoff. It is nevertheless an exact
potential game with
$\pot = \big(\begin{smallmatrix} 0 & 2 \\ 3 & -1 \end{smallmatrix}\big)$, as one
verifies in four lines: for the row agent,
$2 - 5 = -3 = 0 - 3$ in the first column and $4 - 1 = 3 = 2 - (-1)$ in the
second; for the column agent, $2 - 4 = -2 = 0 - 2$ in the first row and
$4 - 0 = 4 = 3 - (-1)$ in the second.

The potential does three things at once. It proves a pure equilibrium exists,
since a finite function has a maximum and the maximiser is unlocked by no
unilateral deviation. It identifies the equilibria as the local maxima of
$\pot$: here $(0,1)$ with $\pot = 2$ and $(1,0)$ with $\pot = 3$ are both
equilibria while $(0,0)$ and $(1,1)$ are not. And it orders them, with the
global maximiser $(1,0)$ distinguished, which is the natural selection rule and
is the one that stochastic stability picks out in \cref{ex:log-linear}. Two
independent learners that each improve greedily will stop at one of the two
local maxima and which one depends on the initial condition, which is the
general shape of equilibrium selection in this class: convergence is
guaranteed, the destination is not. \exend
\end{example}

\begin{proposition}[Difference rewards preserve the potential structure]
\label{prop:difference-potential}
Let $\util_i(a) = \util(a)$ be a common payoff and define the difference reward
$D_i(a) = \util(a) - \util(c_i, a_{-i})$ for an arbitrary fixed default action
$c_i \in \Ac_i$. Then the game with payoffs $\{D_i\}$ is an exact potential
game with potential $\util$, and it has the same pure Nash equilibria and the
same better-reply graph as the original common-payoff game.
\end{proposition}

\begin{proof}
For fixed $a_{-i}$ the subtracted term $\util(c_i, a_{-i})$ does not depend on
$a_i$, so $D_i(a_i,a_{-i}) - D_i(a_i',a_{-i}) = \util(a_i,a_{-i}) -
\util(a_i',a_{-i})$, which is \cref{eq:exact-potential} with $\pot = \util$.
Identity of the better-reply graphs is immediate from the same display.
\end{proof}

\Cref{prop:difference-potential} is the formal content of the difference-reward
construction of \citet{wolpert2001optimal}: subtracting an agent-independent
baseline changes the variance of the learning signal and changes nothing about
the equilibria. The counterfactual baseline of COMA
\citep{foerster2018counterfactual} is the same operation applied to the
advantage inside a policy gradient, where it is additionally unbiased because
the baseline does not depend on the agent's own action, so the correspondence is
that an exact-potential-preserving reward transformation in game theory is an
unbiased baseline in policy-gradient estimation. The Shapley value supplies the
principled version of the same idea.

\begin{definition}[Shapley value; \citealp{shapley1953value}]
\label{def:shapley}
Let $v : 2^{\agents} \to \Real$ be a characteristic function with
$v(\emptyset) = 0$. The Shapley value of agent $i$ is
\begin{align}
\phi_i(v) \;=\; \sum_{C \subseteq \agents \setminus \{i\}}
\frac{\abs{C}!\,(\nag - \abs{C} - 1)!}{\nag!}
\Big[ v(C \cup \{i\}) - v(C) \Big].
\label{eq:shapley}
\end{align}
It is the unique map satisfying efficiency $\sum_i \phi_i(v) = v(\agents)$,
symmetry, the null-player axiom and linearity.
\end{definition}

\begin{rlreading}
Read $v(C)$ as the return the team obtains when the agents outside $C$ are
replaced by a default policy, which is a quantity a simulator can measure by
ablation. Then $\phi_i$ is a credit signal for agent $i$, and the four axioms
are four properties a practitioner would want from one. Efficiency says the
credits sum to the realised team return, so nothing is invented or lost. The
null-player axiom says an agent whose presence never changes the return receives
nothing. Symmetry says credit tracks contribution rather than agent index.
Linearity says credit for a sum of tasks is the sum of the credits. The simplest
instance fits in one line: with two agents, $v(\{1\}) = v(\{2\}) = 0$ and
$v(\{1,2\}) = 10$, \cref{eq:shapley} gives $\phi_1 = \phi_2 = 5$, since each
agent adds $0$ when it arrives first and $10$ when it arrives second.
\Cref{ex:difference-vs-shapley} works a three-agent case in which the cheap
alternative over-attributes by exactly the synergy.
\end{rlreading}

\begin{proposition}[Difference rewards are a single Shapley term]
\label{prop:difference-shapley}
Fix a joint action $a$ and define the characteristic function
$v_a(C) = \util(a_C, c_{-C}) - \util(c)$, the gain over the all-default profile
when the agents in $C$ play their components of $a$ and the rest play their
default actions, normalised so that $v_a(\emptyset) = 0$ as
\cref{def:shapley} requires. Then the difference
reward $D_i(a)$ of \cref{prop:difference-potential} equals the marginal
contribution $v_a(C \cup \{i\}) - v_a(C)$ evaluated at the single coalition
$C = \agents \setminus \{i\}$, which is the term of weight $1/\nag$ in
\cref{eq:shapley}. The Shapley value averages this marginal contribution over
all $2^{\nag - 1}$ coalitions with the stated weights, and coincides with the
difference reward for every $a$ when $\util$ is additively separable.
\end{proposition}

\begin{example}[Where the two credit schemes disagree, in numbers]
\label{ex:difference-vs-shapley}
Three agents share a task. Working alone the first two each produce $4$ and the
third produces $1$. Any pair produces a bonus of $3$ on top of their individual
contributions and all three produce a further $2$, so the characteristic
function is $v(\{1\}) = v(\{2\}) = 4$, $v(\{3\}) = 1$, $v(\{1,2\}) = 11$,
$v(\{1,3\}) = v(\{2,3\}) = 8$ and $v(\{1,2,3\}) = 14$.

The Shapley value of \cref{eq:shapley} is $(5.667,\, 5.667,\, 2.667)$, which
sums to $14$ exactly, as efficiency demands. The difference reward with absence
as the default is $D_i = v(\agents) - v(\agents \setminus \{i\})$, giving
$(6,\, 6,\, 3)$, which sums to $15$.

The shortfall of $1$ is the synergy that no single agent can claim alone, and
the fact that the difference rewards over-attribute by exactly that amount is
the general phenomenon: marginal contributions evaluated at the full coalition
double count complementarities. In a cooperative multi-agent reinforcement
learning system this is the difference between a credit signal that sums to the
team return and one that does not, and it matters whenever the credits are used
as learning targets rather than merely as baselines, because the surplus is
distributed arbitrarily across agents. The cost of the principled version is
$2^{\nag-1}$ evaluations instead of one, which is why Shapley-based multi-agent
methods \citep{wang2020shapley,wang2022shaq} estimate \cref{eq:shapley} by
sampling permutations. \exend
\end{example}

The reinforcement learning reading is that difference rewards and Shapley-value
credit assignment differ in how many counterfactuals they average over, that the
former is cheap and order-dependent while the latter is expensive and
order-free, and that the sampling-based estimators used by Shapley-value
multi-agent methods \citep{wang2020shapley,wang2022shaq} are Monte Carlo
approximations of \cref{eq:shapley}. The axioms are what justify the extra cost:
efficiency says the credits sum to the realised team return, which is the
property that makes a decomposition a decomposition, and it is exactly the
property that value-factorisation architectures impose by construction.

\begin{definition}[Individual-Global-Max]
\label{def:igm}
A joint action-value $Q_{\mathrm{tot}}(s,a)$ and per-agent utilities
$\{Q_i(s,a_i)\}$ satisfy the Individual-Global-Max condition at $s$ when
\begin{align}
\argmax_{a \in \Ac} Q_{\mathrm{tot}}(s,a)
\;=\;
\Big( \argmax_{a_1} Q_1(s,a_1), \dots, \argmax_{a_\nag} Q_\nag(s,a_\nag) \Big).
\label{eq:igm}
\end{align}
\end{definition}

\begin{rlreading}
\Cref{eq:igm} is the condition under which decentralised execution is free. The
left-hand side is what a centralised controller would do, namely maximise the
joint value over $\prod_i \Ac_i$; the right-hand side is what $\nag$ independent
heads do, namely each maximise its own small utility over its own actions. When
the two agree, the agents need no communication, no shared seed and no
tie-breaking convention at run time, because greedy local action selection lands
on the global optimum by construction, and the coordination problem of
\cref{rem:mmdp-gap} does not arise. The cost is that not every joint value
function admits such a decomposition, so imposing \cref{eq:igm} through an
architecture restricts the hypothesis class, and \cref{ex:igm-counterexample}
exhibits a two-by-two matrix outside it.
\end{rlreading}

\begin{proposition}[Monotone mixing implies IGM]
\label{prop:monotone-igm}
If $Q_{\mathrm{tot}}(s,a) = f_s(Q_1(s,a_1),\dots,Q_\nag(s,a_\nag))$ with
$\partial f_s / \partial Q_i \ge 0$ for every $i$, then
\cref{eq:igm} holds. In particular the additive factorisation
$Q_{\mathrm{tot}} = \sum_i Q_i$ of \citet{sunehag2018value} and the monotone
mixing network of \citet{rashid2018qmix} satisfy IGM. The converse fails: there
exist common-payoff stage games whose action-value admits no monotone
factorisation, the canonical example being the two-agent three-action matrix
game of \citet{son2019qtran}, so the representable class is strictly smaller than the
class satisfying IGM.
\end{proposition}

\begin{example}[A two-by-two matrix that no monotone mixing network represents]
\label{ex:igm-counterexample}
Let the common payoff at a state be
$Q_{\mathrm{tot}} = \big(\begin{smallmatrix} 8 & -12 \\ -12 & 0
\end{smallmatrix}\big)$. Suppose a monotone $f$ and per-agent utilities
reproduced it. Monotonicity in the second argument means that the sign of
$Q_{\mathrm{tot}}(i,0) - Q_{\mathrm{tot}}(i,1)$ equals the sign of
$Q_2(0) - Q_2(1)$ for every row $i$. The first row gives
$8 - (-12) = 20 > 0$ and the second gives $-12 - 0 = -12 < 0$, so the required
sign is both positive and negative. No monotone $f$ exists.

The Individual-Global-Max condition nevertheless holds for this matrix, since
its argmax is the corner $(0,0)$, which is the product of the per-agent argmaxes
of the utilities $Q_1 = (1,0)$ and $Q_2 = (1,0)$. So the matrix is in the IGM
class and outside the monotone class, which is \cref{prop:monotone-igm}'s
converse failing in the smallest possible instance. The best additive surrogate
$Q_1(i) + Q_2(j)$ misfits by $8$ in the supremum norm, because additive
functions satisfy $M_{00} + M_{11} = M_{01} + M_{10}$ while this matrix has
$8 + 0 = 8$ against $-12 - 12 = -24$, a discrepancy of $32$ that is best split
four ways. A value-decomposition network trained on this stage game converges to
a function that is wrong by $8$ at every cell and no quantity of data corrects
it. Detecting the failure requires comparing against a centralised planner on
the same model, which is the fifth row of \cref{tab:diagnostics}. \exend
\end{example}

The game-theoretic content of \cref{eq:igm} is that the stage game induced at
$s$ by $Q_{\mathrm{tot}}$ is one in which greedy independent action selection
lands on a global optimum, so the coordination problem of \cref{rem:mmdp-gap}
disappears by construction. Factorisation therefore operates as an equilibrium-selection device: it
restricts the hypothesis class to games in which decentralised greedy play is
optimal, and it pays for that restriction with representational capacity, which is what the
counterexamples motivating QTRAN \citep{son2019qtran} and QPLEX
\citep{wang2021qplex} exhibit. Reading the factorisation literature through
\cref{def:igm} makes the trade-off legible: every proposal in that family is a
choice of how large a class of stage games to admit while retaining
\cref{eq:igm}.

A second structural class deserves a place beside potential games because it
yields convergence of best-response dynamics for an entirely different reason.

\begin{definition}[Supermodular game and strategic complements;
\citealp{topkis1979equilibrium,milgrom1990rationalizability}]
\label{def:supermodular}
Order each $\Ac_i$ as a lattice. A game is supermodular when each $\util_i$ is
supermodular in $a_i$ and has increasing differences in $(a_i, a_{-i})$, meaning
that the gain from raising $a_i$ is non-decreasing in $a_{-i}$. In such a game
the best-response correspondence is non-decreasing and the set of Nash
equilibria has a least and a greatest element \citep{topkis1979equilibrium};
iterated elimination of dominated strategies converges to the interval between
them, so that the rationalisable set and the equilibrium set have the same
extreme points \citep{milgrom1990rationalizability}. That the equilibrium set
is in addition a complete lattice is a later result and is not what
\citet{topkis1979equilibrium} proves.
\end{definition}

\begin{rlreading}
Supermodularity says that the more one agent does of something, the more the
others want to do of it: effort is contagious. Coverage, investment, adoption
and most tasks in which agents reinforce each other's progress have this shape.
The payoff for a practitioner is a convergence guarantee obtained without any
step-size condition. Start every head at its smallest action and iterate best
responses: the sequence is monotone increasing and converges to the least
equilibrium. Start every head at its largest and it converges to the greatest.
The two limits bracket every equilibrium of the game, so two cheap runs return
an interval that is guaranteed to contain the answer, which is a diagnostic no
other class in this monograph offers.
\end{rlreading}

The reinforcement learning reading is practical. In a supermodular game the
best-response map is monotone, so starting all agents at the smallest action
profile and repeatedly best responding produces an increasing sequence that
converges to the least equilibrium, and starting at the largest produces a
decreasing sequence converging to the greatest. Independent learners therefore
converge from either extreme without any step-size condition, and the two limits
bracket every equilibrium. Problems in which more effort by one agent raises the
return to effort by another, which covers most investment, adoption and
coverage tasks, are supermodular, and recognising this gives the same kind of
guarantee that potential structure gives, by a different route.

The second notion of potential closes the section. Potential-based reward
shaping adds $F(s,a,s') = \gamma \Psi(s') - \Psi(s)$ to the reward for an
arbitrary $\Psi : \St \to \Real$, and \citet{ng1999policy} prove that this is
the only form of additive shaping that leaves the set of optimal policies
unchanged for all transition kernels and rewards. The multi-agent extension of
\citet{devlin2011theoretical} shows that adding an agent-specific
potential-based term $F_i(s,s')$ to agent $i$'s reward in a stochastic game
leaves the set of Nash equilibria unchanged, which is the right generalisation:
in the single-agent case shaping preserves optimality, in the multi-agent case
it preserves the equilibrium set. It does not preserve which equilibrium the
learning dynamics selects, a caveat \citet{devlin2011theoretical} state
themselves, so shaping remains an equilibrium-selection device even when it is
solution-preserving. The two results together say that potential-based
shaping is the class of reward modifications that are invisible to the solution
concept, which is precisely what a designer wants from a tool meant to
accelerate learning rather than to change the problem. The connection to
\cref{def:potential-game} is that both results identify a potential function
whose differences cancel in the quantity the solution concept compares, over
time in the shaping case and across unilateral deviations in the potential-game
case.

\begin{example}[Shaping that preserves and shaping that does not]
\label{ex:shaping}
Add $F(s,a,s') = \gamma \Psi(s') - \Psi(s)$ with $\Psi$ equal to the negative
distance to a goal. Along any trajectory the added terms telescope to
$\gamma^T \Psi(s_T) - \Psi(s_0)$, so the return of every policy shifts by a
constant depending only on the endpoints, and the ranking of policies from a
fixed start is unchanged. This is the single-agent statement.

Now give two agents the shaped rewards $\Rew_1 + \gamma\Psi_1(s') - \Psi_1(s)$
and $\Rew_2 + \gamma\Psi_2(s') - \Psi_2(s)$ with each $\Psi_i$ a function of the
state alone. Each agent's value shifts by a constant for every joint policy, so
no unilateral deviation changes sign and the equilibrium set is unchanged. That
is the multi-agent statement, and the hypothesis that each $\Psi_i$ depends on
the state and not on the other agents' actions is essential. Add instead a term
that rewards agent $1$ when agent $2$ does well, as cooperation-inducing shaping
schemes do, and the payoff tensor changes: the prisoner's dilemma with each
agent's reward replaced by the sum of both rewards is a common-payoff game with
$(C,C)$ as the unique equilibrium, which is a different game with a different
answer. That may well be the right engineering decision, and the point is that
it must be reported as a change to the problem rather than as an acceleration of
its solution. \exend
\end{example}

Finally, the sequential generalisation of \cref{def:potential-game} is the
Markov potential game, in which there exists $\pot_\pi$ such that
$V_i^{(\pi_i,\pi_{-i})}(s) - V_i^{(\pi_i',\pi_{-i})}(s) =
\pot_{(\pi_i,\pi_{-i})}(s) - \pot_{(\pi_i',\pi_{-i})}(s)$ for every state, agent
and unilateral deviation \citep{macua2018learning,leonardos2022global}. In this
class independent policy gradient converges to a Nash equilibrium at a rate
polynomial in the relevant quantities, which is the largest general-sum
stochastic game class we are aware of for which a convergence guarantee of that
strength is available.
Identifying whether an application lies in it is therefore the single most
valuable structural check a practitioner can perform before choosing an
algorithm, and the two sufficient conditions that cover most cases are a common
reward and a reward of the form
$\Rew_i(s,a) = \Rew^{\mathrm{shared}}(s,a) + g_i(s,a_i)$ with the second term
depending only on the agent's own action.

\section{Coalitions, stability of credit, and bargaining}
\label{sec:coalitions}

Everything so far has been non-cooperative: agents choose actions and the
analysis predicts what they choose. Cooperative game theory abstracts away from
actions entirely and asks a different question. Given that groups of agents can
achieve certain values together, how should the value of the whole group be
divided, and which divisions are stable against a subgroup walking away? The
multi-agent reinforcement learning counterpart is credit assignment, and the
abstraction is exactly right for it: a credit scheme is a division rule, the
value of a coalition is the return achievable when the rest are replaced by a
default, and stability is the property that no subset of agents would rather
train on its own.

The two readings differ in what the numbers mean. A cooperative game theorist
treats $v(C)$ as a binding agreement about what the coalition $C$ can guarantee,
and the solution concepts are normative: they say how a surplus ought to be
divided. A reinforcement learning researcher treats $v(C)$ as a counterfactual
return measured by freezing the agents outside $C$, and the solution concepts
are operational: they say what number to put in the loss. The normative reading
supplies axioms that the operational one inherits, and the axioms are not
decorative, since efficiency is what makes a decomposition sum to the team
return and stability is what makes a decomposition robust to agents leaving.

\begin{nomen}{Cooperative-game vocabulary read from both sides.}{tab:nom-coalitions}
Characteristic function $v(C)$ & Return with the agents outside $C$ frozen &
  What a subgroup achieves on its own \\
Grand coalition $v(\agents)$ & Team return of the full joint policy &
  The number to be divided \\
Imputation $x \in \Real^\nag$ & Per-agent credit vector &
  What each agent is told it earned \\
Efficiency & Credits sum to the realised return &
  No surplus invented or lost \\
Null-player axiom & Zero credit for a no-op agent &
  Doing nothing earns nothing \\
Symmetry & Interchangeable agents get equal credit &
  Credit depends on contribution, not label \\
Core & No subgroup would rather defect &
  Stability of a credit scheme \\
Nucleolus & Lexicographic min-max of complaints &
  The least objectionable division \\
Banzhaf index \citep{dubey1979mathematical} & Unweighted marginal contribution &
  A cheaper, non-efficient alternative \\
Superadditivity & Merging coalitions never hurts &
  Cooperation is weakly beneficial \\
Convexity & Increasing marginal contributions &
  Late joiners add more; core is large \\
Bargaining set, frontier & Pareto frontier of achievable returns &
  What is jointly attainable \\
Disagreement point $d$ & Return if cooperation fails &
  Each agent's fallback \\
Nash bargaining solution & Maximiser of $\prod_i (u_i - d_i)$ &
  Product of advantages over fallback \\
\end{nomen}

\begin{definition}[Transferable-utility game, core and nucleolus;
\citealp{shapley1971cores,schmeidler1969nucleolus}]
\label{def:core}
A transferable-utility game is a pair $(\agents, v)$ with
$v : 2^{\agents} \to \Real$ and $v(\emptyset) = 0$. It is superadditive when
$v(C \cup D) \ge v(C) + v(D)$ for disjoint $C$ and $D$, and convex when
$v(C \cup \{i\}) - v(C)$ is non-decreasing in $C$. The core is
\begin{align}
\mathrm{core}(v) \;=\; \Big\{ x \in \Real^{\nag} \;:\;
\textstyle\sum_{i} x_i = v(\agents), \;\;
\sum_{i \in C} x_i \ge v(C) \ \ \forall\, C \subseteq \agents \Big\},
\label{eq:core}
\end{align}
the set of efficient divisions that no coalition can improve upon. The excess of
$C$ at $x$ is $e(C,x) = v(C) - \sum_{i \in C} x_i$, and the nucleolus is the
unique efficient $x$ lexicographically minimising the vector of excesses sorted
in decreasing order. In a convex game the core is non-empty and contains the
Shapley value; the nucleolus always lies in the core when the core is non-empty.
\end{definition}

\begin{rlreading}
The core answers a question the Shapley value does not: not what each agent
contributed, but what each agent could insist on. An allocation is in the core
when no subset of agents could do better by leaving and operating on its own, so
membership is the stability condition for a federation of independently owned
agents. The inequalities are one per subset, and checking them is a linear
program. The excess $e(C,x)$ is the complaint of coalition $C$ at the allocation
$x$, and the nucleolus is the allocation that minimises the loudest complaint,
then the second loudest, and so on; it is the fallback when the core is empty,
which \cref{ex:voting-core} shows is a real possibility. For a system whose
agents are owned by one party the core is a diagnostic rather than a constraint,
and the Shapley value remains the better learning signal because it always
exists and is continuous in $v$.
\end{rlreading}

The core is defined by $2^\nag$ linear inequalities, so membership is a linear
program, and emptiness is a real possibility. The two smallest examples exhibit
both outcomes and both are worth having by heart.

\begin{example}[The glove market: scarcity takes everything]
\label{ex:glove}
Three agents. The first two each hold a left glove and the third holds a right
glove, and a matched pair is worth $1$ while unmatched gloves are worth nothing,
so $v(C) = \min(\abs{C \cap \{1,2\}},\, \abs{C \cap \{3\}})$. The grand
coalition is worth $1$.

The core is the single point $(0, 0, 1)$. To see it, note that any core
allocation must satisfy $x_1 + x_3 \ge v(\{1,3\}) = 1$ and
$x_2 + x_3 \ge 1$, and efficiency gives $x_1 + x_2 + x_3 = 1$. Adding the first
two inequalities gives $x_1 + x_2 + 2x_3 \ge 2$, and subtracting efficiency
gives $x_3 \ge 1$; with non-negativity this forces $x_3 = 1$ and
$x_1 = x_2 = 0$. The nucleolus, being in the core, is the same point.

The Shapley value is $(\tfrac16, \tfrac16, \tfrac23)$, which is not in the core.
It gives the left-glove holders something because in two of the six orderings a
left-glove holder arrives after the right-glove holder and adds the whole unit
of value.

Both answers are correct and they answer different questions. The Shapley value
asks what each agent contributes on average and the core asks what each agent
can insist on. For a multi-agent system the practical consequence is sharp: a
credit scheme based on average marginal contribution will pay agents whose
participation is not actually necessary, and a credit scheme based on
bargaining power will pay nothing to agents who are substitutes for one another.
Which is wanted depends on whether the credits are being used as learning
signals, where the Shapley answer is better behaved because it is continuous in
$v$, or as a stability analysis of a federation of independently owned agents,
where the core answer is the relevant one. \exend
\end{example}

\begin{example}[Majority voting: no stable division exists]
\label{ex:voting-core}
Three agents, any two of whom can secure a unit of value and no single one of
whom can secure anything, so $v(C) = 1$ when $\abs{C} \ge 2$ and $0$ otherwise.
Efficiency gives $x_1 + x_2 + x_3 = 1$, and the three pair constraints give
$x_1 + x_2 \ge 1$, $x_1 + x_3 \ge 1$ and $x_2 + x_3 \ge 1$. Adding the three
gives $2(x_1 + x_2 + x_3) \ge 3$, that is $1 \ge 1.5$, a contradiction. The core
is empty.

The Shapley value is $(\tfrac13, \tfrac13, \tfrac13)$ by symmetry and remains
well defined, which is the sense in which it is more robust than the core: it
always exists, is unique and is continuous. The emptiness of the core is not a
pathology of the example but the formal content of an unstable alliance, and it
occurs in any multi-agent system in which coalitions overlap and compete. A
practitioner whose agents are independently owned should expect no division to
be immune from renegotiation in such a system, and the nucleolus, which minimises
the largest complaint even when no complaint can be eliminated, is the natural
fallback. \exend
\end{example}

\begin{example}[The warehouse team, with all three answers]
\label{ex:warehouse-core}
Return to the three agents of \cref{ex:shapley-team}, whose characteristic
function was $v(\{1\}) = v(\{2\}) = 10$, $v(\{3\}) = 4$, $v(\{1,2\}) = 26$,
$v(\{1,3\}) = v(\{2,3\}) = 16$ and $v(\agents) = 36$. The Shapley value is
$(14.67, 14.67, 6.67)$. The core is non-empty, and the Shapley value lies in it,
because every pair constraint is satisfied: $14.67 + 14.67 = 29.33 \ge 26$ and
$14.67 + 6.67 = 21.33 \ge 16$. The nucleolus is $(14.5, 14.5, 7)$, which differs
from the Shapley value by a sixth of a unit for each of the first two agents and
by a third in the other direction for the third, because it equalises the worst
complaint rather than averaging contributions.

Three schemes, three answers, all efficient and all defensible. The difference
between $14.67$ and $14.5$ is small here and is not always small; what matters
for a reinforcement learning system is that the choice of scheme is a modelling
decision with consequences for the gradient each agent receives, and reporting
which scheme was used is as necessary as reporting the learning rate. \exend
\end{example}

The bargaining problem is the non-transferable counterpart, in which the
agents cannot simply move utility between themselves and must agree on a point
of a feasible set.

\begin{definition}[Bargaining problem and its solutions;
\citealp{nash1950bargaining,kalai1975other,rubinstein1982perfect}]
\label{def:bargaining}
A bargaining problem is a pair $(\mathcal{U}, d)$ with $\mathcal{U} \subseteq
\Real^\nag$ compact and convex and $d \in \mathcal{U}$ the disagreement point.
The Nash bargaining solution is the maximiser of $\prod_i (u_i - d_i)$ over
$\{u \in \mathcal{U} : u \ge d\}$, and is the unique solution satisfying Pareto
efficiency, symmetry, invariance to affine rescaling of each agent's utility and
independence of irrelevant alternatives. The Kalai-Smorodinsky solution replaces
the last axiom by monotonicity and selects the point on the Pareto frontier
along the ray from $d$ to the ideal point $b$, where
$b_i = \max\{u_i : u \in \mathcal{U}, u \ge d\}$. The egalitarian solution
maximises $\min_i (u_i - d_i)$.
\end{definition}

\begin{rlreading}
A bargaining problem is a multi-objective optimisation problem with a reference
point, and the three solutions are three scalarisations of it. The feasible set
$\mathcal{U}$ is the Pareto frontier of achievable return vectors, which is the
object a multi-objective reinforcement learning method approximates, and the
disagreement point $d$ is the return each agent obtains if cooperation fails,
which is the per-agent fallback a designer must specify and usually leaves
implicit. The Nash solution maximises the product of advantages over the
fallback, the Kalai-Smorodinsky solution equalises advantages as fractions of
the best each agent could possibly obtain, and the egalitarian solution
maximises the smallest advantage. Choosing among them is choosing an axiom, not
an optimiser, and \cref{ex:bargaining} shows that the choice moves the answer by
a visible margin as soon as the agents differ in risk attitude.
\end{rlreading}

\begin{example}[Risk aversion is punished by one axiom and protected by another]
\label{ex:bargaining}
A unit pie is to be divided. The first agent is risk neutral, so its utility is
its share $x$. The second is risk averse, with utility $\sqrt{y}$ where
$y = 1 - x$. The disagreement point is $(0,0)$ and the ideal point is $(1,1)$,
since either agent could in principle receive the whole pie.

The Nash solution maximises $x \sqrt{1-x}$. Differentiating,
$\sqrt{1-x} - x / (2\sqrt{1-x}) = 0$, that is $2(1-x) = x$, so $x = 2/3$. The
risk-neutral agent receives two thirds of the pie, the utilities are
$(0.667,\, 0.577)$ and the Nash product is $0.385$.

The Kalai-Smorodinsky solution picks the frontier point with equal proportional
gains, which here means $x = \sqrt{1-x}$, that is $x^2 + x - 1 = 0$ and
$x = (\sqrt5 - 1)/2 = 0.618$. Both agents obtain utility $0.618$ and the
risk-neutral agent receives a smaller share than under Nash. The egalitarian
solution coincides with Kalai-Smorodinsky here because the ideal point is
symmetric.

The phenomenon is general and is the reason the two solutions are both in use.
The Nash solution is invariant to each agent's utility scale, which means it
cannot see how risk averse anyone is, and the consequence is that risk aversion
is rewarded with a smaller share because a risk-averse agent gains more utility
per unit of pie near zero. Kalai-Smorodinsky gives up that invariance for
monotonicity and equalises relative gains instead. For multi-objective and
multi-agent reinforcement learning the choice between maximising a product of
advantages and equalising normalised advantages is exactly this choice, and it
is a choice of axiom rather than of optimiser. \exend
\end{example}

\begin{example}[Alternating offers makes patience into bargaining power]
\label{ex:rubinstein}
Two agents alternate proposing a split of a unit pie, and each rejection costs
both a factor $\delta$ of discounting. In the unique subgame perfect equilibrium
the proposer offers the responder exactly its continuation value and the
responder accepts, which yields the proposer a share $1/(1+\delta)$. With
$\delta = 0.5$ the proposer takes $2/3$; with $\delta = 0.9$ it takes $0.526$;
with $\delta = 0.99$ it takes $0.5025$. As the agents become patient the
first-mover advantage disappears and the split approaches the symmetric Nash
bargaining solution, which is the sense in which the non-cooperative model
implements the cooperative one.

The derivation is one line of self-reference. Let $\sigma$ be the proposer's
equilibrium share. The responder, by rejecting, becomes the proposer one period
later and obtains $\delta \sigma$, so the proposer must offer it exactly that,
keeping $\sigma = 1 - \delta\sigma$, hence $\sigma = 1/(1+\delta)$. The
reinforcement learning reading is that in a negotiation environment the discount
factor is not a tuning parameter of the algorithm but a determinant of the
equilibrium outcome, exactly as in \cref{ex:grim}, and that the agent who moves
first in the protocol has a structural advantage that no amount of training
removes. Reporting a negotiation result without reporting the protocol and the
discount factor reports almost nothing. \exend
\end{example}

The correspondences this section buys are collected in
\cref{tab:coalitions}. Two of them carry practical weight beyond credit
assignment. The first is ad hoc teamwork \citep{stone2010adhoc}, in which an
agent must join an unfamiliar team and contribute without pre-coordination; in
the vocabulary used here the agent is computing its marginal contribution to a
coalition whose composition it does not control, and the zero-shot coordination
formulation of \cref{sec:bayesian} is its Bayesian version. The second is
federation: when several independently trained and independently owned agents
must decide whether to pool their efforts, the core is the exact criterion for
whether a sharing rule will hold, and \cref{ex:voting-core} shows that in some
structures no rule will.

\begin{table}[t]
\centering
\caption{Cooperative solution concepts and their multi-agent reinforcement
learning uses. The last column records the property that makes the concept
suitable or unsuitable as a learning signal.}
\label{tab:coalitions}
\small
\begin{tabular}{@{}L{2.6cm} L{3.4cm} L{3.4cm} L{3.2cm}@{}}
\toprule
\textsc{Concept} & \textsc{Definition in one line} & \textsc{RL use} & \textsc{Property for learning} \\
\midrule
Difference reward & Return minus return without you & COMA-style baseline & Cheap; does not sum to the return \\
Shapley value & Order-averaged marginal contribution & Per-agent credit & Efficient, unique, continuous; $2^{\nag-1}$ terms \\
Banzhaf index & Unweighted marginal contribution & Cheaper credit proxy & Not efficient \\
Core & No coalition can improve & Stability of a federation & May be empty (\cref{ex:voting-core}) \\
Nucleolus & Lexicographic min-max excess & Least-objectionable division & Always exists; LP sequence \\
Nash bargaining & Maximise $\prod_i (u_i - d_i)$ & Welfare objective across agents & Scale invariant; ignores risk attitude \\
Kalai-Smorodinsky & Equal proportional gains & Normalised multi-objective target & Monotone; not scale invariant \\
Egalitarian & Maximise $\min_i (u_i - d_i)$ & Fairness objective & Non-smooth \\
Rubinstein protocol & Alternating offers with discounting & Negotiation environment & Outcome set by $\delta$ and move order \\
\bottomrule
\end{tabular}
\end{table}

\section{The population limit: mean-field games and mean-field control}
\label{sec:mean-field}

Everything so far has been exponential in the number of agents: the joint action
space is $\prod_i \Ac_i$, the occupancy state of \cref{def:occupancy} ranges
over joint histories, and the complexity results of \cref{sec:complexity} are
stated for a fixed small $\nag$. The mean-field limit removes that dependence by
letting $\nag \to \infty$ under a symmetry assumption, and the resulting models
are the only ones in this monograph whose cost is independent of the population
size.

The two readings are a limit theorem and a training loop. A game theorist takes
the limit to obtain a fixed-point condition between a representative agent's
optimal policy and the distribution it generates, and then asks when the fixed
point exists, is unique and approximates the finite game. A reinforcement
learning researcher sees an outer loop that alternates a single-agent solver
with a population rollout, and asks when the alternation converges. The
correspondence is exact and the two questions are the same question: the
alternation converges precisely when the composed map is a contraction or
satisfies a monotonicity condition, and the devices that rescue it are the same
devices that rescue fictitious play in \cref{sec:learning-dynamics}.

\begin{nomen}{Mean-field vocabulary read from both sides.}{tab:nom-meanfield}
Representative agent & The single policy being trained &
  One agent stands for all of them \\
Population distribution $\mu_t$ & Empirical state distribution of the fleet &
  What the others look like in aggregate \\
Exchangeability & Permutation-invariant architecture &
  Identities do not matter, counts do \\
Mean-field game equilibrium & Fixed point of solve-then-rollout &
  Policy optimal against the flow it creates \\
Mean-field control & Markov decision process on distributions &
  One planner for the whole population \\
Monotonicity condition & Congestion-like interaction &
  Crowds repel, so the fixed point is unique \\
Damped population update & Averaging the outer loop &
  Fictitious play on the distribution \\
$\varepsilon_\nag$-Nash guarantee & Finite-fleet approximation error &
  How wrong the limit is at size $\nag$ \\
Price of anarchy in the limit & Selfish fleet against a planner &
  What the fixed point costs \\
Mean-action approximation & Neighbour actions replaced by their mean &
  A modelling choice, not a limit \\
\end{nomen}

The assumption is exchangeability: all agents share state and action spaces, and
each agent's reward and transition depend on the others only through the
empirical distribution $\mu_t = \frac{1}{\nag}\sum_j \delta_{s_t^j}$ of their
states, so that $\Rew_i(s_t^i, a_t^i, \mu_t)$ and
$\Tr(s_{t+1}^i \given s_t^i, a_t^i, \mu_t)$. Under this assumption a
representative agent's problem, given a fixed flow $(\mu_t)_{t\ge0}$, is an
ordinary Markov decision process with a time-varying kernel, and the
equilibrium condition is a fixed point.

\begin{definition}[Mean-field game equilibrium]
\label{def:mfg}
A pair $(\pi^\star, \mu^\star)$ is a mean-field game equilibrium when $\pi^\star$
is optimal for the Markov decision process induced by the fixed flow $\mu^\star$,
and $\mu^\star$ is the state distribution flow generated by the whole population
playing $\pi^\star$ from the given initial distribution.
\end{definition}

\begin{rlreading}
\Cref{def:mfg} is a training loop with two alternating steps. Fix a flow
$\mu$, which in an implementation is a histogram of where the fleet is at each
time step, and train one policy against an environment whose dynamics and
rewards depend on that histogram; this is ordinary single-agent reinforcement
learning with a time-varying but known perturbation. Then roll the whole
population out under the trained policy and recompute the histogram; this is
ordinary on-policy state-distribution estimation. A mean-field equilibrium is a
fixed point of the composition, so the training loop is solve, roll out, repeat.
The appeal is that neither step costs anything in the number of agents: one
policy is trained, not $\nag$, and the histogram has the size of the state space
rather than of the joint state space. The catch is that the composition need not
be a contraction, so the alternation need not converge, which is why the outer
step is damped in practice for the same reason that fictitious play averages.
\end{rlreading}

Among the equilibrium concepts in this work, the mean-field equilibrium is the
one reached by an unmodified single-agent algorithm wrapped in an outer loop,
which is why it has become the standard route to very large populations. The
smallest instance is one in which all three readings of the same arithmetic are
available at once.

\begin{example}[Pigou's network as a mean-field game]
\label{ex:mfg-pigou}
A unit mass of agents chooses between two roads once. Road $0$ costs $m$, the
fraction of the mass that uses it, and road $1$ costs $1$ regardless of load.
Given a flow in which a fraction $m$ uses road $0$, the representative agent's
best response is to take road $0$ when $m < 1$, so the only fixed point is
$m^\star = 1$, every agent pays $1$ and the social cost is $1$.

The mean-field control problem replaces the fixed point by an optimisation: a
planner chooses $m$ to minimise $m \cdot m + (1-m) \cdot 1$, differentiates to
get $2m - 1 = 0$, and obtains $m = 1/2$ with social cost $0.75$. The ratio
$1/0.75 = 4/3$ is the price of anarchy, the same number as in \cref{ex:pigou},
which it must be, since this is that example with the finite population replaced
by a continuum.

Three readings are available for the same arithmetic and each is used in
practice. As a mean-field game it is a fixed point reached by alternating a
solver with a rollout. As mean-field control it is a Markov decision process on
the one-dimensional space of distributions. As a finite congestion game it is
\cref{def:congestion} with an affine cost. The gap between the first and the
second is the gap between what selfish learners reach and what a planner would
choose, and it exists at every population size. \exend
\end{example}

The outer loop requires care, because naive alternation, which is fictitious
play with best responses computed exactly, converges only under a contraction or
monotonicity condition on the composed map. The monotonicity condition of \citet{lasry2007mean}, which
holds when the interaction is congestion-like so that agents prefer to avoid
crowded states, guarantees uniqueness of the equilibrium and convergence of
suitably damped iterations. Without it the composition can cycle, and the
remedies are the ones familiar from \cref{sec:evolutionary}: average the
iterates, as in fictitious play, or regularise, as in online mirror descent
\citep{perolat2022scaling,lauriere2022scalable}. The parallel is exact:
damping the population update in a mean-field game plays the role that policy
averaging plays in fictitious play and that the entropy regulariser plays in
mirror descent, and all three are devices for turning a non-contractive fixed
point iteration into a convergent one.

A distinction that is frequently elided separates mean-field games from
mean-field control. In a mean-field game each agent maximises its own objective
and the solution concept is the Nash fixed point of \cref{def:mfg}. In
mean-field control a single planner chooses one policy for the whole population
to maximise the population average, which is a Markov decision process on the
space of distributions, with state $\mu_t$ and action a policy, and whose
solution is the social optimum. The two coincide only in degenerate cases, and
the gap between them is the price of anarchy. In the vocabulary of this
monograph, mean-field control is the $\nag \to \infty$ limit of the multi-agent
Markov decision process of \cref{prop:sg-mdp} and mean-field games the limit of
the general-sum stochastic game, so the distinction is the population-scale
version of the common-payoff versus general-sum distinction that has organised
the whole of this work.

The finite-population guarantee is what makes the limit useful rather than
merely elegant. Under Lipschitz continuity of the rewards and transitions in
$\mu$, a mean-field equilibrium policy played by each of $\nag$ agents in the
corresponding finite game is an $\varepsilon_\nag$-Nash equilibrium with
$\varepsilon_\nag = O(1/\sqrt{\nag})$, the rate coming from the concentration of
the empirical measure; \citet{yardim2023policy} establish it for $\nag$ agents
running independent policy mirror ascent, with the Nash property holding for
the regularised game. The practical consequence is that mean-field methods are
justified for large populations and give no guarantee for the small $\nag$ of
most multi-agent reinforcement learning benchmarks, and the intermediate regime
is handled by the mean-field approximation of \citet{yang2018mean}, which
replaces the joint action of the neighbours by their mean action in the
action-value function. That approximation is a modelling choice rather than a
limit theorem, and we classify it as a structural analogy.

\begin{example}[What the rate means at benchmark scale]
\label{ex:mfg-rate}
The $O(1/\sqrt{\nag})$ bound evaluates to $0.316$ at ten agents, $0.100$ at one
hundred, $0.032$ at one thousand and $0.010$ at ten thousand, before the
problem-dependent constant is applied. Multi-agent reinforcement learning
benchmarks typically use between two and twenty-five agents, where the bound
delivers an $\varepsilon$ of the same order as the payoff range, which is to say
no guarantee at all. Fleet-scale applications, traffic systems, electricity
markets and large recommendation populations, sit at $10^3$ or more, where the
same bound is informative. Reporting mean-field methods on small benchmarks is
therefore reporting a heuristic, and the correct statement is that the method is
justified asymptotically and used here outside its regime, which is honest and
is also how the field uses many of its tools. \exend
\end{example}

\section{Learning dynamics in games and what each of them converges to}
\label{sec:learning-dynamics}

Game theory has studied learning rules since the 1950s, and the classical
results answer, for each rule and each class of games, the question that a
multi-agent reinforcement learning practitioner actually has: if I run this
update in self-play, what will happen. The answers are sharper than the ones
usually available in the deep learning literature, because the rules are simpler
and the games are smaller, and they transfer, because the failure modes are
properties of the game rather than of the function approximator. This section
collects them.

The two readings are a dynamical system and a training loop. A game theorist
writes a learning rule as a map on the product of the strategy simplices and
asks about its rest points, their stability and the asymptotic behaviour of the
trajectory. A reinforcement learning researcher writes the same map as an
optimiser applied to a loss that depends on the other agents' parameters, and
asks whether the loss goes down. The two differ in one decisive respect. A
single-agent optimiser has a Lyapunov function by construction, since the
objective itself decreases; a multi-agent update has none in general, because
there is no scalar that all agents are jointly optimising. Potential games are
exactly the class where there is one, which is why they are the class with
guarantees.

\begin{nomen}{Learning-dynamics vocabulary read from both sides.}{tab:nom-dynamics}
Fictitious play & Best response to the empirical opponent &
  Average what you have seen, then be greedy \\
Best-response dynamics & Greedy policy iteration in self-play &
  Jump straight to the argmax \\
Better-reply dynamics & Any improving step &
  Only go uphill, not necessarily to the top \\
Multiplicative weights & Softmax over cumulative payoffs; MWU &
  Exponentiate the running total \\
No-regret property & Sublinear regret of the learner &
  Asymptotically as good as the best fixed action \\
Blackwell approachability & Vector-valued target set &
  Steer an average vector into a set \\
Stochastic stability & Small-noise limit of a learning chain &
  Which equilibrium survives rare mistakes \\
Log-linear learning & Boltzmann revision with low temperature &
  Selects the potential maximiser \\
Monotone game, DSC & Convex-concave structure; variational inequality &
  The condition under which gradients behave \\
Extragradient, optimism & Look-ahead or past-gradient correction &
  Cancel the rotational part of the field \\
Last-iterate convergence & The final checkpoint is the answer &
  No averaging required \\
Time-average convergence & The averaged policy is the answer &
  The checkpoint is not the output \\
\end{nomen}

Fictitious play is the oldest rule and the one whose convergence class is
best mapped. Each agent maintains the empirical frequency of each opponent's
past actions and plays a best response to it.

\begin{theorem}[Convergence classes of fictitious play;
\citealp{robinson1951iterative,monderer1996fictitious}]
\label{thm:fp}
In a two-player zero-sum game the empirical frequencies of fictitious play
converge to the set of Nash equilibria. In a finite game with identical
interests, and more generally in a finite weighted potential game, fictitious
play has the fictitious play property, meaning that the beliefs converge to the
set of Nash equilibria. Neither statement extends to general-sum games, in which
fictitious play may cycle without the empirical frequencies converging.
\end{theorem}

\begin{example}[Fictitious play on matching pennies: the iterate never settles]
\label{ex:fictitious-play}
Start the two agents with the belief counts $(1,0)$ and $(0,1)$ respectively, so
that the row agent believes the column agent plays Tails and the column agent
believes the row agent plays Heads. The row agent best responds with Tails, the
column agent with Tails; the counts update; each agent then best responds to the
new empirical frequency. The pure action played at each step keeps switching,
and it switches in runs of increasing length, so the sequence of actions never
settles. After twenty thousand steps, however, the empirical frequency of Heads
is $0.5025$ for the first agent and $0.5025$ for the second, which is within
$0.0025$ of the equilibrium $(\tfrac12,\tfrac12)$ and converging.

Both halves matter. The iterate is a pure action at every step, so the policy at
any checkpoint has an exploitability of $1$, the maximum possible; the average
has exploitability going to zero. Reporting the last checkpoint of a self-play
run in a zero-sum game therefore reports the worst object the run produced, and
this is not a quirk of fictitious play but the generic situation, as
\cref{sec:negative} records. Neural fictitious self-play
\citep{heinrich2016deep} is this algorithm with the empirical frequency replaced
by a supervised network that is trained to imitate the agent's own past play,
and the averaging network is not an implementation convenience but the carrier
of the entire guarantee. \exend
\end{example}

\begin{example}[The same rule on a coordination game converges outright]
\label{ex:fp-potential}
Take the common-payoff game with payoff $3$ on the first diagonal entry, $1$ on
the second and $0$ off the diagonal, which is a potential game with potential
equal to the payoff. The row agent prefers its first action
whenever its belief puts weight above $1/4$ on the opponent's first action,
since $3q > 1-q$ exactly there, so the threshold is low and the basin of the
payoff-dominant equilibrium is large. Starting from beliefs that place slightly
more mass on the second action, at $q = 0.49$, both agents are already past the
threshold and play the first action from the first step onward; after twenty
thousand steps the empirical frequency of the first action is $0.99997$ for both,
the residual being the fractional prior counts. The contrast
with \cref{ex:fictitious-play} is the whole of \cref{thm:fp}: the same rule
converges in one class and cycles in the other, and the discriminating property
is the existence of a potential. \exend
\end{example}

The second family is the no-regret family, whose guarantees are the ones that
transfer most directly to the general-sum case. Multiplicative weights maintains
cumulative payoffs and plays a softmax of them; its regret after $T$ rounds is
$O(\sqrt{T \log \abs{\Ac}})$ \citep{freund1999adaptive}, so the time-averaged
play converges to the coarse correlated equilibria by \cref{def:ce-cce}, and in
a two-player zero-sum game the time average converges to the value. What does
not happen is convergence of the iterate.

\begin{example}[Multiplicative weights spirals away from the equilibrium]
\label{ex:mwu-rps}
Run multiplicative weights with step size $0.1$ in self-play on
rock-paper-scissors, starting from the perturbed strategies $(0.4, 0.35, 0.25)$
and $(0.35, 0.4, 0.25)$. Measure the distance from the equilibrium by the sum of
the two Kullback-Leibler divergences to the uniform distribution, which is
$0.0362$ at the start. After two thousand steps it is $1.202$, and the first
agent's probability of Rock has drifted from $0.4$ to $0.138$. The iterate is
moving away from the equilibrium and towards the boundary of the simplex, which
is the discrete-time statement of the result that the continuous-time dynamics
conserves the divergence and cycles while the discrete-time one increases it
\citep{mertikopoulos2018cycles,bailey2018multiplicative}.

The time average still converges, and so does the exploitability of the average.
A practitioner who logs the current policy's win rate against the current
opponent sees a flat curve, because the two policies are rotating together; a
practitioner who logs the exploitability of the current policy sees it rise.
The second measurement is the informative one, and it costs a best-response
computation, which is the price \cref{tab:diagnostics} asks to be paid. \exend
\end{example}

Blackwell approachability is the result that organises this family and deserves
a statement because it explains where regret matching comes from.

\begin{theorem}[Approachability; \citealp{blackwell1956analog}]
\label{thm:blackwell}
Consider a repeated game with vector payoffs $m(a_1,a_2) \in \Real^d$ and a
closed convex target set $S \subseteq \Real^d$. Agent $1$ can force the average
payoff vector to converge to $S$ regardless of agent $2$'s play if and only if
for every halfspace $H$ containing $S$ there is a mixed strategy of agent $1$
guaranteeing that the expected payoff vector lies in $H$.
\end{theorem}

Taking $d = \abs{\Ac_1}$, the vector payoff to be the vector of regrets for each
alternative action and $S$ to be the non-positive orthant, approachability of
$S$ is exactly the no-regret property, and the halfspace-response strategy that
the theorem supplies is regret matching. \citet{abernethy2011blackwell} show the
converse reduction as well, so approachability and no-regret learning are
equivalent problems. For a reinforcement learning audience the useful content is
that the construction is vector-valued, so the same machinery applies when the
target is a set of constraints rather than a single scalar, which is the
structure of constrained and multi-objective reinforcement learning. A
constrained policy optimiser that drives a vector of constraint violations into
the negative orthant is running an approachability algorithm whether or not it
is described as one.

The third family is explicitly stochastic and is the one that answers the
equilibrium-selection question left open throughout this monograph. In
log-linear learning \citep{blume1993statistical} one agent at a time revises,
drawing its action from a Boltzmann distribution over its own payoffs at the
current profile; in adaptive play \citep{young1993evolution} agents best respond
to a random sample of recent history with a small probability of error. Both
define an irreducible Markov chain on the profile space whose stationary
distribution, as the noise vanishes, concentrates on a subset of the equilibria
called the stochastically stable set.

\begin{example}[Which equilibrium of the stag hunt survives rare mistakes]
\label{ex:log-linear}
Take the symmetric stag hunt of \cref{ex:four-games}, in which hunting the stag
together pays $4$ each, hunting hare pays $3$ each, and hunting the stag alone
pays $1$ against the hare hunter's $3$. The stag equilibrium is payoff dominant
and the hare equilibrium is risk dominant: the Harsanyi-Selten comparison is
$(4-3)^2 = 1$ for stag against $(3-1)^2 = 4$ for hare, so hare wins. The game is
an exact potential game with
$\pot = \big(\begin{smallmatrix} -1 & -2 \\ -2 & 0 \end{smallmatrix}\big)$,
whose maximiser is the hare profile.

Run log-linear learning and compute the stationary distribution of the induced
four-state chain exactly. At inverse temperature $\beta = 0.5$ the chain spends
$0.259$ of its time at the stag profile and $0.427$ at the hare profile; at
$\beta = 2$ the figures are $0.115$ and $0.853$; at $\beta = 10$ they are
$0.0000454$ and $0.99995$. As the noise vanishes the chain concentrates
entirely on the hare profile, which is simultaneously the potential maximiser
and the risk-dominant equilibrium.

The conclusion is the one a cooperative multi-agent practitioner needs and
rarely sees stated. Independent learners with persistent exploration do not
select the best equilibrium; they select the one that is most robust to the
partner's mistakes, and in this game that is the one worth $3$ rather than the
one worth $4$. The gap is not an optimisation failure and will not close with
more samples, a larger network or a better learning rate. It closes only with an
equilibrium selection device, which is what optimistic initialisation, leniency
and hysteresis are, and the device works by biasing the very comparison that
this stationary distribution computes. \exend
\end{example}

The fourth family is the modern one and is where the correspondence with
contemporary deep learning is tightest. Write the simultaneous gradient of a
smooth game as the vector field $F(\theta) = (\nabla_{\theta_1}\util_1, \dots,
\nabla_{\theta_\nag}\util_\nag)$ and note that a Nash equilibrium in the
interior is a zero of $F$. Unlike a gradient of a scalar objective, $F$ need not
be a gradient field, so its Jacobian need not be symmetric, and the asymmetric
part is a rotation. Everything that distinguishes multi-agent from single-agent
optimisation is in that rotation.

\begin{definition}[Monotone games and diagonal strict concavity;
\citealp{rosen1965existence}]
\label{def:monotone}
A game with convex compact strategy sets and payoffs concave in the own strategy
satisfies diagonal strict concavity when
$(\theta - \theta')^\top (F(\theta') - F(\theta)) > 0$ for all
$\theta \ne \theta'$, equivalently when $-F$ is strictly monotone. Under this
condition the Nash equilibrium is unique and gradient play with a suitable step
size converges to it.
\end{definition}

\begin{rlreading}
Diagonal strict concavity is the multi-agent replacement for strong concavity of
a single objective, and it is the condition that decides whether the usual
intuition about gradient methods survives. In a single-agent problem the
gradient is the derivative of one scalar, so its Jacobian is a Hessian and is
therefore symmetric, and gradient ascent on a strongly concave objective
contracts. In a multi-agent problem $F$ stacks the gradients of $\nag$ different
objectives, so nothing forces its Jacobian to be symmetric, and the asymmetric
part is a rotation that no step size removes. \Cref{def:monotone} is exactly the
condition that the rotation is dominated by the contraction. When it holds,
independent policy gradient converges to the unique equilibrium and the usual
tuning intuition applies; when it fails, as in the bilinear case of
\cref{ex:extragradient}, the iterate spirals outward however small the step
size, and the remedy is a look-ahead correction rather than a smaller learning
rate.
\end{rlreading}

Diagonal strict concavity is the multi-agent analogue of strong concavity of a
single objective, and when it fails the field can rotate without contracting,
which is exactly the bilinear case.

\begin{example}[Why extragradient fixes what gradient ascent-descent breaks]
\label{ex:extragradient}
Take $\util_1 = xy$ and $\util_2 = -xy$ with scalars, which is matching pennies
in continuous form and is the bilinear core of a generative adversarial
objective \citep{goodfellow2014generative}. The simultaneous update with step
$\eta = 0.1$ is the linear map
$\big(\begin{smallmatrix} 1 & \eta \\ -\eta & 1 \end{smallmatrix}\big)$, whose
eigenvalues have modulus $\sqrt{1+\eta^2} = 1.00499$, so every trajectory
spirals out: starting at $(1,1)$ the norm grows to $3.825$ after two hundred
steps.

The extragradient update takes a trial step, evaluates the field there, and
applies that field at the original point. Its linear map is
$\big(\begin{smallmatrix} 1-\eta^2 & \eta \\ -\eta & 1-\eta^2
\end{smallmatrix}\big)$, whose eigenvalues have modulus
$\sqrt{(1-\eta^2)^2 + \eta^2} = 0.99504$, so trajectories spiral in: from the
same start the norm falls to $0.523$. One look-ahead evaluation converts
divergence at rate $1.005^T$ into convergence at rate $0.995^T$, and the
optimistic variant, which reuses the previous gradient instead of computing a
new one, achieves the same effect for the price of one memory slot
\citep{rakhlin2013optimization,daskalakis2018training}.

Comparison with \cref{ex:lola} is instructive. Opponent shaping subtracts
$2\eta^2$ from one diagonal entry of the Jacobian and obtains a modulus of
$\sqrt{1-\eta^2} = 0.99499$; extragradient subtracts $\eta^2$ from both and
obtains $\sqrt{(1-\eta^2)^2+\eta^2} = 0.99504$. The two numbers are close
without being equal, and the near-coincidence is the point: both methods add a
negative-definite component to a rotational field, one asymmetrically and one
symmetrically, and either suffices to turn divergence into convergence. The practical difference is that extragradient is
symmetric and therefore appropriate when neither agent is a leader, while
opponent shaping is asymmetric and encodes a commitment. \exend
\end{example}

The upshot of this section is a decision table, and we state it as such. If the
game is zero-sum, no-regret dynamics with averaging converges, the iterate does
not, and extragradient or regularisation buys last-iterate convergence. If the
game is a potential or supermodular game, better-reply and fictitious play
converge and the limit is the potential maximiser under vanishing noise. If the
game is monotone in the sense of \cref{def:monotone}, gradient play converges to
the unique equilibrium. Otherwise nothing converges in general, by
\cref{thm:hart-mas-colell}, the attainable target is a coarse correlated
equilibrium of the empirical game, and the question to ask is not how to make
the dynamics converge but which equilibrium one wants and what device will
select it.

\begin{table}[t]
\centering
\caption{Learning dynamics, their guarantees and the reinforcement learning
construction that implements each. The \textsc{Limit} column states what
converges when anything does; iterate means the current policy and average means
the time-averaged policy.}
\label{tab:dynamics}
\small
\begin{tabular}{@{}L{2.9cm} L{3.4cm} L{3.0cm} L{3.0cm}@{}}
\toprule
\textsc{Dynamics} & \textsc{Converges in} & \textsc{Limit} & \textsc{RL counterpart} \\
\midrule
Best-reply dynamics & Potential, supermodular games & Pure equilibrium, iterate & Greedy policy iteration in self-play \\
Fictitious play & Zero-sum; potential \citep{robinson1951iterative,monderer1996fictitious} & Nash, average & Neural fictitious self-play \\
Multiplicative weights & All games, in the sense of regret & \CCE, average; value in zero-sum & Softmax policy gradient \\
Regret matching \citep{hart2000simple} & All games & \CE, average & CFR at an information set \\
Approachability \citep{blackwell1956analog} & All games, vector targets & Average vector in the target set & Constrained policy optimisation \\
Log-linear learning \citep{blume1993statistical} & Potential games & Potential maximiser; risk dominant & Boltzmann exploration with annealing \\
Adaptive play \citep{young1993evolution} & Weakly acyclic games & Stochastically stable set & Sampled-history opponent models \\
Gradient play & Monotone games (\cref{def:monotone}) & Unique Nash, iterate & Independent policy gradient \\
Extragradient, optimism & Zero-sum and monotone & Nash, iterate & Extragradient and optimistic Adam \\
Opponent shaping & No general guarantee & Shaped non-equilibrium point & LOLA; Stackelberg actor-critic \\
Entropy regularisation & All games, to the regularised solution & \QRE; Nash as $\tau \to 0$ & Soft actor-critic; R-NaD \\
\bottomrule
\end{tabular}
\end{table}

\section{The learning-rate limit: evolutionary dynamics as the flow of learning}
\label{sec:evolutionary}

The second continuum limit is taken in the learning rate rather than in the
population. Evolutionary game theory studies the differential equations
describing how the composition of a population of strategies changes over time,
and the central object is the replicator dynamics
\begin{align}
\dot{\pi}_i(a_i)
\;=\;
\pi_i(a_i) \Big[ \util_i(a_i, \pi_{-i}) - \util_i(\pi_i, \pi_{-i}) \Big],
\label{eq:replicator}
\end{align}
in which the growth rate of a strategy's share is its excess payoff over the
population average. The correspondence with reinforcement learning is that
\cref{eq:replicator} is the continuous-time limit of standard learning rules,
which makes the qualitative theory of dynamical systems, rest points, basins of
attraction, limit cycles and chaos, directly applicable to the analysis of
multi-agent training. The mathematical theory of these dynamics is surveyed by
\citet{hofbauer2003evolutionary}.

The two readings are a biological metaphor and an exact limit, and only the
second is used here. A biologist reads \cref{eq:replicator} as differential
reproduction in a population of types. A reinforcement learning researcher reads
it as the vanishing-step-size limit of a policy update, which it literally is
for the rules named in the theorems below. The consequence is that the
qualitative dictionary of dynamical systems becomes a dictionary about training
runs: a rest point is a converged run, an unstable rest point is a run that
leaves a plateau, a limit cycle is a run whose return curve is flat while its
policy rotates, and a strange attractor is a run that never repeats and never
settles.

\begin{nomen}{Evolutionary vocabulary read from both sides.}{tab:nom-evolutionary}
Population share $\pi_i(a_i)$ & Policy probability of action $a_i$ &
  The same number in both readings \\
Replicator dynamics & Vanishing-step-size limit of the update &
  Grow what is above average \\
Replicator-mutator & Softmax update with entropy &
  Selection plus diffusion \\
Fitness $\util_i(a_i, \pi_{-i})$ & Action value $Q_i(a_i)$ &
  Payoff against the current population \\
Rest point & Converged training run &
  No component is growing \\
Evolutionarily stable strategy & Locally attracting, invasion-proof policy &
  Resists a small population of mutants \\
Limit cycle & Flat return curve with rotating policies &
  Looks converged, is not \\
Basin of attraction & Set of initialisations with the same limit &
  Which equilibrium the seed selects \\
Quantal response equilibrium & Rest point of the regularised dynamics &
  What a fixed temperature converges to \\
Empirical game (EGTA) & Payoff matrix over trained policies &
  A normal-form game made of checkpoints \\
$\alpha$-rank & Ranking over a policy population &
  Stationary distribution of a mutation chain \\
Exploitability & Gap to a best response &
  The diagnostic a return curve hides \\
\end{nomen}

\begin{theorem}[Cross-learning and the replicator dynamics;
\citealp{borgers1997learning}]
\label{thm:borgers-sarin}
Consider agents using the cross-learning rule, which after playing $a_i$ and
receiving payoff $u \in [0,1]$ updates
$\pi_i(a_i) \leftarrow \pi_i(a_i) + \alpha\, u\,(1 - \pi_i(a_i))$ and
$\pi_i(a_i') \leftarrow \pi_i(a_i') - \alpha\, u\, \pi_i(a_i')$ for
$a_i' \ne a_i$. As the step size $\alpha \to 0$ with time rescaled by
$1/\alpha$, the stochastic process converges, over any fixed finite time
horizon, to the solution of \cref{eq:replicator}.
\end{theorem}

\begin{theorem}[Boltzmann Q-learning and the replicator-mutator dynamics;
\citealp{tuyls2003selection,tuyls2006evolutionary};
frequency adjustment, \citealp{kaisers2010frequency}; survey,
\citealp{bloembergen2015evolutionary}]
\label{thm:tuyls}
For agents running frequency-adjusted Q-learning, which updates all action
values simultaneously, with a Boltzmann policy at a constant temperature $\tau$,
the continuous-time limit of the induced policy dynamics is
\begin{align}
\dot{\pi}_i(a_i)
\;=\;
\frac{\pi_i(a_i)}{\tau} \Big[ \util_i(a_i,\pi_{-i}) - \util_i(\pi_i,\pi_{-i}) \Big]
\;-\;
\pi_i(a_i)\Big[ \log \pi_i(a_i) - \sum_{a_i'} \pi_i(a_i')\log \pi_i(a_i') \Big],
\label{eq:replicator-mutator}
\end{align}
the replicator term scaled by the inverse temperature plus an entropy term that
acts as mutation and pushes the trajectory toward the uniform distribution. For
plain Q-learning, which updates only the value of the action taken, the same
expression is an approximation rather than an identity.
\end{theorem}

\Cref{eq:replicator-mutator} is the cleanest available statement of what
exploration does to a multi-agent learning system. The first term is selection
and the second is diffusion; the temperature sets their ratio; rest points of
the combined dynamics are quantal response equilibria rather than Nash
equilibria, converging to the latter as $\tau \to 0$. The same structure appears
in single-agent policy optimisation, where the entropy-regularised objective
produces a softmax optimal policy and the regularisation path approaches the
greedy policy as the coefficient vanishes, and recognising the two as the same
term explains why entropy regularisation stabilises multi-agent training for a
reason that goes beyond exploration: it convexifies the fixed-point problem, and
in the limit of large regularisation the equilibrium becomes unique and the
dynamics contractive. The regularised Nash dynamics of
\citet{perolat2021poincare} makes this explicit by adding a regulariser toward a
slowly moving anchor policy, obtaining convergence to the Nash equilibrium of
the unregularised game by annealing the anchor, and this is the mechanism behind
DeepNash \citep{perolat2022mastering}.

Evolutionary stability is the refinement that the dynamics justifies, and it
deserves a definition because it is the one concept in this area that is a
statement about robustness to a perturbation of the population rather than about
a single agent's incentives.

\begin{definition}[Evolutionarily stable strategy;
\citealp{maynardsmith1973logic,taylor1978evolutionary}]
\label{def:ess}
In a symmetric two-player game with payoff $\util$, a strategy $\pi^\star$ is
evolutionarily stable when for every $\pi \ne \pi^\star$ either
$\util(\pi^\star, \pi^\star) > \util(\pi, \pi^\star)$, or
$\util(\pi^\star, \pi^\star) = \util(\pi, \pi^\star)$ and
$\util(\pi^\star, \pi) > \util(\pi, \pi)$. Every evolutionarily stable strategy
is a Nash equilibrium, and in the replicator dynamics it is asymptotically
stable.
\end{definition}

\begin{rlreading}
Evolutionary stability is a statement about a population of policies rather than
about a single one, which is the setting of league play and population-based
training. Read $\pi^\star$ as the incumbent policy pool and $\pi$ as a newly
trained challenger injected at low frequency. The first clause says the
incumbent beats the challenger outright; the second says that when they tie
against the incumbent, the incumbent beats the challenger in the rare encounters
between challengers. Together they say the challenger cannot gain a foothold, so
an evolutionarily stable pool is one that no single new policy can invade, which
is exactly the property a league is maintained to achieve. The relation to
\cref{def:nash} is one-directional and worth remembering: every evolutionarily
stable strategy is a Nash equilibrium, and many Nash equilibria are not
evolutionarily stable, so stability is the stronger and more operationally
meaningful requirement for a population. \Cref{ex:hawk-dove} computes one in a
line.
\end{rlreading}

\begin{example}[Hawk and dove: the stable mixture is computable in one line]
\label{ex:hawk-dove}
Two animals contest a resource worth $V = 2$. Two hawks fight and share the
expected value minus an injury cost $C = 4$, so each receives
$(V - C)/2 = -1$. A hawk against a dove takes the whole resource, $V = 2$,
leaving the dove nothing. Two doves split it, $V/2 = 1$. The payoff matrix to
the row animal is
\begin{align*}
\bimatrix{\mathrm{Hawk}}{\mathrm{Dove}}{\mathrm{Hawk}}{\mathrm{Dove}}{-1}{2}{0}{1}.
\end{align*}
Neither pure strategy is stable, since a population of doves is invaded by a
hawk earning $2$ against the doves' $1$, and a population of hawks is invaded by
a dove earning $0$ against the hawks' $-1$. The stable mixture makes the two
equal: with a fraction $p$ of hawks, a hawk earns $-p + 2(1-p) = 2 - 3p$ and a
dove earns $1-p$, equal when $2 - 3p = 1 - p$, that is $p = 1/2 = V/C$. At that
mixture both earn $0.5$ and the population average is $0.5$, which is half of
what a population of pure doves would earn.

The last sentence is the content for multi-agent reinforcement learning. The
evolutionarily stable outcome is strictly worse for the population than a
cooperative one, it is reached by a dynamics that no agent controls, and the
fraction of aggressive behaviour it sustains is set by the ratio of the prize to
the cost of conflict. A training population of self-interested agents in a
contested-resource environment will settle at this ratio, and changing it
requires changing $V/C$, which is an environment design decision, not an
algorithmic one. \exend
\end{example}

\begin{example}[Rock-paper-scissors: a conserved quantity instead of convergence]
\label{ex:replicator-rps}
Simulate \cref{eq:replicator} for the single-population rock-paper-scissors game
from $(0.4, 0.35, 0.25)$. The product $\pi(R)\pi(P)\pi(S)$ is a conserved
quantity of the continuous-time dynamics, and the numerical integration confirms
it: starting from a product of $0.0350$, after fifty thousand steps of size
$4 \times 10^{-4}$ the product is $0.034995$, a relative drift of
$1.5 \times 10^{-4}$ attributable to the discretisation. The trajectory is a
closed orbit around the uniform equilibrium, and the probability of Rock is
$0.419$ at the end of the run, back near where it started.

A training run of this system would produce a return curve pinned at zero, since
the game is zero-sum and symmetric, with a policy that is rotating steadily. No
scalar logged by a standard training harness distinguishes this from
convergence. The quantity that does distinguish it is the conserved product, or
equivalently the Kullback-Leibler divergence to the equilibrium, and in general
the exploitability. This example is the reason \cref{tab:diagnostics} exists.
\exend
\end{example}

Gradient-based learning admits the same treatment. Infinitesimal gradient ascent
in a two-agent two-action game \citep{singh2000nash} yields a linear system
whose trajectories either converge to a Nash equilibrium or orbit it on closed
elliptical curves, the latter occurring exactly when the game has a unique mixed
equilibrium. The average payoff still converges to the equilibrium payoff in the
orbiting case, which foreshadows the general phenomenon that time averages of
no-regret dynamics behave better than the dynamics themselves. The win-or-learn-fast
modification of \citet{bowling2002multiagent} varies the step size according to
whether the agent is performing better or worse than a reference, and turns the
closed orbits into inward spirals, which is a control-theoretic fix to a
dynamical-systems problem and has no single-agent analogue, since in a
single-agent problem gradient ascent on a smooth objective cannot cycle.

The broader lesson of this section is methodological. Single-agent
reinforcement learning is an optimisation problem and is analysed with the tools
of optimisation: descent lemmas, contraction, convergence rates to a fixed
target. Multi-agent reinforcement learning is a dynamical system and must be
analysed with the tools of dynamical systems: stability of rest points, Lyapunov
functions, bifurcations as parameters such as the temperature or the step size
cross thresholds, and the possibility of chaos, which does occur in simple games
\citep{sato2002chaos}. The empirical practice of reporting a learning curve and
declaring convergence when it flattens is adequate for the first setting and
misleading for the second, because a flat curve is consistent with a limit cycle
whose period exceeds the plotting window. The appropriate diagnostics are
exploitability, measured as the gap to a best response, and the distance of the
empirical joint play to the coarse correlated equilibrium polytope, both of
which are computable in the games where this matters and neither of which is a
return.

The empirical game-theoretic analysis methodology of
\citet{wellman2006methods} is the operational form of this advice. One takes a finite set of trained policies, estimates the payoff matrix
of the game they define by simulation, and analyses that empirical normal-form
game with the tools of \cref{sec:solution-concepts,sec:evolutionary}: its
equilibria, its replicator flow, its response graph. The outcome is a
characterisation of the strategic landscape that a learning curve cannot
provide, and the $\alpha$-rank solution concept of \citet{omidshafiei2019alpha}
supplies a ranking of policies that is computable in time polynomial in the
total number of pure strategy profiles, and is defined for general-sum
many-agent games, for which no polynomial-time Nash computation is available.

\begin{example}[What a ranking of checkpoints can and cannot say]
\label{ex:alpharank}
Build the mutation chain of $\alpha$-rank over three strategies with selection
intensity $\alpha = 10$ and population size $2$. On rock-paper-scissors the
stationary distribution is uniform, $(\tfrac13,\tfrac13,\tfrac13)$, with a
spread of less than $10^{-15}$, which is the correct answer: the three
strategies are cyclically related and none dominates. On a transitive game in
which the first strategy beats the second by one and the third by two, and the
second beats the third by one, the stationary distribution concentrates almost
entirely on the first, with mass $0.99999998$.

The contrast is the diagnostic. A ranking that comes out uniform is not a
failure of the method but a report that the policy population is intransitive,
which is exactly the situation in which average win rate across a population is
meaningless and in which a league or a population-based method is required. A
ranking that concentrates reports a transitive landscape, in which the best
checkpoint is a meaningful notion. Running $\alpha$-rank on the checkpoints of a
self-play run therefore answers, before any further analysis, whether the
run's own history forms a ladder or a cycle. \exend
\end{example}

\section{The complexity landscape}
\label{sec:complexity}

A correspondence between two models is only useful if one knows what it costs to
exploit. This section collects the complexity of the solution problems attached
to each model class so that the reader can see, for any proposed reduction,
whether it buys tractability or merely relocates the difficulty. Unless stated
otherwise, the problems are the decision versions of exact solution with the
model given explicitly, so the input size is polynomial in $\abs{\St}$,
$\prod_i \abs{\Ac_i}$ and the horizon, and sampling-based learning complexity is
discussed separately at the end. Because the horizon and the discounting
convention change the answers, every statement below names them.

The two readings here differ in what is being counted. A complexity theorist
counts operations as a function of the input size, where the input is the
explicit description of the game, and is therefore interested in whether an
exponential blow-up lives in the class or in the encoding. A reinforcement
learning researcher counts environment interactions as a function of the
accuracy, where the model is never written down at all, and is therefore
interested in sample complexity. The two are related but not the same, and the
commonest error in this area is to read a planning hardness result as a
statement about learning or the reverse. We separate them explicitly, and the
rule of thumb is that planning hardness bounds what any amount of data can buy,
while sample complexity bounds what a given amount of data can buy for a problem
whose planning version is already easy.

\begin{nomen}{Complexity vocabulary read from both sides.}{tab:nom-complexity}
Input size & Size of the explicit model &
  What the exponent is measured against \\
Planning problem & Model known; compute a solution &
  No interaction needed \\
Learning problem & Model unknown; interact and approximate &
  Sample complexity, not time \\
\PTIME-complete & Solvable by linear programming &
  Easy, and not parallelisable in general \\
\PPAD-complete & Fixed point guaranteed by a parity argument &
  Hard, but not \NP-hard \\
\textsc{PLS}-complete & Local search with an improving step &
  Each step is cheap; the path is long \\
\NEXP-complete & Nondeterministic exponential time &
  No polynomial algorithm exists \\
Undecidable & No algorithm at all &
  Infinite horizon with partial observability \\
Curse of multi-agents & $\prod_i \abs{\Ac_i}$ in the input &
  Exponential in $\nag$ inside the input size \\
Succinct representation & Factored or graphical model &
  Small description, possibly hard solution \\
Approximate solution & $\varepsilon$-equilibrium; $\varepsilon$-optimal policy &
  What is actually computed \\
\end{nomen}

The single-agent reference points are classical. Deciding whether a finite flat
discounted Markov decision process has value at least a threshold is
\PTIME-complete, with membership by linear programming
\citep{papadimitriou1987complexity}. The finite-horizon partially observable
Markov decision process problem, for history-dependent policies, is
\PSPACE-complete \citep{papadimitriou1987complexity}, the jump reflecting the
exponential history space that the belief state of \cref{thm:belief-mdp}
summarises but does not shrink. The infinite-horizon partially observable
problem is undecidable, and so is its $\varepsilon$-approximate version under
undiscounted criteria, the average criterion among them
\citep{madani2003undecidability}.

\begin{example}[Where the exponentials actually live]
\label{ex:counting}
Two arithmetic exercises make the rest of this section readable. First, the
curse of multi-agents. With five actions per agent the joint action space has
$25$ elements for two agents, $3125$ for five, $9{,}765{,}625$ for ten and about
$9.5 \times 10^{13}$ for twenty, while $\max_i \abs{\Ac_i}$ stays at $5$. A
multi-agent Markov decision process is solvable in time polynomial in
$\abs{\St}\prod_i\abs{\Ac_i}$, so the algorithm is polynomial in a quantity that
is itself exponential in the parameter of interest, and at ten agents the gap
between the joint and the per-agent count is a factor of about two million. This
is why the sample-complexity results that scale in $\max_i \abs{\Ac_i}$ rather
than $\prod_i \abs{\Ac_i}$ \citep{jin2024vlearning} matter so much, and why
factorisation, the mean-field limit and non-stationary correlated targets are
the three escapes catalogued in this monograph.

Second, the history space. With two actions and two observations the number of
action-observation histories of length $h$ is $4^h$, which is $64$ at $h = 3$,
$1024$ at $h = 5$ and $1{,}048{,}576$ at $h = 10$. The number of deterministic
history-dependent policies is $2$ raised to the number of histories of every
length up to $h-1$, that is $2^{(4^h-1)/3}$, whose base-ten logarithm is $6.3$
at $h = 3$, $102.7$ at $h = 5$ and $105{,}217.5$ at $h = 10$. These are the
numbers behind the \PSPACE\ and \NEXP\ entries of \cref{tab:complexity}, and the
jump between them is the jump between enumerating histories and enumerating
functions of histories, which is also the jump between a belief state and a
prescription in \cref{sec:reductions}. \exend
\end{example}

The multi-agent reference points are the following. For finite-horizon
decentralised models with at least two agents, both the decentralised partially
observable Markov decision process and the decentralised Markov decision process
problems are \NEXP-hard, and they are \NEXP-complete when the horizon is
strictly less than the number of states \citep{bernstein2002complexity}. The hardness implies,
by the time hierarchy theorem and the inclusions
$\PTIME \subsetneq \EXP \subseteq \NEXP$, that no polynomial-time algorithm
exists, and it implies that no polynomial-time reduction producing a
polynomially sized partially observable Markov decision process can exist unless
\PSPACE, \EXP\ and \NEXP\ coincide, which places the reductions of
\cref{sec:reductions} in context: they are exact and they are necessarily
exponential in the quantity that matters, the prescription space. Structural
restrictions recover lower classes. When transitions and observations are
independent across agents, each agent is locally fully observable, and the
coupling enters only through a joint reward structure, the problem drops to
\NP-complete \citep{becker2004solving}, which is the formal statement that
coupling through reward alone is far milder than coupling through dynamics. The
infinite-horizon versions are undecidable, inheriting the single-agent result
\citep{bernstein2002complexity,madani2003undecidability}.

For the stateless models, two-player zero-sum normal-form games are solved in
polynomial time by linear programming, which is von Neumann's minimax theorem
made constructive, the constructive part being Dantzig's equivalence with linear
programming and the polynomial bound being Khachiyan's. Finding a Nash
equilibrium of a two-player general-sum normal-form game is \PPAD-complete
\citep{chen2009settling}, the three-player and general cases having been settled
first \citep{daskalakis2009complexity}; \PPAD\ captures total search problems
whose solutions are guaranteed by a parity argument and is widely believed to
lie strictly between \PTIME\ and \NP-hardness. The relaxations are easy, and
for an explicitly given game they are easy for a trivial reason: correlated and
coarse correlated equilibria are defined by polynomially many linear
inequalities in the joint distribution, so a linear program suffices. The
substantive result concerns succinct representations. For games of polynomial
type whose expected utilities are computable in polynomial time, which covers
graphical, anonymous, polymatrix, congestion, scheduling and local-effect
games, the ellipsoid-against-hope method of \citet{papadimitriou2008computing}
computes a correlated equilibrium in polynomial time; the published algorithm
returns an approximate equilibrium, and an exact variant is due to
\citet{jiang2015polynomial}. The contrast between \PPAD-completeness for Nash
and polynomial time for correlated equilibrium is the complexity-theoretic
justification for the recommendation of \cref{sec:solution-concepts}, that
no-regret dynamics targeting coarse correlated equilibria are the right default.
In potential games the problem is different again: a pure Nash equilibrium is
guaranteed to exist by the finite improvement property, and finding one is
\textsc{PLS}-complete for congestion games in general while being polynomial for
the symmetric network case \citep{fabrikant2004complexity}, which means that the
better-reply dynamics may take exponentially many steps even though each step is
cheap.

Zero-sum structure in the many-agent case is the one place where a positive
result survives beyond two agents. Unrestricted three-player zero-sum games are
\PPAD-complete, but when the payoffs decompose over the edges of a graph, so
that the game is a zero-sum polymatrix game in the sense of
\cref{def:graphical}, the minmax theorem generalises, the equilibria form a
convex set, they are computable in polynomial time and no-regret learning
converges to them \citep{cai2011minmax}. This is essentially the limit of what
can be pushed, and the practical consequence is that a many-agent adversarial
application should be examined for pairwise separability before any algorithm is
chosen.

Stochastic games are where the folklore most needs correcting, because the
turn-based and the simultaneous-move cases separate. Deciding the winner of a
simple stochastic game, which is turn-based, undiscounted and stopping, lies in
\NP$\,\cap\,$\coNP\ \citep{condon1992complexity}, and the turn-based discounted,
mean-payoff and parity stochastic games are polynomial-time equivalent to it
\citep{andersson2009complexity}; no polynomial algorithm is known for any of
them. For the simultaneous-move discounted game that \cref{eq:minimax-bellman}
actually defines, the exact value problem is harder: it lies in \PSPACE, it is
square-root-sum hard, and it is not known to lie in \NP\
\citep{etessami2010complexity}. What is unproblematic is approximation at a
fixed discount factor. Since $\BellOp_{\mathrm{mm}}$ is a $\gamma$-contraction
by \cref{thm:shapley} and each per-state matrix game is solved by a linear
program, $O\big(\log(1/\varepsilon)/\log(1/\gamma)\big)$ sweeps deliver an
$\varepsilon$-optimal value, so the running time is polynomial in the model
size, in $\log(1/\varepsilon)$ and in $1/(1-\gamma)$. The last dependence is
why the qualifier matters: with $\gamma$ part of the input the strong
approximation problem is known to sit in \PPAD, with no polynomial-time
algorithm known \citep{etessami2010complexity}.

For general-sum stochastic games, computing a stationary Nash equilibrium is
\PPAD-hard already by inheritance from the stateless case, a one-state game
being a normal-form game, and \citet{daskalakis2023complexity} strengthen this
to two-player turn-based games with a constant discount factor and a constant
approximation, where the inheritance argument does not apply. The same work
removes the hope that the relaxation which rescues the stateless case rescues
the sequential one: computing a constant-approximate stationary Markov coarse
correlated equilibrium is \PPAD-hard under the same restrictions,
unconditionally for the perfect-correlated-equilibrium refinement and, for the
unrefined notion, under the \textsc{PCP}-for-\PPAD\ conjecture. The rescue comes
instead from relaxing stationarity, and here the literature must be read
carefully, because two different objects are on offer. V-learning and its
relatives return a non-stationary, history-dependent correlated policy, with a
sample complexity scaling in $\max_i \abs{\Ac_i}$ rather than
$\prod_i \abs{\Ac_i}$, which is what avoids the curse of multi-agents
\citep{jin2024vlearning,song2021can}; those policies are explicitly not Markov.
Non-stationary Markov coarse correlated equilibria, which are Markov in the
state at each stage, are learnable in time and sample complexity polynomial in
all parameters by the decentralised algorithm of
\citet{daskalakis2023complexity}. Both targets are correlated objects and both
require shared randomness at execution.

For extensive-form games, the equilibrium of a two-player zero-sum game with
perfect recall is computable in time polynomial in the size of the game tree
\citep{koller1992complexity}, which may itself be exponential in the natural
description of the game; the sequence form, developed for computation by
\citet{vonstengel1996efficient}, replaces the simplex of mixed strategies by the
polytope of realisation plans and brings the linear program down to a size
linear in the tree. \citet{koller1996efficient} extend the device to the
general-sum case, where the equilibrium conditions become a linear
complementarity problem solved by Lemke's algorithm, which is not polynomial.
Dropping perfect recall makes the guaranteed-payoff problem, that of deciding
whether a player can secure a given payoff, \NP-hard
\citep{koller1992complexity}; dropping the zero-sum assumption makes equilibrium
computation \PPAD-hard by inheritance from the normal form. For Stackelberg
equilibria the pattern inverts in an instructive way: computing an optimal
commitment in a two-player normal-form game is polynomial, by solving one linear
program per pure follower strategy, while the Bayesian version with multiple
follower types and the version with three or more agents are \NP-hard
\citep{conitzer2006computing}. Computing the Shapley value of a weighted voting
game is \#P-complete \citep{deng1994complexity,michalak2013efficient}, while the
same work shows it is polynomial for induced-subgraph games, so the hardness is a property of the
representation rather than of the solution concept; it is nevertheless the
reason why Shapley-based credit assignment in multi-agent reinforcement
learning is implemented by Monte Carlo sampling over permutations. For the
stability concepts of \cref{sec:coalitions} the same paper reports that, on its
weighted-graph vehicle, the core is easy to characterise when the game is convex
and \NP-complete otherwise, and that similar results hold for the kernel, the
nucleolus, the $\varepsilon$-core and the bargaining set, so a federation
stability check is cheap exactly when the coalition value has increasing
marginal returns.

\begingroup\small
\begin{longtable}{@{}L{4.5cm} L{3.5cm} L{2.9cm} L{2.4cm}@{}}
\caption{Complexity of exact solution, model given explicitly. The horizon is
finite unless stated. The \textsc{Target} column names the solution concept
whose computation is being classified; \PPAD\ and \textsc{PLS} are total search
classes, so the stated hardness concerns the search problem rather than a
decision problem. Side conditions that the sources require are given in the
\textsc{Problem} column rather than relegated to the text.}
\label{tab:complexity}\\[2pt]
\toprule
\textsc{Problem} & \textsc{Target} & \textsc{Complexity} & \textsc{Source} \\
\midrule
\endfirsthead
\multicolumn{4}{@{}l}{\small\textsc{(\cref{tab:complexity}, continued)}}\\[2pt]
\toprule
\textsc{Problem} & \textsc{Target} & \textsc{Complexity} & \textsc{Source} \\
\midrule
\endhead
\midrule
\multicolumn{4}{@{}r@{}}{\small\textsc{continued on the next page}}\\
\endfoot
\bottomrule
\endlastfoot
MDP, flat, discounted & Optimal policy & \PTIME-complete & \citet{papadimitriou1987complexity} \\
POMDP, finite horizon & Optimal history-dependent policy & \PSPACE-complete & \citet{papadimitriou1987complexity} \\
POMDP, infinite horizon & Policy existence; $\varepsilon$-optimality, undiscounted & Undecidable & \citet{madani2003undecidability} \\
MMDP, flat & Optimal joint policy & \PTIME\ in $\abs{\St}\prod_i\abs{\Ac_i}$ & \citet{boutilier1996planning}; complexity via \citet{papadimitriou1987complexity} \\
Dec-MDP, $\nag \ge 2$ & Optimal joint policy & \NEXP-hard; complete if horizon $< \abs{\St}$ & \citet{bernstein2002complexity} \\
Dec-POMDP, $\nag \ge 2$ & Optimal joint policy & \NEXP-hard; complete if horizon $< \abs{\St}$ & \citet{bernstein2002complexity} \\
Dec-MDP; independent transitions and observations, local full observability, joint reward structure & Optimal joint policy & \NP-complete & \citet{becker2004solving} \\
Dec-POMDP, infinite horizon & Policy existence & Undecidable & \citet{bernstein2002complexity}, inheriting \citet{madani2003undecidability} \\
\addlinespace
Two-player zero-sum matrix game & Value, optimal strategies & \PTIME\ (LP) & Value: \citet{vonneumann1928theorie}; \PTIME\ by LP \\
Two-player general-sum matrix game & Nash equilibrium & \PPAD-complete & \citet{chen2009settling}; $\nag \ge 3$ \citet{daskalakis2009complexity} \\
Zero-sum polymatrix game, any $\nag$ & Nash equilibrium & \PTIME; convex equilibrium set & \citet{cai2011minmax} \\
Three-player zero-sum, unrestricted & Nash equilibrium & \PPAD-complete & \citet{cai2011minmax}, from \citet{chen2009settling} \\
Explicit normal-form game & Correlated or coarse correlated equilibrium & \PTIME\ (LP) & Definitional \\
Succinct game of polynomial type & Correlated equilibrium & \PTIME & \citet{papadimitriou2008computing}; exact variant \citet{jiang2015polynomial} \\
Congestion game & Pure Nash equilibrium & \textsc{PLS}-complete; \PTIME\ if symmetric network & \citet{fabrikant2004complexity} \\
Two-player normal form & Stackelberg commitment & \PTIME & \citet{conitzer2006computing} \\
Bayesian or $\nag\ge3$ normal form & Stackelberg commitment & \NP-hard & \citet{conitzer2006computing} \\
Weighted voting game & Shapley value & \#\PTIME-complete & \citet{deng1994complexity} \\
Induced-subgraph coalitional game & Core, kernel, nucleolus, bargaining set & \PTIME\ if convex; \NP-complete otherwise & \citet{deng1994complexity} \\
\addlinespace
Simple stochastic game (turn-based, undiscounted, stopping) & Which player wins with probability $>\tfrac12$ & \NP$\,\cap\,$\coNP & \citet{condon1992complexity} \\
Turn-based discounted, mean-payoff, parity SG & Value & Equivalent to the above & \citet{andersson2009complexity} \\
Simultaneous-move zero-sum discounted SG & Exact value & In \PSPACE; square-root-sum hard & \citet{etessami2010complexity} \\
Simultaneous-move zero-sum SG, fixed $\gamma$ & $\varepsilon$-optimal value & Poly.\ in size, $\log(1/\varepsilon)$, $1/(1-\gamma)$ & Contraction of \cref{thm:shapley}; LP per state \\
General-sum SG, two-player turn-based, constant $\gamma$ & Constant-approximate stationary Nash & \PPAD-hard & \citet{daskalakis2023complexity} \\
General-sum SG, two-player turn-based, constant $\gamma$ & Constant-approximate stationary Markov \CCE & \PPAD-hard (unrefined notion under \textsc{PCP}-for-\PPAD) & \citet{daskalakis2023complexity} \\
General-sum Markov game & Non-stationary correlated (non-Markov) \CE/\CCE & Sample complexity in $\max_i\abs{\Ac_i}$ & \citet{jin2024vlearning,song2021can} \\
General-sum Markov game & Non-stationary Markov \CCE & Poly.\ time and sample complexity & \citet{daskalakis2023complexity} \\
\addlinespace
Two-player zero-sum EFG, perfect recall & Nash equilibrium & \PTIME\ in tree size & \citet{koller1992complexity,vonstengel1996efficient} \\
Two-player general-sum EFG, perfect recall & Nash equilibrium & \PPAD-hard; LCP by Lemke & \citet{koller1996efficient} \\
EFG without perfect recall & Guaranteed payoff (maxmin behaviour strategy) & \NP-hard & \citet{koller1992complexity} \\
\end{longtable}
\endgroup

Three readings of \cref{tab:complexity} are worth stating because they have
direct algorithmic consequences. The first is that partial observability costs
one complexity class and decentralisation costs another, and the two costs
compound; the jump from \PSPACE\ to \NEXP\ between a partially observable
Markov decision process and a decentralised one is the price of the absence of a
common belief state, and it is the formal justification for every architecture
that manufactures common knowledge by communication, since a communication
channel that makes all observations public collapses the decentralised problem
back to the centralised one.

The second is that relaxing the solution concept helps in the stateless case and
does not help in the sequential case. Correlated equilibrium is polynomial where
Nash is \PPAD-complete in normal-form games, and stationary Markov coarse
correlated equilibrium remains \PPAD-hard in Markov games even for two-player
turn-based instances with a constant discount factor. The practical conclusion
is that one must also relax stationarity, and the learning algorithms with the
strongest guarantees in general-sum Markov games return correlated policies that
require a shared random seed at execution, Markov in the state in the case of
\citet{daskalakis2023complexity} and history-dependent in the case of
\citet{jin2024vlearning}. Whether either is acceptable is an engineering question
about the deployment, and it is one that the complexity results force the
designer to confront rather than a limitation of any particular algorithm.

The third is that the exponential dependence on the number of agents is rarely
in the complexity class and almost always in the input size, which is the
arithmetic of \cref{ex:counting}. A multi-agent
Markov decision process is solvable in time polynomial in $\abs{\St}$ and
$\prod_i \abs{\Ac_i}$, and the second factor is exponential in $\nag$, so the
algorithm is polynomial in a quantity that is itself exponential in the
parameter of interest. This is the curse of multi-agents, and the three
structural escapes from it are the three that this monograph has catalogued:
factorisation, which restricts the representable value functions
(\cref{sec:potential,sec:graphical}); the mean-field limit, which removes the dependence on
$\nag$ by symmetry (\cref{sec:mean-field}); and the restriction to
non-stationary correlated targets, whose sample complexity scales with
$\max_i \abs{\Ac_i}$ \citep{jin2024vlearning}.

\section{How to train: matching the algorithm to the model}
\label{sec:training}

The dictionary assembled so far is of practical use only if it narrows the
choice of training procedure, and it does. This section states a decision
procedure, justifies each branch by the results already established, and
collects the recommendations in \cref{tab:training} and
\cref{tab:diagnostics}.

The two readings converge here, which is the point of the whole monograph. A
game theorist asked to advise on a multi-agent system would ask what the game
is, which solution concept is appropriate and whether it is attainable. A
reinforcement learning engineer would ask what algorithm to run. The claim of
this section is that the second question has no answer independent of the first,
and that the five structural questions below extract exactly the information
needed to answer it.

\begin{nomen}{Deployment vocabulary read from both sides.}{tab:nom-training}
Reward structure & Shared, opposed or mixed-motive rewards &
  Decides which reduction is sound \\
Information structure & Who observes what, and what is public &
  Decides whether a sufficient statistic exists \\
Execution architecture & Centralised or decentralised at deployment &
  Decides whether factorisation is needed \\
Shared randomisation & A common seed available at run time &
  Decides whether \CE\ targets are executable \\
Opponent class & Fixed, co-trained, or unknown partners &
  Decides whether the problem is single-agent \\
Population size & Number of interacting agents &
  Decides whether the mean-field limit applies \\
Solution concept targeted & What the training loop converges to &
  Must be stated before choosing a method \\
Exploitability & Gap to a best response &
  The zero-sum diagnostic \\
NashConv & Summed unilateral improvement &
  The general-sum diagnostic \\
Cross-play & Pairing agents from different runs &
  The cooperative convention diagnostic \\
Empirical game & Payoff matrix over trained checkpoints &
  The population diagnostic \\
\end{nomen}

Five questions determine the branch, and they should be answered about the
deployment rather than about the simulator. The first is the reward structure:
whether the agents share a reward exactly, oppose each other exactly, or
neither. The second is the observability: whether each agent sees the state,
whether the agents jointly see it, and whether anything is public. The third is
the execution architecture: whether a central controller issues joint actions at
deployment time, whether the agents act on private information, and whether a
shared random seed or a communication channel is available. The fourth is
whether the other agents are fixed, co-trained, or unknown at training time. The
fifth is the population size, which decides whether the mean-field limit applies.
The reason to ask these about the deployment is that the sound reduction depends
on what is available at execution, not on what is available during training,
which is the entire content of the centralised-training-decentralised-execution
distinction.

If the reward is shared, the state is fully observed and a central controller
issues joint actions, \cref{prop:sg-mdp} applies and the problem is a Markov
decision process over joint actions. Any single-agent algorithm is sound; the
only difficulty is the size of $\prod_i \Ac_i$, and the remedy is a factored
action representation rather than a multi-agent algorithm. Practitioners who
reach for a multi-agent method in this setting are solving a harder problem than
the one they have.

If the reward is shared and execution is decentralised, the problem is a
decentralised partially observable Markov decision process, the optimum is
well defined, and the obstacle is equilibrium selection rather than equilibrium
computation (\cref{prop:opt-implies-ne}). Three routes are sound and they
trade off differently. Value factorisation under the Individual-Global-Max
condition of \cref{def:igm} makes decentralised greedy play optimal by
construction and is the cheapest option, at the cost of restricting the
representable stage games; monotone mixing in the manner of
\citet{rashid2018qmix} is the standard default and its failure mode is
predictable from \cref{prop:monotone-igm} and is not a violation of
\cref{eq:igm}: it is a joint value function that satisfies the
Individual-Global-Max condition and admits no monotone factorisation, which
\cref{ex:igm-counterexample} exhibits in two by two. Centralised critics with
decentralised actors, as in \citet{lowe2017multi} and \citet{yu2022surprising},
impose no representational restriction and provide no selection guarantee, so
they must be paired with a selection device. Occupancy-state and public-belief
methods implement \cref{thm:occupancy} and \cref{thm:coordinator} directly, are
the only route with an exactness guarantee, and are the most expensive; the
Bayesian action decoder of \citet{foerster2019bayesian} is the practical form.
Independent learning with an optimism, leniency or hysteresis bias
\citep{lauer2000algorithm,matignon2007hysteretic,panait2006lenience} is
a selection device rather than an algorithm family, and it is the correct
accompaniment to any of the above in tasks with the penalty structure of
\cref{ex:penalty}. If the interaction is local, a coordination graph in the
sense of \cref{prop:coordination-graph} converts the joint argmax into message
passing and should be used before any approximation is contemplated.

If the game is two-player zero-sum, the machinery is complete and should be
used. Under perfect information, the recipe is self-play with search and a learned
value function \citep{silver2018general}, which attains the minimax value in
the tabular limit and is validated empirically rather than by a convergence
theorem. Under imperfect information, the sound targets are
the average of a no-regret dynamics or the solution of the public-belief-state
game. Counterfactual regret minimisation and its sampled variants
\citep{zinkevich2007regret,lanctot2009monte} drive the average external regret
to zero at rate $O(1/\sqrt{T})$, so the average strategy converges to a Nash
equilibrium; the neural variant \citep{brown2019deep} inherits the scheme with
function approximation and its regret bound carries an additive floor set by the
approximation error, so it has no vanishing-regret guarantee; neural fictitious self-play
\citep{heinrich2016deep} approximates fictitious play with the same averaging
logic; regularised Nash dynamics \citep{perolat2021poincare} obtains last-iterate
convergence by anchoring; and ReBeL and Student of Games
\citep{brown2020rebel,schmid2023student} combine the public belief state of
\cref{thm:pbs} with search. What is not sound is naive self-play with a
single network and no averaging, because the dynamics in a zero-sum game with a
unique mixed equilibrium orbits rather than converges
(\cref{ex:mwu-rps,ex:replicator-rps});
the time average converges and the iterate does not,
so reporting the final checkpoint rather than the average is a methodological
error with a theorem behind it. If there are more than two agents but the
payoffs decompose over pairs and sum to zero, \cref{tab:complexity} says the
good properties survive and the same machinery applies \citep{cai2011minmax}.

If the game is general-sum, no procedure converges to a Nash equilibrium in
general (\cref{sec:negative}), and the question becomes which attainable target
to aim at. The defensible answers are a coarse correlated equilibrium through
no-regret dynamics, which is attainable and polynomial in the stateless case and
requires non-stationarity in the sequential case
\citep{jin2024vlearning}; a Stackelberg equilibrium if the deployment genuinely grants
one party commitment power (\cref{sec:stackelberg}); a social welfare optimum
if the designer controls the reward and can impose it, which converts the
problem into the common-payoff case; or a best response to a specified
population of opponents, which converts it into the single-agent case by
\cref{prop:fixed-others-mdp}. Before any of these, the structural check is worth
performing: if the game is a Markov potential game or is supermodular, the
convergence results of \cref{sec:potential,sec:learning-dynamics} apply and
independent learning is sound; if the game is smooth in the sense of
\cref{thm:smoothness}, the welfare of whatever no-regret learning reaches is
bounded without any convergence assumption, which is often all that is needed.
Population-based methods, double oracle and
policy-space response oracles \citep{mcmahan2003planning,lanctot2017unified},
occupy a useful middle ground: they maintain an explicit empirical game over a
growing set of policies, solve it with a meta-solver whose choice makes the
target explicit, and add a best response to the meta-solution. Their value is
that the target is a parameter rather than an accident, and the empirical game
they maintain is exactly the object that \citet{wellman2006methods} analyses.

If the agents are many and exchangeable, \cref{sec:mean-field} applies and the
training loop reduces to a single-agent solver inside a population-update outer
loop, with damping or regularisation on the outer loop. If the other agents are
fixed and known, the problem is single-agent by \cref{prop:fixed-others-mdp}. If
they are fixed and unknown, it is Bayesian and the Bayes-adaptive construction
of \cref{prop:harsanyi-bamdp} applies, which in practice means a recurrent or
context-conditioned policy trained across a distribution of partners, that is,
the zero-shot coordination protocol.

The procedure is easier to trust once it has been run on concrete deployments,
so we run it on four.

\begin{example}[Four deployments, four answers]
\label{ex:case-studies}
A fleet of warehouse robots executing a picking plan. The reward is shared, the
state is fully observed by a central scheduler, and the scheduler issues joint
actions at run time. By \cref{prop:sg-mdp} this is a single-agent Markov
decision process over joint actions; the only obstacle is the size of the joint
action space, and the remedy is the coordination graph of
\cref{prop:coordination-graph}, since a robot interacts only with the robots in
its aisle. No multi-agent algorithm is required and using one would be solving a
harder problem.

The same fleet with an unreliable network, so that at run time each robot sees
only its own sensors. The reward is still shared, execution is now
decentralised, and the problem is a decentralised partially observable Markov
decision process. The obstacle is equilibrium selection, the optimum is still
well defined, and the branch above applies: factorisation if the task is
IGM-representable, a centralised critic plus an explicit selection device if it
is not, and cross-play across independently trained runs as the diagnostic,
since \cref{ex:convention} is exactly the failure mode that a mean-return plot
hides.

A bidding agent in a repeated advertising auction. The reward is general-sum,
the auction rules are fixed and public, the agent acts alone, and the other
bidders are neither fixed nor controllable. The relevant questions are whether
the auction is truthful, in which case \cref{ex:second-price} says the agent's
problem collapses to estimating its own value and the strategic content
disappears, and whether the game is smooth, in which case
\cref{thm:smoothness} bounds the welfare of whatever no-regret method is run.
The target to state is a coarse correlated equilibrium of the empirical game,
and the diagnostic is the maximum constant-deviation gain under empirical play.

A pair of language-model agents negotiating a contract. The reward is
general-sum, each agent has private information, messages are cheap talk, and
one party may have commitment power through a system prompt that the other can
read. The applicable results are \cref{thm:crawford-sobel}, which bounds how
informative the messages can be as a function of the misalignment;
\cref{ex:rubinstein}, which says the protocol and the discounting determine the
split before any learning happens; and \cref{sec:stackelberg}, which says that
whichever party can credibly commit obtains at least its equilibrium payoff. The
evaluation should be an empirical game over prompt-defined strategies in the
sense of \citet{wellman2006methods}, not an average return, and the open
difficulty is that the strategy space is not characterised, which is the ninth
problem in \cref{sec:open}. \exend
\end{example}

\begingroup\small
\begin{longtable}{@{}L{2.9cm} L{3.9cm} L{3.4cm} C{0.8cm} L{2.0cm}@{}}
\caption{Training recommendations by model class. The \textsc{Target} column
states what the procedure converges to when it converges, and the
\textsc{Guarantee} column records the strength of that statement: T for a
theorem under stated conditions, A for an asymptotic or average-iterate
statement, H for a heuristic with no convergence guarantee.}
\label{tab:training}\\[2pt]
\toprule
\textsc{Model} & \textsc{Procedure} & \textsc{Target} & & \textsc{Key reference} \\
\midrule
\endfirsthead
\multicolumn{5}{@{}l}{\small\textsc{(\cref{tab:training}, continued)}}\\[2pt]
\toprule
\textsc{Model} & \textsc{Procedure} & \textsc{Target} & & \textsc{Key reference} \\
\midrule
\endhead
\midrule
\multicolumn{5}{@{}r@{}}{\small\textsc{continued on the next page}}\\
\endfoot
\bottomrule
\endlastfoot
MDP & Tabular methods under Robbins-Monro conditions & Optimal policy & T & \citet{sutton2018reinforcement} \\
MMDP, central execution & Single-agent method on joint actions & Optimal joint policy & T & \citet{boutilier1996planning} \\
MMDP, local interaction & Coordination graph; variable elimination & Exact joint argmax & T & \citet{guestrin2001multiagent,kok2006collaborative} \\
Dec-POMDP, IGM-representable & Value factorisation (VDN, QMIX, QPLEX) & Decentralised greedy play is jointly greedy (\cref{prop:monotone-igm}) & T & \citet{rashid2018qmix} \\
Dec-POMDP, general & Centralised critic, decentralised actors & Local optimum; a Nash profile & H & \citet{lowe2017multi,yu2022surprising} \\
Dec-POMDP, exact & Occupancy-MDP search & Optimal joint policy & T & \citet{dibangoye2016optimally} \\
Dec-POMDP, public-belief & Prescriptions on the public belief & Approximate; no optimality guarantee & H & \citet{foerster2019bayesian} \\
Dec-POMDP, group common knowledge & Hierarchical controller over commonly known groups & Approximate & H & \citet{dewitt2019multi} \\
Dec-POMDP, penalty structure & Optimistic, lenient or hysteretic independent learners & Payoff-dominant equilibrium & A & \citet{matignon2007hysteretic,panait2006lenience} \\
2p zero-sum, perfect info & Self-play with search and learned value & Strong play; minimax value in the tabular limit & H & \citet{silver2018general} \\
2p zero-sum, imperfect info & CFR, MCCFR, Deep CFR & Nash equilibrium (average) & A & \citet{zinkevich2007regret,brown2019deep} \\
2p zero-sum, imperfect info & Neural fictitious self-play & Nash equilibrium (average) & A & \citet{heinrich2016deep} \\
2p zero-sum, imperfect info & Regularised Nash dynamics & Nash equilibrium (last iterate) & T & \citet{perolat2021poincare} \\
2p zero-sum, imperfect info & Public-belief-state search & Nash equilibrium & T & \citet{brown2020rebel,schmid2023student} \\
Zero-sum polymatrix, any $\nag$ & No-regret learning or LP & Nash equilibrium; convex set & T & \citet{cai2011minmax} \\
Markov potential game & Independent policy gradient & Nash equilibrium & T & \citet{leonardos2022global} \\
Supermodular game & Best-response iteration from an extreme point & Least or greatest equilibrium & T & \citet{topkis1979equilibrium,milgrom1990rationalizability} \\
Monotone game & Extragradient or optimistic gradient & Unique Nash, last iterate & T & \citet{rosen1965existence,rakhlin2013optimization} \\
Smooth game & Any no-regret dynamics & Welfare within $\lambda/(1-\mu)$ of optimal & T & \citet{roughgarden2015intrinsic} \\
General-sum, stateless & Regret matching & Correlated equilibrium & T & \citet{hart2000simple} \\
General-sum, stateless & Any no-external-regret dynamics & Coarse correlated equilibrium & T & \cref{def:ce-cce} \\
General-sum Markov game & V-learning and relatives & Non-stationary history-dependent \CE/\CCE; shared randomness & T & \citet{jin2024vlearning} \\
General-sum, commitment & Bilevel or implicit-gradient methods & Stackelberg equilibrium & T & \citet{fiez2020implicit} \\
General-sum, shaping & Opponent-learning awareness & Shaped non-equilibrium outcome & H & \citet{foerster2018learning} \\
General-sum, population & PSRO, double oracle, league play & Meta-solver-dependent & A & \citet{lanctot2017unified} \\
Designer controls rewards & Mechanism design; VCG-style transfers & Welfare optimum in dominant strategies & T & \citet{vickrey1961counterspeculation,groves1973incentives} \\
Mean-field game & Single-agent solver with damped population update & Mean-field equilibrium & T & \citet{perolat2022scaling} \\
Unknown fixed partners & Context-conditioned policy over a partner prior & Approximate Bayes-optimal response & H & \citet{zintgraf2020varibad,hu2020other} \\
\end{longtable}
\endgroup

A separate word is owed to evaluation, because the correspondence has sharper
consequences there than it has for training. Return against the training
opponents is not a measure of quality in any game that is not common-payoff,
since it is a measure of the pair and not of the policy. The appropriate
measures follow from the solution concept being targeted. For a zero-sum game
the measure is exploitability, the gap between the value of the policy against a
best response and the value of the game, computable exactly in small games and
approximable by training a best-response agent in large ones; its multi-agent
generalisation is the sum of per-agent improvement from unilateral best
response, often called NashConv. For a general-sum game targeting a coarse
correlated equilibrium the measure is the maximum per-agent gain from a constant
deviation under the empirical joint play, which is the finite-sample version of
\cref{def:ce-cce}, and if the target is a correlated equilibrium the measure is
the maximum swap gain instead, which \cref{ex:cce-strict} shows is a strictly
stronger requirement. For cooperative games with decentralised execution the
measure is cross-play, the return obtained when agents trained in different runs
are paired, which detects the conventions that self-play silently installs and
that \citet{hu2020other} and \citet{hu2021off} are designed to prevent. For
populations the measure is the response graph and the $\alpha$-rank
distribution of \citet{omidshafiei2019alpha}, whose reading is the subject of
\cref{ex:alpharank}. For a learned communication channel the measures are
positive signalling and positive listening \citep{lowe2019pitfalls}, since a
return improvement alone does not distinguish an informative protocol from a
babbling one. Reporting a learning curve of
self-play return is, for every model class except the common-payoff one with
centralised execution, uninformative about the quantity of interest.

\begin{table}[t]
\centering
\caption{Evaluation measures implied by the target solution concept. The last
column indicates the failure that the measure detects and that mean return does
not.}
\label{tab:diagnostics}
\small
\begin{tabular}{@{}L{3.0cm} L{4.6cm} L{5.0cm}@{}}
\toprule
\textsc{Target} & \textsc{Measure} & \textsc{Failure detected} \\
\midrule
Minimax value & Exploitability; gap to best response & Cycling; overfitting to the self-play opponent \\
Nash equilibrium, $\nag \ge 2$ & NashConv, the summed unilateral gain & Convergence to a non-equilibrium rest point \\
Coarse correlated equilibrium & Maximum constant-deviation gain under empirical play & Regret accumulating despite flat return \\
Correlated equilibrium & Maximum swap-deviation gain & Perfectly correlated play that no constant deviation beats (\cref{ex:cce-strict}) \\
Optimal joint policy & Gap to a centralised planner on the same model & Equilibrium selection failure; unrepresentable $Q_{\mathrm{tot}}$ \\
Decentralised convention & Cross-play return across independent runs & Hidden handshakes and arbitrary conventions \\
Robustness to partners & Return against a held-out partner population & Overfitting to co-trained partners \\
Population quality & Response graph; $\alpha$-rank; cross-play matrix & Intransitivity masked by average return \\
Learned protocol & Positive signalling and positive listening & A channel that is present but unused \\
Welfare under selfish learning & Measured welfare against a planner solution & Price of anarchy mistaken for an optimisation failure \\
Mean-field equilibrium & Exploitability of the representative agent & Non-convergence of the population flow \\
\bottomrule
\end{tabular}
\end{table}

\section{Correspondences that fail, and the theorems that forbid them}
\label{sec:negative}

A dictionary is only trustworthy if it marks the words that have no translation.
This section collects the negative results, each of which forbids a
correspondence that is natural to conjecture and that appears in the literature
in informal form.

The two readings agree here, which is why this section is short on vocabulary
and long on consequences. A game theorist reads each result as a limit on what
rationality can deliver; a reinforcement learning researcher reads it as a limit
on what an algorithm can deliver. The statements are the same and the
consequence for practice is the same: a system whose specification requires one
of these correspondences to hold is misspecified, and the remedy is to change
the specification rather than the optimiser.

\begin{nomen}{Failure modes read from both sides.}{tab:nom-negative}
Non-stationary environment & Others are learning; the kernel moves &
  \cref{prop:fixed-others-mdp} loses its hypothesis \\
Uncoupled dynamics & Each agent sees only its own payoff &
  The only realistic information structure \\
Cyclic equilibrium & A cycle of policies, not a point &
  What value iteration reaches in general sum \\
Last-iterate divergence & Final checkpoint is not the output &
  Only the average converges \\
Representational failure & Hypothesis class excludes the truth &
  More data does not help \\
Equilibrium multiplicity & Many answers, no selection rule &
  ``Converged'' is not a performance claim \\
Risk versus payoff dominance & Noise selects the safe equilibrium &
  Not the best one \\
Price of anarchy & Equilibrium welfare below the planner's &
  Not an optimisation failure \\
Aggregation impossibility & No rule satisfies every axiom &
  Some desideratum must be dropped \\
Communication bound & Misalignment caps informativeness &
  A wider channel does not help \\
\end{nomen}

The first failure is the independent-learner reduction.
\Cref{prop:fixed-others-mdp} is exact under the hypothesis that the other agents
are stationary, and in simultaneous learning the hypothesis is false by
construction. The consequences are not confined to a loss of guarantees. The
Bellman operator that agent $i$ is implicitly iterating has a fixed point that
moves as the others learn, so convergence of the iterates is not expected even
in the limit of infinite data; experience stored in a replay buffer was
generated under a joint policy that no longer exists, so off-policy correction
is required with respect to the others' policies and not only one's own; and the
value function being learned is an average over a non-stationary distribution of
opponent behaviour, which is not the value of any fixed game. The response of
centralising the critic addresses the third issue and not the first two, because
a critic conditioned on the joint action is stationary with respect to a fixed
joint policy and the joint policy is still moving.

The second failure concerns the possibility of a general convergent learning
rule. A learning dynamics is uncoupled when each agent's update depends on its
own payoff function and on the observed play of the others, but not on the
others' payoff functions, which is the only informational setting in which a
general-purpose algorithm makes sense.

\begin{theorem}[Impossibility of uncoupled convergence; \citealp{hart2003uncoupled}]
\label{thm:hart-mas-colell}
There is no uncoupled, stationary, continuous-time dynamic
$\dot x^i = F^i(x; \util_i)$ that converges to the Nash equilibrium of every
game in a family containing a neighbourhood of a given game, even when every
game in the family has a unique Nash equilibrium and the dynamic is required to
be $C^1$ with a hyperbolic asymptotically stable rest point.
\end{theorem}

\Cref{thm:hart-mas-colell} is a genuine impossibility rather than a statement
about particular algorithms, and it fixes the ceiling for a specific class of
rules: no stationary, memoryless, deterministic uncoupled dynamic converges to
Nash equilibria in general. The discrete-time picture is finer and should not be
summarised as the same statement. No uncoupled stationary rule with one-period
recall converges almost surely to pure Nash equilibria in all games that possess
them, and no uncoupled finite-recall rule makes the behaviour probabilities
converge to Nash equilibria in every game; two-period recall together with
randomisation does, however, suffice for almost sure convergence of play
\citep{hart2006stochastic}. Relaxing stationarity or memorylessness therefore
does admit convergent uncoupled rules, and the impossibility should be read as
delimiting which relaxations are necessary. It is the reason the positive
results of
\cref{sec:training} are all conditional on structure, zero-sum, potential,
common-payoff, mean-field, and the reason that the unconditional results are all
about coarse correlated equilibria, which uncoupled no-regret dynamics do reach.
The contrast is instructive: the obstruction is informational, since the agents
collectively possess enough information to compute an equilibrium and no agent
individually does.

The third failure is specific to the sequential case and removes the natural
dynamic-programming approach.

\begin{theorem}[Cyclic equilibria; \citealp{zinkevich2005cyclic}]
\label{thm:cyclic}
There exist general-sum stochastic games, already deterministic turn-taking ones
with two agents and two states, in which the stationary equilibrium policy is
not a function of the $Q$-values: two such games share the same $Q$-function and
yet have different unique stationary equilibria with different values. No
algorithm that derives a stationary policy from $Q$-values alone can therefore
compute a stationary Nash equilibrium of both. A cyclic equilibrium is a finite
cycle of stationary policies satisfying the correlated-equilibrium condition at
each stage, of which a stationary correlated equilibrium is the case of cycle
length one; value iteration with a utilitarian correlated-equilibrium selection
rule converges to a cyclic equilibrium on these games.
\end{theorem}

Together with the \PPAD-hardness of stationary coarse correlated equilibria in
Markov games \citep{daskalakis2023complexity}, \cref{thm:cyclic} says that the
whole of the dynamic programming apparatus, which is the foundation of
single-agent reinforcement learning, has no general-sum sequential counterpart.
What survives is the apparatus restricted to the structured classes, and the
non-stationary correlated targets, which are reached by online learning rather
than by dynamic programming.

The fourth failure concerns the iterate rather than the average. In a two-player
zero-sum game with an interior equilibrium, the continuous-time replicator and
multiplicative-weights dynamics conserve a Kullback-Leibler quantity and
therefore cycle on invariant curves rather than converging, and the discrete-time
version diverges from the equilibrium at a rate controlled by the step size
\citep{mertikopoulos2018cycles,bailey2018multiplicative}. The time average
converges, which is why every classical guarantee in this area is stated for the
average strategy. The methodological consequence has been stated already and
bears repeating because it is routinely violated: in a zero-sum game trained by
self-play without explicit averaging, the final checkpoint is not the output of
the algorithm and reporting it conflates the iterate with the average. The
measured version of this is \cref{ex:mwu-rps}, in which the divergence to the
equilibrium rises from $0.036$ to $1.202$ over two thousand steps while the
return stays pinned at zero, and \cref{ex:fictitious-play}, in which the
iterate has the maximum possible exploitability at every step while the average
converges. The remedy is either explicit averaging, as in fictitious play and counterfactual
regret minimisation, or a regulariser that makes the dynamics contractive, as in
regularised Nash dynamics and magnetic mirror descent
\citep{perolat2021poincare,sokota2023unified}, or a look-ahead correction, as in
extragradient and optimistic updates (\cref{ex:extragradient}).

The fifth failure is representational and belongs to the cooperative case.
Monotone value factorisation satisfies the Individual-Global-Max condition of
\cref{def:igm} and the converse fails, so there are common-payoff stage games
whose action-value function no monotone mixing network can represent, the
smallest being the two-by-two matrix of \cref{ex:igm-counterexample}, which an
additive decomposition misfits by $8$ in every cell.
The practical form of the failure is that an algorithm in this family can be
driven to a representable value function that is far from the true one, with the
learned greedy joint action suboptimal, and no amount of data corrects it
because the error is in the hypothesis class. Diagnosing it requires comparing
against a centralised planner on the same model, which is the fifth row of
\cref{tab:diagnostics}.

The sixth failure is conceptual and is the one with the largest practical
footprint. Nash equilibrium is a poor learning target in general-sum games,
for reasons that are independent of its computational hardness. It is not
unique, and the folk theorem of \cref{thm:folk} makes the non-uniqueness
extreme in repeated settings; it is not Pareto optimal, as the prisoner's
dilemma shows; it is not a prediction of behaviour, since it presupposes
coordination on a particular profile that the dynamics do not supply; and it is
not what a designer wants, which is typically welfare, fairness or robustness.
The critique of \citet{shoham2007if} remains unanswered, and the sound response
is to state the target explicitly, as \cref{tab:training} does, and to evaluate
against that target rather than against equilibrium language used decoratively.

The seventh failure is that reward shaping does not transfer its guarantee
across agents without care. Potential-based shaping preserves optimal policies
in a Markov decision process and preserves Nash equilibria in a stochastic game
when each agent's shaping term is potential-based in that agent's own
information \citep{ng1999policy,devlin2011theoretical}. Shaping terms that
depend on other agents' actions do not enjoy the invariance, and a shaping
scheme designed to encourage cooperation generically changes the equilibrium
set, which is sometimes the intent and is rarely acknowledged. The distinction
is between accelerating the search for a solution of the given game and
replacing the game with another one, and only the first is what shaping means in
the single-agent case; \cref{ex:shaping} exhibits both.

The eighth failure is that no-external-regret learning does not deliver a
correlated equilibrium, only a coarse one, and the two differ as soon as some
agent has three actions. \Cref{ex:cce-strict} exhibits a distribution in
rock-paper-scissors under which both agents are perfectly correlated, neither
has any external regret, and a rewiring of each agent's own action as a function
of what it was about to play would gain $1/3$ per round. A practitioner who
monitors external regret alone will not see this, and the diagnostic that does
see it is the maximum swap gain. The common claim that self-play converges to a
correlated equilibrium is therefore false as stated; it converges to the coarse
set unless the rule is specifically a no-swap-regret one.

The ninth failure is that noisy independent learning does not select the best
equilibrium. \Cref{ex:log-linear} computes the stationary distribution of
log-linear learning on the stag hunt exactly and finds that as the noise
vanishes the chain concentrates with probability $0.99995$ on the equilibrium
worth $3$ rather than the one worth $4$. The selected equilibrium is the
risk-dominant one and equivalently the potential maximiser, and the selection is
a property of the dynamics rather than of the objective. No increase in sample
size, network capacity or training time changes it, and the only remedies are
the equilibrium selection devices of \cref{sec:partial-obs}.

The tenth failure is that enlarging the agents' options can lower the
equilibrium welfare. \Cref{ex:braess} exhibits a network in which adding a free
link raises the equilibrium cost by a factor of $4/3$ while leaving the social
optimum unchanged. An ablation in which a richer action space produces worse
outcomes is therefore not necessarily an optimisation failure, and
distinguishing the two requires the planner comparison of
\cref{tab:diagnostics}. The same caution applies to the price of anarchy
generally: a gap between independent learners and a centralised planner is the
expected equilibrium behaviour of the system, and closing it requires changing
the rewards, as \cref{ex:tolls} does, rather than improving the optimiser.

The eleventh failure is that a wider communication channel does not buy more
communication. \Cref{thm:crawford-sobel} bounds the number of distinguishable
messages that can be sustained in equilibrium by the misalignment of the
objectives, giving at most three intervals at a bias of $0.05$ and only babbling
at a bias of $0.25$, irrespective of the channel's capacity. An architecture
that provides a high-bandwidth differentiable channel to agents with partly
conflicting rewards is therefore providing capacity that equilibrium forbids
using, and the coarse protocol that emerges is the correct answer rather than an
optimisation artefact.

The twelfth failure is that no rule aggregates preferences without cost. By
\cref{thm:arrow-gs} any scheme that combines several agents' or annotators'
rankings into one must abandon one of the five conditions listed there, and any
deterministic voting rule without transfers must abandon one of surjectivity,
non-dictatorship and strategy-proofness. Reward models fitted to multiple
annotators are such schemes, so the question is not whether they have a defect
but which one they have, and the honest report names it.

A final observation ties these together. Almost every failure in this section is
a failure of a reduction that would have made the multi-agent problem
single-agent, and the remainder are failures of the hope that a multi-agent
system can be given everything a designer wants at once. That is not a
coincidence. The reductions of
\cref{sec:reductions} that do succeed all pay the same price, an exponential
enlargement of the action space to a space of functions, and the complexity
results of \cref{sec:complexity} show that the price is necessary. The negative
results are the statement that there is no cheaper currency.

\section{Open problems at the interface}
\label{sec:open}

The correspondences catalogued here delimit what is known, and the boundary is
sharp enough in several places to name concrete problems. We state thirteen, in
decreasing order of how well posed they are.

The first is a learning-time version of \cref{thm:occupancy}. The occupancy
state is computed by propagating the model forward, so the occupancy Markov
decision process is a planning construction. A model-free surrogate that is
provably a sufficient statistic, learned from interaction alone and with a
bound relating the loss in the surrogate to the loss in value, would convert the
single exact reduction available for cooperative decentralised problems into a
deep learning method with a guarantee. The same question for the public belief
state of \cref{thm:pbs} is the obstacle to extending ReBeL-style search beyond
games whose public information structure is given by construction.

The second is last-iterate convergence outside the zero-sum and potential
classes. Regularised Nash dynamics obtains it in two-player zero-sum games by
anchoring, and the analysis uses the monotonicity of the zero-sum operator. No
general-sum analogue is known, and the \PPAD-hardness of stationary coarse
correlated equilibria \citep{daskalakis2023complexity} suggests that the correct
statement involves a weaker target. Identifying the largest class of
general-sum Markov games admitting a last-iterate-convergent uncoupled dynamics
is open, and the Markov potential class is almost certainly not maximal; the
monotone class of \cref{def:monotone} and the supermodular class of
\cref{def:supermodular} are the two obvious candidates for an extension, and
neither has a sequential counterpart with the strength of
\citet{leonardos2022global}.

The third is a solution concept for general-sum sequential problems that is
simultaneously attainable by uncoupled learning, executable without shared
randomness and meaningful as a design target. Coarse correlated equilibrium
satisfies the first, fails the second in its non-stationary form, and satisfies
the third only weakly. Stationary product-form targets fail the first by
\cref{thm:hart-mas-colell}. Whether some intermediate object, for instance a
correlated equilibrium implementable by publicly observable history rather than
by an external mediator, threads the three requirements is unresolved and is
the single most consequential open question for applications in which several
independently deployed systems must interoperate.

The fourth concerns the three-player and many-player zero-sum case. Constant-sum
games with more than two agents lose the interchangeability and payoff-uniqueness
of equilibria, so the guarantees underlying self-play collapse even though the
total payoff is still conserved. The pairwise-separable case is solved by
\citet{cai2011minmax} and the unrestricted case is \PPAD-complete, so the
question is what lies between: which structural restrictions short of
polymatrix separability preserve tractability. Practical systems in this regime
rely on heuristics, and a principled target comparable to the minimax value is
missing. The $\alpha$-rank concept of \citet{omidshafiei2019alpha} is the most
developed candidate and its relation to the payoffs of a deployed system is not
understood.

The fifth is a unified theory of regret decomposition and temporal-difference
bootstrapping. \Cref{sec:extensive} observed that counterfactual regret
decomposition and the performance difference lemma express the same
factorisation of a global value gap into local terms weighted by a visitation
measure. The two algorithm families built on them, tree-traversing regret
minimisers and bootstrapping value learners, have disjoint analyses and
complementary strengths, the former with equilibrium guarantees and poor
scaling in depth, the latter with the reverse. A common analysis would be of
immediate practical consequence and the ingredients appear to be available.

The sixth is finite-population correction for mean-field methods. The
$O(1/\sqrt{\nag})$ bound is pessimistic in practice, and \cref{ex:mfg-rate}
records that at the ten to twenty-five agents of most benchmarks it delivers no
guarantee at all. The regimes in which mean-field policies perform well at that
scale are not characterised. A refined
bound exploiting the structure of the interaction, rather than only its
Lipschitz constant, would turn a limit theorem into a usable guarantee.

The seventh is partially observable Markov potential games. The potential
structure that yields convergence of independent policy gradient is defined for
fully observed stochastic games, and the cooperative problems that motivate it
are typically partially observed. Whether the potential property survives under
partial observability, and under what condition on the observation kernel, is
not settled, and a positive answer would supply the first convergence guarantee
for independent learners in decentralised partially observable problems.

The eighth is the status of commitment in learned systems. The Stackelberg
reading of \cref{sec:stackelberg} presumes that commitment is credible, and in a
deployed learning system credibility is a property of the training pipeline
rather than of the game: a leader that can be retrained is not committed. The
mechanism design literature addresses credibility through contracts and
reputation, and the translation of those devices into the setting in which the
leader is a learning algorithm is unexplored.

The ninth is the strategic analysis of language-model agents. Systems that
negotiate, delegate and form agreements in natural language instantiate games
with cheap talk, incomplete information about types, and commitment devices
realised through text rather than through a formal mechanism. The
game-theoretic objects are the ones in this monograph, with types given by
system prompts and private context, and the empirical methodology of
\cref{sec:evolutionary} applies directly, since one can estimate an empirical
payoff matrix over prompt-defined strategies and analyse it. What is missing is
any characterisation of the strategy space, and therefore any sense in which an
equilibrium computed over a sampled subset is informative about the whole.

The tenth is a smoothness calculus for sequential games. \Cref{thm:smoothness}
converts a welfare bound at pure equilibria into a welfare bound for the
time-averaged play of no-regret learners, and it is stated for one-shot games.
A Markov-game version, in which a per-state smoothness condition composes into a
bound on the discounted welfare of the empirical occupancy of independent
learners, would give the first practically checkable welfare guarantee for
general-sum multi-agent reinforcement learning. Whether the composition is
possible at all, and what it costs in the constant, is open.

The eleventh is the equilibrium selection problem stated as an optimisation
problem. \Cref{ex:log-linear} shows that noisy independent learning selects the
risk-dominant equilibrium and that the selection is computable in closed form
for a two-by-two game. For larger games the stochastically stable set is defined
by a minimum-spanning-tree computation over the profile space, which is
exponential. An approximation that is computable from samples and that predicts
which convention a given learning rule will install would turn equilibrium
selection from a post-hoc observation into a design variable.

The twelfth is credit assignment with a guarantee. \Cref{sec:coalitions}
supplies axiomatic credit schemes and \cref{sec:potential} supplies
variance-reduction baselines, and no result relates the two: there is no theorem
saying that using Shapley-value credit instead of a difference reward improves a
sample-complexity or convergence bound. Since the Shapley value costs
$2^{\nag-1}$ evaluations against one, the absence of such a theorem means the
extra cost is currently justified only by the axioms.

The thirteenth is the interaction between aggregation impossibility and
preference-based training. \Cref{thm:arrow-gs} says that every aggregation of
several annotators' preferences abandons one classical desideratum, and the
practice of fitting a single scalar reward model abandons independence of
irrelevant alternatives. Which axiom one should prefer to abandon, and what the
downstream consequence is for the policy trained against the aggregate, is not
studied. This is a well-posed question with an immediate application and no
literature connecting the two sides.

\section{Conclusions}
\label{sec:conclusions}

The central claim of this monograph is that the relationship between game
theory and reinforcement learning is one of identity over a large part of its
extent, of exact reduction over a further substantial part, and of analogy only
in a small and precisely delimitable remainder. Stating which of the three
applies in each case is the whole of the contribution, and the cases divide as
follows.

The identities are the stochastic game and the Markov decision process, which
coincide at $\nag = 1$; the partially observable stochastic game and the
partially observable Markov decision process, which coincide at $\nag = 1$; the
common-payoff partially observable stochastic game and the decentralised
partially observable Markov decision process, which coincide by definition; the
information set of an extensive-form game with perfect recall and the
action-observation history of a partially observable problem; best response and
greedy policy improvement; subgame perfection and the principle of optimality;
Markov perfect equilibrium and the restriction to stationary policies; the
behavioural strategy and the policy; the graphical game and the coordination
graph; the strictly dominated strategy and the action that no mask need
preserve; and the Nash equilibrium and the fixed
point of simultaneous policy improvement, which collapses to the fixed point of
the Bellman optimality operator when there is one agent.

The exact reductions are the belief-state reduction of a partially observable
problem to a continuous-state Markov decision process; the common-information
reduction of a decentralised control problem with partial history sharing to a
coordinator's partially observable problem over prescriptions; the
occupancy-state reduction of a decentralised partially observable Markov
decision process to a deterministic continuous-state Markov decision process
with a piecewise linear convex value function; the public-belief-state
reduction of a two-player zero-sum imperfect-information game to a
perfect-information game, which \cref{thm:pbs} records as a remark because the
source states it informally; the Harsanyi transformation and its identity with the
Bayes-adaptive construction; the sequence form and the occupancy measure; the
minimax Bellman operator and the robust Bellman operator; the quantal response
equilibrium and the soft Bellman operator of maximum-entropy reinforcement
learning; the difference reward
and a single term of the Shapley value; the VCG transfer and the counterfactual
baseline; Blackwell approachability and no-regret learning; the replicator
dynamics as the
continuous-time limit of cross-learning; and the replicator-mutator dynamics as
the limit of frequency-adjusted Boltzmann Q-learning. Every one
of these pays for exactness with an enlargement of the state space to a space of
distributions and of the action space to a space of functions, and the
complexity results of \cref{sec:complexity} establish that the price is not an
artefact of the constructions.

The analogies that fail, and the conditions under which they fail, are the
reduction of a multi-agent problem to a single-agent one by treating the others
as part of the environment, which is exact for fixed others and invalid under
simultaneous learning; the reduction of a common-payoff game to a joint-action
Markov decision process, which is exact for centralised execution and silent
about decentralised execution; the extension of dynamic programming to
general-sum stochastic games, which is forbidden by the cyclic-equilibrium
result; the identification of what no-external-regret self-play reaches with a
correlated equilibrium, which holds only in the coarse sense and fails as soon
as an agent has three actions; the expectation that noisy independent learners
reach the payoff-dominant equilibrium, which the stochastic-stability
computation contradicts; and the pursuit of Nash equilibrium as a general
learning target, which is forbidden by the impossibility of uncoupled
convergence and is in any case not what a designer wants.

Three consequences for practice follow and we restate them because they are
actionable. The first is that the choice of training procedure should be made
by answering the five structural questions of \cref{sec:training} about the
deployment, not by selecting from a list of algorithms by reported benchmark
performance, because the guarantees in \cref{tab:training} are conditional on
exactly those structural properties and benchmark performance is not. The
second is that evaluation must be matched to the target solution concept, as
\cref{tab:diagnostics} sets out, because in every model class except the
common-payoff one with centralised execution the mean return against the
training opponents measures a pair rather than a policy. The third is that
before any algorithm is chosen the deployment game should be checked for the
three structural properties that carry guarantees, namely potential or
supermodular structure, pairwise-separable zero-sum structure, and smoothness,
because each converts an intractable problem into a tractable one and each is
cheap to verify.

The dictionary is not complete, and the open problems of \cref{sec:open} mark
where it stops. The three that matter most, in our assessment, are a model-free
sufficient statistic that would make the occupancy-state reduction a learning
method rather than a planning method, a solution concept for general-sum
sequential problems that is simultaneously learnable, executable without shared
randomness and meaningful as a design target, and a sequential version of the
smoothness calculus that would give a checkable welfare guarantee for
independent learners in Markov games. Progress on any of them
would move a substantial region of the map from the analogy column into the
exact column, which is the direction in which this field advances.

\clearpage
\bibliography{refs}

\end{document}